\documentclass[11pt]{report}

\usepackage{graphicx}
\usepackage{color}
\usepackage{amsmath}
\usepackage{verbatim}
\usepackage[numbers]{natbib}
\usepackage{tikz}
\usetikzlibrary{calc}
\usepackage{enumitem}
\usepackage{fancyhdr,fancybox}
\usepackage{varioref,float,amssymb,amsmath}
\usepackage{float,amsfonts,amsthm}
\usepackage{changebar,longtable}
\usepackage{makeidx}
\usepackage[ruled, vlined,linesnumbered,resetcount,algochapter]{algorithm2e}
\usepackage[below]{placeins}
\usepackage{flafter}
\usepackage[para,online,flushleft]{threeparttable}
\usepackage{booktabs}
\usepackage{hyperref}
\usepackage{mathdots}
\usepackage{verbatim}
\usepackage{appendix}

\usepackage{xy}
\xyoption{matrix}
\xyoption{frame}
\xyoption{arrow}
\xyoption{arc}
\usepackage{ifpdf}
\ifpdf
\else
\PackageWarningNoLine{Qcircuit}{Qcircuit is loading in Postscript mode.  The Xy-pic options ps and dvips will be loaded.  If you wish to use other Postscript drivers for Xy-pic, you must modify the code in Qcircuit.tex}
\xyoption{ps}
\xyoption{dvips}
\fi
\entrymodifiers={!C\entrybox}
\newcommand{\bra}[1]{{\left\langle{#1}\right\vert}}
\newcommand{\ket}[1]{{\left\vert{#1}\right\rangle}}
\newcommand{\qw}[1][-1]{\ar @{-} [0,#1]}
\newcommand{\qwx}[1][-1]{\ar @{-} [#1,0]}
\newcommand{\cw}[1][-1]{\ar @{=} [0,#1]}

\newcommand{\gate}[1]{*+<.6em>{#1} \POS ="i","i"+UR;"i"+UL **\dir{-};"i"+DL **\dir{-};"i"+DR **\dir{-};"i"+UR **\dir{-},"i" \qw}
\newcommand{\meter}{*=<1.8em,1.4em>{\xy ="j","j"-<.778em,.322em>;{"j"+<.778em,-.322em> \ellipse ur,_{}},"j"-<0em,.4em>;p+<.5em,.9em> **\dir{-},"j"+<2.2em,2.2em>*{},"j"-<2.2em,2.2em>*{} \endxy} \POS ="i","i"+UR;"i"+UL **\dir{-};"i"+DL **\dir{-};"i"+DR **\dir{-};"i"+UR **\dir{-},"i" \qw}

\newcommand{\control}{*!<0em,.025em>-=-<.2em>{\bullet}}
\newcommand{\controlo}{*+<.01em>{\xy -<.095em>*\xycircle<.19em>{} \endxy}}
\newcommand{\ctrl}[1]{\control \qwx[#1] \qw}
\newcommand{\ctrlo}[1]{\controlo \qwx[#1] \qw}
\newcommand{\targ}{*+<.02em,.02em>{\xy ="i","i"-<.39em,0em>;"i"+<.39em,0em> **\dir{-}, "i"-<0em,.39em>;"i"+<0em,.39em> **\dir{-},"i"*\xycircle<.4em>{} \endxy} \qw}
\newcommand{\qswap}{*=<0em>{\times} \qw}
\newcommand{\multigate}[2]{*+<1em,.9em>{\hphantom{#2}} \POS [0,0]="i",[0,0].[#1,0]="e",!C *{#2},"e"+UR;"e"+UL **\dir{-};"e"+DL **\dir{-};"e"+DR **\dir{-};"e"+UR **\dir{-},"i" \qw}
\newcommand{\ghost}[1]{*+<1em,.9em>{\hphantom{#1}} \qw}

\newcommand{\gategroup}[6]{\POS"#1,#2"."#3,#2"."#1,#4"."#3,#4"!C*+<#5>\frm{#6}}
\newcommand{\rstick}[1]{*!L!<-.5em,0em>=<0em>{#1}}
\newcommand{\lstick}[1]{*!R!<.5em,0em>=<0em>{#1}}
\newcommand{\ustick}[1]{*!D!<0em,-.5em>=<0em>{#1}}

\newcommand{\Qcircuit}{\xymatrix @*=<0em>}

\newcommand{\pureghost}[1]{*+<1em,.9em>{\hphantom{#1}}}

\newcommand{\braket}[2]{\left\langle #1 \big| #2 \right\rangle}

\newcommand{\e}{\textrm{e}}
\newcommand{\ii}{\textrm{i}}

\theoremstyle{definition}
\newtheorem{exercise}{\bf Exercise}[chapter]
\newtheorem{theorem}{Theorem}[chapter]
\newtheorem{lemma}[theorem]{Lemma}
\newtheorem{proposition}[theorem]{Proposition}
\newtheorem{corollary}[theorem]{Corollary}
\newtheorem{definition}[theorem]{Definition}

\newtheorem{conjecture}[theorem]{Conjecture}

\title{\Huge\bf Basic Quantum Algorithms\\
\vspace{2cm}
\mbox{}
}
\author{\bf Renato Portugal\\
\\
{\small Full Researcher at the}\\
{\small National Laboratory of Scientific Computing}\\
{\small LNCC/MCTI}
\\
\vspace{4cm}\\
}

\begin{document}

\maketitle

\begin{abstract}
Quantum computing is evolving so rapidly that it forces us to revisit, rewrite, and update the foundations of the theory. \emph{Basic Quantum Algorithms} revisits the earliest quantum algorithms. The journey began in 1985 with Deutsch attempting to evaluate a function at two domain points simultaneously. Then, in 1992, Deutsch and Jozsa created a quantum algorithm that determines whether a Boolean function is constant or balanced. The following year, Bernstein and Vazirani realized that essentially the same algorithm could be used to identify a specific Boolean function within a set of linear Boolean functions. In 1994, Simon introduced a novel quantum algorithm that determines whether a function is one-to-one or two-to-one exponentially faster than any classical algorithm for the same problem. That same year, Shor developed two groundbreaking quantum algorithms for integer factoring and calculating discrete logarithms, posing a threat to widely used cryptographic methods. In 1995, Kitaev proposed an alternative formulation based on phase estimation that proved valuable in numerous applications. The following year, Grover devised a quantum search algorithm that is quadratically faster than its classical counterpart. More than a decade later, Harrow, Hassidim, and Lloyd proposed a quantum algorithm for solving systems of linear equations, now known as the HHL algorithm. With an emphasis on the circuit model, this work provides a detailed description of all these remarkable algorithms.

\end{abstract}

\addtocontents{toc}{\protect\thispagestyle{empty}}
\tableofcontents
\thispagestyle{empty}
\newpage
\setcounter{page}{1}

\chapter{Introduction}\label{chap:Intro}

Quantum algorithms are a rapidly evolving subarea of quantum computing, not only in terms of new algorithms but also in applications and implementations. The basic algorithms serve as the pillars of this new edifice. The construction began with a change in the rules of the game. Instead of storing information in bits, which take either zero or one, we are allowed to store information in qubits, the state of which is a superposition of zeros and ones. The rules based on classical mechanics were replaced by rules based on quantum mechanics.

The first breakthrough came with Deutsch's 1985 proposal to evaluate a one-bit Boolean function at two points simultaneously using quantum parallelism, which exploits the superposition of zeros and ones. At the time, a framework for creating new algorithms was missing, which Deutsch established in 1989 with the introduction of quantum gates and circuits, taking the place of well-known classical gates such as AND, OR, and NOT.

In 1992, Deutsch and Jozsa developed an algorithm to determine whether a Boolean function is balanced or constant, giving momentum to the field of quantum algorithms and inspiring the development of oracle-based algorithms. The goal is to find a hidden property of a function with as few queries as possible.

Bernstein and Vazirani observed in 1993 that the Deutsch--Jozsa algorithm could be used to identify a specific Boolean function within a set of linear Boolean functions. The Bernstein--Vazirani algorithm outperforms its classical counterpart without exploiting entanglement, relying solely on quantum parallelism.

The momentum continued to grow as Simon published a quantum algorithm in 1994 that exponentially outperformed classical algorithms in determining whether a function is one-to-one or two-to-one. This algorithm exploited entanglement and had applications in finding hidden subgroups within specific classes of groups.

In the same year, Shor developed two groundbreaking quantum algorithms for factoring composite integers and calculating discrete logarithms, which posed a significant threat to the cryptographic methods widely used today. Shor's algorithms brought quantum computing to the spotlight, and since then the field has been growing at an astonishing rate. Shor's algorithm can also be formulated as an oracle-based algorithm with a function that is periodic. The goal is to find the period by evaluating the function as few times as possible. Finding periods is a task well suited for the Fourier transform, which in classical computation has complexity $O(N \log N)$, where $N$ represents the data size. In the quantum domain, however, the Fourier transform can be implemented using $O(\log^2 N)$ universal gates, and it is the quantum superposition that makes this possible.

In 1995, Kitaev introduced another formulation of Shor's algorithms after developing a quantum algorithm for phase estimation. Given a unitary operator and one of its eigenvectors, the algorithm efficiently finds the corresponding eigenvalue, which is completely characterized by its phase. Kitaev's algorithm proved useful for other applications, such as quantum counting.

Grover focused on unsorted databases in 1996 and developed a quantum algorithm that can locate an item quadratically faster than classical searching. Grover's algorithm can also be formulated as an oracle-based algorithm with a Boolean function that is constant except for a single point in the domain. The goal is to find that point by evaluating the function as few times as possible. When written as a black box algorithm, it becomes clear that Grover's algorithm has wide applicability.

More recently, Harrow, Hassidim, and Lloyd proposed a quantum algorithm for solving systems of linear equations, now known as the HHL algorithm. Given a matrix $A$ and a vector $\vec{b}$, the algorithm prepares a quantum state proportional to the solution $\vec{x}$ of the linear system $A \vec{x} = \vec{b}$. Under certain conditions on the matrix $A$, the algorithm can achieve an exponential speedup over the best known classical algorithms for this task.

\emph{Basic Quantum Algorithms} details the remarkable contributions mentioned above. There is no hope of describing these algorithms properly without using the correct language: mathematics, and more specifically linear algebra. Concepts such as quantum superposition and entanglement acquire precise meaning when expressed in this language. Measurements are described by projectors, gates by unitary operators, and qubits by vectors. Projectors, unitary operators, and vectors are the words of the language of linear algebra. When describing quantum algorithms or anything related to quantum computing, it is better to rely on mathematics; otherwise someone will probably utter nonsense.

\emph{Basic Quantum Algorithms} follows as much as possible the historical ordering, which also corresponds roughly to an order of increasing complexity. We feel as if we are climbing steps of increasing height, strengthening our muscles and preparing for the challenge of understanding more complex quantum algorithms. Each chapter is designed to be as independent as possible, allowing readers already familiar with some algorithms to skip certain sections.

Lastly, do not hesitate to contact the author (\verb|portugal@lncc.br|) if there are errors or problems in terms of imprecision or missing citations. Suggestions are also welcome.

\subsection*{Acknowledgments}

The author thanks P. H. G. Lug\~ao, G. A. Bezerra, and G. A. Bridi for useful discussions.

\chapter{Quantum Circuits}\label{chap:2}


The goal of this Chapter is to define the concepts of qubit, logic gate, and quantum circuit. Before that, we briefly review key facts of linear algebra~\cite{SA97,GS88} using Dirac notation from the beginning. References for this Section are~\cite{RP11,YM08,Buch24}. Additional references for quantum mechanics and linear algebra for quantum computing are Sections~2.1 and~2.2 of~\cite{NC00}.

\section{Review of linear algebra using Dirac notation}

There are several notations to show that a variable $v$ is a vector, for example, $\vec{v}$, $\textbf{v}$, and so on. In quantum computing, the most common notation is Dirac's: $\ket{v}$. A sequence of vectors is denoted by $\ket{v_0}$, $\ket{v_1}$, $\ket{v_2}$ and so forth. It is very common to abuse this notation and denote the same sequence as $\ket{0}$, $\ket{1}$, $\ket{2}$ and so on.

The canonical basis of a two-dimensional vector space has two vectors, denoted by $\{\ket{0},\ket{1}\}$ in Dirac notation, where $\ket{0}$ and $\ket{1}$ have the following representation
\[
\ket{0}=\begin{bmatrix} 1 \\ 0 \end{bmatrix} \,\,\,\text{ and }\,\,\,
\ket{1}=\begin{bmatrix} 0 \\ 1 \end{bmatrix}.
\]
These vectors have two entries or components, unit length, and are orthogonal. Then, this basis is orthonormal. It is called the \textit{canonical basis} in linear algebra and the \textit{computational basis} in quantum computing. Note that $\ket{0}$ is not the null vector, but the first vector of the canonical basis. All entries of the null vector are equal to 0. In the two-dimensional case, it is
\[
\begin{bmatrix} 0 \\ 0 \end{bmatrix}
\]
without any special designation in Dirac notation.

A generic vector in a two-dimensional vector space is obtained via the \textit{linear combination} of the basis vectors,
\[
\ket{\psi}=\alpha\ket{0}+\beta\ket{1},
\]
where $\alpha$ and $\beta$ are complex numbers. These numbers are the entries of vector $\ket{\psi}$, as can be seen from the notation
\begin{equation*}\label{eq:psi}
\ket{\psi}=\begin{bmatrix} \alpha \\ \beta \end{bmatrix}.
\end{equation*}
The \textit{dual vector} (with respect to $\ket{\psi}$) is denoted by $\bra{\psi}$ and is obtained by transposing $\ket{\psi}$ and conjugating each entry. Using the previous equation, we obtain
\begin{equation*}\label{eq:brapsi}
\bra{\psi}=\begin{bmatrix} \alpha^* & \beta^* \end{bmatrix},
\end{equation*}
which can be written as
\[
\bra{\psi}=\alpha^*\bra{0}+\beta^*\bra{1},
\]
where
\[
\bra{0}=\begin{bmatrix} 1& 0 \end{bmatrix} \,\,\,\text{ and }\,\,\,
\bra{1}=\begin{bmatrix} 0 & 1 \end{bmatrix}.
\]
The dual vector $\bra{\psi}$ is a $1\times 2$ matrix and vector $\ket{\psi}$ is a $2\times 1$ matrix. At this point, we introduce the \textit{dagger} symbol, denoted by $\dagger$, which is the notation for the conjugate transpose vector (transpose the vector and then conjugate each entry or vice versa). Then, we may write $\bra{\psi}=\ket{\psi}^\dagger$ and $\ket{\psi}=\bra{\psi}^\dagger$. Applying the dagger twice gives back the original vector.

Suppose that $\ket{\psi_1}$ and $\ket{\psi_2}$ are two-dimensional vectors given by
\[
\ket{\psi_1}=\begin{bmatrix} \alpha \\ \beta \end{bmatrix} \,\,\,\,\text{ and }\,\,\,\,
\ket{\psi_2}=\begin{bmatrix} \gamma \\ \delta \end{bmatrix}.
\]
The \textit{inner product} of two vectors $\ket{\psi_1}$ and $\ket{\psi_2}$ is a complex number denoted by $\braket{\psi_1}{\psi_2}$ and defined as the matrix product of the dual vector $\bra{\psi_1}$ by $\ket{\psi_2}$, as follows
\begin{equation*}
\braket{\psi_1}{\psi_2}=\begin{bmatrix} \alpha^* & \beta^* \end{bmatrix}\begin{bmatrix} \gamma \\ \delta \end{bmatrix}=\alpha^*\gamma+\beta^*\delta.
\end{equation*}
In Dirac notation, the calculation of the inner product is performed by distributing the matrix product over the sum of vectors, as follows
\begin{equation*}
\braket{\psi_1}{\psi_2}=\big(\alpha^*\bra{0}+\beta^*\bra{1}\big)\cdot \big(\gamma\ket{0}+\delta\ket{1}\big) =\alpha^*\gamma\,\braket{0}{0}+\beta^*\delta\,\braket{1}{1}=\alpha^*\gamma+\beta^*\delta.
\end{equation*}
The \textit{norm} of vector $\ket{\psi_1}$ is denoted by $\|\,\ket{\psi_1}\,\|$ and defined as
\begin{equation*}
\|\,\ket{\psi_1}\,\|=\sqrt{\braket{\psi_1}{\psi_1}}=\sqrt{|\alpha|^2+|\beta|^2},
\end{equation*}
where $|\alpha|$ is the \textit{absolute value} of $\alpha$, that is
\[
|\alpha|=\sqrt{\alpha\cdot \alpha^*}.
\]
If $\alpha=a+b\,\ii $, where $\ii$ is the imaginary unit $(\ii =\sqrt{-1})$, $a$ is the real part and $b$ is the imaginary part, then
\[|\alpha|=\sqrt{(a+b\,\ii )\cdot (a-b\,\ii )}=\sqrt{a^2+b^2}.
\]
A complex number $\alpha$ such that $|\alpha|=1$ is called a \textit{unit complex number} and can be written as $\e^{\ii\theta}=\cos\theta+\ii\sin\theta$, where $\theta$ is an angle.
In real vector spaces, the inner product is called \textit{scalar product} and is given by
\begin{equation*}
\braket{\psi_1}{\psi_2}=\|\,\ket{\psi_1}\|\,\,\|\,\ket{\psi_2}\|\,\cos\theta,
\end{equation*}
where $\theta$ is the angle between vectors $\ket{\psi_1}$ and $\ket{\psi_2}$.

Using these definitions, we can show that the basis ${\ket{0}, \ket{1}}$ is orthonormal, meaning that the vectors $\ket{0}$ and $\ket{1}$ are orthogonal and each has a norm of 1, that is
\begin{eqnarray*}
\braket{0}{0} = 1,\,\,\,\,\,\,\,\,
\braket{0}{1} = 0,\,\,\,\,\,\,\,\,
\braket{1}{0} = 0,\,\,\,\,\,\,\,\,
\braket{1}{1} = 1.
\end{eqnarray*}
An algebraic way of denoting orthonormality and of compacting the last four equations into one is
\[
\braket{k}{\ell}=\delta_{k\ell},
\]
where $k$ and $\ell$ are bits ($k,\ell \in {0,1}$) and $\delta_{k\ell}$ is the \textit{Kronecker delta}, defined as
\[
\delta_{k\ell}=\begin{cases} 1, \text{ if }k=\ell, \\ 0, \text{ if }k\neq \ell. \end{cases}
\]

\begin{exercise}
Let $\theta \in \mathbb{R}$ and define
\[
\ket{\psi} = \frac{1}{\sqrt{3}}\ket{0} + \sqrt{\frac{2}{3}}\,\e^{\ii\theta}\ket{1}.
\]
\begin{enumerate}
\item[(a)] Write $\ket{\psi}$ in column vector form.
\item[(b)] Compute $\bra{\psi}$.

\item[(c)] Show explicitly that $\|\ket{\psi}\| = 1$.

\item[(d)] Compute $\braket{0}{\psi}$ and $\braket{1}{\psi}$.

\item[(e)] For which values of $\theta$ does $\ket{\psi}$ become orthogonal to
\[
\ket{\phi}=\sqrt{\frac{2}{3}}\ket{0}+\frac{1}{\sqrt{3}}\ket{1}?
\]
\end{enumerate}
\end{exercise}

\section{Qubit and superposition}

The basic memory unit of a classical computer is the \textit{bit}, which takes on the values 0 or 1. Usually, the bit is implemented using two distinct voltages, following the convention that null or low voltage represents bit 0 and high voltage represents bit 1. To determine whether the output is bit 0 or 1 at the end of the computation, it is necessary to measure the voltage.

The basic memory unit of a quantum computer is the \textit{qubit}, which also yields 0 or 1 at the end of the computation. The qubit can be implemented using an electric current in a small superconductor, following the convention that clockwise current represents 0 and counter-clockwise current represents 1, or vice versa. The difference from the classical device appears during the computation, since the qubit allows the simultaneous coexistence of 0 and 1. During the computation, or before the measurement, the \textit{state} of a qubit is represented by a norm-1 two-dimensional vector and the states of a qubit corresponding to 0 and 1 are $\ket{0}$ and $\ket{1}$. The definition of \textit{state} is a vector of norm 1 in a complex vector space endowed with the inner product presented in the previous Section.\footnote{A finite-dimensional vector space with an inner product is a \textit{Hilbert space}.} The state can be thought of as the ``value'' of the qubit before the measurement. Quantum coexistence is represented mathematically by a linear combination of orthonormal vectors as follows
\[
\ket{\psi}=\alpha\ket{0}+\beta\ket{1},
\]
where $\alpha$ and $\beta$ are complex numbers that obey the constraint
\[
|\alpha|^2+|\beta|^2=1.
\]
The state of the qubit is the vector $\ket{\psi}$, which has norm 1 and entries $\alpha$ and $\beta$. The complex numbers $\alpha$ and $\beta$ are the \textit{amplitudes} of the state $\ket{\psi}$.

The coexistence of bits 0 and 1 cannot be implemented in a classical device, since it is not possible to have low and high voltage simultaneously, as everyone knows. In quantum mechanics, though hard to believe, it is possible to have a quantum system (usually microscopic) in a \textit{superposition} of the states corresponding to low and high voltage. This coexistence cannot be understood within our classical intuition; it is something fundamentally new that becomes clear only through the mathematical formalism. This superposition can only be fully maintained if the quantum system is sufficiently isolated from the surrounding macroscopic environment. When we measure the quantum system to determine the voltage value, the measuring device inevitably interacts with the system, producing a stochastic result, which is either low or high voltage, similar to the classical bit. In other words, superposition is maintained only as long as no measurement interaction reveals which outcome would be obtained.

Note that quantum mechanics is a \textit{scientific theory}, meaning its laws and results can be tested objectively in laboratories. In addition, unnecessary laws and statements are promptly discarded. Therefore, the statement that superposition requires isolation has practical consequences and has been tested and re-tested for over 100 years in thousands of quantum mechanics laboratories worldwide. On the other hand, alternative theories that attempt to reproduce classical intuition without superposition have been ruled out by experimental tests.

From a computational point of view, we have a qubit in superposition and we use this feature in a circuit. For example, the circuit
\[
\Qcircuit @C=2.3em @R=1.9em {
\lstick{\ket{\psi}}       &\qw & \meter & \rstick{
0 \textrm{ or }
1}\cw
}\vspace*{3pt}
\]
tells us that the initial ``value'' of the qubit is $\ket{\psi}$ and this information is conveyed unchanged from left to right until a measurement is performed, as shown by the \textit{meter} (the display of a voltmeter). The measurement outputs 0 or 1. Classical information is conveyed by a double wire. If the state of the qubit is $\ket{\psi}=\ket{0}$, a measurement will necessarily output 0 and if the state is $\ket{1}$, a measurement will necessarily output 1. In the general case, if the state is $\alpha\ket{0}+\beta\ket{1}$, a measurement will return 0 with probability $|\alpha|^2$ or 1 with probability $|\beta|^2$, as shown in the circuit \vspace{0.3cm}
\[
\Qcircuit @C=2.3em @R=1.9em {
\lstick{\alpha\ket{0}+\beta\ket{1}}       &\qw & \meter & \rstick{\begin{cases} 0, \text{ with probability }|\alpha|^2, \\ 1, \text{ with probability }|\beta|^2. \end{cases}}\cw
}\hspace{2.5cm}\vspace*{7pt}
\]
The output can be depicted by a histogram of the probability distribution.
It is important to repeat the fact that $\alpha$ and $\beta$ are called \textit{amplitudes} of the state $\alpha\ket{0}+\beta\ket{1}$ and are numbers that can be negative and may have an imaginary part. On the other hand, $|\alpha|^2$ and $|\beta|^2$ are positive real numbers in the interval $[0,1]$ and are called probabilities. A careless interchange between amplitudes and probabilities creates unforgivable errors.

\begin{figure}[!ht]
\centering
\includegraphics[trim=0 0 0 0,clip,scale=0.9]{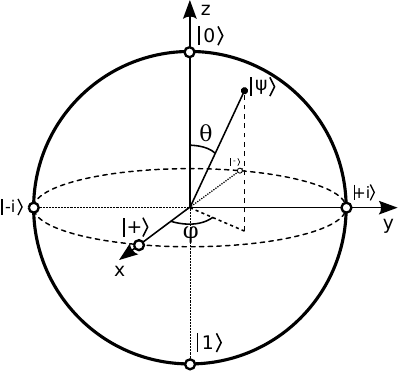}
\caption{Bloch sphere and the location of states $\ket{0}$, $\ket{1}$, $\ket{\pm}$, and $\ket{\pm \ii}$. An arbitrary state $\ket{\psi}$ is shown with spherical angles $\theta$ and $\varphi$.}
\label{fig:esfera_Bloch}
\end{figure}

The state of a qubit can be characterized by two angles $\theta$ and $\varphi$ as follows
\begin{equation*}
    \ket{\psi}=\cos\frac{\theta}{2}\,\ket{0} +  \textrm{e}^{\ii \varphi}\sin\frac{\theta}{2}\, \ket{1},
\end{equation*}
where $0\le \theta \le \pi$ and $0\le \varphi < 2\pi$. This notation shows that there is a one-to-one correspondence between the physically distinct states of a qubit and points on the surface of a sphere of radius 1, called the \textit{Bloch sphere}. The angles $\theta$ and $\varphi$ are spherical angles that describe the location of the state $\ket{\psi}$, as shown in Fig.~\ref{fig:esfera_Bloch}. A point on the Bloch sphere is described by a three-dimensional vector with real entries
\begin{equation*}
    \begin{bmatrix}
  \sin \theta \cos\varphi \\
  \sin \theta \sin\varphi \\
  \cos \theta \\
\end{bmatrix}.
\end{equation*}
The locations of the states
\[
\ket{\pm}=\frac{\ket{0}\pm\ket{1}}{\sqrt{2}},
\qquad
\ket{\pm \ii}=\frac{\ket{0}\pm \ii\ket{1}}{\sqrt{2}}
\]
correspond to the spherical angles
\[
\begin{aligned}
\ket{+} &: (\theta,\varphi)=(\pi/2,0), \\
\ket{-} &: (\theta,\varphi)=(\pi/2,\pi), \\
\ket{+\ii} &: (\theta,\varphi)=(\pi/2,\pi/2), \\
\ket{-\ii} &: (\theta,\varphi)=(\pi/2,3\pi/2).
\end{aligned}
\]
Therefore, $\ket{\pm}$ lie on the $x$-axis and $\ket{\pm \ii}$ lie on the $y$-axis of the Bloch sphere.

If we have an arbitrary single-qubit state $\alpha\ket{0}+\beta\ket{1}$ and we want to find the spherical angles $\theta$ and $\varphi$, the first thing to do is to write $\alpha$ and $\beta$ as $r_1\e^{\ii\varphi_1}$ and $r_2\e^{\ii\varphi_2}$, respectively, where $r_1=|\alpha|$ and $r_2=|\beta|$. Now we multiply the state by $\e^{-\ii\varphi_1}$ to obtain $r_1\ket{0}+\e^{\ii(\varphi_2-\varphi_1)}r_2\ket{1}$. Then, we take $\varphi=\varphi_2-\varphi_1$ and $\theta=2\arccos r_1$. Note that $r_2=\sin(\theta/2)$ because $r_1^2+r_2^2=1$. There is no problem in multiplying the state by a unit complex number such as $\e^{-\ii\varphi_1}$ because in quantum mechanics two quantum states that differ by a global factor are considered equivalent and have the same location on the Bloch sphere. The global factor must be a unit complex number and is usually called \textit{global phase factor}.

\begin{exercise}
Let
\[
\ket{\psi}=\cos\frac{\theta}{2}\ket{0}
+\e^{\ii\varphi}\sin\frac{\theta}{2}\ket{1},
\qquad
0\le \theta \le \pi, \; 0\le \varphi < 2\pi.
\]

\begin{enumerate}
\item[(a)] Show that $\|\ket{\psi}\|=1$.

\item[(b)] Compute the measurement probabilities of obtaining $0$ and $1$ in the computational basis.

\item[(c)] Show that
\[
\ket{\psi'}=\e^{\ii\gamma}\ket{\psi}
\]
has the same measurement probabilities as $\ket{\psi}$ for any real $\gamma$.
\item[(d)] Find the Bloch sphere angles $(\theta,\varphi)$ corresponding to the state
\[
\ket{\phi}=\frac{\sqrt{3}}{2}\ket{0}+\frac{1}{2}\,\e^{\ii\pi/4}\ket{1}.
\]
\end{enumerate}
\end{exercise}

\section{Single-qubit gates}

A single-qubit gate is a $2\times 2$ unitary matrix. A matrix $U$ is unitary if $U^\dagger U = I$. Equivalently, unitary operators preserve inner products and therefore preserve norms, that is, $\| U\ket{\psi} \| = \| \ket{\psi} \|$ for all $\ket{\psi}$. Formally, suppose that $\ket{\psi'}=U\ket{\psi}$, where $\ket{\psi}$ is a norm-1 two-dimensional vector. If $U$ is a unitary matrix, then $\ket{\psi'}$ will have norm 1. For example, the Hadamard matrix
$$
H=\frac{1}{\sqrt 2}\begin{bmatrix}
    1 & \,\,\,1 \\
    1 & -1 \\
\end{bmatrix}
$$
is unitary. Therefore, the multiplication of $H$ by the basis vectors has to result in norm-1 vectors. In fact,
\[
H\ket{0}\,\,=\,\,H\begin{bmatrix} 1 \\ 0 \end{bmatrix}\,\,=\,\,\frac{1}{\sqrt 2}\begin{bmatrix} 1 \\ 1 \end{bmatrix}\,\,=\,\,\frac{1}{\sqrt 2}\ket{0}+\frac{1}{\sqrt 2}\ket{1}.
\]
We denote this vector by $\ket{+}$, that is
\[\ket{+}=\frac{1}{\sqrt 2}\ket{0}+\frac{1}{\sqrt 2}\ket{1}.
\]
Multiplying $H$ by $\ket{1}$ yields vector $\ket{-}$ defined as
\[
\ket{-}=\frac{1}{\sqrt 2}\ket{0}-\frac{1}{\sqrt 2}\ket{1},
\]
which also has norm 1. These calculations are important because we need to learn the output of the gate. If the input is $\ket{0}$ then the output is $\ket{+}$. If the input is $\ket{1}$, the output is $\ket{-}$. If the input is a superposition $\alpha\ket{0}+\beta\ket{1}$, the output is the superposition of $\ket{+}$ and $\ket{-}$ with the same amplitudes, $\alpha\ket{+}+\beta\ket{-}$, because we use the linearity property of the gate, that is, instead of thinking that $H$ is a matrix, we use that $H$ is a linear operator and if $H$ is applied to a linear combination of vectors $\ket{0}$ and $\ket{1}$ with amplitudes $\alpha$ and $\beta$, the result is a linear combination of $H\ket{0}$ and $H\ket{1}$ with the same amplitudes $\alpha$ and $\beta$. While we could avoid this abstract perspective, when multiplying a matrix by a sum of vectors, we must distribute the multiplication over the sum of vectors.

Verifying that $H$ maps the vectors of an orthonormal basis to norm-1 vectors is not sufficient to prove that $H$ is unitary. It is also necessary to check that the resulting vectors remain orthogonal, that is, to verify that $\braket{-}{+}=0$. A more direct way to establish that $H$ is unitary is to compute $H H^\dagger$, where $H^\dagger$ is obtained by transposing $H$ and conjugating each entry. If $H H^\dagger = I$, then $H$ is unitary. The matrix $H^\dagger$ is called the \textit{Hermitian transpose} of $H$.

A \textit{quantum circuit} is a graphical representation of a quantum algorithm. The input qubit is located on the left, and the information (qubit's state) is transmitted unaltered from left to right until it encounters a logic gate. The gate receives the input from the left, acts on the qubit's state, and the resulting state is then passed to the right. The gate processing is accomplished by multiplying the unitary matrix, which represents the gate, by the vector that represents the qubit's state. For example, the expression
 $\ket{+}=H\ket{0}$ is represented by the following circuit:
\[
\Qcircuit @C=2.3em @R=1.9em {
\lstick{\ket{0}}        &\gate{H}  &  \rstick{\ket{+}.}\qw
}\vspace{0.1cm}
\]
The input is vector $\ket{0}$, which is conveyed unchanged by the wire to $H$, which acts on the input and transforms it into $\ket{+}$, which is then conveyed to the right. The gate action is calculated by multiplying $H$ by $\ket{0}$. Therefore, the result of the computation is $\ket{+}$. If at the end of the computation we perform a measurement, the circuit is \vspace* {0.3cm}
\[
\Qcircuit @C=2.3em @R=1.9em {
\lstick{\ket{0}}    & \gate{H} & \meter & \rstick{\begin{cases} 0, \text{ with probability }\frac{1}{2}, \vspace{0.1cm}\\ 1, \text{ with probability }\frac{1}{2}. \end{cases}}\cw
}\hspace{3cm}\vspace*{0.5cm}
\]
The circuit shows that measuring the qubit in the state $\ket{+}$ yields 0 with probability $1/2$ and 1 with probability $1/2$. Fig.~\ref{fig:histograma_H} shows the histogram of the probability distribution generated in Qiskit.\footnote{Qiskit is open-source software for running programs on IBM quantum computers.}

\begin{figure}[!ht]
\centering
\includegraphics[scale=0.4]{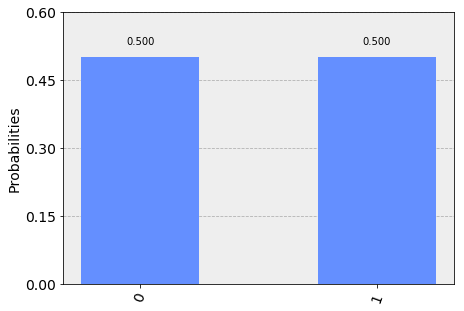}
\caption{Histogram of the probability distribution generated by measuring a qubit whose state is $\ket{+}$.}
\label{fig:histograma_H}
\end{figure}

An example that is simpler than the previous one is the $X$ gate, defined as
$$X=\begin{bmatrix}
    0 & 1 \\
    1 & 0 \\
\end{bmatrix}.
$$ $X$ is the quantum NOT gate because $\ket{1}=X\ket{0}$ and $\ket{0}=X\ket{1}$. We can verify these equations by multiplying the matrix $X$ by $\ket{0}$ and $\ket{1}$. In a more compact form, we can write $\ket{j\oplus 1}=X\ket{j}$, where $\oplus$ is the XOR operation or sum modulo 2. Because of this, the gate $X$ is also represented as $\oplus$. A circuit using the $X$ gate is
\[
\Qcircuit @C=2.3em @R=1.9em {
\lstick{\ket{0}}    & \gate{X} & \meter & \rstick{1 \text{ with probability 1.}}\cw
}\hspace{3cm}\vspace*{0.2cm}
\]

Now we can increase the complexity. How can we generate a superposition such that the amplitudes of $\alpha\ket{0}+\beta\ket{1}$ are different and nonzero? For example, how to generate a state $\alpha\ket{0}+\beta\ket{1}$ such that $|\alpha|^2=25\%$ and $|\beta|^2=75\%$ before the measurement? The answer is to use the most general single-qubit gate, whose algebraic expression is
\[
U(\theta,\phi,\lambda)=
\left[\begin{array}{cc}
 \cos\frac{\theta}{2} &  -\e^{ \ii \lambda}\sin\frac{\theta}{2} \\
 \e^{ \ii \phi}\sin\frac{\theta}{2}   &  \e^{ \ii (\lambda+\phi)}\cos\frac{\theta}{2}
\end{array}\right].
\]
After applying $U(\theta,0,0)$ on $\ket{0}$, we obtain
\[
U(\theta,0,0)\ket{0}=\cos\frac{\theta}{2}\,\ket{0}+\sin\frac{\theta}{2}\,\ket{1}.
\]
We must choose $\theta = 2\pi/3$, since we require
\[
\cos^2\frac{\theta}{2} = \frac{1}{4},
\]
which implies $\frac{\theta}{2} = \frac{\pi}{3}$. The previous example using the Hadamard gate can be reproduced by taking $\theta=\pi/2$ and $\lambda=\pi$ because $H=U(\pi/2,0,\pi)$.

$U(\theta,\phi,\lambda)$ is a universal single-qubit gate because every $2\times 2$ unitary matrix can be written as a global phase factor times $U(\theta,\phi,\lambda)$. For instance, three useful gates obtained from $U$ are
\begin{align*}
R_x(\theta) &= U\Big(\theta,-\frac{\pi}{2},\frac{\pi}{2}\Big) = \left[\begin{array}{cc}
 \,\,\,\,\,\,\cos\frac{\theta}{2} &  -\ii\sin\frac{\theta}{2} \vspace{2pt}\\
 -\ii\sin\frac{\theta}{2}   &  \,\,\,\,\,\,\cos\frac{\theta}{2}
\end{array}\right],\\
R_y(\theta) &= U(\theta,0,0) = \left[\begin{array}{cc}
 \cos\frac{\theta}{2} &  -\sin\frac{\theta}{2} \vspace{2pt}\\
 \sin\frac{\theta}{2}   &  \,\,\,\,\cos\frac{\theta}{2}
\end{array}\right],\\
R_z(\lambda) &= \frac{U(0,0,\lambda)}{\e^{\ii \lambda/2}} =
\left[\begin{array}{cc}
 \e^{-\ii \lambda/2} &  0 \\
 0   &  \e^{ \ii \lambda/2}
\end{array}\right].
\end{align*}
Here $R_x$, $R_y$, and $R_z$ are the operators that rotate the Bloch sphere about the $x$-, $y$-, and $z$-axes, respectively. Although $U(\theta,\phi,\lambda)$ provides a convenient mathematical parametrization of all single-qubit unitary operators, in practical implementations one often works with a smaller set of elementary gates.

\begin{exercise}
Let $\theta$ be an angle and $U$ a matrix such that $U^2=I$. Using the Taylor expansion of the exponential function
\[
\e^{x}\,=\,\sum_{k=0}^\infty\frac{x^k}{k!},
\]
show that
\[
\text{e}^{\ii\theta U}=\cos(\theta) I + \ii\sin(\theta) U.
\]
Show that
\begin{align*}
R_x(\theta) &= \text{e}^{-\ii  \frac{\theta}{2} X}=\cos\frac{\theta}{2} \,I - \ii \sin\frac{\theta}{2} \,X,\\
R_y(\theta) &= \text{e}^{-\ii  \frac{\theta}{2} Y}=\cos\frac{\theta}{2} \,I - \ii \sin\frac{\theta}{2} \,Y,\\
R_z(\theta) &= \text{e}^{-\ii  \frac{\theta}{2} Z}=\cos\frac{\theta}{2} \,I - \ii \sin\frac{\theta}{2} \,Z.
\end{align*}
\end{exercise}

When constructing quantum circuits, the most important single-qubit gates are
\[
   I_2 \,\,=\,\, \begin{bmatrix}
                 1 & 0 \\
                 0 & 1 \\
               \end{bmatrix},\hspace{0.4cm}
    X \,\,=\,\,  \begin{bmatrix}
                 0 & 1 \\
                 1 & 0 \\
               \end{bmatrix}, \hspace{0.4cm}
   Y \,\,=\,\,  \begin{bmatrix}
                 0 & -\ii \\
                 \ii & \,\,\,\,0 \\
               \end{bmatrix}, \hspace{0.4cm}
   Z \,\,=\,\,  \begin{bmatrix}
                 1 & \,\,\,\,0 \\
                 0 & -1 \\
               \end{bmatrix}
\]
known as \textit{Pauli matrices},
\[
H=\frac{1}{\sqrt 2}\begin{bmatrix}
    1 & \,\,\,\,1 \\
    1 & -1 \\
\end{bmatrix},\hspace{0.3cm}
   S= \begin{bmatrix}
                 1 & 0 \\
                 0 & \ii \\
               \end{bmatrix},\hspace{0.2cm}
    S^\dagger= \begin{bmatrix}
                 1 & \,\,\,\,0 \\
                 0 & -\ii \\
               \end{bmatrix},\hspace{0.3cm}
    T= \begin{bmatrix}
                 1 & 0 \\
                 0 & \textrm{e}^{\frac{\ii \pi}{4}} \\
               \end{bmatrix},\hspace{0.2cm}
    T^\dagger= \begin{bmatrix}
                 1 & 0 \\
                 0 & \textrm{e}^{-\frac{\ii \pi}{4}} \\
               \end{bmatrix}
\]
known as the Hadamard gate, the phase gate, its conjugate, the $\pi/8$ gate (or $T$ gate), and its conjugate.\footnote{The $T^\dagger$ gate is in fact the transpose-conjugate gate, but since $T$ is diagonal, $T^\dagger$ is simply the conjugate gate.} The complex numbers
$\textrm{e}^{\pm\frac{\ii \pi}{4}} $
are equal to
\[
\textrm{e}^{\pm\frac{\ii \pi}{4}} = \frac{1\pm \ii}{\sqrt 2}.
\]
Every quantum circuit without measurements corresponds to a unitary operator and therefore has an equivalent algebraic representation. For example, if $A$, $B$, $C$ are single-qubit gates, the circuit
\[
\Qcircuit @C=2.3em @R=1.9em {
\lstick{\ket{0}}    & \gate{A}  & \gate{B}  & \gate{C} & \rstick{\ket{\psi}} \qw
}\vspace*{0.1cm}
\]
is equivalent to the algebraic expression
\[
\ket{\psi} \,=\, C\cdot B\cdot A \cdot \ket{0},
\]
where $\,\cdot\,$ is the matrix product, which is usually omitted. The algebraic expression equivalent to the circuit has the reverse order. Therefore, the last circuit can also be written as
\[
\Qcircuit @C=2.3em @R=1.9em {
\lstick{\ket{0}}    & \gate{CBA}  & \rstick{\ket{\psi},} \qw
}\vspace*{0.1cm}
\]
where $CBA$ is a $2\times 2$ unitary matrix. For example, the following circuits are equivalent:
\[
\Qcircuit @C=2.3em @R=1.9em {
\lstick{}    & \gate{H}  & \gate{X}  & \gate{H} & \rstick{\,\,\equiv} \qw & & \gate{Z}   & \rstick{,}\qw
}\vspace*{0.1cm}
\]
because $Z=HXH$. The equivalent algebraic expression can be used to simplify the circuit and predict its output.

In quantum computing, it is convenient to use classical notation when describing computational basis states. For instance, let $\ell \in \{0,1\}$ be a classical bit. Then, $\ket{\ell}$ denotes either $\ket{0}$ or $\ket{1}$. With this notation, we can write $\ket{\ell \oplus 1} = X\ket{\ell}$, where $\oplus$ denotes addition modulo 2. Similarly, we have $(-1)^\ell \ket{\ell} = Z\ket{\ell}$ and $\ii(-1)^\ell \ket{\ell \oplus 1} = Y\ket{\ell}$. These expressions are useful for determining the action of the gates $X$, $Y$, and $Z$ on computational basis states.

For the gate $U(0,0,\lambda)$, we have $\e^{\ii \lambda \ell}\ket{\ell} = U(0,0,\lambda)\ket{\ell}$, which includes $S$, $T$, and their conjugate gates as special cases. When acting on computational basis states, these gates do not create superposition. In contrast, the Hadamard gate transforms computational basis states into superposition states. In fact,
\[
H\ket{\ell} = \frac{\ket{0} + (-1)^\ell \ket{1}}{\sqrt{2}}.
\]

\begin{exercise}
Verify that the outputs of the following circuits are correct:

\[
\Qcircuit @C=1.5em @R=0.9em {
\lstick{\ket{\ell}}  & \qw  & \gate{X} & \gate{H}
& \rstick{\dfrac{\ket{0}-(-1)^{\ell}\ket{1}}{\sqrt 2}} \qw
}
\]

\[
\Qcircuit @C=1.5em @R=0.9em {
\lstick{\ket{\ell}}  & \qw  & \gate{H} & \gate{X}
& \rstick{\dfrac{(-1)^\ell\ket{0}+\ket{1}}{\sqrt 2}} \qw
}
\]

Show also that, up to a global phase, the output of

\[
\Qcircuit @C=1.5em @R=0.9em {
\lstick{\ket{0}}  & \qw  & \gate{H} & \gate{T} & \gate{H}
& \rstick{\dfrac{\ket{0}-\ii(\sqrt{2}-1)\ket{1}}
{\sqrt{2}\sqrt{2-\sqrt{2}}}} \qw
}
\]

is correct.
\end{exercise}

\section{Quantum states and entanglement}

The state of two qubits is described by a vector of norm 1 that belongs to a four-dimensional vector space, consistent with the four possible results after measuring the qubits in the computational basis: 00, 01, 10, 11.  The first bit refers to the first qubit and the second bit to the second qubit, as is usual in textbooks on quantum computing, and the least significant bit is on the right, as usual.

Following the laws of quantum mechanics, there is a one-to-one correspondence between the canonical basis and the possible measurement outcomes as follows:
\begin{equation*}
\ket{00}  = \begin{bmatrix}
1 \\ 0 \\ 0 \\ 0
\end{bmatrix}
,\,\,\ \
\ket{01}  =\begin{bmatrix}
0 \\ 1 \\ 0 \\ 0
\end{bmatrix}
,\,\,\ \
\ket{10}  =\begin{bmatrix}
0 \\ 0 \\ 1 \\ 0
\end{bmatrix}
,\,\,\ \
\ket{11}  =\begin{bmatrix}
0 \\ 0 \\ 0 \\ 1
\end{bmatrix}
.
\end{equation*}
In the classical case, the state of two bits is either 00 or 01 or 10 or 11, exclusively. In the quantum case, the state of two qubits is the linear combination
\[
\ket{\psi}= a_0\ket{00}+a_1\ket{01}+a_2\ket{10}+a_3\ket{11},
\]
where $a_0$, $a_1$, $a_2$, and $a_3$ are complex numbers. When the state of two qubits is $\ket{\psi}$, the outcome of a measurement in the computational basis is either 00 or 01 or 10 or 11, exclusively and stochastically. In general, there is no way of predicting the measurement outcome deterministically even knowing $\ket{\psi}$, unless $\ket{\psi}$ is one of the computational basis states. However, if we know $\ket{\psi}$, then we know the probabilities of outcomes (via the Born rule), which are
\begin{align*}
&\text{prob}(00)=\left|\braket{00}{\psi}\right|^2=\left|a_0\right|^2,\,\,\,\,\\
&\text{prob}(01)=\left|\braket{01}{\psi}\right|^2=\left|a_1\right|^2,\,\,\,\,\\
&\text{prob}(10)=\left|\braket{10}{\psi}\right|^2=\left|a_2\right|^2,\,\,\,\,\\
&\text{prob}(11)=\left|\braket{11}{\psi}\right|^2=\left|a_3\right|^2.
\end{align*}
The sum of those probabilities is 1. If we do not know the state $\ket{\psi}$, a single measurement does not allow the determination of $\ket{\psi}$, that is, we cannot find the amplitudes $a_0$, $a_1$, $a_2$, and $a_3$. There is an important theorem in quantum mechanics known as the \textit{non-cloning theorem}~\cite{Die82,Par70,WZ82}.

\begin{theorem} (No cloning)
Using unitary operators, it is impossible to make an identical copy of an arbitrary unknown quantum state that is available to us.
\end{theorem}

This theorem severely restricts any possibility of determining $\ket{\psi}$ through measurements. However, if we can generate $\ket{\psi}$ again and again, for example, through a circuit, we can repeat the whole process several times and obtain an approximation for $\text{prob}(00)$, $\text{prob}(01)$, $\text{prob}(10)$, and $\text{prob}(11)$. For example, by repeating 1000 times, we can determine these probabilities with two digits. Unfortunately, we are still unable to determine $\ket{\psi}$ exactly because knowing $|a_0|^2$ doesn't allow us to determine $a_0$ exactly. It may seem that this is unimportant---false. Let's consider a critical example. Suppose that
\[
\text{prob}(00)=\frac{1}{2},\,\,\,\,\,\,\,\,\,\,  \text{prob}(01)=0,\,\,\,\,\,\,\,\,\,\, \text{prob}(10)=0,\,\,\,\,\,\,\,\,\,\,  \text{prob}(11)=\frac{1}{2}.
\]
We have at least two possibilities for $\ket{\psi}$:
\[
\ket{\psi_1}=\frac{\ket{00}+\ket{11}}{\sqrt 2}\,\,\,\,\,\,\, \text{ and }\,\,\,\,\,\,\, \ket{\psi_2}=\frac{\ket{00}-\ket{11}}{\sqrt 2}.
\]
Note that $\ket{\psi_1}$ and $\ket{\psi_2}$ are orthogonal. This shows that we can make a serious mistake. We cannot conclude that two circuits are equivalent merely because they produce the same probability distribution in the computational basis.

At this point, the following question is relevant: Suppose that we know $\ket{\psi}$, is it possible to determine the state of each qubit? The answer is ``depends on $\ket{\psi}$''. If $\ket{\psi}$ is one of the states of the computational basis then we know the state of each qubit. For example, suppose $\ket{\psi}=\ket{10}$. We have to factorize $\ket{\psi}$ as
\[
\ket{10}\,=\,\ket{1}\ket{0}\,=\,\ket{1}\otimes\ket{0},
\]
where $\otimes$ is called \textit{Kronecker product}. When factorization is successful, we know the state of each qubit. In this case, the state of the first qubit is $\ket{1}$ and the state of the second is $\ket{0}$. When we write $\ket{1}\ket{0}$, the Kronecker product is implicitly assumed.

The Kronecker product\footnote{There is a more abstract formulation of the Kronecker product called \textit{tensor product}. In this work, we use these terms interchangeably. The notation $A\otimes B$ reads ``A tensor B''; however, note that the terms \textit{tensor} and \textit{tensor product} are used in other areas of mathematics with differing meaning, such as in differential geometry.} of two vectors or two matrices is defined as follows. Let $A$ be a $m\times n$ matrix and $B$ a $p\times q$ matrix. Then,
\begin{equation*}
    A\otimes B = \begin{bmatrix}
                   a_{11} B & \cdots & a_{1n} B \\
                    & \ddots &  \\
                   a_{m1} B & \cdots & a_{mn} B \\
                 \end{bmatrix}.
\end{equation*}
The result is a $mp\times nq$ matrix. The Kronecker product of vectors $\ket{1}$ and $\ket{0}$ is calculated by viewing these vectors as $2\times 1$ matrices and is given by
\begin{equation*}
\ket{1}\otimes\ket{0}=
    \begin{bmatrix}
      0 \\
      1 \\
    \end{bmatrix}\otimes \begin{bmatrix}
                           1 \\
                           0 \\
                         \end{bmatrix}\,=\,
    \begin{bmatrix}
      0\begin{bmatrix}
                           1 \\
                           0 \\
                         \end{bmatrix} \\ \\
      1\begin{bmatrix}
                           1 \\
                           0 \\
                         \end{bmatrix}
    \end{bmatrix}  \,=\,
    \begin{bmatrix}
      0 \\
      0 \\
      1 \\
      0
    \end{bmatrix}=\ket{10}.
\end{equation*}
Note that the Kronecker product is noncommutative. For example, $\ket{1}\otimes\ket{0}\neq \ket{0}\otimes\ket{1}$. An important hint is to never change the order of the Kronecker product.

From the laws of quantum mechanics, if the state of a qubit is $\ket{\psi_1}$ and the state of a second one is $\ket{\psi_2}$, then the state $\ket{\psi}$ of the composite system of the two qubits will initially be
\[
\ket{\psi}=\ket{\psi_1}\otimes\ket{\psi_2}.
\]
We can always obtain the state of the composite system when we know the states of the parts. However, the reverse process is not possible in general. For example, suppose that the state of two qubits is
\[
\ket{\psi}=\frac{\ket{00}+\ket{11}}{\sqrt 2}.
\]
We want to find single-qubit states $\ket{\psi_1}=\alpha\ket{0}+\beta\ket{1}$ and $\ket{\psi_2}=\gamma\ket{0}+\delta\ket{1}$ such that
\[
(\alpha\ket{0}+\beta\ket{1})\otimes(\gamma\ket{0}+\delta\ket{1}) = \frac{\ket{00}+\ket{11}}{\sqrt 2}.
\]
Expanding the left-hand side, we obtain the following system of equations:
\begin{align*}
\alpha\gamma=\frac{1}{\sqrt 2},\ \
\alpha\delta=0,\ \
 \beta\gamma=0,\ \
 \beta\delta=\frac{1}{\sqrt 2}.
\end{align*}
Since this system has no solution, the two-qubit state $\ket{\psi}$ cannot be written as the Kronecker product of single-qubit states. In quantum mechanics, a composite quantum system may have a definite pure state $\ket{\psi}$ while a subsystem does not have a definite pure state.

\begin{exercise}
The Bell states are
\[
\ket{\Phi^\pm} = \frac{ \ket{00} \pm \ket{11} }{\sqrt{2}}
\,\,\,\text{   and   }\,\,\,
\ket{\Psi^\pm} = \frac{\ket{01} \pm \ket{10}}{\sqrt{2}}.
\]
Show that all of these states are entangled.
\end{exercise}

A quantum state of a composite system that cannot be factorized in terms of the Kronecker product is called an \textit{entangled state}. Entangled states are very important in quantum computing because without them the computational power of the quantum computer would be badly impaired. However, the presence of entanglement in a quantum algorithm does not guarantee that this algorithm is more efficient than its classical counterpart.

The term ``definite state'' $\ket{\psi}$ in quantum mechanics means \textit{pure state}. A state of a quantum system is called a pure state if we are 100\% sure that the system is described by a norm-1 vector $\ket{\psi}$. On the other hand, if we are not 100\% sure, that is, if we know that the state of the system is $\ket{\psi_1}$ with probability $0<p<1$ or $\ket{\psi_2}$ with probability $1-p$, then the state is \textit{mixed} and is represented by an \textit{ensemble} or a positive matrix $\rho$ such that $\mathrm{Tr}(\rho)=1$. Mixed states are often used to describe the state of a sub-system of an entangled system.

We can generalize the discussion of this Section to $n$ qubits, where $n\ge 1$. The computational basis has $2^n$ vectors, each vector with $2^n$ entries,
\begin{equation*}
\ket{0\cdots 00}  = \begin{bmatrix}
1 \\ 0 \\ \vdots \\ 0
\end{bmatrix},\,\,\ \
\ket{0 \cdots 01}  =\begin{bmatrix}
0 \\ 1 \\ \vdots \\ 0
\end{bmatrix},\,\,\ \
\cdots,\,\,\ \
\ket{1\cdots 11}  =\begin{bmatrix}
0 \\ 0 \\ \vdots \\ 1
\end{bmatrix}.
\end{equation*}
Note that the binary number inside the ket, for instance, $0\cdots 0$ in $\ket{0\cdots 0}$, has $n$ bits and the state itself is the Kronecker product of $n$ single-qubit states. The binary number $0\cdots 0$ can be written in the decimal notation as $\ket{0\cdots 0}\rightarrow\ket{0}$. Each binary number inside the kets can be written in the decimal notation as
\[
\ket{0\cdots 00}\rightarrow\ket{0}, \,\,\,\,\,\ket{0\cdots 01}\rightarrow\ket{1}, \,\,\,\,\,\ket{0\cdots 10}\rightarrow\ket{2}, \,\,\,\,\, ..., \,\,\,\,\, \ket{1\cdots 11}\rightarrow\ket{2^n-1}.
\]
A generic state $\ket{\psi}$ belongs to a $2^n$-dimensional vector space. Then,
\[
\ket{\psi}=a_0\,\ket{0}+a_1\,\ket{1}+a_2\,\ket{2}+ \cdots +a_{2^n-1}\,\ket{2^n-1},
\]
where
\[
\left|a_0\right|^2+\left|a_1\right|^2+\left|a_2\right|^2+\cdots+\left|a_{2^n-1}\right|^2=1.
\]
After a measurement of all qubits, we obtain a random $n$-bit string with the following probability distribution: The outcome is either the $n$-bit string $0\cdots 00$ with probability $\left|a_0\right|^2$, or the $n$-bit string $0\cdots 01$ with probability $\left|a_1\right|^2$, and so on. We have a \textit{sample space} comprising those $n$-bit strings and a probability distribution given by prob$(\ell)=\left|a_\ell\right|^2$, where $\ell$ is a $n$-bit string. The measurement outcome is a random variable that takes a value $\ell$ in this sample space with probability prob$(\ell)$.

As mentioned earlier, a computational basis state of $n$ qubits can be written as the Kronecker product of $n$ single-qubit basis states. For example, for $n=3$ qubits, we can obtain the second vector of the computational basis using the Kronecker product as
\begin{equation*}
\ket{0}\otimes\ket{0}\otimes\ket{1}\,\,=\,\,
    \begin{bmatrix}
      1 \\
      0 \\
    \end{bmatrix}\otimes \begin{bmatrix}
                           1 \\
                           0 \\
                         \end{bmatrix}\otimes
                         \begin{bmatrix}
                           0 \\
                           1 \\
                         \end{bmatrix}\,=\,
    \begin{bmatrix}
      0 \\
      1 \\
      0 \\
      0 \\
      0 \\
      0 \\
      0 \\
      0
    \end{bmatrix}\,\,=\,\,\ket{001}.
\end{equation*}
In the decimal notation, $\ket{0}$ can be confused with the state of 1 qubit. To avoid confusion, we have to know what is the number of qubits. For example, if $\ket{0}$ refers to the state of 3 qubits in the decimal notation, then in binary we have $\ket{000}$.

In this Section, we have defined entangled states, a key concept in quantum computing that is best understood with mathematics. Learning the basic definition of entanglement is somewhat similar to understanding prime numbers in arithmetic. After mastering addition and multiplication, we learn factorization, discovering that numbers like 17 cannot be factored into smaller integers, which leads us to the concept of prime numbers. In quantum mechanics, the analogy lies in understanding that entangled states are the \textit{irreducible} vectors of composite quantum systems in terms of the Kronecker product. Specifically, given a state $\ket{\psi}$, we aim to determine whether it can be factored into two smaller vectors, $\ket{\psi_1}$ and $\ket{\psi_2}$, such that $\ket{\psi} = \ket{\psi_1} \otimes \ket{\psi_2}$. If no such factorization is possible, the state is entangled. While a single-qubit state is not entangled (since a single qubit is not a composite system), entanglement becomes meaningful with two or more qubits. In the next Section, we will explore how to produce an entangled state using a quantum computer.

\begin{exercise}
Let
\[
\ket{\psi}=a_0\ket{00}+a_1\ket{01}+a_2\ket{10}+a_3\ket{11}
\]
be a state of two qubits.

\begin{enumerate}
\item[(a)] Prove that $\ket{\psi}$ is separable (not entangled) if and only if $a_0 a_3 = a_1 a_2$.

\item[(b)] Decide whether the state
\[
\ket{\phi}=\frac{1}{\sqrt{6}}\Big(2\ket{00}+\ket{01}+\ket{10}\Big).
\]
is entangled.
\item[(c)] Decide whether the state
\[
\ket{\phi}=\frac{1}{2}\Big(\ket{00}+ \mathrm{e}^{\ii\theta}\ket{01}+ \mathrm{e}^{\ii\varphi}\ket{10}+ \mathrm{e}^{\ii(\theta+\varphi)}\ket{11}\Big),
\]
is entangled, where $\theta,\varphi\in\mathbb{R}$.
\end{enumerate}
\end{exercise}

\section{Two-qubit quantum gates}

The most important two-qubit gate is CNOT or controlled-NOT gate, also denoted by $C(X)$ or $CX$. It is defined as
$$\textrm{CNOT}\,\ket{k}\ket{\ell}=\ket{k}X^k\ket{\ell},$$
and is represented by the circuit
\[
\Qcircuit @C=2.3em @R=1.9em {
\lstick{\ket{k}}    & \ctrl{1} &  \rstick{\ket{k}} \qw \\
\lstick{\ket{\ell}}  & \targ    &  \rstick{X^k\ket{\ell}=\ket{\ell \oplus k},}  \qw \\
}\vspace{0.2cm}
\]
where $k$ and $\ell$ are bits. The state of the first qubit (\textit{control}) doesn't change after applying CNOT. The state of the second qubit (\textit{target}) changes only if bit $k$ is 1. In this case the output is $X\ket{\ell}=\ket{\ell\oplus 1}$. If $k=0$ then $X^0=I_2$ and $I_2\ket{\ell}=\ket{\ell}$, where $I_2$ is the $2\times 2$ identity matrix. We have defined CNOT by showing its action on the vectors of the computational basis. In linear algebra, this definition is complete, because to know the action of CNOT on an arbitrary vector, which is a linear combination of vectors of the computational basis, we use the linearity of this gate. For example, in the circuit below the first input is in superposition:
\[
\Qcircuit @C=3.3em @R=1.em {
\lstick{\frac{\ket{0}+\ket{1}}{\sqrt 2}}   & \ctrl{2} &   \qw \\
  &  &  & \frac{\ket{00}+\ket{11}}{\sqrt 2}. \\
\lstick{\ket{0}}                           & \targ    &   \qw  \gategroup{1}{1}{3}{3}{.7em}{\}}
}
\]
What is the output? The best way to determine the output is via algebraic calculations. After using the distributive property of the Kronecker product over the sum of vectors, the input to the circuit is
\[
\frac{\ket{0}+\ket{1}}{\sqrt 2}\otimes \ket{0} \,\,\,=\,\,\, \frac{1}{\sqrt 2}\,\ket{0}\otimes\ket{0}+ \frac{1}{\sqrt 2}\,\ket{1}\otimes\ket{0} \,\,\,=\,\,\, \frac{\ket{00}+\ket{10}}{\sqrt 2}.
\]
To calculate the action of CNOT on a sum of vectors, we use the linearity of the matrix product, that is,
\[
\text{CNOT}\cdot\left( \frac{\ket{00}+\ket{10}}{\sqrt 2}\right)\,\,=\,\, \frac{1}{\sqrt 2}\,\text{CNOT}\cdot\ket{00}+ \frac{1}{\sqrt 2}\,\text{CNOT}\cdot\ket{10},
\]
where $\text{CNOT}\cdot\ket{00}$ denotes the multiplication of the CNOT matrix by vector $\ket{00}$. Using the definition given at the beginning of the Section, we obtain $\text{CNOT}\ket{00}=\ket{00}$ and  $\text{CNOT}\ket{10}=\ket{11}$, and we confirm that the output is
\[
\frac{\ket{00}+\ket{11}}{\sqrt 2}.
\]
Since this result is an entangled state, we cannot factorize it and therefore we cannot write the output for each qubit.

The same result is obtained if we use the matrix representation, which is
\[
\text{CNOT} =
\begin{bmatrix}
    I_2 &  \\
    & X \\
\end{bmatrix}
=
\begin{bmatrix}
                 1 & 0 & 0 & 0 \\
                 0 & 1 & 0 & 0 \\
                 0 & 0 & 0 & 1 \\
                 0 & 0 & 1 & 0
               \end{bmatrix},
\]
and the representation of $\ket{00}$ and $\ket{10}$ as 4-dimensional vectors, that is,
\[
\frac{1}{\sqrt 2}
\begin{bmatrix}
  1 & 0 & 0 & 0 \\
  0 & 1 & 0 & 0 \\
  0 & 0 & 0 & 1 \\
  0 & 0 & 1 & 0
\end{bmatrix}
\begin{bmatrix}
  1  \\
  0  \\
  0  \\
  0
\end{bmatrix}
+
\frac{1}{\sqrt 2}
\begin{bmatrix}
  1 & 0 & 0 & 0 \\
  0 & 1 & 0 & 0 \\
  0 & 0 & 0 & 1 \\
  0 & 0 & 1 & 0
\end{bmatrix}
\begin{bmatrix}
  0  \\
  0  \\
  1  \\
  0
\end{bmatrix}
=
\frac{1}{\sqrt 2}
\begin{bmatrix}
  1  \\
  0 \\
  0  \\
  1
\end{bmatrix}.
\]

The complete circuit that implements the entangled state above when the initial state of the qubits is $\ket{00}$ is
\[
\Qcircuit @C=2.3em @R=0.7em {
\lstick{\ket{0}} & \gate{H}   & \ctrl{2} &  \meter \qw & \cw \\
&  &  &  &  & &  {\substack{\text{   00 with probability 0.5,} \\ \text{   11 with probability 0.5.}}} \\
\lstick{\ket{0}} &\qw            & \targ    &  \meter \qw &\cw \gategroup{1}{1}{3}{5}{.7em}{\}}
}
\]
Since the state of the first qubit is initially $\ket{0}$, we have to use $H$ to generate $(\ket{0}+\ket{1})/\sqrt 2$. In fact, we have
\[
\text{CNOT}\cdot (H \otimes I)\,\ket{00}=\text{CNOT}\cdot\big(H\ket{0}\otimes\ket{0})=\frac{\ket{00}+\ket{11}}{\sqrt 2}.
\]

The next example is simpler than the previous one. Consider the circuit without measurements
\[
\Qcircuit @C=2.3em @R=0.7em {
\lstick{\ket{0}}    & \gate{H} &  \rstick{\ket{+}} \qw \\
\lstick{\ket{0}}  & \gate{H}  &  \rstick{\ket{+}.}  \qw \\
}\vspace{0.2cm}
\]
What is the output? Usually, to calculate the output, we convert the circuit to its equivalent algebraic expression. For this circuit, we have
$(H\otimes H)\ket{00}.$
The calculation is performed in the following way:
\[
(H\otimes H)\ket{00}\,\,=\,\,(H\otimes H)\cdot(\ket{0}\otimes\ket{0})\,\,=\,\,(H\ket{0})\otimes (H\ket{0})\,\,=\,\,\ket{+}\otimes\ket{+}.
\]
In the second equality, we use the following property of the Kronecker product:
\[
(A\otimes B)\cdot(C\otimes D)\,\,=\,\,(A\cdot C)\otimes (B\cdot D),
\]
for matrices $A$, $B$, $C$, $D$, or
\[
(A\otimes B)\cdot(\ket{\psi_1}\otimes\ket{\psi_2})\,\,=\,\,(A\ket{\psi_1})\otimes (B\ket{\psi_2}),
\]
which is valid for any matrices $A$ and $B$ and vectors $\ket{\psi_1}$ and $\ket{\psi_2}$ as long as the number of entries of the vectors is equal to the corresponding number of columns of the matrices.
So the output of the circuit is
\[
\ket{+}\otimes\ket{+}\,\,=\,\,\frac{\ket{0}+\ket{1}}{\sqrt 2} \otimes \frac{\ket{0}+\ket{1}}{\sqrt 2}\,\,=\,\,\frac{\ket{0}+\ket{1}+\ket{2}+\ket{3}}{2}.
\]

Let us consider another example. Take the following circuit without measurements:
\[
\Qcircuit @C=2.3em @R=1.0em {
\lstick{\ket{0}}    & \gate{H} &  \rstick{\ket{+}} \qw \\
\lstick{\ket{0}}  & \qw  &  \rstick{\ket{0}.}  \qw \\
}\vspace{0.2cm}
\]
How to calculate the output using the equivalent algebraic expression? The hint is to use the following equivalent circuit:
\[
\Qcircuit @C=2.3em @R=0.7em {
\lstick{\ket{0}}    & \gate{H} &  \rstick{\ket{+}} \qw \\
\lstick{\ket{0}}  & \gate{I_2}  &  \rstick{\ket{0},}  \qw \\
}\vspace{0.1cm}
\]
where $I_2$ is the two-dimensional identity matrix. The algebraic calculation is done as follows: \[
(H\otimes I_2)\ket{00}\,\,=\,\,(H\ket{0})\otimes (I_2\ket{0})\,\,=\,\,\ket{+}\otimes\ket{0}\,\,=\,\,\frac{\ket{00}+\ket{10}}{\sqrt 2}.
\]

When we convert a quantum circuit to its equivalent algebraic expression, we must use the Kronecker product for gates in the same column and the matrix product for gates in the same wire or in sequence; however, we must reverse the order of the gates in the second case. For example, the algebraic expression equivalent to the circuit
\[
\Qcircuit @C=2.3em @R=0.7em {
\lstick{\ket{0}}  & \gate{A} & \gate{B} & \multigate{2}{D} & \qw \\
\lstick{}         &          &          &                  &  & {\ket{\psi},\,\,\,\,\,\,}  \\
\lstick{\ket{0}}  & \qw      & \gate{C} & \ghost{D}        &  \qw
\gategroup{1}{1}{3}{5}{.7em}{\}}
}\vspace{0.2cm}
\]
where $A$, $B$, and $C$ are single-qubit gates and $D$ is a irreducible two-qubit gate, is
\[
\ket{\psi} \,=\, D\cdot (B\otimes C)\cdot (A\otimes I_2)\cdot \ket{00}.
\]
We can simplify this expression a little and write
\[
\ket{\psi} \,=\, D \cdot (B A\,\otimes \,C) \ket{00}.
\]
Only $D$ can create or destroy entanglement. Operators $A$, $B$, and $C$ neither create nor destroy entanglement. The proof that $A$ cannot create or destroy entanglement is as follows. Suppose that the input state is $\ket{\psi_1}\otimes \ket{\psi_2}$ (unentangled). The action of $A$ outputs $(A\ket{\psi_1})\otimes \ket{\psi_2}$, which is unentangled. Now suppose that the input is an entangled state $\ket{\psi}$. The action of $A$ outputs $(A\otimes I)\ket{\psi}$. If this state is unentangled, that is, there exist $\ket{\psi_1}$ and $\ket{\psi_2}$ such that $(A\otimes I)\ket{\psi}=\ket{\psi_1}\otimes \ket{\psi_2}$, we reach a contradiction because the last equation is equivalent to $\ket{\psi}=(A^\dagger\ket{\psi_1})\otimes \ket{\psi_2}$.

CNOT is so important that we describe a variant that is CNOT activated by 0. It is defined by
$$\ket{k}\ket{\ell}\longrightarrow\ket{k}X^{(1-k)}\ket{\ell},$$
and is represented by the circuit
\[
\Qcircuit @C=2.3em @R=1.9em {
\lstick{\ket{k}}    & \ctrlo{1} &  \rstick{\ket{k}} \qw \\
\lstick{\ket{\ell}}  & \targ    &  \rstick{X^{(1-k)}\ket{\ell}=\ket{\ell\oplus k\oplus 1}.}  \qw \\
}\vspace{0.2cm}
\]
Note that the control qubit is denoted by the empty circle indicating that the CNOT's control is inactive if the state of the control qubit is $\ket{1}$. This gate is obtained from the usual CNOT by multiplying $(X\otimes I_2)$ on both sides, as shown in the following circuit equivalence:
\[
\Qcircuit @C=1.5em @R=0.7em {
             & \gate{X}      & \ctrl{2} & \gate{X}   &          \qw   & & & \ctrlo{2} & \qw   \\
             &          &          &       & \rstick{\,\,\equiv} & &  \\
\lstick{}    & \qw & \targ & \qw  &  \qw       & & & \targ   & \rstick{.}\qw
}\vspace*{0.1cm}
\]
This gate is represented by a block matrix of the form
$$ (X\otimes I_2)\cdot CNOT \cdot (X\otimes I_2) =
\begin{bmatrix}
     & I_2 \\
    I_2 &  \\
\end{bmatrix}
\begin{bmatrix}
    I_2 &  \\
    & X \\
\end{bmatrix}
\begin{bmatrix}
     & I_2 \\
    I_2 &  \\
\end{bmatrix} =
\begin{bmatrix}
    X &  \\
    & I_2 \\
\end{bmatrix}.
$$
An alternative way to obtain the matrix representation is by listing sequentially the output of each vector of the computational basis
\begin{align*}
\ket{00}\xrightarrow{\circ\text{\textemdash}\oplus}\ket{01},\\
\ket{01}\xrightarrow{\circ\text{\textemdash}\oplus}\ket{00},\\
\ket{10}\xrightarrow{\circ\text{\textemdash}\oplus}\ket{10},\\
\ket{11}\xrightarrow{\circ\text{\textemdash}\oplus}\ket{11}.
\end{align*}
Then convert each output into the vector notation and join them side by side to form a matrix:
$$
\begin{bmatrix}
  0  \\
  1 \\
  0  \\
  0
\end{bmatrix}
\begin{bmatrix}
  1  \\
  0 \\
  0  \\
  0
\end{bmatrix}
\begin{bmatrix}
  0  \\
  0 \\
  1  \\
  0
\end{bmatrix}
\begin{bmatrix}
  0  \\
  0 \\
  0  \\
  1
\end{bmatrix}\longrightarrow
\begin{bmatrix}
  0 & 1 & 0 & 0 \\
  1 & 0 & 0 & 0 \\
  0 & 0 & 1 & 0 \\
  0 & 0 & 0 & 1
\end{bmatrix}=
\begin{bmatrix}
    X &  \\
    & I_2 \\
\end{bmatrix}.
$$

The second most important two-qubit gate is the controlled $Z$ gate, denoted by $C(Z)$ or $CZ$. Its circuit representations are
\[
\Qcircuit @C=1.5em @R=0.7em {
         & \ctrl{2} & \qw  & & & \gate{Z} & \qw  & & & \ctrl{2} & \qw   \\
         &          & \rstick{\,\,\equiv} & & &  &     \rstick{\,\,\equiv} & &  \\
\lstick{}& \gate{Z} & \qw  & & & \ctrl{-2} & \qw   & & & \ctrl{-2}   & \rstick{.}\qw
}\vspace*{0.1cm}
\]
It does not matter which qubit is the control or target, that is, $Z$ may be controlled by the first qubit, and target on the second, or the other way around. The third representation is interesting because the qubits are on an equal footing. There is only one matrix representation given by
\[
C(Z) =
\begin{bmatrix}
    I_2 &  \\
    & Z \\
\end{bmatrix}
=
\begin{bmatrix}
                 1 & 0 & 0 & 0 \\
                 0 & 1 & 0 & 0 \\
                 0 & 0 & 1 & 0 \\
                 0 & 0 & 0 & -1
               \end{bmatrix}.
\]
A useful method for proving that two circuits are equivalent is to apply both to the states in the computational basis. If you can verify that both circuits produce the same output for every input state in the computational basis, the circuits are equivalent. This is because, by the linearity of linear algebra, the circuits will also produce the same output for any arbitrary input state.

The CNOT and $C(Z)$ gates are connected as shown by the following circuit equivalence:
\[
\Qcircuit @C=1.5em @R=0.9em {
             & \qw      & \ctrl{2} & \qw   &          \qw   & & & \ctrl{2} & \qw   \\
             &          &          &       & \rstick{\,\,\equiv} & &  \\
\lstick{}    & \gate{H} & \targ & \gate{H}  &  \qw       & & & \gate{Z}   & \rstick{.}\qw
}\vspace*{0.1cm}
\]
The equivalence follows from
\[ (I\otimes H)\,\text{CNOT}\,(I\otimes H)=C(Z), \]
which can be shown using the matrix representation. There is an alternative way of showing the equivalence by using the fact that $H^2=I$ and then $I\otimes H$ can be replaced by $C(H)$ because two $H$'s act trivially if the CNOT's control is inactive. Then
\[ (I\otimes H)\,\text{CNOT}\,(I\otimes H)=C(H)C(X)C(H)=C(HXH)=C(Z), \]
because $HXH=Z$ and $C(A)C(B)=C(AB)$ for any single-qubit gates $A$ and $B$ acting on the same target and controlled by the same qubit.

Now we are ready to show that the control and target of a CNOT invert when we multiply CNOT by $H\otimes H$ on both sides. Indeed,
\[ (H\otimes H)\,\text{CNOT}\,(H\otimes H)= (H\otimes I)(I\otimes H)\,\text{CNOT}\,(I\otimes H)(H\otimes I)=(H\otimes I)C(Z)(H\otimes I). \]
Consider the first qubit as the target of $C(Z)$. Then we have $HZH=X$ controlled by the second qubit, and we are done. The circuit representation is
\[
\Qcircuit @C=1.5em @R=0.7em {
             & \gate{H}      & \ctrl{2} & \gate{H}   &          \qw   & & & \targ & \qw   \\
             &          &          &       & \rstick{\,\,\equiv} & &  \\
\lstick{}    & \gate{H} & \targ & \gate{H}  &  \qw       & & & \ctrl{-2}   & \rstick{.}\qw
}\vspace*{0.1cm}
\]

The third most important two-qubit gate is the SWAP gate. Its circuit representations are
\[
\Qcircuit @C=1.5em @R=0.9em {
         &\ctrl{2}& \targ  &\ctrl{2} & \qw           &&&\targ    &\ctrl{2}& \targ   & \qw            & & & \qswap & \qw   \\
         &         &        &          & \rstick{\,\,\equiv}&&&         &        &         & \rstick{\,\,\equiv} & & &  \qwx       &\\
\lstick{}&\targ    &\ctrl{-2}& \targ    & \qw           &&&\ctrl{-2}& \targ  &\ctrl{-2}& \qw            & & & \qswap \qwx  & \rstick{.}\qw
}
\]
The matrix representation is
\[
\text{SWAP}
=
\begin{bmatrix}
                 1 & 0 & 0 & 0 \\
                 0 & 0 & 1 & 0 \\
                 0 & 1 & 0 & 0 \\
                 0 & 0 & 0 & 1
               \end{bmatrix}.
\]
This gate inverts an unentangled state $\ket{\psi_1}\ket{\psi_2}$ into $\ket{\psi_2}\ket{\psi_1}$, which follows from the fact that
\[
\text{SWAP}\ket{00}=\ket{00},\,\,\,\,\text{SWAP}\ket{01}=\ket{10},\,\,\,\,\text{SWAP}\ket{10}=\ket{01},\,\,\,\,\text{SWAP}\ket{11}=\ket{11}.
\]

\begin{exercise}
Let the input state of two qubits be
\[
\ket{\psi}=
(\alpha\ket{0}+\beta\ket{1})
\otimes
(\gamma\ket{0}+\delta\ket{1}),
\]
where $\alpha,\beta,\gamma,\delta\in\mathbb{C}$ and the state is normalized.

\begin{enumerate}
\item[(a)] Compute explicitly $\mathrm{CNOT}\ket{\psi}$.

\item[(b)] Show that the output state is entangled if and only if
\[
\beta\gamma \neq 0.
\]

\item[(c)] Conclude that CNOT creates entanglement only when the control qubit is in superposition and the target qubit is not an eigenstate of $X$.
\end{enumerate}
\end{exercise}

\subsection*{Preparing an arbitrary two-qubit state}

An arbitrary single-qubit state is prepared up to a global phase with gates $R_y(\theta)$ and $R_z(\varphi)$, as follows
\[
R_z(\varphi)R_y(\theta)\ket{0}=\e^{-\frac{\ii \varphi}{2}}\left(\cos\frac{\theta}{2}\,\ket{0}+\e^{\ii \varphi}\sin\frac{\theta}{2}\,\ket{1} \right),
\]
where
\begin{align*}
R_y(\theta)=\left[\begin{array}{cc}
 \cos\frac{\theta}{2} &  -\sin\frac{\theta}{2} \vspace{3pt}\\
 \sin\frac{\theta}{2}   &  \,\,\,\,\cos\frac{\theta}{2}
\end{array}\right]
\end{align*}
and
\begin{align*}
R_z(\varphi)=\e^{-\frac{\ii \varphi}{2}}
\left[\begin{array}{cc}
 1 &  0 \\
 0   &  \e^{ \ii \varphi}
\end{array}\right].
\end{align*}
This method can be extended to the two-qubit case.

Supposing that the initial state of 2 qubits is $\ket{00}$, how do we prepare an arbitrary two-qubit state $\ket{\psi}=a_0\ket{00}+a_1\ket{01}+a_2\ket{10}+a_3\ket{11}$? This new state can be used as the initial state of an algorithm instead of $\ket{00}$. There no trivial circuit to achieve this task. First let us rewrite $\ket{\psi}$ as
\[
\ket{\psi}=\e^{\ii \alpha_0}|a_0|\,\ket{00}+\e^{\ii \alpha_1}|a_1|\,\ket{01}+\e^{\ii \alpha_2}|a_2|\,\ket{10}+\e^{\ii \alpha_3}|a_3|\,\ket{11}.
\]
Now each amplitude is a non-negative number times a unit complex number.

The output of the circuit
\[
\Qcircuit @C=1.3em @R=.3em {
\lstick{\ket{0}} & \gate{R_y(2\theta_1)}&\ctrlo{2}           &\ctrl{2}            & \rstick{} \qw \\
\lstick{}         &          &          &                  &  & {\,\,\,\ket{\psi}}  \\
\lstick{\ket{0}} & \qw                 &\gate{R_y(2\theta_2)}&\gate{R_y(2\theta_3)}& \rstick{}  \qw
\gategroup{1}{1}{3}{5}{.7em}{\}} }
\]
is
\[
\ket{\psi}=\cos\theta_1\cos\theta_2\,\ket{00}+\cos\theta_1\sin\theta_2\,\ket{01}+\sin\theta_1\cos\theta_3\,\ket{10}+\sin\theta_1\sin\theta_3\,\ket{11}.
\]
Then, after solving the system of equations
\begin{align*}
\cos\theta_1\cos\theta_2 &= |a_0|,\\
\cos\theta_1\sin\theta_2 &= |a_1|,\\
\sin\theta_1\cos\theta_3 &= |a_2|,\\
\sin\theta_1\sin\theta_3 &= |a_3|,
\end{align*}
we find $\theta_1$, $\theta_2$, and $\theta_3$, which must be used in the circuit as parameters of $R_y$. To obtain the solution, we need to square the equations, add iteratively two by two equations, and take the square root. Since all numbers on the right-hand side are non-negative, the process ends up with the correct solution. For instance, $\theta_1$ is given by
\[
\theta_1=\arccos\sqrt{|a_0|^2+|a_1|^2}.
\]
The state produced by the circuit is
\[
\ket{\psi}=|a_0|\,\ket{00}+|a_1|\,\ket{01}+|a_2|\,\ket{10}+|a_3|\,\ket{11}.
\]
The phases are missing.

For example, to prepare state $\frac{1}{2}\left(\ket{00}+\ket{01}+\ket{10}+\ket{11}\right)$, which has only positive amplitudes, the solution of the system of equations is $\theta_1=\theta_2=\theta_3=\pi/4$. When $\theta_2=\theta_3$, the product of the controlled gates simplifies to only one gate $R_y(2\theta_2)$ applied on the second qubit. Then,
\[
R_y\left(\frac{\pi}{2}\right)\otimes R_y\left(\frac{\pi}{2}\right)\,\ket{00}\,=\,\frac{1}{2}\left(\ket{00}+\ket{01}+\ket{10}+\ket{11}\right).
\]

To obtain the phases, we have to augment the previous circuit with the following sequence of gates
\[
\Qcircuit @C=1.3em @R=.3em {
\lstick{} & \gate{R_z(\beta_1)}&\ctrlo{2}           &\ctrl{2}            & \rstick{} \qw \\
\lstick{}         &          &          &                  &  & {\,\,\,}  \\
\lstick{} & \qw                 &\gate{R_z(\beta_2)}&\gate{R_z(\beta_3)}& \rstick{,}  \qw
}
\]
where
\begin{align*}
\beta_1&=\frac{\alpha_2+\alpha_3-\alpha_0-\alpha_1}{2},\\
\beta_2&={\alpha_1-\alpha_0},\\
\beta_3&={\alpha_3-\alpha_2}.
\end{align*}
After using those two circuits one after the other, we obtain the output $\e^{\ii \beta_0}\ket{\psi}$, where
\[
\beta_0=\frac{\alpha_0+\alpha_1+\alpha_2+\alpha_3}{4}.
\]

\begin{exercise}\label{ex:state-prep-two-qubit}
Find $\theta_2$ and $\theta_3$ in terms of $|a_0|,|a_1|,|a_2|,|a_3|$, and show that the output of the circuit
\[
\Qcircuit @C=1.3em @R=.3em {
\lstick{\ket{0}} & \gate{R_y(2\theta_1)}&\ctrlo{2}           &\ctrl{2}            & \gate{R_z(\beta_1)}&\ctrlo{2}           &\ctrl{2}            & \rstick{} \qw \\
\lstick{}         &          &          &                  &  &     &          &  &{\,\,\,}  \\
\lstick{\ket{0}} & \qw                 &\gate{R_y(2\theta_2)}&\gate{R_y(2\theta_3)}& \qw                 &\gate{R_z(\beta_2)}&\gate{R_z(\beta_3)}& \rstick{}  \qw
 }
\]
is
\[
\e^{\ii \beta_0}\left(\e^{\ii \alpha_0}|a_0|\,\ket{00}+\e^{\ii \alpha_1}|a_1|\,\ket{01}+\e^{\ii \alpha_2}|a_2|\,\ket{10}+\e^{\ii \alpha_3}|a_3|\,\ket{11}\right).
\]
\end{exercise}

\begin{exercise}\label{ex:decompRy}
The decomposition of $C(R_y(\theta))$ in terms of CNOT gates and single-qubit gates is shown in the following circuit equivalence:
\[
\Qcircuit @C=1.5em @R=0.7em {
             & \ctrl{2}&  \qw                & & &\ctrl{2}& \qw                  &\ctrl{2}& \qw& \qw&   \\
             &         & \rstick{\,\,\equiv} & &  \\
\lstick{}    & \gate{R_y(\theta)}   &  \qw                & & & \targ  &\gate{R_y(-\theta/2)} & \targ  &\gate{R_y(\theta/2)}& \rstick{.}\qw
}\vspace*{0.1cm}
\]
Show the equivalence between the circuits.
\end{exercise}

\section{Multiqubit gates} \label{subsec:porta3oumais}

The most important three-qubit gate is the Toffoli gate, denoted by CCNOT or $C^2(X)$, defined by
$$C^2(X)\,\ket{j}\ket{k}\ket{\ell}=\ket{j}\ket{k}X^{jk}\ket{\ell},$$
and represented by the circuit
\[
\Qcircuit @C=2.3em @R=1.3em {
\lstick{\ket{j}}    & \ctrl{1} &  \rstick{\ket{j}} \qw \\
\lstick{\ket{k}}    & \ctrl{1} &  \rstick{\ket{k}} \qw \\
\lstick{\ket{\ell}}  & \targ    &  \rstick{X^{jk}\ket{\ell}=\ket{\ell\oplus(j \text{ AND } k)}.}  \qw \\
}\vspace*{0.2cm}
\]
This gate has two controls and one target. $X$ acts on the target qubit if and only if both control qubits are set to one. If one control is set to zero, the target doesn't change. This gate can be seen as the quantum version of the classical AND gate because if $\ell=0$, the output of the third qubit is $(j \text{ AND } k)$.
The matrix representation of the Toffoli gate is a diagonal block matrix given by
\[
C^2(X)= \begin{bmatrix}
                 I_2 &  &  &  \\
                  & I_2 &  &  \\
                  &  & I_2 &  \\
                  &  &  & X
               \end{bmatrix}.
\]

There are variants of the Toffoli gate active when the control qubits are set to zero. For example, the circuit
\[
\Qcircuit @C=2.3em @R=1.3em {
\lstick{\ket{j}}    & \ctrlo{1} &  \rstick{\ket{j}} \qw \\
\lstick{\ket{k}}    & \ctrlo{1} &  \rstick{\ket{k}} \qw \\
\lstick{\ket{\ell}}  & \targ    &  \rstick{X^{(1-j)(1-k)}\ket{\ell}}  \qw \\
}\vspace*{0.2cm}
\]
implements a gate that applies $X$ on the third qubit if and only if the first two control qubits are set to zero. It can be implemented using a standard Toffoli gate and $X$ gates, as shown in the following circuit equivalence:
\[
\Qcircuit @C=0.8em @R=0.6em {
\lstick{}&\ctrlo{1}&\rstick{} \qw           & & & & &\gate{X}    &\ctrl{1}   &\gate{X}  & \qw            \\
\lstick{}&\ctrlo{1}&\rstick{\,\,\,\,\,\equiv} \qw& & & & &\gate{X}    &\ctrl{1} &\gate{X}   &\qw          \\
\lstick{}&  \targ &\rstick{} \qw            & & & & &\qw         &\targ   & \qw &\rstick{.}\qw
}\vspace{0.2cm}
\]

The multi-controlled NOT\footnote{The multi-controlled NOT is also called multiqubit Toffoli or generalized Toffoli.} gate $C^{n}(X)$ is a $(n+1)$-qubit gate with $n$ control qubits and one target. It is defined by
$$
C^{n}(X)\ket{q_0}\ket{q_1}\cdots\ket{q_{n-1}}\ket{q_n}=\ket{q_0}\ket{q_1}\cdots\ket{q_{n-1}}X^{q_0q_1\cdots q_{n-1}}\ket{q_n},
$$
and is represented by the circuit
\[
\Qcircuit @C=2.3em @R=1.5em {
\lstick{\ket{q_0}}    & \ctrl{1} &  \rstick{\ket{q_0}} \qw \\
\lstick{\ket{q_1}}    & \ctrl{0} &  \rstick{\ket{q_1}} \qw \\
\lstick{\vdots}    & \vdots &  \rstick{\vdots} \\
\lstick{\ket{q_{n-1}}}    & \ctrl{1} &  \rstick{\ket{q_{n-1}}} \qw \\
\lstick{\ket{q_n}}  & \targ    &  \rstick{X^{q_0q_1\cdots q_{n-1}}\ket{q_n},}  \qw \\
}\vspace*{0.2cm}
\]
where $q_0q_1\cdots q_{n-1}$ is the product of bits $q_0$, $q_1$, ..., $q_{n-1}$. Therefore, the state of the target qubit changes if and only if all control qubits are set to 1. The state of each control qubit doesn't change when we describe this gate acting on the computational basis. The action on a generic vector is obtained by writing the vector as a linear combination of vectors of the computational basis and then using linearity.

The simplest way to decompose the multi-controlled NOT gate in terms of the usual Toffoli gate is by using $(n-2)$ draft qubits called \textit{ancillas}.\footnote{An extra auxiliary qubit with a fixed input state is usually called a \textit{clean ancilla}. In this work, we refer to it simply as an \textit{ancilla}.} The ancillas are interlaced with the control qubits, the first ancilla being inserted between the second and third qubits. The best way to explain the decomposition is to show an example. Consider the gate $C^5(X)$, whose decomposition requires three \textit{ancillas}, as shown in the following circuit equivalence:
\[
\Qcircuit @C=1.3em @R=0.7em {
\lstick{}& \ctrl{1}& \rstick{} \qw      & & &                & &\ctrl{1}& \qw     & \qw     & \qw     & \qw    &  \qw   & \ctrl{1} & \qw \\
\lstick{}& \ctrl{2}& \rstick{} \qw      & & &                & &\ctrl{1} & \qw     & \qw     & \qw     & \qw    &  \qw   & \ctrl{1} & \qw\\
\lstick{}&    ---- &                    & & &&\lstick{\ket{0}} &\targ    &\ctrl{1} & \qw     & \qw     & \qw    &\ctrl{1}& \targ& \rstick{\ket{0}}\qw \\
\lstick{}& \ctrl{2}&  \rstick{} \qw     & & &                & &\qw      &\ctrl{1}& \qw     & \qw     & \qw    &\ctrl{1}& \qw &\qw\\
\lstick{}&    ---- &\rstick{\,\,\,\,\,\equiv}& & &&\lstick{\ket{0}} &\qw      &\targ    &\ctrl{1} &\qw      &\ctrl{1}& \targ& \qw& \rstick{\ket{0}}\qw \\
\lstick{}& \ctrl{2}&  \rstick{} \qw     & & &                & &\qw      &\qw      &\ctrl{1}&\qw      &\ctrl{1}& \qw& \qw& \qw \\
\lstick{}&  ----   &                    & & &&\lstick{\ket{0}} &\qw      &\qw      &\targ    &\ctrl{1} &\targ& \qw& \qw& \rstick{\ket{0}}\qw \\
\lstick{}& \ctrl{1}&  \rstick{} \qw     & & &                & &\qw      &\qw      &\qw      &\ctrl{1} &\qw& \qw& \qw& \qw \\
\lstick{}& \targ   &  \rstick{} \qw     & & &                & &\qw      &\qw      &\qw      &\targ    &\qw& \qw& \qw& \rstick{.}\qw \\
}\vspace{0.1cm}
\]

The multi-controlled NOT gate can also be activated when the control qubits are set to zero. In this case, the control qubit is represented by an empty circle. Since each control can be either empty or full, and besides there are $n+1$ target qubits, the total number of multi-controlled NOT gates $C^n(X)$ is $(n+1)2^n$. The combination of these multi-controlled NOT gates can be used to implement any Boolean function of $n$ bits, as described in Section~\ref{sec:circ_Boolean_fcn}.

\subsection*{Preparing an arbitrary three-qubit state}

Suppose the initial state of the 3 qubits is $\ket{000}$. We prepare an arbitrary three-qubit state $\ket{\psi}=\sum_{j=0}^7 a_j \ket{j}$ with non-negative real amplitudes using the circuit
\[
\Qcircuit @C=0.7em @R=1.0em {
\lstick{\ket{0}} & \gate{R_y(2\theta_1)}&\ctrlo{1}                          &\ctrl{1}                             &\ctrlo{1}                         &\ctrlo{1}                          &\ctrl{1}                          &\ctrl{1}                          & \rstick{} \qw \\
\lstick{\ket{0}} & \qw                                  &\gate{R_y(2\theta_2)}&\gate{R_y(2\theta_3)}&\ctrlo{1}                          &\ctrl{1}                             &\ctrlo{1}                          &\ctrl{1}                         &  \qw& \rstick{\ket{\psi}} \\
\lstick{\ket{0}} & \qw                                  &\qw                                   &\qw                                   &\gate{R_y(2\theta_4)}&\gate{R_y(2\theta_5)}&\gate{R_y(2\theta_6)}&\gate{R_y(2\theta_7)}& \rstick{} \qw
\gategroup{1}{1}{3}{9}{.7em}{\}}}\vspace{0.2cm}
\]
and by solving the system of equations
\begin{align*}
\cos\theta_1\cos\theta_2\cos\theta_4 &= a_0,\,\,\,\,\,\,\,\,\,\,\,\cos\theta_1\cos\theta_2\sin\theta_4 = a_1,\\
\cos\theta_1\sin\theta_2\cos\theta_5 &= a_2,\,\,\,\,\,\,\,\,\,\,\,\,\cos\theta_1\sin\theta_2\sin\theta_5 = a_3,\\
\sin\theta_1\cos\theta_3\cos\theta_6 &= a_4,\,\,\,\,\,\,\,\,\,\,\,\,\sin\theta_1\cos\theta_3\sin\theta_6 = a_5,\\
\sin\theta_1\sin\theta_3\cos\theta_7 &= a_6,\,\,\,\,\,\,\,\,\,\,\,\,\,\sin\theta_1\sin\theta_3\sin\theta_7 = a_7,
\end{align*}
to determine the angles $\theta_1,...,\theta_7$, which are used in the circuit as parameters for $R_y$. This method is valid because, if we carefully apply all the gates in the circuit iteratively to the initial state $\ket{000}$, we can verify that the output is
\[
\ket{\psi} \,=\,
\begin{bmatrix}
\cos\theta_1\cos\theta_2\cos\theta_4 \\
\cos\theta_1\cos\theta_2\sin\theta_4 \\
\cos\theta_1\sin\theta_2\cos\theta_5 \\
\cos\theta_1\sin\theta_2\sin\theta_5 \\
\sin\theta_1\cos\theta_3\cos\theta_6 \\
\sin\theta_1\cos\theta_3\sin\theta_6 \\
\sin\theta_1\sin\theta_3\cos\theta_7 \\
\sin\theta_1\sin\theta_3\sin\theta_7
\end{bmatrix}.
\]
The process can be extended to more qubits, and the general pattern can be obtained recursively.

It is also straightforward to find the solution to the previous equations. For instance, to find $\theta_1$, we square the equations of the first line and add them up to obtain
\[
\cos^2\theta_1\cos^2\theta_2=a_0^2+a_1^2.
\]
Repeating the same procedure with the second line, we obtain
\[
\cos^2\theta_1\sin^2\theta_2=a_2^2+a_3^2.
\]
Then, adding these equation, we find
\[
\theta_1 = \arccos \sqrt{a_0^2+a_1^2+a_2^2+a_3^2}.
\]

\begin{exercise}
Find the remaining angles $\theta_j$ for $j$ from 2 to 7.
\end{exercise}

\begin{exercise}
Show that $C^2(R_y(\theta))$ can be decomposed into two Toffoli gates and $R_y$ gates, similar to the circuit in Exercise~\ref{ex:decompRy}.
\end{exercise}

\subsection*{Partial measurement}

It is not necessary to measure all qubits at the end of a quantum circuit. For instance, consider the following circuit:
\[
\Qcircuit @C=2.3em @R=1.7em {
\lstick{\ket{0}}    & \multigate{1}{U} &  \meter \qw & \cw \\
\lstick{\ket{0}}  & \ghost{U}    & \qw & \rstick{.}\qw\\
\ustick{\hspace{0.9cm}\ket{\psi_0}} & \ustick{\hspace{1.6cm}\ket{\psi}}&
}
\]
In this example, only the first qubit is measured. Although we illustrate the concept with two qubits, it generalizes to an arbitrary number of qubits. The states $\ket{\psi_0}$ and $\ket{\psi}$ shown at the bottom of the circuit represent the state of the quantum system at different stages. Initially, $\ket{\psi_0} = \ket{00}$, and after applying the unitary operation $U$, the state becomes $\ket{\psi} = U\ket{\psi_0}$. The measurement occurs when the quantum system is in state $\ket{\psi}$.

Mathematically, each measurement is represented by a set of orthogonal projection operators. In this case, the projections are
\begin{align*}
P_0 &=\ket{0}\bra{0}\otimes I,\\
P_1 &=\ket{1}\bra{1}\otimes I.
\end{align*}
By definition, an operator $P$ is an orthogonal projection if it is Hermitian and satisfies $P^2=P$. The expression for $\ket{0}\bra{0}$ is obtained by computing the matrix product of $\ket{0}$ and $\bra{0}$, which results in a 2-dimensional matrix. When a set of projections corresponds to a quantum measurement, they must sum to the identity operator
\begin{equation}
P_0+P_1\,=\,I\otimes I.
\end{equation}
According to the measurement postulate, the probability of obtaining outcome 0 or 1 is given by
\begin{align*}
p_0 &= \bra{\psi}P_0\ket{\psi},\\
p_1 &= \bra{\psi}P_1\ket{\psi},
\end{align*}
respectively. Immediately after the measurement, the quantum state collapses to
\begin{align*}
\ket{\psi'} =
\begin{cases}
\frac{1}{\sqrt{p_0}} P_0 \ket{\psi}, & \text{if the outcome is } 0, \\
\frac{1}{\sqrt{p_1}} P_1 \ket{\psi}, & \text{if the outcome is } 1.
\end{cases}
\end{align*}
Note that the post-measurement state $\ket{\psi'}$ is still a two-qubit state; in general, measuring the first qubit also changes the (conditional) state of the second qubit.

\begin{exercise}
Find the matrix representations of $P_0$ and $P_1$, verify that they are orthogonal projection operators, and check that $P_0 + P_1 = I \otimes I$.
\end{exercise}

\begin{exercise}
Generalize the mathematical description of a partial measurement of $n'$ qubits in an $n$-qubit circuit. You will need $2^{n'}$ projection operators. Then, express your answer compactly using an $n'$-bit string $i$, such as in $p_i$ and $P_i$.
\end{exercise}

\begin{exercise}
The measurement postulate states that an observable \( \mathcal{O} \) must be a linear combination of orthogonal projection operators:
\[
\sum_\ell \ell P_\ell,
\]
where \( \ell \in \mathbb{R} \) and \( P_\ell \) are orthogonal projection operators satisfying \( \sum_\ell P_\ell = I \). If the physical system is in the state \( \ket{\psi} \) just before measuring \( \mathcal{O} \), the outcome \( \ell \) occurs with probability
$p_\ell = \bra{\psi} P_\ell \ket{\psi}$. Show that after many repetitions of the measurement process, the average observed value converges to $\bra{\psi} \mathcal{O} \ket{\psi}$.
\end{exercise}

\section{Circuit of a Boolean function}\label{sec:circ_Boolean_fcn}

Suppose we want to evaluate a Boolean expression on a quantum computer. How do we implement the circuit? We need to know the types of bricks we have, some of which are shown in Fig.~\ref{Fig:3bricks}. \vspace{-15pt}
\begin{figure}[!ht]
\centering
\[
\Qcircuit @C=1.0em @R=1.9em {
\lstick{a}    & \ctrl{1} &  \rstick{a} \qw \\
\lstick{b}    & \ctrl{1} &  \rstick{b} \qw \\
\lstick{0}  & \targ    &  \rstick{a\wedge b}  \qw \\
&{\text{AND}}\\
}\hspace{2.5cm}
\Qcircuit @C=1.0em @R=1.8em {
\lstick{a}  & \ctrlo{1} &  \rstick{a} \qw \\
\lstick{b}  & \ctrlo{1} &  \rstick{b} \qw \\
\lstick{1}  & \targ &  \rstick{a\vee b}  \qw \\
&{\text{OR}}\\
} \hspace{2.5cm}
\Qcircuit @C=1.0em @R=1.93em {
\lstick{}    &  &  \rstick{} \\
\lstick{a}    & \ctrl{1} &  \rstick{a} \qw \\
\lstick{b}  & \targ    &  \rstick{a\oplus b}  \qw \\
&{\text{XOR}}\\
}\hspace{2.5cm}
\Qcircuit @C=1.em @R=2.0em {
\lstick{}   &  &  \rstick{} \\
\lstick{}  &  &  \rstick{} \\
\lstick{a}  &\targ    &  \rstick{\bar a}  \qw \\
&{\text{NOT}}\\
}
\]
\caption{Basic bricks to build a Boolean expression.  To implement the AND and OR gates, we must add an extra ancilla bit initialized as 0 and 1, respectively.}\label{Fig:3bricks}
\end{figure}

The Boolean function\footnote{We use the following equivalent notations: $\text{AND}\,\equiv\wedge$, $\text{OR}\,\equiv\vee$, $\text{NOT}\,\equiv\,\bar{\,}$ .} $f(a,b)=a\wedge b$ is not reversible; therefore, we need an extra bit to implement $U_f$, so that $U_f\ket{a,b,0}=\ket{a,b,a\wedge b}$. Note that we also have to keep the input bits. That is why the AND and OR gates require an extra ancilla bit, while the XOR and NOT gates don't. The ancilla bit must be initialized as 0 for the AND gate and as 1 for the OR gate. For the OR gate, the ancilla bit can be initialized as 0 followed by a NOT gate. We obtain other useful gates with these bricks, for instance, the NAND gate is obtained by placing a NOT gate before or after a AND gate in the following way:
\[
\Qcircuit @C=1.3em @R=1.6em {
\lstick{a}     &\qw & \ctrl{1}&  \rstick{a} \qw \\
\lstick{b}    &\qw& \ctrl{1}  &  \rstick{b} \qw \\
\lstick{0}  & \targ    & \targ &  \rstick{\overline{a\wedge b}.}
 \qw \\
}\vspace*{0.2cm}
\]
Alternatively, we can simply initialize the ancilla as 1 instead of adding the NOT gate. The output $\overline{a\wedge b}$ is equivalent to $\bar a\vee \bar b$ by De Morgan's laws, which is described in Fig.~\ref{Fig:simp-rules}.  The exact same methods can be applied to obtain the NOR gate.

\begin{figure}[!ht]
\centering
\fbox{
\begin{tabular}{lllll}
Involution & $\bar{\bar a} = a$ \\ \vspace{2pt}
Idempotent & $a\wedge a = a$ & $a\vee a = a$ \\ \vspace{2pt}
Commutative\ \ \ & $a\wedge b = b\wedge a$ & $a\vee b = b\vee a$ \\ \vspace{2pt}
De Morgan & $\overline{a\wedge b} = \bar a\vee \bar b$ & $\overline{a\vee b}=\bar a\wedge \bar b$ \\ \vspace{2pt}
Associative & $(a\wedge b)\wedge c = a\wedge(b\wedge c)$   & $(a\vee b)\vee c = a\vee(b\vee c)$ \\ \vspace{2pt}
Distributive & $a\wedge (b\vee c) = (a\wedge b)\vee(a\wedge c)$   & $a\vee(b\wedge c) = (a\vee b)\wedge(a\vee c)$ \\ \vspace{2pt}
Complement & $a\wedge \bar a = 0$\hspace{1cm} $\bar 0=1$  & $a\vee \bar a = 1$\hspace{1cm}  $0=\bar 1$  \\ \vspace{2pt}
Identity & $a\wedge 0 = 0$\hspace{1cm}  $a\wedge 1 = a$ & $a\vee 0 = a$\hspace{1cm}   $a\vee 1 = 1$
\end{tabular}
}
\caption{Simplification rules for Boolean expressions.}\label{Fig:simp-rules}
\end{figure}

The output label of the XOR gate in Fig.~\ref{Fig:3bricks} does not explicitly display both input bits. This introduces a challenge if we need to reuse the input bits. A useful implementation of the XOR gate, at the cost of losing efficiency, is
\[
\Qcircuit @C=1.3em @R=1.6em {
\lstick{a}     &\qw & \ctrl{2}&  \rstick{a} \qw \\
\lstick{b}    &\ctrl{1}& \qw  &  \rstick{b} \qw \\
\lstick{0}  & \targ    & \targ &  \rstick{ a\oplus b.}  \qw \\
}\vspace*{0.2cm}
\]
The same problem happens if we implement the Boolean expression $(\bar a\wedge b)$ as follows
\[
\Qcircuit @C=1.3em @R=1.6em {
\lstick{a}     &\targ & \ctrl{1}&  \rstick{\bar a} \qw \\
\lstick{b}    &\qw& \ctrl{1}  &  \rstick{b} \qw \\
\lstick{0}  & \qw    & \targ &  \rstick{\bar a\wedge b.}  \qw \\
}\vspace*{0.cm}
\]
We cannot reuse the first input bit straightway. It is usually better to implement this Boolean expression placing a NOT before and after the first control, which is equivalent to an  empty control as follows
\[
\Qcircuit @C=1.3em @R=1.6em {
\lstick{a}    & \ctrlo{1}&  \rstick{a} \qw \\
\lstick{b}    & \ctrl{1}  &  \rstick{b} \qw \\
\lstick{0}    & \targ &  \rstick{\bar a\wedge b.}  \qw \\
}\vspace*{0.2cm}
\]

What about the ancilla bits, how can we reuse them? This is also possible but we need a larger example to show how it is done. Consider the following Boolean function:
\begin{equation}\label{eq:Boolean-expr}
f(a,b,c)\,=\, a\wedge (c\vee (a\wedge \bar b\wedge \bar c) ),
\end{equation}
which is implemented in the following way:
\[
\Qcircuit @C=1.9em @R=1.em {
\lstick{a}&\ctrl{1} & \qw     &\ctrl{4}& \qw     &\ctrl{1}    &\rstick{a} \qw \\
\lstick{b}&\ctrlo{1}& \qw     &\qw     & \qw     &\ctrlo{1}   &\rstick{b} \qw \\
\lstick{c}&\ctrlo{1}&\ctrlo{1}&\qw     &\ctrlo{1}&\ctrlo{1}   &\rstick{c} \qw \\
\lstick{0}& \targ   &\ctrlo{1}&\qw     &\ctrlo{1}& \targ      &\rstick{0} \qw \\
\lstick{1}& \qw     & \targ   &\ctrl{1}& \targ   & \qw        &\rstick{1} \qw \\
\lstick{0}& \qw     & \qw     &\targ   & \qw     & \qw        &\rstick{f(a,b,c).} \qw
\gategroup{1}{5}{5}{6}{1.5em}{--}
}\vspace*{0.cm}
\]
We start the implementation by selecting the expression inside the innermost parentheses, which is $(a\wedge \bar b\wedge \bar c)$. We use 3 wires for the inputs $a$, $b$, and $c$ and an ancilla wire (4th bit) for the output of this sub-expression. We use a multi-controlled NOT gate with full control for $a$ and empty controls for $\bar b$ and $\bar c$. Besides, the ancilla bit is initialized as 0 for the AND case. The second gate implements the OR of the output of the previous gate with $c$. All controls are empty and the new ancilla bit is initialized as 1 for the OR case. The third gate implements the AND of the last output with $a$. All controls are full and the new ancilla bit is initialized as 0 for the AND case. The output of the last ancilla bit is $f(a,b,c)$. To reuse the ancilla bits (except the last one) we have to include the gates in the dashed box, which is the mirror of the gates of the first part.  


Before building the circuit of a Boolean expression, it is a good idea to simplify the expression by applying the simplification rules listed in Fig.~\ref{Fig:simp-rules}. Expression~(\ref{eq:Boolean-expr}) simplifies to
\[
f(a,b,c)\,=\, a\wedge (\bar b\vee c),
\]
which can be implemented using 3 Toffoli gates and 2 ancilla bits in the following way:
\[
\Qcircuit @C=1.9em @R=0.9em {
\lstick{a}&\qw      &\ctrl{3}&\qw    &\rstick{a} \qw \\
\lstick{b}&\ctrl{1} &\qw     &\ctrl{1}   &\rstick{b} \qw \\
\lstick{c}&\ctrlo{1}&\qw     &\ctrlo{1}   &\rstick{c} \qw \\
\lstick{1}& \targ   &\ctrl{1}& \targ      &\rstick{1} \qw \\
\lstick{0}& \qw     &\targ   & \qw        &\rstick{f(a,b,c).} \qw
}\vspace*{0.cm}
\]

This last example shows that we obtain equivalent circuits by manipulating the Boolean expression. For example, we know that
\[
a\wedge b\wedge c\wedge d\wedge e = (((a\wedge b)\wedge c)\wedge d)\wedge e.
\]
The left-hand side is implemented with one multi-controlled NOT gate and one ancilla bit, while the right-hand side is implemented with 7 standard Toffoli gates and 4 ancillas bits. This provides a decomposition of the multi-controlled NOT gate in terms of the Toffoli gates, identical to the one presented at the end of Section~\ref{subsec:porta3oumais}.

\subsection*{Circuit of a truth table}

Let's show how to obtain the quantum circuit of a truth table. We only need multi-controlled NOT gates. To show that the multi-controlled NOT gates can implement any truth table on a quantum computer, we take a 3-bit Boolean function $f(a,b,c)$ defined by the following truth table as an example:
\begin{displaymath}
\begin{array}{c c c|c}
a & b &c & f(a,b,c)\\ 
\hline 
0 & 0 & 0 & 0 \\
0 & 0 & 1 & 1 \\
0 & 1 & 0 & 0 \\
0 & 1 & 1 & 0 \\
1 & 0 & 0 & 0 \\
1 & 0 & 1 & 0 \\
1 & 1 & 0 & 1 \\
1 & 1 & 1 & 0 \\
\end{array}
\end{displaymath}
After this example, it is evident how the general case is obtained. Since $f$ has three input bits, we use multi-controlled NOT gates with three controls. The 4th qubit is the target. The output of $f$ is the output of a measurement of the target qubit. Since $f$ has two clauses in the \textit{disjunctive normal form} (there are two outputs 1 in the truth table), we use two multi-controlled NOT gates.\footnote{In the disjunctive normal form, we only consider rows of the truth table whose output is 1. For each row with output 1, we write a conjunction of three literals, using NOT for each input 0. Then we write a disjunction of the conjunctions, like this, $f(a,b,c)=(\bar a\wedge\bar b\wedge c)\vee (a\wedge b\wedge \bar c)$, where $(\bar a\wedge\bar b\wedge c)$ is true only for input 001 and $(a\wedge b\wedge \bar c)$ only for input 110. Then $f$ is true only for these inputs.}
The first gate must be active when the input is $\ket{001}$ and the second when the input is $\ket{110}$, which correspond to the rows of the truth table whose output is 1. The following circuit implements $f$:
\[
\Qcircuit @C=2.3em @R=1.6em {
\lstick{\ket{0}}    & \ctrlo{1} & \ctrl{1} &  \rstick{\ket{0}} \qw \\
\lstick{\ket{0}}    & \ctrlo{1} & \ctrl{1} &  \rstick{\ket{0}} \qw \\
\lstick{\ket{1}}    & \ctrl{1} & \ctrlo{1} & \rstick{\ket{1}} \qw \\
\lstick{\ket{0}}  & \targ  & \targ  &  \rstick{\ket{1}.}  \qw \\
}\vspace*{0.2cm}
\]
If the input is $\ket{a,b,c}\ket{0}$ then the output is $\ket{a,b,c}\ket{f(a,b,c)}$.
This shows that the quantum computer can calculate any $n$-bit Boolean function using a multi-controlled NOT gate with $n$ control qubits and one target qubit for each output 1 of the truth table. The goal here is not to implement classical algorithms on quantum computers because it makes no sense to build a much more expensive machine to run only classical algorithms. However, the implementation we have just described can be used for inputs in superposition, which is not allowed on a classical computer. Unfortunately, this quantum circuit construction technique for calculating truth tables is not efficient in general, since the number of multi-controlled NOT gates increases exponentially as a function of the number of qubits in the worst case.

\section{Quantum parallelism}

The simplest model of quantum computing is described by the circuit
\[
\Qcircuit @C=2.3em @R=0.7em {
\lstick{\ket{0}} & \gate{H} & \multigate{5}{\,\,\,U\,\,\,} &\meter  & \rstick{0 \text{ or } 1} \cw \\
\lstick{}  &  &        &        & \rstick{}  \\
\lstick{\vdots \,\,\,}  & {\vdots} &        &    {\vdots}   & & \lstick{\vdots}  \\
\lstick{}  &  &        &        & \rstick{}  \\
\lstick{}  &  &        &        & \rstick{}  \\
\lstick{\ket{0}} & \gate{H} & \ghost{\,\,\,U\,\,\,}        & \meter & \rstick{0 \text{ or } 1.} \cw
}\vspace{0.2cm}
\]
The initial state is $\ket{0,...,0}$, which is a state from the computational basis, as it represents a classical state. Choosing $\ket{0}$ as the initial state for each qubit is no loss of generality. In the next step, the Hadamard gate is applied to each qubit to enable \textit{quantum parallelism}. Following this, the unitary matrix $U$ is applied to all qubits. Finally, a measurement meter on each qubit returns a bit string. Note that the number of output qubits must match the number of input qubits, as unitary gates are invertible, making the computation process fully reversible in the absence of measurement, meaning no information is erased.

From the circuit above, we can understand the mechanism and advantage of quantum parallelism in the standard quantum computing model. A quantum circuit can simultaneously perform computations on multiple inputs through superposition, a feat that would require an exponential number of classical processors to match. The Hadamard gates initially transform each qubit into a superposition of all possible states, representing every possible classical input. The unitary operator $U$ is then applied to this superposition, executing the computation for all inputs simultaneously due to quantum parallelism. However, the final measurement yields only one classical outcome (an $n$-bit string), representing a single result from the superposition of potential outcomes. This highlights the unique computational power of quantum systems, albeit with the limitation that only a single result is observable after measurement.

More precisely, the first step of the circuit is to perform the following computation: $(H\ket{0})\otimes\cdots\otimes(H\ket{0})$, that is
\[
\left(\frac{\ket{0}+\ket{1}}{\sqrt 2}\right)\otimes\cdots\otimes\left(\frac{\ket{0}+\ket{1}}{\sqrt 2}\right).
\]
After expanding this product, we obtain \[
\frac{1}{\sqrt{2^n}}\big(\ket{0\cdots 0}+\ket{0\cdots 01}+\ket{0\cdots 10}+\cdots+\ket{1\cdots 1}   \big).
\]
Each $n$-bit string in the sum above represents a possible input of a classical algorithm. The next step is to apply $U$. Due to linearity, $U$ acts \textit{simultaneously} on all terms in the sum, allowing $2^n$ simultaneous computations on an $n$-qubit quantum computer. Note, however, that the result of these computations is a superposition state, and after measurement, the final output is a single $n$-bit string.

The input to the circuit is fixed; it is the state $\ket{0,...,0}$, which represents the lowest energy state, making it the easiest to prepare experimentally. Why is this choice no loss of generality? Unlike in classical computing, where inputs (like $3$ and $4$ in $3+4$) are directly fed into a circuit, quantum computing uses gates to prepare the desired input states. For example, to prepare the input $3$ and $4$ on a 6-qubit quantum computer, we initialize it in the state $\ket{0,0,0}\ket{0,0,0}$ and apply $I\otimes X \otimes X$ to the first register and $X\otimes I\otimes I$ to the second, resulting in $\ket{0,1,1}\ket{1,0,0}$, which represent the desired numbers. In quantum computing, the classical inputs are therefore produced by quantum gates rather than being directly assigned. Following this reasoning, applying Hadamard gates for quantum parallelism would scatter the specific initial values across all possible states, losing the desired input numbers among them. Without applying Hadamard gates, we effectively use the quantum computer as a classical computer, which is the standard approach for algorithms that do not benefit from quantum enhancements.

\section{Big $O$ notation}

Since $O$ and $\Omega$ notation will be used frequently in the next Section, this is a good time to define them. To determine the computational complexity of an algorithm, we typically consider the number of qubits $n$ in the circuit that implements it and examine whether the circuit depth scales as a polynomial or exponential function of $n$. Circuit depth is defined as the number of layers (or columns) of gates from a universal gate set, where gates that can be processed in parallel are placed in the same column. If the depth is a polynomial function of $n$, the algorithm is considered efficient in the complexity-theoretic sense; if it scales exponentially, the algorithm is considered inefficient. For polynomial depth, it is also important to know the degree of the polynomial, since lower-degree polynomials indicate more efficient algorithms.

Let $f(n)$ represent the depth of a circuit with $n$ qubits, or more generally any function defined for sufficiently large integers $n$. For most algorithms, $f(n)$ is an intricate function that is often too difficult to express exactly. However, if we can find a function $g(n)$ and constants $c > 0$ and $n_0 > 0$ such that $|f(n)| \le c\,|g(n)|$ for all $n \ge n_0$, we say that $f(n)$ is $O(g(n))$, and we write
\[
f(n)=O(g(n)).
\]
For example, if $f(n) = 7n^2 + 3n + 5$, we can correctly say that $f(n) = O(n^2)$ because, by choosing $c = 10$ and $n_0 = 2$, we have $f(n) \le 10\,n^2$ for all $n \ge 2$. It would also be correct to say that $f(n) = O(n^3)$, but this would miss an opportunity to describe the computational complexity more precisely. On the other hand, it would be incorrect to say that $f(n) = O(n)$, because no constants $c > 0$ and $n_0 > 0$ satisfy $f(n) \le c\,n$ for all sufficiently large $n$. Stating that the computational complexity of an algorithm is $O(n^a)$ for some $a > 0$ leaves open the possibility of finding a more efficient algorithm with complexity $O(n^b)$, where $b < a$.

If we can find a function $g(n)$ and constants $c > 0$ and $n_0 > 0$ such that $|f(n)| \ge c\,|g(n)|$ for all $n \ge n_0$, we say that $f(n)$ is $\Omega(g(n))$, and we write
\[
f(n)=\Omega(g(n)).
\]
For example, if every algorithm that solves a given computational problem has complexity $f(n)=\Omega(\e^n)$, then no algorithm with a polynomial-depth circuit can solve that problem.

Using these definitions, we have the equivalence
\[
f(n)=O(g(n)) \Longleftrightarrow g(n)=\Omega(f(n)).
\]
We then define the big $\Theta$ notation as
\[
f(n)=\Theta(g(n)) \Longleftrightarrow f(n)=O(g(n)) \text{ and } f(n)=\Omega(g(n)).
\]

There is also the little $o$ notation, defined as follows: $f(n)=o(g(n))$ if, for every constant $c>0$, there exists a constant $n_0>0$ such that $|f(n)| \le c\,|g(n)|$ for all $n \ge n_0$. This means that $f(n)$ becomes negligible compared to $g(n)$ as $n \to \infty$, or equivalently
\[
f(n)=o(g(n)) \Longleftrightarrow \lim_{n \to \infty} \frac{f(n)}{g(n)} = 0,
\]
provided that $g(n)\neq 0$ for all sufficiently large $n$. For instance, if $p(n)$ is a probability that depends on $n$, stating that $p(n)=1+o(1)$ implies that $\lim_{n \to \infty} p(n)=1$.

It is important to note a key distinction between the definitions of little $o$ and big $O$: in the definition of little $o$, the inequality must hold for every constant $c>0$, whereas in the definition of big $O$, it is sufficient that it holds for at least one constant $c>0$.

\section{Decomposition into universal gates}\label{sec:decomp-Univ}

When working with quantum algorithms, we typically reason in terms of $n$-qubit unitary operators, and eventually, we need to implement these unitary operators on quantum computers. Quantum computers provide a set of universal gates, which usually includes CNOT, SWAP, $R_x(\theta)$, $R_y(\theta)$, $R_z(\theta)$, $H$, $S$, $T$, Pauli matrices, and a few other gates. This implies that we must decompose the unitary operators we wish to implement into a sequence of operators from the set of available gates, using tensor products, matrix products, and multiplication by unit complex numbers (global phases). An interesting question to consider is: what is the smallest set of universal gates? If we aim to decompose any unitary operator exactly, the set must be infinite, with one of the simplest choice being~\cite{BBCD95}
\[
S_{\text{universal}}=\{\text{CNOT},R_y(\theta),R_z(\theta)\},
\]
where $0\le \theta<2\pi$.

A sketch of the proof that $S_{\text{universal}}$ is universal is as follows. We first use the result of the following exercise.

\begin{exercise}\label{Exer:univ-gates}
Using the fact that an arbitrary $(2\times 2)$ matrix $U$ with complex entries can be written in the form
\[
\begin{bmatrix}
 a\,\e^{\ii\alpha_a} &  b\,\e^{\ii\alpha_b} \vspace{2pt}\\
 c\,\e^{\ii\alpha_c} &  d\,\e^{\ii\alpha_d}
\end{bmatrix},
\]
where $a$, $b$, $c$, and $d$ are nonnegative real numbers and $\alpha_a$, $\alpha_b$, $\alpha_c$, $\alpha_d$ are angles, and assuming that $U$ is unitary ($UU^\dagger=U^\dagger U=I$), which imposes constraints on these parameters, show that $U$ can be re-parametrized as
\begin{equation}\label{eq:arb-U}
U\,=\,\e^{\ii \alpha}
\begin{bmatrix}
 \e^{-\ii\frac{\beta+\delta}{2}} \cos\frac{\gamma}{2} &
 - \e^{\ii\frac{\delta-\beta}{2}}\sin\frac{\gamma}{2}\vspace{2pt}\\
 \e^{\ii\frac{\beta-\delta}{2}} \sin\frac{\gamma}{2} &
 \e^{\ii\frac{\beta+\delta}{2}}\cos\frac{\gamma}{2}
\end{bmatrix},
\end{equation}
where $\alpha$, $\beta$, $\gamma$, and $\delta$ are angles.
\end{exercise}

The first part of the proof that $S_{\text{universal}}$ is universal is to use the definitions of $R_y$ and $R_z$ to obtain any one-qubit unitary operator $U$ in the form given by Eq.~(\ref{eq:arb-U}). This corresponds to the standard ZYZ Euler decomposition of a single-qubit unitary and yields
\[
U \,=\, \e^{\ii \alpha} R_z(\beta)\,R_y(\gamma)\,R_z(\delta).
\]
Hence, $U$ can be implemented up to a global phase $\e^{\ii\alpha}$ by a sequence of three rotations. The corresponding circuit is
\[
\Qcircuit @C=2.3em @R=1.6em {
\lstick{} & \gate{R_z(\delta)} & \gate{R_y(\gamma)} & \gate{R_z(\beta)} & \rstick{.} \qw
}\vspace*{0.2cm}
\]
Given a unitary matrix, we first express it in the form of Eq.~(\ref{eq:arb-U}), from which we extract the angles $\beta$, $\gamma$, and $\delta$, and then construct the circuit accordingly.

The second part of the proof deals with the decomposition of an $n$-qubit unitary gate $U$, where $n\ge 2$. We begin by expressing $U$ as a product of $m$ two-level unitary matrices, denoted by $U_1, \dots, U_m$, in the computational basis. Let $U_j$ be a matrix in this sequence. A two-level unitary matrix is defined as a unitary operator that acts non-trivially on exactly two states of the computational basis. Note that this definition depends on the chosen basis. Specifically, there exist two distinct computational basis states $\ket{k}$ and $\ket{\ell}$ such that
\begin{align*}
U_j\ket{k} &= a\ket{k} + c\ket{\ell},\\
U_j\ket{\ell} &= b\ket{k} + d\ket{\ell},\\
U_j\ket{i} &= \ket{i} \text{ for all } i \notin \{k,\ell\}.
\end{align*}
Thus, the diagonal entries are equal to $1$ for all indices $i \notin \{k,\ell\}$, and
the four entries in the rows and columns indexed by $k$ and $\ell$ form a $2\times 2$ unitary matrix
\[
\tilde{U}_j \,=\,
\begin{bmatrix}
a & b \\
c & d
\end{bmatrix}.
\]
All remaining matrix entries of $U_j$ are zero.

For example, the matrix
\begin{equation}\label{eq:U_j-}
U_j =
\begin{bmatrix}
1 & 0 & 0 & 0 & 0 & 0 & 0 & 0 \\
0 & a & 0 & 0 & 0 & 0 & b & 0 \\
0 & 0 & 1 & 0 & 0 & 0 & 0 & 0 \\
0 & 0 & 0 & 1 & 0 & 0 & 0 & 0 \\
0 & 0 & 0 & 0 & 1 & 0 & 0 & 0 \\
0 & 0 & 0 & 0 & 0 & 1 & 0 & 0 \\
0 & c & 0 & 0 & 0 & 0 & d & 0 \\
0 & 0 & 0 & 0 & 0 & 0 & 0 & 1
\end{bmatrix}
\end{equation}
is a two-level three-qubit unitary matrix, where $\ket{k}=\ket{001}$ and $\ket{\ell}=\ket{110}$. Note that both $b$ and $c$ may be zero, in which case $\tilde{U}_j$ is diagonal and $a$ and $d$ are complex numbers of unit modulus, and vice versa: both $a$ and $d$ may be zero, in which case $\tilde{U}_j$ is anti-diagonal and $b$ and $c$ are complex numbers of unit modulus.

The decomposition of an $n$-qubit $U$ as a product of two-level unitary matrices $U_1, \dots, U_m$ is inefficient in the generic case, because in the generic case $m=\Omega(4^n)$. An example of how this decomposition is obtained can be found in Section~4.5.1 of Ref.~\cite{NC00}.

\begin{exercise}
Show that any decomposition of $H^{\otimes n}$ into a product of $m$ two-level unitary matrices $U_1,\dots,U_m$ requires at least $m \ge 2^n - 1$ factors. To solve this exercise, it is not necessary to know the explicit procedure for constructing the matrices $U_1,\dots,U_m$. As a hint, consider the action of the circuit on the computational basis state $\ket{0^n}$. Show that when a two-level unitary is applied to any state, the number of computational-basis states with nonzero amplitude in the resulting superposition can increase by at most one. Compare this with the number of nonzero amplitudes in $H^{\otimes n}\ket{0^n}$.
\end{exercise}

In contrast, a two-level unitary admits an efficient decomposition, producing a circuit with $O(n^2)$ gates from the universal gate set. We therefore concentrate on the circuit construction for two-level unitaries. The method is best explained through an example, which generalizes directly.

Consider the two-level unitary matrix~(\ref{eq:U_j-}). The circuit construction uses the basis states $\ket{k}$ and $\ket{\ell}$. We take $\ket{\ell}$ as a reference state and start from $\ket{k}$, with the goal of transforming it into a state that differs from $\ket{\ell}$ in only one qubit, where the $2\times2$ unitary block $\tilde{U}_j$ will be applied.

In the example, $\ket{k}=\ket{001}$ and $\ket{\ell}=\ket{110}$. Since their first qubits differ, we apply a multi-controlled NOT gate that flips the first qubit, controlled by the remaining qubits matching those of $\ket{k}$. This operation is shown as the first gate in Fig.~\ref{fig:two-level-decomp} and maps $\ket{001}$ to $\ket{101}$. Next, the second qubit of $\ket{101}$ differs from that of $\ket{\ell}=\ket{110}$. We therefore apply a second multi-controlled NOT gate that flips the second qubit, again controlled by the remaining qubits. This is the second gate in Fig.~\ref{fig:two-level-decomp}, and the state becomes $\ket{111}$.

At this point, the current state and $\ket{\ell}$ differ in only one qubit (the third qubit). When only one qubit differs, we replace the NOT operation by a multi-controlled single-qubit gate $C^2(\tilde{U}_j)$ acting on that qubit, with the other qubits as controls. This gate implements the nontrivial $2\times2$ unitary block and corresponds to the third gate in Fig.~\ref{fig:two-level-decomp}. Finally, the preceding multi-controlled NOT gates are applied again in reverse order to uncompute the intermediate transformations and restore the remaining basis states, ensuring symmetry of the construction.

\begin{figure}[h!]
\[
\Qcircuit @C=2.3em @R=1.4em {
\lstick{} & \targ     & \ctrl{1} &\ctrl{1}           & \ctrl{1} & \targ &  \rstick{} \qw \\
\lstick{} & \ctrlo{-1}& \targ    & \ctrl{1}          &\targ      & \ctrlo{-1}& \rstick{} \qw \\
\lstick{} & \ctrl{-1} & \ctrl{-1}& \gate{\tilde{U}_j}& \ctrl{-1}& \ctrl{-1}& \rstick{} \qw \\
}\vspace*{0.2cm}
\]
\caption{Quantum circuit decomposition of the unitary operator $U_j$ from Eq.~(\ref{eq:U_j-}).
}\label{fig:two-level-decomp}
\end{figure}

\begin{exercise}
Show that the circuit in Fig.~\ref{fig:two-level-decomp} implements the unitary operator $U_j$ in Eq.~(\ref{eq:U_j-}).
(Hint: verify the action on all computational basis states and use linearity.)
\end{exercise}

\begin{exercise}
Let $U_j$ be a two-level four-qubit unitary matrix defined on the computational basis by
\begin{align*}
U_j\ket{1010} &= \frac{1}{\sqrt{2}}\bigl(\ket{1010}+\ket{1100}\bigr),\\
U_j\ket{1100} &= \frac{1}{\sqrt{2}}\bigl(\ket{1010}-\ket{1100}\bigr),
\end{align*}
and $U_j\ket{i}=\ket{i}$ for all computational basis states $\ket{i}\notin\{\ket{1010},\ket{1100}\}$.
\begin{enumerate}
\item[(a)] By applying the construction described above, derive the circuit
\[
\Qcircuit @C=2.3em @R=1.1em {
\lstick{} & \ctrl{1}   & \ctrl{1}   & \ctrl{1}          & \rstick{} \qw \\
\lstick{} & \targ      & \ctrl{1}   & \targ             & \rstick{} \qw \\
\lstick{} & \ctrl{-1}  & \gate{H}   & \ctrl{-1}         & \rstick{} \qw \\
\lstick{} & \ctrlo{-1} & \ctrlo{-1} & \ctrlo{-1}        & \rstick{.} \qw
}\vspace*{0.2cm}
\]
\item[(b)] Verify that the circuit implements $U_j$ by checking its action on all computational basis states.
\end{enumerate}
\end{exercise}

\begin{exercise}
Decompose the following two-qubit unitary matrix
\[
\begin{bmatrix}
\frac{1}{\sqrt{2}} & \frac{1}{\sqrt{2}} & 0 & 0 \\
\frac{1}{\sqrt{2}} & -\frac{1}{\sqrt{2}} & 0 & 0 \\
0 & 0 & 0 & -\ii \\
0 & 0 & \ii & 0
\end{bmatrix}
\]
as a product of two-level unitary matrices. Then, following the construction described above, derive a quantum circuit that implements this unitary using controlled one-qubit gates. (Since this is a two-qubit operator, multi-controlled NOT gates reduce to CNOT gates.)
\end{exercise}

The next step in the proof that $S_{\text{universal}}$ is universal is the decomposition of the multi-controlled gate $C^n(\tilde{U}_j)$. In Section~\ref{subsec:porta3oumais}, we described the decomposition of the multi-controlled NOT gate $C^n(X)$ into $(2n-3)$ Toffoli gates using $(n-2)$ ancillas. A similar construction can be used to implement $C^n(\tilde{U}_j)$ using $2(n-1)$ Toffoli gates, $(n-1)$ ancillas, and one controlled-$\tilde{U}_j$ gate. The circuit structure for the case $n=4$ is illustrated below:
\[
\Qcircuit @C=1.3em @R=0.7em {
\lstick{}& \ctrl{1}& \rstick{} \qw      & & &                & &\ctrl{1}& \qw     & \qw     & \qw     & \qw    &  \qw   & \ctrl{1} & \qw \\
\lstick{}& \ctrl{2}& \rstick{} \qw      & & &                & &\ctrl{1} & \qw     & \qw     & \qw     & \qw    &  \qw   & \ctrl{1} & \qw\\
\lstick{}&    ---- &                    & & &&\lstick{\ket{0}} &\targ    &\ctrl{1} & \qw     & \qw     & \qw    &\ctrl{1}& \targ& \rstick{\ket{0}}\qw \\
\lstick{}& \ctrl{2}&  \rstick{} \qw     & & &                & &\qw      &\ctrl{1}& \qw     & \qw     & \qw    &\ctrl{1}& \qw &\qw\\
\lstick{}&    ---- &\rstick{\,\,\,\,\,\equiv}& & &&\lstick{\ket{0}} &\qw      &\targ    &\ctrl{1} &\qw      &\ctrl{1}& \targ& \qw& \rstick{\ket{0}}\qw \\
\lstick{}& \ctrl{2}&  \rstick{} \qw     & & &                & &\qw      &\qw      &\ctrl{1}&\qw      &\ctrl{1}& \qw& \qw& \qw \\
\lstick{}&  ----   &                    & & &&\lstick{\ket{0}} &\qw      &\qw      &\targ    &\ctrl{1} &\targ& \qw& \qw& \rstick{\ket{0}}\qw \\
\lstick{}&\gate{\tilde{U}_j}    &  \rstick{} \qw     & & &                & &\qw      &\qw      &\qw      &\gate{\tilde{U}_j}    &\qw& \qw& \qw& \rstick{.}\qw \\
}\vspace{0.1cm}
\]
Assuming that $\tilde{U}_j$ is written in the form of Eq.~(\ref{eq:arb-U}), the controlled gate $C(\tilde{U}_j)$ admits the following decomposition:
\[
\Qcircuit @C=1.5em @R=0.7em {
              & \ctrl{2}                     & \qw  & & & \qw          & \ctrl{2} & \qw          & \ctrl{2}& \gate{D} & \qw   \\
               &          & \rstick{\,\,\equiv} &     \\
\lstick{}& \gate{\tilde{U}_j} & \qw  & & & \gate{C} & \targ    & \gate{B}  & \targ    & \gate{A} &  \rstick{,}\qw
}\vspace*{0.1cm}
\]
where
\begin{align*}
A &= R_z(\beta)\,R_y\left(\frac{\gamma}{2}\right),\\
B &= R_y\left(-\frac{\gamma}{2}\right)\,R_z\left(-\frac{\delta+\beta}{2}\right),\\
C &= R_z\left(\frac{\delta-\beta}{2}\right),\\
D &= \left[\begin{array}{cc}
 1 &  0\\
 0   &  \e^{\ii\alpha}
\end{array}\right].
\end{align*}

\begin{exercise}
Assume that $\tilde{U}_j$ is written in the form of Eq.~(\ref{eq:arb-U}). Prove that $ABC = I$ and that $\e^{\ii\alpha} A X B X C = \tilde{U}_j$. Use these identities to verify that the proposed decomposition of $C(\tilde{U}_j)$ is correct by evaluating the circuit on the inputs $\ket{0}\ket{i}$ and $\ket{1}\ket{i}$, where $i$ is an arbitrary bit. Recall that two circuits are equivalent if they produce the same output for all computational-basis inputs. See also Corollary~4.2 in Ref.~\cite{NC00}.
\end{exercise}

As a final step in the decomposition, the Toffoli gate can be implemented using only single-qubit gates and CNOT gates, as shown below:
\[
\Qcircuit @C=0.8em @R=0.8em {
\lstick{}&\ctrl{1}&\rstick{} \qw           & & & & &\qw      &\qw      & \qw             &\ctrl{2}& \qw     & \qw     &  \qw            & \ctrl{2}& \qw    &\ctrl{1} &\gate{T}&\ctrl{1} &\qw\\
\lstick{}&\ctrl{1}&\rstick{\,\,\,\,\,=} \qw& & & & &\qw      &\ctrl{1} & \qw             &\qw     & \qw     & \ctrl{1}&  \qw            & \qw     &\gate{T}&\targ    &\gate{T^\dagger}&\targ&\qw\\
\lstick{}&  \targ &\rstick{} \qw           & & & & &\gate{H} &\targ    &\gate{T^\dagger} &\targ   &\gate{T} & \targ   &\gate{T^\dagger} & \targ   &\gate{T}& \gate{H}  & \qw & \qw &\rstick{.}\qw
}\vspace{5pt}
\]
In the next exercise, you are asked to show that the gates $H$ and $T$ can be expressed in terms of the rotation gates $R_y(\theta)$ and $R_z(\theta)$, thereby completing the sketch of the proof that $S_\text{universal}$ is a universal gate set.

\begin{exercise}
Show that, up to a global phase, the following identities hold:
\begin{align*}
H &= \e^{\frac{\pi\ii}{2}} \, R_y\!\left(\frac{\pi}{2}\right) R_z(\pi),\\
S &= \e^{\frac{\pi\ii}{4}} \, R_z\!\left(\frac{\pi}{2}\right),\\
T &= \e^{\frac{\pi\ii}{8}} \, R_z\!\left(\frac{\pi}{4}\right).
\end{align*}
\end{exercise}

\begin{exercise}
Show that the gate set $\{\mathrm{CNOT},\, H,\, R_z(\theta)\}$, for $0 \le \theta < 2\pi$,
is universal.
\end{exercise}

\begin{exercise}\label{exe:C(AUAdagger)}
Let $U$ and $A$ be unitary operators acting on the target system. Show that
\[
C^n(A^\dagger U A) = (I\otimes A^\dagger)\, C^n(U)\, (I\otimes A).
\]
When $n=1$, we have
\[
\Qcircuit @C=1.5em @R=0.7em {
         & \ctrl{2}             & \qw                 & & & \qw      & \ctrl{2} & \qw  &\qw     \\
         &                      & \rstick{\,\,\equiv} & \\
\lstick{}& \gate{A^\dagger U A } & \qw                 & & & \gate{A} & \gate{U}    & \gate{A^\dagger}& \rstick{.}\qw}\vspace*{0.1cm}
\]
\end{exercise}

\begin{exercise}
Determine the gates $A$, $B$, $C$, and $D$ in the controlled-unitary decomposition when $\tilde{U}_j = Y$, and draw a circuit that implements $C(Y)$ using gates from the universal set. Then construct an alternative circuit using only one CNOT gate and single-qubit gates by employing the identity
\[
Y = R_z(\pi/2)\, X \, R_z(-\pi/2)
\]
together with the relation $C(U X U^\dagger) = (I\otimes U)\, C(X)\, (I\otimes U^\dagger)$. Draw this alternative circuit.
\end{exercise}

\begin{exercise}
Determine the gates $A$, $B$, $C$, and $D$ in the controlled-unitary decomposition when $\tilde{U}_j = H$, and draw a circuit that implements $C(H)$ using gates from the universal set. Then derive an alternative circuit using only one CNOT gate and single-qubit gates based on the identity
\[
H = S\, H\, T\, X\, T^\dagger\, H\, S^\dagger.
\]
Draw this alternative circuit.
\end{exercise}

\begin{exercise}
The Toffoli gate plays an important role in the decomposition of unitary operators into universal gates. Moreover, in many constructions, Toffoli gates appear in conjugate pairs acting on the same qubits. In such cases, one may use simplified decompositions that differ from the exact Toffoli gate by a global phase, since the global phases cancel and do not affect the overall unitary.
\end{exercise}

\subsection*{Using only one ancilla}

In the decomposition into universal gates described above, we introduced $(n-2)$ auxiliary qubits initialized to the state $\ket{0}$. All multi-controlled NOT gates in the circuit can share the same auxiliary qubits because they return to the state $\ket{0}$ after each gate. If an auxiliary qubit is initialized to $\ket{1}$, the process fails. Therefore, it seems that we cannot use qubits already present in the circuit as auxiliary qubits. For instance, when decomposing $C^{\lceil n/2 \rceil}(X)$, we require $\lceil n/2 \rceil - 2$ auxiliary qubits, even though there are enough qubits available in the circuit. The issue is that qubits already in the circuit may be in a state different from $\ket{0}$.

To overcome this limitation, we modify the decomposition of the multi-controlled NOT gate so that the auxiliary-qubit input may be in an arbitrary state $\ket{i}$, where $i \in \{0,1\}$. If such an auxiliary qubit is part of the original circuit and you have borrowed it for some purpose, it is referred to as a \textit{work qubit}.\footnote{An extra qubit added to the circuit with a known fixed initial state, typically $\ket{0}$, is called a \textit{clean ancilla}. A qubit already present in the circuit that assists in the decomposition of a multiqubit gate is called a \textit{dirty ancilla}. In this work, we use the term \textit{ancilla} to mean \textit{clean ancilla} and the term \textit{work qubit} instead of \textit{dirty ancilla}. Therefore, the initial state of a work qubit is not known a priori.}

Let us now show how to decompose $C^{m}(X)$ into Toffoli gates using $(m-2)$ work qubits, where
$m \le \left\lceil\frac{n}{2}\right\rceil$ or $m \le \left\lfloor\frac{n}{2}\right\rfloor + 1$.
In both cases, a sufficient number of work qubits is available in the circuit. We convert the ancillas from the decomposition of the multi-controlled NOT gate described in Section~\ref{subsec:porta3oumais} into work qubits by adding $O(m)$ Toffoli gates inside a dashed box, as shown in the following example with five controls:
\[
\Qcircuit @C=0.8em @R=0.7em {
\lstick{q_0}&\ctrl{1}&\rstick{}\qw&               &&&&&                &\ctrl{1} & \qw     & \qw     & \qw     & \qw    &  \qw   & \ctrl{1}& \qw     & \qw     & \qw     & \qw    &  \qw &\qw \\
\lstick{q_1}&\ctrl{2}&\rstick{}\qw&               &&&&&                &\ctrl{1} & \qw     & \qw     & \qw     & \qw    &  \qw   & \ctrl{1}& \qw     & \qw     & \qw     & \qw    &  \qw &\qw \\
\lstick{w_0}& \qw & \qw           &               &&&&&\lstick{\ket{i}}&\targ    &\ctrl{1} & \qw     & \qw     & \qw    &\ctrl{1}& \targ   &\ctrl{1} & \qw     & \qw     & \qw    &\ctrl{1} &\rstick{\ket{i}}\qw \\
\lstick{q_2}&\ctrl{2}&\rstick{}\qw&               &&&&&                &\qw      &\ctrl{1} & \qw     & \qw     & \qw    &\ctrl{1}& \qw     &\ctrl{1} & \qw     & \qw     & \qw    &\ctrl{1} &\qw \\
\lstick{w_1}&  \qw  &  \qw          &\rstick{\equiv}&&&&&\lstick{\ket{j}}&\qw      &\targ    &\ctrl{1} &\qw      &\ctrl{1}& \targ  & \qw     &\targ    &\ctrl{1} &\qw      &\ctrl{1}& \targ  &\rstick{\ket{j}}\qw \\
\lstick{\vdots\,\,}&\ctrl{2}&\rstick{}\qw&               &&&&&                &\qw      &\qw      &\ctrl{1} &\qw      &\ctrl{1}& \qw    & \qw     &\qw      &\ctrl{1} &\qw      &\ctrl{1}& \qw    &\qw \\
\lstick{}& \qw   &  \qw          &               &&&&&\lstick{\ket{k}}&\qw      &\qw      &\targ    &\ctrl{1} &\targ   & \qw    & \qw     &\qw      &\targ    &\ctrl{1} &\targ   & \qw   &\rstick{\ket{k}}\qw \\
\lstick{}&\ctrl{1}&\rstick{}\qw&               &&&&&                &\qw      &\qw      &\qw      &\ctrl{1} &\qw     & \qw    & \qw     &\qw      &\qw      &\ctrl{1} &\qw     & \qw   &\qw \\
\lstick{}&\targ   &\rstick{}\qw&               &&&&&                &\qw      &\qw      &\qw      &\targ    &\qw     & \qw    & \qw     &\qw      &\qw      &\targ    &\qw     & \qw   &\rstick{.}\qw
\gategroup{3}{17}{9}{21}{.9em}{--}
}\vspace{0.1cm}
\]
If the gates inside the dashed box are removed, the circuit reduces to the one described in Section~\ref{subsec:porta3oumais} and operates correctly only when $\ket{i} = \ket{j} = \ket{k} = \ket{0}$. By including the gates inside the dashed box, the two circuits become fully equivalent. That is, they implement the same unitary operator and therefore produce identical outputs not only for arbitrary $\ket{i}$, $\ket{j}$, and $\ket{k}$, but for any input state. The proof is left as an exercise.

\begin{exercise}\label{exe-2-Toff-decomp}
Prove that the circuits shown in the last figure are equivalent for an arbitrary number of control qubits. As a hint, it is useful to observe that the right-hand circuit can be obtained recursively as follows. The dashed box is a $C^{m-1}(X)$ gate. Let us formally describe the first recursive step. Let $q_0,\ldots,q_{m-1}$ be the control qubits and let $w_0$ be the first work qubit. The left-hand circuit is the multi-controlled NOT gate
\[
C^m_{q_0,\ldots,q_{m-1}}(X_{q_m}),
\]
which specifies the number of controls $m$, the control qubits $q_0,\ldots,q_{m-1}$, and the target qubit $q_m$. The first recursive step is
\[
C^m_{q_0,\ldots,q_{m-1}}(X_{q_m})
=
\bigl(C^{m-1}_{w_0,q_2,\ldots,q_{m-1}}(X_{q_m})\bigr)^\dagger\,
C^2_{q_0,q_1}(X_{w_0})\,
C^{m-1}_{w_0,q_2,\ldots,q_{m-1}}(X_{q_m})\,
C^2_{q_0,q_1}(X_{w_0}).
\]
To verify this recursive step, it suffices to assume that the work qubit $w_0$ is in an arbitrary computational-basis state $\ket{i}$, where $i\in\{0,1\}$. Explain why the adjoint operation $(\cdot)^\dagger$ (i.e., the conjugate transpose) is essential. Finally, draw the corresponding circuit identity omitting the work qubits $w_1,\ldots,w_{m-3}$ (which are introduced during the recursive process).
\end{exercise}

\begin{exercise}\label{exe-2-U2-I-decomp}
Extend the proof of Exercise~\ref{exe-2-Toff-decomp} to include the case $C^n(U)$ for the special case where $U^2=I$. In this context, $U$ represents a single-qubit unitary operator, which includes the multi-controlled NOT gate when $U=X$. Modify the circuit accordingly by replacing the target NOT gates with $U$ gates and then show the first recursive step:
\[
C^m_{q_0,\ldots,q_{m-1}}(U_{q_m}) =
\bigl(C^{m-1}_{w_0,q_2,\ldots,q_{m-1}}(U_{q_m})\bigr)^\dagger\,
C^2_{q_0,q_1}(X_{w_0})\,
C^{m-1}_{w_0,q_2,\ldots,q_{m-1}}(U_{q_m})\,
C^2_{q_0,q_1}(X_{w_0}).
\]
Note that $U^\dagger = U$, and therefore $C^k(U)^\dagger = C^k(U)$. At the end of the recursion, there are two $C^2(U)$ gates and Toffoli gates.
\end{exercise}

\begin{exercise}\label{exe-2-U2-I-decomp-2}
Provide an alternative proof of Exercise~\ref{exe-2-U2-I-decomp}, up to a global phase, for the case of a multi-controlled gate $C^m(U)$, where $U$ is a one-qubit unitary gate such that $U^2=I$.
\begin{enumerate}
\item[(a)] Show that there exists a one-qubit unitary $A$ and a global phase $\e^{\ii\phi}$ such that
\[
U=\e^{\ii\phi}\,A X A^\dagger .
\]
\item[(b)] Use part (a) together with the identity
\[
C^m(A^\dagger V A)=(I\otimes A^\dagger)\,C^m(V)\,(I\otimes A),
\]
where $A$ acts on the target qubit, to reduce the decomposition of $C^m(U)$ to that of $C^m(X)$, and hence obtain a Toffoli-based decomposition using work qubits.
\end{enumerate}
\end{exercise}\vspace{8pt}

The key technique for using only one ancilla (or one work qubit) to decompose $C^n(X)$, which is our primary goal, is to modify the initial step as follows. We partition the control qubits into two halves (two disjoint subsets whose sizes differ by at most one), as shown in the following circuit equivalence:
\[
\Qcircuit @C=1.3em @R=0.9em {
\lstick{}& \ctrl{1}& \rstick{} \qw      & & &                & &\ctrl{1}& \qw      & \ctrl{1} & \qw   & \qw      \\
\lstick{}& \ctrl{1}& \rstick{} \qw      & & &                & &\ctrl{1}& \qw      & \ctrl{1} & \qw   & \qw      \\
\lstick{}& \ctrl{1}& \rstick{} \qw      & & &                & &\ctrl{1}& \qw      & \ctrl{1} & \qw   & \qw      \\
\lstick{}& \ctrl{2}& \rstick{} \qw      & & &                & &\ctrl{1}& \qw      & \ctrl{1} & \qw   & \qw      \\
\lstick{}&    ---- &      \rstick{\,\,\,\,\,\equiv}              & & &&\lstick{\ket{i}} &\targ       &\ctrl{1}   & \targ   &\ctrl{1}   & \rstick{\ket{i}}\qw \\
\lstick{}& \ctrl{1}&  \rstick{} \qw     & & &                & &\qw      &\ctrl{1}  & \qw     &\ctrl{1}   & \qw \\
\lstick{}& \ctrl{1}&  \rstick{} \qw     & & &                & &\qw      &\ctrl{1}  & \qw     &\ctrl{1}   & \qw \\
\lstick{}& \ctrl{1}&  \rstick{} \qw     & & &                & &\qw      &\ctrl{1}  & \qw     &\ctrl{1}   & \qw \\
\lstick{}& \targ   &  \rstick{} \qw     & & &                & &\qw      &\targ     & \qw     &\targ       & \rstick{.}\qw \\
}\vspace{0.1cm}
\]
To verify the equivalence of the circuits, we consider four cases: (1)~both halves have inputs $\ket{1}$; (2)~the first half has inputs $\ket{1}$ while the second half has at least one input $\ket{0}$; (3)~the first half has at least one input $\ket{0}$ while the second half has inputs $\ket{1}$; and (4)~both halves have at least one input $\ket{0}$. In each case, assume that the work qubit may be in either state, $\ket{i}=\ket{0}$ or $\ket{i}=\ket{1}$, one at a time.

We have shown how to decompose $C^n(X)$ into four $C^m(X)$ gates using one work qubit, where $m = \lceil n/2 \rceil$ for two of them and $m = \lfloor n/2 \rfloor + 1$ for the other two. The remaining qubits can serve as work qubits to decompose each $C^m(X)$ gate into $O(m)$ Toffoli gates, as described earlier. As a result, $C^n(X)$ can be decomposed into $O(n)$ Toffoli gates using only one ancilla (or one work qubit).

\begin{exercise}
Determine the exact number of Toffoli gates required to decompose $C^n(X)$ using one work qubit, following the method described above.
\end{exercise}

\begin{exercise}\label{exe-2-U2-I-decomp-3}
Extend the decomposition described above to the case $C^n(U)$, where $U$ is a one-qubit unitary satisfying $U^2 = I$ (which includes the multi-controlled NOT gate as the special case $U = X$). Modify the circuit accordingly and show the equivalence of the resulting circuits.
\end{exercise}

The decomposition of $C^n(U)$ using only one ancilla is still missing when $U$ is an arbitrary one-qubit gate. This can be achieved using the following circuit equivalence:
\[
\Qcircuit @C=1.3em @R=0.7em {
\lstick{}& \ctrl{1}&\rstick{}\qw&                         &&&&&\ctrl{1} &\qw                   &\ctrl{1} & \qw      \\
\lstick{}& \ctrl{2}&\rstick{}\qw&                         &&&&&\ctrl{1} &\qw                   &\ctrl{1} & \qw\\
\lstick{}& \ctrl{2}&\rstick{}\qw&                         &&&&&\ctrl{1} &\qw                   &\ctrl{1} & \qw\\
\lstick{}& \ctrl{1}&\rstick{}\qw&                         &&&&&\ctrl{1} &\qw                   &\ctrl{1} & \qw  \\
\lstick{}& \ctrl{2}&\rstick{}\qw&&&&&&\ctrl{1} &\qw                   &\ctrl{1} &\qw   \\
\lstick{}&   -----  &\rstick{}   & \rstick{\,\equiv}                            &&&&\lstick{\ket{0}}&\targ    &\ctrl{1}              &\targ    &\rstick{\ket{0}}\qw     \\
\lstick{}& \gate{U}&\rstick{}\qw&                         &&&&&\qw      &\gate{U}&\qw      & \rstick{.}\qw \\
}
\]
There are two $C^n(X)$ gates, which can be decomposed into $O(n)$ Toffoli gates using the method described above, since one work qubit is available for them. Therefore $C^n(U)$ can be decomposed into $O(n)$ Toffoli gates and one controlled-$U$ gate using only one ancilla. Within the decomposition framework described above, any unitary operator can thus be implemented using only one ancilla, although not necessarily with $O(n)$ basic gates. Next, we show that all ancillas can be eliminated.

\begin{exercise}\label{exe-2-U2-I-decomp-4}
The aim of this exercise is to modify the decomposition of $C^n(U)$ described above so that the ancilla is replaced by a work qubit in the special case where $U^2 = I$. Modify the circuit accordingly and show the equivalence of the resulting circuits, assuming that the input to the work qubit is the state $\ket{i}$, where $i\in\{0,1\}$. Draw the corresponding circuit equivalence.
\end{exercise}

\subsection*{Eliminating all ancillas}

To avoid using ancillas in the decomposition of $C^n(X)$ and $C^n(U)$ when $n>2$, we modify the first step of the construction. An alternative first step for $C^n(X)$ (which also applies to $C^n(U)$ with the appropriate replacements) is given by the following circuit equivalence:
\[
\Qcircuit @C=1.3em @R=0.7em {
\lstick{}& \ctrl{1}&\rstick{}\qw&                         &&&&\qw        &\ctrl{1} &\qw                    &\ctrl{1} &\ctrl{1} & \qw      \\
\lstick{}& \ctrl{2}&\rstick{}\qw&                         &&&&\qw       &\ctrl{1} &\qw                    &\ctrl{1} &\ctrl{1} & \qw\\
\lstick{}& \ctrl{1}&\rstick{}\qw&                         &&&&\qw       &\ctrl{1}  &\qw                    &\ctrl{1} &\ctrl{1} & \qw  \\
\lstick{}& \ctrl{1}&\rstick{}\qw&\rstick{\,\,\,\,\,\equiv}&&&&\qw       &\ctrl{1}  &\qw                    &\ctrl{1} &\ctrl{2} &\qw   \\
\lstick{}& \ctrl{1}&\rstick{}\qw&                         &&&&\ctrl{1} &\targ      &\ctrl{1}              &\targ &\qw &\qw     \\
\lstick{}& \gate{X}   &\rstick{}\qw&                         &&&&\gate{\sqrt X} &\qw        &\gate{\sqrt X^\dagger}&\qw&\gate{\sqrt X}    & \rstick{.}\qw \\
}\vspace{0.1cm}
\]
After this decomposition, we can use the last qubit (the target qubit) as a work qubit in the decomposition of the multi-controlled NOT gates $C^{n-1}(X)$ into $O(n)$ Toffoli gates. The decomposition of $C^{n-1}\big(\sqrt{X}\big)$ is obtained by a recursive method, in which $C^{n-1}\big(\sqrt{X}\big)$ is broken down into $C^{n-2}\big(X^{1/4}\big)$, $C^{n-2}(X)$, $C\big(X^{1/4}\big)$, and $C\big(X^{1/4}\big)^\dagger$. After completing $(n-1)$ recursive steps, the process yields $O(n^2)$ Toffoli gates and $O(n)$ controlled gates.

Eliminating all ancillas comes at a significant cost with this method. With one ancilla, the decomposition of a multi-controlled gate $C^n(U)$ into universal gates requires $O(n)$ gates; without ancillas, the cost increases to $O(n^2)$. Using alternative techniques, Refs.~\cite{SP22,ICKHC16,SP13} give an $O(n)$-depth decomposition for $C^n(X)$ or $C^n(U)$, while Ref.~\cite{CZFDPP24,KG25} provides a polylogarithm-depth decomposition with and without auxiliary qubits.


\begin{exercise}\label{exe:sqrt-X}
Show that
\[
\sqrt{R_x(\theta)} = R_x\!\left(\frac{\theta}{2}\right),
\]
and then show that
\begin{align*}
\sqrt X &= \e^{\frac{\pi\ii}{4}} \, R_x\!\left(\frac{\pi}{2}\right)
= \e^{\frac{\pi\ii}{4}} \, H R_z\!\left(\frac{\pi}{2}\right) H
= H S H,\\
\sqrt X^\dagger &= \e^{-\frac{\pi\ii}{4}} \, R_x\!\left(-\frac{\pi}{2}\right)
= \e^{-\frac{\pi\ii}{4}} \, H R_z\!\left(-\frac{\pi}{2}\right) H
= H S^\dagger H.
\end{align*}
Show that this construction extends to
\[
X^{1/2^n}
= \e^{\,{i\pi/2^{\,n+1}}}\,
H R_z\!\left(\frac{\pi}{2^n}\right) H,
\]
where $n$ is a natural number. Use this result to obtain the decomposition of $X^{1/4}$ in terms of $H$ and $T$.
\end{exercise}

\begin{exercise}
To compute the square root of an arbitrary one-qubit gate $U$, we use the concepts of eigenvalues and eigenvectors, which are also essential in the Kitaev algorithm discussed later in this work. The spectral theorem (see Theorem~2.1 of~\cite{NC00}) for two-dimensional matrices states that for any one-qubit unitary operator there exists an orthonormal basis $\{\ket{\psi_1}, \ket{\psi_2}\}$ such that
\[
U = \lambda_1 \ket{\psi_1}\bra{\psi_1} + \lambda_2 \ket{\psi_2}\bra{\psi_2},
\]
where $\lambda_1,\lambda_2$ are the eigenvalues and $\ket{\psi_1},\ket{\psi_2}$ are the corresponding eigenvectors.
Show that
\[
\sqrt U = \sqrt{\lambda_1}\ket{\psi_1}\bra{\psi_1} + \sqrt{\lambda_2}\ket{\psi_2}\bra{\psi_2}.
\]
Find the eigenvalues and an orthonormal basis of eigenvectors of the gate $X$. That is, find $\lambda_1,\lambda_2$ and $\{\ket{\psi_1}, \ket{\psi_2}\}$ (each with unit norm) such that $X\ket{\psi_1} = \lambda_1\ket{\psi_1}$ and $X\ket{\psi_2} = \lambda_2\ket{\psi_2}$. Check that $\braket{\psi_1}{\psi_2}=0$. Then compute $\sqrt{X}$ and compare it with the result of Exercise~\ref{exe:sqrt-X}.
\end{exercise}

\begin{exercise}
Find the simplest circuits that implement $C\!\big(\sqrt{X}\big)$ and $C\!\big(\sqrt{X}^\dagger\big)$ using CNOT gates and one-qubit gates. Then find a decomposition of $C\!\big(X^{1/2^n}\big)$, where $n$ is a natural number.
\end{exercise}


\section{Final remarks}

In this Chapter, we reviewed fundamental concepts of linear algebra and quantum computation with the goal to prepare us for a better understanding of the basic quantum algorithms. We defined the concepts of superposition and entanglement, which are key resources used by quantum algorithms to achieve greater efficiency compared to classical algorithms. Quantum circuits were the main focus of this Chapter. We began by describing gates involving a few qubits and then presented the most important multiqubit gates. These gates can be used to implement any Boolean function on a quantum computer, which is necessary for many algorithms such as Deutsch-Jozsa, Bernstein-Vazirani, and Simon. The final Section discussed the decomposition of unitary operators into basic gates, which is part of the quantum compilation process. We emphasized the decomposition of multi-controlled gates, as they are central to building quantum circuits. The material covered in this Chapter provides a solid foundation for understanding the basic quantum algorithms.

\chapter{Deutsch's Algorithm}\label{chap:Deutsch}

Deutsch's algorithm is the first algorithm to exploit quantum parallelism. It uses two qubits (only one in the economical version) and has a modest gain, but it has inspired the development of several new quantum algorithms that are more efficient than their classical counterparts. Deutsch's problem was posed in 1985~\cite{Deu85} without yet using the quantum circuit model. The concepts of universal gates and quantum circuits were first presented in~\cite{Deu89}, approximately four years later. A generalization of Deutsch's algorithm was described in~\cite{DJ92} and is known as the Deutsch-Jozsa algorithm (see the next Chapter). The version described here follows the modern view of Deutsch's algorithm~\cite{CEMM98,NC00}, which differs slightly from the original one. In the final Section, we describe an implementation of Deutsch's algorithm with only one qubit.

\section{Problem formulation}

Suppose we have a 1-bit Boolean function $f:\{0,1\}\longrightarrow\{0,1\}$ without knowing the details of the implementation of $f$. We want to determine whether this function is \textit{balanced} or \textit{constant}. A 1-bit Boolean function is balanced if $f(0)\neq f(1)$; otherwise, the function is constant, in which case $f(0)=f(1)$. There are four Boolean functions $f$, whose truth tables are described in Fig.~\ref{fig:truthtablesf0-f3} with names $f_0$ to $f_3$.

\begin{figure}[!ht]\label{fig:truthtablesf0-f3}
\begin{center}
\begin{tabular}{c|c|c|c|c}
$x$ & $f_0(x)$ & $f_1(x)$ & $f_2(x)$ & $f_3(x)$\\
\hline
0 & 0 & 0 & 1 & 1\\
1 & 0 & 1 & 0 & 1
\end{tabular}\vspace{-13pt}
\end{center}
\caption{Truth tables of all 1-bit Boolean functions.}
\end{figure}

\noindent
The \textit{disjunctive normal forms} are
\begin{align*}
 &f_0(x)=0,\\
 &f_1(x)=x,\\
 &f_2(x)=\bar x,\\
 &f_3(x)=\bar x\vee x,
\end{align*}
where $\bar x=\text{NOT}\, x$. The Boolean expression of $f_3$ can be simplified since $\bar x\vee x=1$.

A classical algorithm that finds the solution to this problem needs to evaluate $f$ twice, meaning it must evaluate both $f(0)$ and $f(1)$. However, Deutsch's algorithm uses a unitary operator $U_f$ that implements $f$ and calls this operator only once. In this case, $f(0)$ and $f(1)$ are also evaluated, but the difference is that the evaluations are performed simultaneously. This idea is used repeatedly in quantum algorithms.

In the quantum case, $f$ is implemented through a two-qubit unitary operator $U_f$ defined as
\[
U_f\ket{x}\ket{j}=\ket{x}\ket{j\oplus f(x)},
\]
where $\oplus$ is the XOR operation or addition modulo 2. This is a recipe that can be used to implement an arbitrary $n$-bit Boolean function. We must ensure that $x$ is the input to the first register, and to obtain $f(x)$, we set $j=0$ as the input to the second register and then we look at the output of the second register. Now we use the technique described in Chapter~\ref{chap:2} to obtain the quantum circuit of functions $f_0$ to $f_3$. We use their disjunctive normal forms and we have to use one multi-controlled NOT gate for each output 1 in the truth table. In the case of two qubits, the multi-controlled NOT gate reduces to a CNOT activated by either 0 or 1. There is no output 1 in the truth table of $f_0$. Then,
$$U_{f_0}=I\otimes I.$$
There is one output 1 in the truth table of $f_1$, which corresponds to input 1. We use the standard CNOT, which is active when the control is set to 1. Then,
$$U_{f_1}=\text{CNOT}.$$
There is one output 1 in the truth table of $f_2$, which corresponds to input 0. We use the CNOT that is active when the control is set to 0. Then,
$$U_{f_2}=(X\otimes I)\cdot \text{CNOT} \cdot (X\otimes I).$$
Finally, there are two outputs 1 in the truth table of $f_3$ with inputs 0 and 1. We use the standard CNOT and the CNOT that is active when the control is set to 0. Then,
$$U_{f_3}=(X\otimes I)\cdot \text{CNOT} \cdot (X\otimes I) \cdot \text{CNOT}.$$
It is evident that using the general recipe to implement Boolean functions based on their truth tables can produce unnecessarily large circuits. In such cases, it is necessary to simplify the Boolean expression before building the circuit. At this point, there is no specific recipe to guide us, except for our skill in handling Boolean functions and designing circuits. For $f_3$, we know that $f_3(x)=1$ and the output must be 1 regardless of the input $x$ to the first qubit. Since the input to the second qubit is 0, the output 1 is obtained by using an $X$ gate. Then, the simplified version of $U_{f_3}$ is
$$U_{f_3}=I\otimes X.$$
Note that $U_f$ is unitary in all cases.

\section{The algorithm}

\begin{algorithm}[!ht]
\caption {Deutsch's algorithm} \label{algo_Deu}
\KwIn{A Boolean function $f:\{0,1\}\longrightarrow \{0,1\}$.}
\KwOut{$0$ if $f$ is constant, 1 if $f$ is balanced.} \BlankLine
Prepare the initial state $\ket{0}\ket{1}$\;
Apply $H\otimes H$\;
Apply $U_f$\;
Apply $H\otimes H$\;
Measure the first qubit in the computational basis.
\end{algorithm}

Deutsch's algorithm is described in Algorithm~\ref{algo_Deu} and the (non-economical) circuit is
\vspace* {3pt}
\[
\Qcircuit @C=2.0em @R=1.6em {
\lstick{\ket{0}}        & \gate{H} & \multigate{1}{\,\,\,U_f\,\,\,} & \gate{H} & \meter  & \rstick{f(0)\oplus f(1)} \cw \\
\lstick{\ket{1}}        & \gate{H} & \ghost{\,\,\,U_f\,\,\,}        & \gate{H}&   \rstick{\ket{1}.} \qw \\
\ustick{\hspace{0.7cm}\ket{\psi_0}} & \ustick{\hspace{1.4cm}\ket{\psi_1}}& \ustick{\hspace{2.0cm}\ket{\psi_2}} & \ustick{\hspace{1.4cm}\ket{\psi_3}}
}\vspace{0.0cm}
\]
We easily check that the circuit corresponds exactly to the steps of Algorithm~\ref{algo_Deu}. The output $f(0)\oplus f(1)$ is $0$ if $f$ is constant, and 1 if $f$ is balanced. Note that the last Hadamard gate applied to the second qubit can be eliminated without affecting the algorithm. This gate is included here because the central part of the circuit is symmetric and the analysis of the algorithm is neater with it than without it. The states at the bottom of the circuit are used in the analysis of the algorithm.

\section{Analysis of the algorithm}

After the first step, the state of the qubits is
\[
\ket{\psi_0}=\ket{0}\otimes\ket{1}.
\]
After the second step, the state of the qubits is
\begin{eqnarray*}
\ket{\psi_1} &=& (H\ket{0})\otimes(H\ket{1})\\
&=&\frac{\ket{0}+\ket{1}}{\sqrt 2}\otimes \ket{-},
\end{eqnarray*}
where $\ket{-}=(\ket{0}-\ket{1})/\sqrt 2$. After the third step, the state of the qubits is
\begin{eqnarray*}
\ket{\psi_2} &=& U_f\ket{\psi_1}\\
&=&\frac{U_f\ket{0}\ket{-}+U_f\ket{1}\ket{-}}{\sqrt 2}.
\end{eqnarray*}
To simplify $\ket{\psi_2}$, let us show the following proposition:
\begin{proposition}\label{prop:U_fxminus}
Let $f:\{0,1\}^n\longrightarrow\{0,1\}$ be a $n$-bit Boolean function and define $U_f$ as
\[
U_f=\sum_{x\in \{0,1\}^n}\sum_{j=0}^1 \, \ket{x,j\oplus f(x)}\bra{x,j}.
\]
Then
\[
U_f\big(\ket{x}\otimes\ket{-}\big) \,=\,(-1)^{f(x)}\ket{x}\otimes\ket{-}.
\]
\end{proposition}
\begin{proof}
Using the definition of $\ket{-}$, we obtain
\[
U_f\big(\ket{x}\otimes\ket{-}\big) \,=\, \frac{U_f\ket{x}\ket{0}-U_f\ket{x}\ket{1}}{\sqrt 2}.
\]
Using the definition of $U_f$, we obtain
\[
U_f\big(\ket{x}\otimes\ket{-}\big) \,=\, \frac{\ket{x}\ket{f(x)}-\ket{x}\ket{1\oplus f(x)}}{\sqrt 2}.
\]
If $f(x)=0$, the right-hand side is $\ket{x}\ket{-}$ and if $f(x)=1$, the right-hand side is $-\ket{x}\ket{-}$. Then, joining these results together, we have $(-1)^{f(x)}\ket{x}\otimes\ket{-}$.
\end{proof}

Using the proposition above, $\ket{\psi_2}$ simplifies to
\begin{eqnarray*}
\ket{\psi_2}
&=&\frac{(-1)^{f(0)}\ket{0}+(-1)^{f(1)}\ket{1}}{\sqrt 2}\otimes \ket{ -}\\
&=& \begin{cases}
\pm\,\ket{+}\otimes\ket{-}, & \text{if }f(0)=f(1),\\
\pm\,\ket{-}\otimes\ket{-}, & \text{if }f(0)\neq f(1).
\end{cases}
\end{eqnarray*}
In the last passage, we have simplified the state of the first qubit without specifying exactly what is the sign of the amplitude. This sign has no effect on the output of the algorithm.
After the fourth step, the state of the qubits is
\begin{eqnarray*}
\ket{\psi_3} &=& (H\otimes H)\ket{\psi_2}\\
&=& \begin{cases}
\pm\,\ket{0}\otimes\ket{1}, & \text{if }f(0)=f(1),\\
\pm\,\ket{1}\otimes\ket{1}, & \text{if }f(0)\neq f(1),
\end{cases}
\end{eqnarray*}
where we have used the fact that $H\ket{+}=\ket{0}$ and $H\ket{-}=\ket{1}$, which can be deduced using $\ket{+}=H\ket{0}$, $\ket{-}=H\ket{1}$, and $H^2=I$. The state $\ket{\psi_3}$ can be further simplified to
\begin{eqnarray*}
\ket{\psi_3} &=& \pm \,\ket{f(0)\oplus f(1)}\otimes\ket{1}.
\end{eqnarray*}
After the fifth step, the measurement of the first qubit in the computational basis returns $f(0)\oplus f(1)$, which is $0$ if $f$ is constant and 1 if $f$ is balanced, concluding the analysis of the algorithm.

Note that the sign of $\ket{\psi_3}$ has no influence on the measurement result because the probability of obtaining $f(0)\oplus f(1)$ is $|\pm 1|^2=1.$ It is easy to check, if relevant, that this sign is $(-1)^{f(0)}$.

\section{Analysis of the entanglement}

Summarizing the set of states of the qubits after each step, we have
\begin{align*}
\ket{\psi_0}&=\ket{0}\otimes\ket{1},\\
\ket{\psi_1}&=\ket{+}\otimes\ket{-},\\
\ket{\psi_2}&=\begin{cases}
\pm\,\ket{+}\otimes\ket{-}, & \text{if }f(0)=f(1),\\
\pm\,\ket{-}\otimes\ket{-}, & \text{if }f(0)\neq f(1),
\end{cases}\\
\ket{\psi_3}&=\pm \,\ket{f(0)\oplus f(1)}\otimes\ket{1}.
\end{align*}
Regardless of $f$, the qubits are unentangled during the execution of Deutsch's algorithm because each state $\ket{\psi_i}$ for $i$ from 1 to 3 is a tensor product of single-qubit pure states. The qubits remain unentangled throughout the algorithm. This means that Deutsch's algorithm is faster than classical algorithms making use of quantum parallelism only.

There is an alternate route to analyze entanglement. No two-qubit operator $A\otimes B$ creates or destroys entanglement. In Deutsch's algorithm, only CNOT can create entanglement and it is used at most one time. Then, for each function $f$ we can simplify the whole algorithm in order to obtain only one final operator, whose input is $\ket{01}$. If we simplify the algorithm for each $f$, we obtain
\begin{align*}
(H\otimes H)\, U_{f_0} (H\otimes H) &= (H\otimes H)\cdot (H\otimes H) = I\otimes I  ,\\
(H\otimes H)\, U_{f_1} (H\otimes H) &= (H\otimes H)\cdot\text{CNOT}\cdot (H\otimes H) = \text{CNOT}_{10}  ,\\
(H\otimes H)\, U_{f_2} (H\otimes H) &= (Z\otimes I)\cdot \text{CNOT}_{10}\cdot (Z\otimes I),\\
(H\otimes H)\, U_{f_3} (H\otimes H) &=  I\otimes Z,
\end{align*}
where the control of $\text{CNOT}_{10}$ is the second qubit and the target is the first qubit. We conclude right away that no entanglement is created when $f$ is $f_0$ or $f_3$. For $f_1$ and $f_2$, the CNOT's control qubit is set to 1 because the input is $\ket{01}$, in these cases, no entanglement is created because CNOT can create entanglement only if the control qubit is in a superposition state, although this condition alone is not sufficient.

We can take advantage of the simplified version of the algorithm to check the output again and reanalyze the algorithm. The states of the qubits just before the measurement are
\begin{align*}
\ket{\psi_3}\big|_{f_0} &= (I\otimes I)\ket{0,1}=\ket{0,1},\\
\ket{\psi_3}\big|_{f_1} &= \text{CNOT}_{10}\ket{0,1} =\ket{1,1},\\
\ket{\psi_3}\big|_{f_2} &= (Z\otimes I)\cdot \text{CNOT}_{10}\cdot (Z\otimes I)\ket{0,1} =-\ket{1,1},\\
\ket{\psi_3}\big|_{f_3} &= (I\otimes Z)\ket{0,1} =-\ket{0,1}.
\end{align*}
The output of a measurement of the first qubit is 0 for $f_0$ and $f_3$, which are the constant functions, and 1 for $f_1$ and $f_2$, which are the balanced functions. We can also check that the sign is $(-1)^{f(0)}$ because
$(-1)^{f_0(0)}=(-1)^{f_1(0)}=+1$ and
$(-1)^{f_2(0)}=(-1)^{f_3(0)}=-1$.

\section{Who implements the oracle?}

Deutsch's algorithm is said to be more efficient than its classical counterpart in a limited context known as  ``query complexity''~\cite{KLM07}. In this context, we must compare the number of evaluations of $f$ in the classical case with the number of times $U_f$ is used in the quantum case. This is a rule of the game. Deutsch's algorithm applies $U_f$ only once, whereas the classical algorithm needs to evaluate $f$ twice. Therefore, Deutsch's is faster.

The analysis of Deutsch's algorithm shows that $f$ is evaluated at two distinct points in the domain simultaneously. This is not possible using a sequential classical algorithm. However, this can be carried out on a classical computer with parallel processors if the number of simultaneous threads does not scale up. Since Deutsch's algorithm uses a fixed number of qubits, we cannot scale it up and it is not possible to perform an asymptotic analysis. In terms of time complexity, the quantum version has no gain. The importance of Deutsch's algorithm lies in the fact that it stimulated the search for generalizations, such as the Deutsch-Jozsa, Bernstein-Vazirani, and Simon algorithms, which inspired Shor in the development of a quantum algorithm for factoring composite integers and a quantum algorithm for calculating discrete logarithms.

Lastly, it is important to note that implementing the oracle is not our responsibility. It is someone else's job. Without this understanding, we face a contradiction—we need to know the answer before even starting to implement the algorithm that will find the answer. We are allowed to query function $f$ implemented by someone else without looking at its implementation details. We are given what is called a \textit{black box} quantum computer with $U_f$ already implemented, which we can use more than once, and we may add new gates, but cannot examine the inner workings of the black box that implements $U_f$.

\section{Economical circuit of Deutsch's algorithm}

Deutsch's algorithm can be implemented with only one qubit in the following way:
\vspace* {3pt}
\[
\Qcircuit @C=2.0em @R=1.6em {
\lstick{\ket{0}}        & \gate{H} & \gate{\,\,\,{U'_f}^{\mbox{}}\,\,\,} & \gate{H} & \meter  & \rstick{f(0)\oplus f(1),} \cw
}\vspace{1pt}
\]
where
\[
U'_f = \sum_{x=0}^1 (-1)^{f(x)}\ket{x}\bra{x}.
\]
From Proposition~\ref{prop:U_fxminus}, we see that the second qubit is not necessary for the algorithm, although it is necessary if we want to obtain $f(x)$ in the classical sense, as shown below. By expanding the sum, we obtain
\[
U'_f = \left[\begin{array}{cc}
(-1)^{f(0)} & 0 \\
0 & (-1)^{f(1)}
\end{array}\right].
\]
Then,
\[
U'_f=\begin{cases} \pm I, \text{ if }f(0)=f(1), \\ \pm Z, \text{ if }f(0)\neq f(1). \end{cases}
\]

The analysis of the algorithm reduces to calculating
\begin{eqnarray*}
HU'_fH \ket{0}= \begin{cases}
\pm\,I\ket{0}, & \text{if }f(0)=f(1),\\
\pm\,X\ket{0}, & \text{if }f(0)\neq f(1).
\end{cases}
\end{eqnarray*}
Using that $f(0)\oplus f(1)=0$ if $f(0)=f(1)$, and $f(0)\oplus f(1)=1$ if $f(0)\neq f(1)$, we obtain
\[
HU'_fH \ket{0}= \pm \ket{f(0)\oplus f(1)}.
\]
After a measurement, the output is $f(0)\oplus f(1)$ with probability $|\pm 1|^2=1$.

Now it is straightforward to check that there is no entanglement in Deutsch's algorithm because a qubit cannot entangle with itself. Entanglement requires at least two qubits.

\subsection*{Consulting the oracle}

In the non-economical circuit, if the input to the first qubit of $U_f$ is $\ket{x}$, $U_f$ returns $f(x)$ when we perform a measurement of the second qubit. In the economical circuit, we obtain $f(x)$ in the classical sense using the circuit\vspace* {3pt}
\[
\Qcircuit @C=2.0em @R=1.6em {
\lstick{\ket{x}}        & \qw & \gate{U'_f}        & \qw&   \rstick{\ket{x}} \qw \\
\lstick{\ket{0}}        & \gate{H} & \ctrl{-1} & \gate{H} & \meter  & \rstick{f(x).} \cw
}\vspace{5pt}
\]
Indeed, the steps of this circuit are\vspace{1pt}
\begin{align*}
\ket{x}\ket{0} \xrightarrow[\text{}]{I\otimes H} \ket{x}\ket{+} \xrightarrow[\text{}]{\fbox{$U'_f$}-\bullet} \frac{\ket{x}\ket{0}+(-1)^{f(x)}\ket{x}\ket{1}}{\sqrt 2} \xrightarrow[\text{}]{I\otimes H} \ket{x}\ket{f(x)}.
\end{align*}
To calculate the last step, we use that $f(x)$ is either 0 or 1 for a fixed $x$. First, suppose that $f(x)=0$; then $I\otimes H$ is applied to $\ket{x}\ket{+}$, resulting in $\ket{x}\ket{f(x)}$. Second, suppose that $f(x)=1$; then $I\otimes H$ is applied to $\ket{x}\ket{-}$, resulting in $\ket{x}\ket{f(x)}$. After measuring the second qubit in the computational basis, we obtain $f(x)$ with probability 1.

\chapter{Deutsch-Jozsa Algorithm}\label{chap:DJ}

The Deutsch-Jozsa algorithm is a deterministic quantum algorithm, a generalization of Deutsch's algorithm, and the first example that is exponentially faster than its equivalent classical deterministic algorithm. It was published in 1992~\cite{DJ92} and revisited in 1998~\cite{CEMM98}. Many books~\cite{KLM07,LR14,MM04,Mer07,NC00,Wil08} and papers~\cite{KMR06,MJM15,QZ18,QZ20} have reviewed and generalized this algorithm.

\section{Problem formulation}

Let $f:\{0,1\}^n \longrightarrow \{0,1\}$ be a $n$-bit Boolean function, $n\ge 2$, with the following property: $f$ is either balanced or constant. A Boolean function is balanced if the inverse image of point 0 has cardinality $2^{n-1}$, and it is constant if the inverse image of point 0 has cardinality 0 or $2^{n}$. Suppose we are able to evaluate this function at any point in the domain; however, we do not have access to the details of the implementation of $f$, that is, we are given a \textit{black box} quantum computer with $f$ already implemented. The Deutsch-Jozsa algorithm solves the following problem: Determine whether $f$ is balanced or constant using this black box quantum computer.

There are only two constant functions, namely $f(x)=0$ and $f(x)=1$ for every $n$-bit string $x$, but there are many balanced functions. The best classical deterministic algorithm that solves this problem, given a black box function $f$ that is either balanced or constant implemented on a $n$-bit classical computer, is the following: Evaluate $f$ at $2^{n-1}+1$ distinct points in its domain and check whether the output is always the same ($f$ is constant) or not ($f$ is balanced).

The Deutsch-Jozsa algorithm, on the other hand, is a deterministic quantum algorithm that uses a black box unitary operator $U_f$ only once. The action of $U_f$ on the computational basis is
\[
U_f\ket{x}\ket{j}=\ket{x}\ket{j\oplus f(x)},
\]
where $x\in\{0,1\}^n$ and $j\in\{0,1\}$. The qubits are split into two quantum registers with sizes $n$ and 1.\footnote{A quantum register is a set of qubits.} After creating a superposition of the vectors $\ket{x}$ for all $x$, the Deutsch-Jozsa algorithm is able to compute $f(x)$ for all $x$ with a single application of $U_f$ and, after quick post-processing, determine whether $f$ is constant or balanced.

In the next proposition, we show that $U_f$ is unitary. Then, since each entry of $U_f$ is either 0 or 1, $U_f$ is a $2^{n+1}$-dimensional permutation matrix.\footnote{A \textit{permutation matrix} is a square binary matrix such that each row and each column has exactly one entry equal to 1 and zeroes elsewhere.}
\begin{proposition} \label{prop:U_f_is_unitary}
$U_f$ is unitary for any $n$-bit Boolean
 function.
\end{proposition}
\begin{proof}
Let us show that $U_f^\dagger U_f=I$. Using the definition of $U_f$, we have
\begin{align*}
\bra{x'}\bra{j'}U_f^\dagger U_f\ket{x}\ket{j} &= \braket{x'}{x}\,\braket{j'\oplus f(x')}{j\oplus f(x)}\\
&= \delta_{xx'}\,\braket{j'\oplus f(x')}{j\oplus f(x)}.
\end{align*}
Using that $\delta_{xx'}\neq 0$ only if $x=x'$, we have $\delta_{xx'}\braket{j'\oplus f(x')}{j\oplus f(x)}=\delta_{xx'}\braket{j'}{j}$. Then
\begin{align*}
\bra{x'}\bra{j'}U_f^\dagger U_f\ket{x}\ket{j} &= \delta_{xx'}\delta_{jj'}.
\end{align*}
Since $j,j'$ are arbitrary bits and $x,x'$ are arbitrary $n$-bit strings, the proof is complete.
\end{proof}

\section{The algorithm}

\begin{algorithm}[!ht]
\caption {Deutsch-Jozsa algorithm} \label{algo_DeuJoz}
\KwIn{A black box $U_f$ implementing a $n$-bit Boolean function $f:\{0,1\}^n\longrightarrow \{0,1\}$, which is either balanced or constant.}
\KwOut{$0$ if $f$ is constant; otherwise, $f$ is balanced.} \BlankLine
Prepare the initial state $\ket{0}^{\otimes n}\ket{1}$\;
Apply $H^{\otimes(n+1)}$\;
Apply $U_f$\;
Apply $H^{\otimes(n+1)}$\;
Measure the first register in the computational basis.
\end{algorithm}

The Deutsch-Jozsa algorithm is described in Algorithm~\ref{algo_DeuJoz} and the $(n+1)$-qubit circuit is
\vspace* {3pt}
\[
\Qcircuit @C=2.0em @R=1.6em {
\lstick{\ket{0}}        & \gate{H} & \multigate{3}{\,\,\,U_f\,\,\,} & \gate{H} & \meter  & \rstick{0 \text{ or } 1} \cw \\
\lstick{\vdots}         & \vdots   &                                & \vdots   &  \vdots  & \rstick{}  \\
\lstick{\ket{0}}        & \gate{H} & \ghost{\,\,\,U_f\,\,\,}        & \gate{H} & \meter  & \rstick{0 \text{ or } 1} \cw \\
\lstick{\ket{1}}        & \gate{H} & \ghost{\,\,\,U_f\,\,\,}        & \gate{H} & \rstick{\ket{1}.} \qw \\
\ustick{\hspace{0.7cm}\ket{\psi_0}} & \ustick{\hspace{1.4cm}\ket{\psi_1}}& \ustick{\hspace{2.0cm}\ket{\psi_2}} & \ustick{\hspace{1.4cm}\ket{\psi_3}}
}\vspace{0pt}
\]
We easily check that the circuit corresponds exactly to the steps of Algorithm~\ref{algo_DeuJoz}. The output is the $n$-bit string $0$ if $f$ is constant, and different from $0$ if $f$ is balanced. Note that the last Hadamard gate applied to the last qubit can be eliminated without affecting the algorithm. This gate is included here because the central part of the circuit is symmetric and the analysis of the algorithm is neater with it than without it. The states at the bottom of the circuit are used in the analysis of the algorithm.

\section{Analysis of the algorithm}

After the first step, the state of the qubits is
\[
\ket{\psi_0}=\ket{0}^{\otimes n}\otimes\ket{1}.
\]
After the second step, the state of the qubits is
\begin{eqnarray*}
\ket{\psi_1} &=& (H\ket{0})^{\otimes n}\otimes(H\ket{1})\\
&=&\frac{1}{\sqrt{2^n}}\sum_{x=0}^{2^n-1}\ket{x}\ket{-},
\end{eqnarray*}
where $\ket{-}=(\ket{0}-\ket{1})/\sqrt 2$. After the third step, the state of the qubits is
\begin{eqnarray*}
\ket{\psi_2} &=& U_f\ket{\psi_1}\\
&=&\frac{1}{\sqrt{2^n}}\sum_{x=0}^{2^n-1}U_f\big(\ket{x}\ket{-}\big).
\end{eqnarray*}
To simplify $\ket{\psi_2}$, we use Proposition~\ref{prop:U_fxminus} on Page~\pageref{prop:U_fxminus}, which states that
\[
U_f\big(\ket{x}\ket{-}\big) \,=\, (-1)^{f(x)}\ket{x}\ket{-}.
\]
We obtain
\[
\ket{\psi_2} \,=\,\frac{1}{\sqrt{2^n}}\sum_{x=0}^{2^n-1} (-1)^{f(x)}\ket{x}\ket{-}.
\]
After the fourth step, the state of the qubits is
\begin{eqnarray*}
\ket{\psi_3} &=& H^{\otimes (n+1)}\ket{\psi_2}\\
&=&  \frac{1}{\sqrt{2^n}}\sum_{x=0}^{2^n-1} (-1)^{f(x)} \big(H^{\otimes n}\ket{x}\big)\otimes \big(H\ket{-}\big).
\end{eqnarray*}
To simplify $\ket{\psi_3}$, let us show the following proposition:
\begin{proposition}\label{prop:Hnx}
Let $x\in\{0,1\}^n$ be a $n$-bit string $x_0\cdots x_{n-1}$.
Then
\[
H^{\otimes n}\ket{x} \,=\, \frac{1}{\sqrt{2^n}}\sum_{y=0}^{2^n-1}(-1)^{x\cdot y}\ket{y},
\]
where $x\cdot y=x_{0}y_{0}+  \cdots + x_{{n-1}}y_{{n-1}}\mod 2$.
\end{proposition}
\begin{proof}
Using that $x=(x_0\cdots x_{n-1})_2$,
we obtain
\begin{eqnarray*}
H^{\otimes n}\ket{x} &=& \big(H\ket{x_0}\big)\otimes \cdots \otimes \big(H\ket{x_{n-1}}\big) \\
&=&  \Big(\frac{1}{\sqrt{2}}\sum_{y_0=0}^{1} (-1)^{x_0y_0} \ket{y_0}\Big)\otimes \cdots \otimes \Big(\frac{1}{\sqrt{2}}\sum_{y_{n-1}=0}^{1} (-1)^{x_{n-1}y_{n-1}} \ket{y_{n-1}}\Big).
\end{eqnarray*}
Putting all sums at the beginning, we obtain
\begin{eqnarray*}
H^{\otimes n}\ket{x} &=& \frac{1}{\sqrt{2^n}}\sum_{y_0=0}^{1}\,\,\cdots \sum_{y_{n-1}=0}^{1} (-1)^{\left(x_0y_0+\cdots+x_{n-1}y_{n-1}\right)} \ket{y_0}\otimes \cdots \otimes \ket{y_{n-1}}.
\end{eqnarray*}
Using the definition of $x\cdot y$ and converting to the decimal notation, we complete the proof.
\end{proof}

Using the proposition above, $\ket{\psi_3}$ simplifies to
\begin{equation*}
\ket{\psi_3}\,=\,\frac{1}{{2^n}}\sum_{y=0}^{2^n-1}\left(\sum_{x=0}^{2^n-1}(-1)^{x\cdot y+f(x)}\right)\ket{y}\otimes\ket{1}.
\end{equation*}
The amplitude of state $\ket{0}\ket{1}$ (0 in decimal) is
\begin{equation*}
\frac{1}{{2^n}}\sum_{x=0}^{2^n-1}(-1)^{f(x)}.
\end{equation*}
The probability that a measurement of the first register returns $y=0$ (in decimal) is
\begin{equation*}
p(0)\,=\,\left|\frac{1}{{2^n}}\sum_{x=0}^{2^n-1}(-1)^{f(x)}\right|^2.
\end{equation*}
If $f$ is constant, $p(0)=1$ and we know with certainty that the output is $y=0$. If $f$ is balanced, $p(0)=0$ and we know with certainty that the output is $y\neq 0$.

\section{Analysis of the entanglement}

There is no entanglement between the registers because each of the states $\ket{\psi_0}$ through $\ket{\psi_3}$ can be written either as $\ket{\psi}\otimes\ket{1}$ or $\ket{\psi}\otimes\ket{-}$, for some state $\ket{\psi}$. We have to check the first register only.
From the circuit of the algorithm, we realize that the only operator that creates or destroys entanglement is $U_f$. Then, it is enough to analyze
\[
\ket{\psi} \,=\,\frac{1}{\sqrt{2^n}}\sum_{x=0}^{2^n-1} (-1)^{f(x)}\ket{x},
\]
which is the state of the first register after applying $U_f$. We ask ourselves whether there are $a_i$ and $b_i$ so that
\[
\ket{\psi} \,=\,\frac{a_0\ket{0}+b_0\ket{1}}{\sqrt{2}}\otimes\cdots\otimes\frac{a_{n-1}\ket{0}+b_{n-1}\ket{1}}{\sqrt{2}},
\]
where each pair $(a_i,b_i)$ must obey $|a_i|^2+|b_i|^2=2$.  This is the only way of having no entanglement at all.
Equivalently, we ask whether the system of equations
\begin{align*}
a_0...a_{n-2}a_{n-1} &= (-1)^{f(0...00)}\\
a_0...a_{n-2}b_{n-1} &= (-1)^{f(0...01)}\\
a_0...b_{n-2}a_{n-1} &= (-1)^{f(0...10)}\\
a_0...b_{n-2}b_{n-1} &= (-1)^{f(0...11)}\\
&\vdots\\
b_0...b_{n-2}b_{n-1} &= (-1)^{f(1...11)}
\end{align*}
admits a solution or not. It is straightforward to check that $a_i\neq 0$ and $b_i\neq 0$ for all $i$. Besides, we need not worry about phases of $a_i$. In fact, without loss of generality, we consider $a_i$ real and positive because a global factor can be discarded. Then, by selecting equation
\begin{align*}
b_0a_1a_2...a_{n-2}a_{n-1} &= (-1)^{f(1...00)},
\end{align*}
we conclude that $b_0$ is also real. The same applies to the other $b_i$. Now, let us show that $a_i= 1$ and $b_i=\pm 1$. Let us start with $a_0$ and $b_0$ by selecting the following equations:
\begin{align*}
a_0a_1a_2...a_{n-2}a_{n-1} &= (-1)^{f(0...00)},\\
b_0a_1a_2...a_{n-2}a_{n-1} &= (-1)^{f(1...00)}.
\end{align*}
Dividing them, we obtain $a_0\pm b_0 =0$. This result together with the constraint $a_0^2+b_0^2=2$ implies that $a_0=1$ and $b_0=\pm 1$. The same applies to the other $a_i$ and $b_i$. Now, counting the number of unentangled states $\ket{\psi}$ with $a_i=1$ and $b_i=\pm 1$, we obtain at most $2^{n}$, up to a global sign. Or, at most $2^{n+1}$.

There are only two constant functions: $f(x)=0$ and $f(x)=1$ for all $x$, which correspond to $\ket{\psi}=\big(H\ket{0}\big)^{\otimes n}$ and $\ket{\psi}=-\big(H\ket{0}\big)^{\otimes n}$, respectively. In both cases, there is no entanglement. Next, let us count the number of balanced functions. Let $S=\{0,1\}^n$. The exact number of balanced functions is the number of subsets of $S$ with cardinality $2^{n-1}$ because as soon as we find such a subset $S'$, we define a balanced function by setting $f(x)=0$ if $x\in S'$ and $f(x)=1$ if $x\in S\setminus S'$. When we cover all such subsets, we have obtained all balanced functions. Since the domain of $f$ has cardinality $|S|=2^n$, the number of balanced functions is $\binom{2^n}{2^{n-1}}$.\footnote{An alternative way of counting the number of balanced functions is by considering the number of permutations of a list with $2^{n-1}$ zeros and $2^{n-1}$ ones, which yields $2^n!/(2^{n-1}!)^2$.}

The number of balanced functions grows faster than the number of states with no entanglement at all. Indeed, using the asymptotic approximation\footnote{\url{https://en.wikipedia.org/wiki/Binomial_coefficient}}
\[
\binom{2p}{p}\approx \frac{2^{2p}}{\sqrt{\pi p}}
\]
we obtain
\[
\binom{2^n}{2^{n-1}}\approx \frac{\sqrt{2}}{\sqrt{\pi}}\frac{2^{2^n}}{\sqrt{2^n}}.
\]
We conclude that when $n$ increases, there are more and more balanced functions $f$ with $U_f$ creating entanglement.

The discussion above shows that there exists an $n_0\ge 2$, such that for all $n\ge n_0$, there are balanced functions $f$ for which $U_f$ creates entanglement. Since, $\binom{2^n}{2^{n-1}}>2^{n+1}$ for $n\ge 3$, we can take $n_0=3$. The only remaining case that may have no entanglement is $n=2$. For $n=2$, there are 6 balanced functions, whose truth tables are described in Fig.~\ref{fig:truthtablesf0-f5} with function names $f_0$ to $f_5$.

\begin{figure}[!ht]\label{fig:truthtablesf0-f5}
\begin{center}
\begin{tabular}{c|c|c|c|c|c|c}
$x_0\,\,x_1$ & $f_0(x)$ & $f_1(x)$ & $f_2(x)$ & $f_3(x)$ & $f_4(x)$ & $f_5(x)$\\
\hline
0\,\,\,\,0 & 0 & 0 & 0 & 1 & 1 & 1\\
0\,\,\,\,1 & 0 & 1 & 1 & 0 & 0 & 1\\
1\,\,\,\,0 & 1 & 0 & 1 & 0 & 1 & 0\\
1\,\,\,\,1 & 1 & 1 & 0 & 1 & 0 & 0
\end{tabular}\vspace{-13pt}
\end{center}
\caption{Truth tables of all 2-bit balanced functions.}
\end{figure}

It is enough to analyze $f_0$ to $f_2$ because the other functions complement those ones. For instance, $f_5$ is the complement of $f_0$, which means that state $\ket{\psi}$ created by $U_{f_5}$ is equal to the state created by $U_{f_0}$ up to a global phase. States $\ket{\psi}$ that correspond to $f_0$ to $f_2$ (without normalization) are
\begin{align*}
\ket{00}+\ket{01}-\ket{10}-\ket{11} &= \big(\ket{0}-\ket{1}\big)\otimes\big(\ket{0}+\ket{1}\big),\\
\ket{00}-\ket{01}+\ket{10}-\ket{11} &= \big(\ket{0}+\ket{1}\big)\otimes\big(\ket{0}-\ket{1}\big),\\
\ket{00}-\ket{01}-\ket{10}+\ket{11} &= \big(\ket{0}-\ket{1}\big)\otimes\big(\ket{0}-\ket{1}\big),
\end{align*}
respectively. All those states are unentangled. We conclude that there is no entanglement in the Deutsch-Jozsa algorithm when $n=2$.

\begin{exercise} \label{exe:DJ-oracles} \mbox{}
\begin{enumerate}
\item[(a)] Show that the Boolean functions $f(x) = b\oplus  (s \cdot x)$, where $s$ is an $n$-bit string and $b$ is a bit, and
 \[
s \cdot x = s_0 x_0\oplus \ldots \oplus  s_{n-1} x_{n-1}= s_0 x_0 + \ldots + s_{n-1} x_{n-1} \mod 2,
\]
are either constant or balanced.
\item[(b)] Show that all Boolean functions described in Fig.~\ref{fig:truthtablesf0-f5} can be written as $f(x) = b\oplus (s \cdot x)$ for some $2$-bit string $s$ and some bit $b$.
\item[(c)] Find a 3-bit balanced function that cannot be written as $f(x) = b\oplus (s \cdot x) $ for any $3$-bit string $s$ and any bit $b$, and show that in this case there is entanglement in the Deutsch-Jozsa algorithm between at least two qubits of the first register.
\item[(d)] Show that there is no entanglement at all in the Deutsch-Jozsa algorithm if and only if the oracle is $f(x) = b\oplus (s \cdot x) $ for some $s \in \{0,1\}^n$ and $b \in \{0,1\}$.
\end{enumerate}
\end{exercise}

\section{Implementing the oracle}

In general, it is not up to us to implement the oracle, unless there is a formula for $f$ that does not reveal the solution beforehand. Without this understanding, we face a contradiction---we have to know the answer before starting to implement the algorithm that will find the answer. We are allowed to evaluate the function $f$ without looking at its implementation details. We are given what is called a \textit{black box} quantum computer with $U_f$ already implemented, which we can use and add new gates, but cannot see inside. In the classical case, we have to count the number of evaluations of $f$. In the quantum case, we have to count the number of applications of $U_f$. This is how we calculate the query complexity of oracle-based algorithms. Note that it does not matter whether the evaluation of $f$ is efficient or not.

With this understanding, let us show how to implement the oracle $U_f$ for the following entangling balanced function: $f(x)=0$ if $x\in\{000,010,100,101\}$ and $f(x)=1$ if $x\in\{001,011,110,111\}$. Using the disjunctive normal form, we add to the circuit one multi-controlled NOT gate for each point in $\{001,011,110,111\}$, as follows
\[
\Qcircuit @C=2.3em @R=1.5em {
\lstick{\ket{x_0}}    & \ctrlo{1} & \ctrlo{1} & \ctrl{1} & \ctrl{1} &  \rstick{\ket{x_0}} \qw \\
\lstick{\ket{x_1}}    & \ctrlo{1} & \ctrl{1} & \ctrl{1} & \ctrl{1} &  \rstick{\ket{x_1}} \qw \\
\lstick{\ket{x_2}}    & \ctrl{1} & \ctrl{1} & \ctrlo{1} & \ctrl{1} & \rstick{\ket{x_2}} \qw \\
\lstick{\ket{0}}  & \targ  & \targ  & \targ  & \targ  &  \rstick{\ket{f(x)}.}  \qw \\
}\vspace*{0.2cm}
\]
Since the first and second gates have opposite controls in the second qubit and identical controls in the first and third qubits (the third and fourth gates have a similar feature), this circuit can be simplified to
\[
\Qcircuit @C=2.3em @R=1.5em {
\lstick{\ket{x_0}}    & \ctrlo{2} & \ctrl{1} &  \rstick{\ket{x_0}} \qw \\
\lstick{\ket{x_1}}    & \qw & \ctrl{2} &  \rstick{\ket{x_1}} \qw \\
\lstick{\ket{x_2}}    & \ctrl{1} & \qw & \rstick{\ket{x_2}} \qw \\
\lstick{\ket{0}}  & \targ  & \targ  &  \rstick{\ket{f(x)}.}  \qw \\
}\vspace*{0.2cm}
\]

This example helps to show how to implement any balanced oracle using $n+1$ qubits. It is possible to implement the Deutsch-Jozsa algorithm with only $n$ qubits, but this discussion is postponed and fully addressed in Chapter~\ref{chap:Grover}.

\section{Final remarks}

There are efficient randomized classical algorithms that solve the Deutsch-Jozsa problem.
It is possible to find the correct solution with high probability by consulting the classical oracle a few times. More formally, consider the following randomized algorithm: (1)~Select uniformly at random $k\ge 2$ points $x_0$, ..., $x_{k-1}$ in the domain, repetitions are allowed, (2)~if $f(x_0)=\cdots=f(x_{k-1})$, return ``$f$ is constant''; otherwise, return ``$f$ is balanced''. This algorithm returns the output ``$f$ is balanced'' with certainty because as soon as we obtain two different values after evaluating $f$, we claim the promise that $f$ is either balanced or constant---it must be balanced. The probability that the output ``$f$ is constant'' is correct is $1-1/2^{k-1}$.\footnote{To calculate the probability that the output ``$f$ is constant'' is correct, we start by calculating the probability that ``$f$ is constant'' is wrong. The output ``$f$ is constant'' is wrong when $f$ is balanced and $f(x_0)=\cdots=f(x_{k-1})$. The probability that $f(x_0)=\cdots=f(x_{k-1})=0$ is $1/2^k$ because each evaluation is independent. Likewise, the probability that $f(x_0)=\cdots=f(x_{k-1})=1$ is $1/2^k$. Then, the probability that ``$f$ is constant'' is wrong is $2/2^k$ because all results equal to 0 and all results equal to 1 are mutually exclusive. Then, the probability that ``$f$ is constant'' is correct is $1-1/2^{k-1}$.} The success probability quickly tends to 1, for instance, take $k=10$, the success probability is at least 99.8\%.

The number of queries required to solve the Deutsch-Jozsa problem using deterministic classical algorithms is $\Omega(2^n)$, which means we need an exponential number of queries in the worst case. Using quantum algorithms or randomized classical algorithms with a small fixed error $\epsilon$, it is $O(1)$, which means the number of queries is fixed independently of the size of the problem.

We have shown that depending on the oracle, there is entanglement in the Deutsch-Jozsa algorithm. On the other hand, there is a restricted version of the Deutsch-Jozsa algorithm that has no entanglement at all. In this version, we have the promise that the oracle is $f(x) = (s \cdot x) \oplus b$ for some $s \in \{0,1\}^n$ and $b \in \{0,1\}$, and the goal is to determine whether $f$ is constant or balanced. Unfortunately, for this restricted version, there is an efficient deterministic classical algorithm.

\chapter{Bernstein-Vazirani Algorithm}\label{chap:BV}

The Bernstein-Vazirani algorithm was first presented at a conference in 1993~\cite{BV93}, and the full paper was published in 1997~\cite{BV97}. It was the first deterministic quantum algorithm to demonstrate a linear advantage over the best deterministic or randomized classical algorithm. The Bernstein-Vazirani algorithm exploits quantum parallelism but, interestingly, does not involve entanglement at all. This algorithm is described in several books~\cite{Mer07,NO08,RP11}.

\section{Problem formulation}

Let $s=s_0...s_{n-1}$ be an unknown $n$-bit string. Although we do not know $s$, we have at our disposal a Boolean function $f:\{0,1\}^n\longrightarrow\{0,1\}$ defined as
\[
f(x)=s\cdot x=s_0x_0+...+s_{n-1}x_{n-1}\mod 2,
\]
where $x_0,...,x_{n-1}$ are the bits of $x$. Note that $f$ is a linear function and each linear Boolean function is characterized by a specific hidden $s$. Our goal is to find $s$ by evaluating $f$ without knowing its implementation details. In quantum computing, the operator that implements this function is called an \textit{oracle}, because an oracle reveals $f(x)$ without showing $s$ explicitly, and $f(x)$ can be used to determine $s$. When we describe the circuit of the algorithm, $f$ is implemented by another person, because we do not know $s$. It is important to understand this, as it helps avoid the mistaken impression that we need to know the answer to find the answer, which is absurd.

In the classical version of this problem, we have to consult the \textit{classical oracle} at least $n$ times. In fact, we choose $x=10...0$ and ask the oracle what $f(10...0)$ is. The answer is $s_0$. Next we choose $x=010...0$ and ask the oracle what $f(010...0)$ is. The answer is $s_1$. The last query is $f(0...01)$, whose answer is $s_{n-1}$. This shows that we need to consult the oracle $n$ times and there is no way to reduce this number without introducing an error in the algorithm. In the quantum case, we consult the \textit{quantum oracle} only once, which allows us to find all bits of $s$, as described below.

In the quantum case, the function $f(x)=s\cdot x$ is implemented using the unitary operator $U_f$ of $n+1$ qubits, defined as
\[
U_f\ket{x}\ket{j}=\ket{x}\ket{j\oplus f(x)},
\]
where $x\in\{0,1\}^n$, $j$ is a bit, and $\oplus$ is the XOR operation or sum modulo 2. This operator uses two registers, with sizes $n$ and 1, respectively. We can use $U_f$ as many times as we wish. However, it is used only once in the Bernstein-Vazirani algorithm.

Summing up, the classical algorithm queries the classical oracle $n$ times using a $n$-bit classical computer. The quantum algorithm queries the quantum oracle only once using a $(n+1)$-qubit quantum computer. In the last Section of this Chapter, we show that the algorithm can be implemented on a $n$-qubit quantum computer.

\section{The algorithm}

\begin{algorithm} [!ht]
\caption{Bernstein-Vazirani algorithm} \label{algo_BV}
\KwIn{A Boolean function $f:\{0,1\}^n\longrightarrow \{0,1\}$ such that $f(x)=s\cdot x$.}
\KwOut{$s$ with probability equal to 1.} \BlankLine
Prepare the initial state $\ket{0}^{\otimes n}\ket{1}$\;
Apply $H^{\otimes (n+1)}$\;
Apply $U_f$\;
Apply $H^{\otimes (n+1)}$\;
Measure the first register in the computational basis.
\end{algorithm}

The Bernstein-Vazirani algorithm is described in Algorithm~\ref{algo_BV} and the circuit is \vspace{3pt}
\[
\Qcircuit @C=2.3em @R=0.9em {
\lstick{\ket{0}}        & \gate{H} & \multigate{4}{\,\,\,U_f\,\,\,} & \gate{H} & \meter  & \rstick{s_0} \cw \\
\lstick{\vdots \,\,\,}  & {\vdots} &                                &{\vdots}&{\vdots} & \rstick{\vdots}  \\
\lstick{}               &          &                                &&         & \rstick{}  \\
\lstick{\ket{0}}        & \gate{H} & \ghost{\,\,\,U_f\,\,\,}        & \gate{H}& \meter  & \rstick{s_{n-1}}  \cw\\
\lstick{\ket{1}}        & \gate{H} & \ghost{\,\,\,U_f\,\,\,}        & \gate{H}&   \rstick{\ket{1}.} \qw  \\
\lstick{}               &          &                                &&         & \rstick{}  \\
\ustick{\hspace{0.7cm}\ket{\psi_0}} & \ustick{\hspace{1.4cm}\ket{\psi_1}}& \ustick{\hspace{2.0cm}\ket{\psi_2}} & \ustick{\hspace{1.4cm}\ket{\psi_3}}
}\vspace{0.0cm}
\]
This circuit is identical to the circuit of the Deutsch-Jozsa algorithm, with the only difference being the function used in the Bernstein-Vazirani algorithm, which need not be either constant or balanced. Furthermore, in the Deutsch-Jozsa algorithm, we check whether all outputs are 0 to conclude that the function is constant; otherwise, it is balanced. In the Bernstein-Vazirani algorithm, each output is valuable for determining the unknown $s$. The states at the bottom of the circuit are used in the analysis of the algorithm.

\section{Analysis of the algorithm}

After the first step, the state of the qubits is
\[
\ket{\psi_0}=\ket{0}^{\otimes n}\ket{1}.
\]
After the second step, the state of the qubits is
\begin{eqnarray*}
\ket{\psi_1} &=& (H^{\otimes n}\ket{0}^{\otimes n})\otimes(H\ket{1})\\
&=& \frac{1}{\sqrt{2^n}}\sum_{x=0}^{2^n-1}\ket{x}\otimes \ket{-},
\end{eqnarray*}
where $\ket{-}=(\ket{0}-\ket{1})/\sqrt{2}$ and $x$ is written in the decimal notation. After the third step, the state of the qubits is
\begin{eqnarray*}
\ket{\psi_2} &=& U_f\ket{\psi_1}\\
&=& \frac{1}{\sqrt{2^n}}\sum_{x=0}^{2^n-1} U_f\ket{x}\otimes\ket{-}.
\end{eqnarray*}
To simplify $\ket{\psi_2}$, we use Proposition~\ref{prop:U_fxminus} on Page~\pageref{prop:U_fxminus}, which states that
\[
U_f\big(\ket{x}\ket{-}\big) \,=\, (-1)^{f(x)}\ket{x}\ket{-}.
\]
We obtain
\begin{eqnarray*}
\ket{\psi_2} &=& \frac{1}{\sqrt{2^n}}\sum_{x=0}^{2^n-1}(-1)^{s\cdot x}\ket{x}\otimes\ket{-}.
\end{eqnarray*}
After the fourth step, the state of the qubits is
\begin{eqnarray}\label{eq:BV-psi3}
\ket{\psi_3} &=& H^{\otimes (n+1)}\ket{\psi_2}\nonumber\\
&=&  H^{\otimes n}\left(\frac{1}{\sqrt{2^n}}\sum_{x=0}^{2^n-1}(-1)^{s\cdot x}\ket{x}\right)\otimes(H\ket{-}).
\end{eqnarray}
To simplify $\ket{\psi_3}$, we use Proposition~\ref{prop:Hnx} on Page~\pageref{prop:Hnx}, which states that for any $s=(s_0...s_{n-1})_2$
\[
H^{\otimes n}\ket{s} \,=\, \frac{1}{\sqrt{2^n}}\sum_{x=0}^{2^n-1}(-1)^{s\cdot x}\ket{x},
\]
where $s\cdot x=s_{0}x_{0}+ \cdots + s_{{n-1}}x_{{n-1}}\mod 2$. Then we use $H^2=I$ to obtain
\[
\ket{s}=H^{\otimes n}\left(\frac{1}{\sqrt{2^n}}\sum_{x=0}^{2^n-1}(-1)^{s\cdot x}\ket{x}\right).
\]
We replace this result in Eq.~(\ref{eq:BV-psi3}) to obtain
\begin{eqnarray*}
\ket{\psi_3} &=&  \ket{s}\otimes\ket{1}.
\end{eqnarray*}
The output of the measurement of the first register is $s_0,...,s_{n-1}$ with probability 1 because $\ket{s}=\ket{s_0}\otimes\cdots\otimes\ket{s_{n-1}}$.

$U_f$ is applied only once, so we say that a single query was made to the quantum oracle. Note that $U_f$ is applied to a superposition of all vectors of the computational basis, which means that $f$ is evaluated simultaneously at all points in the domain---all $n$-bit strings.  The result of this massive evaluation is a superposition state, which is often useless in many cases. However, in the Bernstein-Vazirani algorithm, applying $H^{\otimes n}$ to the first register at the end reveals $s$. Before the last step, the information about $s$ is encoded in relative phases. After applying $H^{\otimes n}$, interference converts those relative phases into the computational-basis state $\ket{s}$, which can be read out directly by measurement.

\section{Bernstein-Vazirani algorithm has no entanglement}

The analysis of entanglement initially follows the same approach used in the Deutsch-Jozsa algorithm.
Examining the algorithm's circuit, we realize that the only operator that creates or destroys entanglement is $U_f$. Thus, it is sufficient to analyze $\ket{\psi_2}$. There is entanglement in the Bernstein-Vazirani algorithm if and only if $\ket{\psi_2}$ is either totally or partially entangled. State $\ket{\psi_2}$ is given by
\[
\ket{\psi_2} \,=\,\frac{1}{\sqrt{2^n}}\sum_{x=0}^{2^n-1} (-1)^{s\cdot x}\ket{x}\otimes\ket{-}.
\]
Proposition~\ref{prop:Hnx} on Page~\pageref{prop:Hnx} states that for any $s=(s_0...s_{n-1})_2$
\[
H^{\otimes n}\ket{s} \,=\, \frac{1}{\sqrt{2^n}}\sum_{x=0}^{2^n-1}(-1)^{s\cdot x}\ket{x}.
\]
Then $\ket{\psi_2}$ can be written as
\[
\ket{\psi_2} \,=\,\big(H^{\otimes n}\ket{s}\big)\otimes\ket{-},
\]
which can be fully factorized as
\[
\ket{\psi_2} \,=\,\big(H\ket{s_0}\big)\otimes\cdots\otimes \big(H\ket{s_{n-1}}\big)\otimes\ket{-}.
\]
Note that $\ket{\psi_2}$ has no entanglement at all because it is the Kronecker product of single-qubit pure states~\cite{Mey00,Duetal01}.

\section{Circuit of the oracle}

In oracle-based algorithms, the oracle is not implemented by us---it is implemented by someone else. However, it is important to know how it is done in order to understand the whole process. Let us show an example with $n=4$ and $s=1011$, which is enough to understand the general case. In this particular case, $f$ is
\[
f(x)=s\cdot x= x_0\oplus x_2\oplus x_3 = x_0+x_2+x_3\mod 2
\]
and the action of $U_f$ on an arbitrary vector of the computational basis $\ket{x}\ket{j}$ is
\[
 U_f\ket{x_0x_1x_2x_3}\ket{j}=\ket{x_0x_1x_2x_3}\ket{j\oplus f(x)}=
 \ket{x_0x_1x_2x_3}\ket{j\oplus x_0 \oplus x_2 \oplus x_3}.
\]
Let CNOT$_{04}$ be the CNOT gate acting on qubits 0 and 4. Without showing qubits 1, 2, and 3, we have
\[
\text{CNOT}_{04}\ket{x_0}\ket{j}=\ket{x_0}X^{x_0}\ket{j}=\ket{x_0}\ket{j\oplus x_0}.
\]
If we use $\text{CNOT}_{04}$, $\text{CNOT}_{24}$, and $\text{CNOT}_{34}$, we generate the expected result $j\oplus x_0 \oplus x_2 \oplus x_3$ in the second register. Then,
\[
U_f\,=\,\text{CNOT}_{34}\cdot\text{CNOT}_{24}\cdot\text{CNOT}_{04},
\]
whose circuit is \vspace{-3pt}
\[
\Qcircuit @C=1.3em @R=0.9em {
\lstick{\ket{x_0}}  & \ctrl{4}&\qw     & \qw       &\rstick{\ket{x_0}}\qw \\
\lstick{\ket{x_1}}  & \qw     &\qw     & \qw       &\rstick{\ket{x_1}}\qw \\
\lstick{\ket{x_2}}  & \qw     &\ctrl{2}& \qw       &\rstick{\ket{x_2}}\qw \\
\lstick{\ket{x_3}}  & \qw     &\qw     & \ctrl{1}  &\rstick{\ket{x_3}}\qw \\
\lstick{\ket{j}}    & \targ   &\targ   & \targ     &\rstick{\ket{j\oplus x_0 \oplus x_2 \oplus x_3}.}\qw \\
}\vspace{0.2cm}
\]
From this example, we can generalize the implementation to an arbitrary $s=s_0\cdots s_{n-1}$. Just consider
\[
U_f\,=\,\left(\text{CNOT}_{0n}\right)^{s_0}\left(\text{CNOT}_{1n}\right)^{s_1}\,\cdots\,\left(\text{CNOT}_{{n-1},n}\right)^{s_{n-1}},
\]
where CNOT$_{ij}$ is controlled by qubit $i$ and the target is qubit $j$, and  $\left(\text{CNOT}_{ij}\right)^{s_k}$ is the identity operator if $s_k=0$, and the standard $\text{CNOT}_{ij}$ if $s_k=1$.
Note that the order of the CNOT gates is irrelevant. There is a CNOT for each bit 1 of $s$.

Let us finish the circuit of the Bernstein-Vazirani algorithm for our example. The whole circuit is
\[
\Qcircuit @C=1.3em @R=0.3em {
\lstick{\ket{0}}& \gate{H} & \ctrl{4}&\qw     & \qw      & \gate{H} & \meter  & \rstick{1} \cw \\
\lstick{\ket{0}}& \gate{H} & \qw     &\qw     & \qw      & \gate{H} & \meter  & \rstick{0} \cw \\
\lstick{\ket{0}}& \gate{H} & \qw     &\ctrl{2}& \qw      & \gate{H} & \meter  & \rstick{1} \cw \\
\lstick{\ket{0}}& \gate{H} & \qw     &\qw     & \ctrl{1} & \gate{H} & \meter  & \rstick{1} \cw \\
\lstick{\ket{1}}& \gate{H} & \targ   &\targ   & \targ    & \gate{H} &\rstick{\ket{1},}\qw \\
}\vspace{0.2cm}
\]
which is equivalent to
\[
\Qcircuit @C=1.3em @R=0.3em {
\lstick{\ket{0}}& \gate{H} & \ctrl{4}& \gate{H} &\qw&\qw &\qw &\qw &\qw  & \qw      & \meter  & \rstick{1} \cw \\
\lstick{\ket{0}}& \qw & \qw  &\qw   &\qw  &\qw   & \qw   &\qw  &\qw & \qw & \meter  & \rstick{0} \cw \\
\lstick{\ket{0}}& \qw    &\qw &\qw& \gate{H} &\ctrl{2}& \gate{H} & \qw  &\qw  &\qw  & \meter  & \rstick{1} \cw \\
\lstick{\ket{0}}& \qw  &\qw   &\qw &\qw &\qw &\qw  & \gate{H} & \ctrl{1} & \gate{H} & \meter  & \rstick{1} \cw \\
\lstick{\ket{1}}& \gate{H} & \targ &\gate{H} &\gate{H}  &\targ &\gate{H} &\gate{H}  & \targ    & \gate{H} &\rstick{\ket{1},}\qw \\
}\vspace{0.2cm}
\]
because $H^2=I$. Using that $(H\otimes H)\cdot \text{CNOT}_{ij}\cdot (H\otimes H) = \text{CNOT}_{ji}$, the last circuit simplifies to
\vspace{0pt}
\[
\Qcircuit @C=1.9em @R=0.3em {
\lstick{\ket{0}}  & \targ    &\qw      & \qw       & \meter  & \rstick{1} \cw \\
\lstick{\ket{0}}  & \qw      &\qw      & \qw       & \meter  & \rstick{0} \cw \\
\lstick{\ket{0}}  & \qw      &\targ    & \qw       & \meter  & \rstick{1} \cw \\
\lstick{\ket{0}}  & \qw      &\qw      & \targ     & \meter  & \rstick{1} \cw \\
& \\
\lstick{\ket{1}}  & \ctrl{-5}&\ctrl{-3}&\ctrl{-2}  &\rstick{\ket{1}.}  \qw \\
}\vspace{0.2cm}
\]
Using this example, we can easily derive the generic case, which can always be expressed as CNOTs controlled by the last qubit and with targets being the qubits corresponding to the bits of $s$ that are equal to 1. This simplification helps to understand that the Bernstein-Vazirani algorithm has no entanglement because $U_f$ is multiplied by $H^{\otimes (n+1)}$, which neither creates nor destroys entanglement. The depiction of $U_f$ in the circuit of the algorithm gives us the impression that $U_f$ creates entanglement, which is misleading.

\begin{exercise} \label{exe:BV-oracles} \mbox{}
What are the modifications to the Bernstein-Vazirani algorithm that are necessary so that the modified version finds $s \in \{0,1\}^n$ and $b \in \{0,1\}$ given the promise that the oracle is a Boolean function $f(x) = (s \cdot x) \oplus b$, where $s$ is an $n$-bit string and $b$ is a bit?
\end{exercise}

\section{Economical circuit of the Bernstein-Vazirani algorithm}

The Bernstein-Vazirani algorithm can be implemented with $n$ qubits instead of $n+1$. Indeed, we have learned that $U_f$ is a product of CNOT's
\[
U_f\,=\,\left(\text{CNOT}_{0n}\right)^{s_0}\,\cdots\,\left(\text{CNOT}_{{n-1},n}\right)^{s_{n-1}},
\]
where CNOT$_{ij}$ is controlled by qubit $i$ and the target is qubit $j$, and  $\left(\text{CNOT}_{ij}\right)^{s_k}$ is the identity if $s_k=0$, and $\text{CNOT}_{ij}$ if $s_k=1$.  In the algorithm, just before the action of $U_f$, the state of the $n$-th qubit is $\ket{-}$. The circuit that describes the action of $\text{CNOT}_{0n}$ is
\[
\Qcircuit @C=2.0em @R=1.6em {
\lstick{\ket{x_0}}       &\ctrl{1} \qw&   \rstick{(-1)^{x_0}\ket{x_0}} \qw \\
\lstick{\ket{-}}         & \targ  \qw  & \rstick{\ket{-},} \qw
}\vspace{5pt}
\]
where we have depicted only the first and the last qubits and we have placed $(-1)^{x_0}$ in the first qubit because this is mathematically allowed. This circuit is equivalent to
\[
\Qcircuit @C=2.0em @R=1.6em {
\lstick{\ket{x_0}}       &\gate{Z} \qw&   \rstick{(-1)^{x_0}\ket{x_0}.} \qw
}\vspace{5pt}
\]
We can convert all CNOT's  of $U_f$ into $Z$'s and then
\[
U'_f \,=\, Z^{s_0}\otimes\cdots\otimes  Z^{s_{n-1}}.
\]
If a bit $s_i$ of $s$ is 0, $Z^{s_i}=I$. Then
\[
U'_f\ket{x}\,=\, (-1)^{s\cdot x}\ket{x}.
\]
Function $f$ is the same as before. What changes is the way we implement $f$ as a unitary operator.

The circuit of the economical version of the Bernstein-Vazirani algorithm is \vspace{3pt}
\[
\Qcircuit @C=2.3em @R=0.9em {
\lstick{\ket{0}}        & \gate{H} & \multigate{3}{\,\,\,U'_f\,\,\,} & \gate{H} & \meter  & \rstick{s_0} \cw \\
\lstick{\vdots \,\,\,}  & {\vdots} &                                &{\vdots}&{\vdots} & \rstick{\vdots}  \\
\lstick{}               &          &                                &&         & \rstick{}  \\
\lstick{\ket{0}}        & \gate{H} & \ghost{\,\,\,U'_f\,\,\,}        & \gate{H}& \meter  & \rstick{s_{n-1}.}  \cw\\
}\vspace{0.0cm}
\]
As before, this circuit simplifies to
\[
H^{\otimes n}U'_fH^{\otimes n} \,=\, X^{s_0}\otimes\cdots\otimes  X^{s_{n-1}}
\]
because $HZH=X$. Now, it is straightforward to check that there is no entanglement in the Bernstein-Vazirani algorithm. $U'_f$ is not a genuine $n$-qubit gate, but instead is the tensor product of $n$ single-qubit gates. The depiction of the circuit is misleading.

\subsection*{Consulting the oracle}

In the non-economical circuit, if the input to the first register of $U_f$ is $\ket{x}$, $U_f$ returns $f(x)=s\cdot x$ when we perform a measurement of the last qubit. In the economical circuit, we obtain $f(x)$ using the circuit\vspace*{3pt}
\[
\Qcircuit @C=2.0em @R=1.6em {
\lstick{\ket{x_0}}        & \qw & \multigate{2}{U'_f}        & \qw&   \rstick{\ket{x_0}} \qw \\
\lstick{\vdots\,\,\,\,\,}        &  &      & &   \rstick{\,\,\,\,\vdots}  \\
\lstick{\ket{x_{n-1}}}        & \qw & \ghost{U'_f}        & \qw&   \rstick{\ket{x_{n-1}}} \qw \\
\lstick{\ket{0}}        & \gate{H} & \ctrl{-1} & \gate{H} & \meter  & \rstick{f(x).} \cw
}\vspace{5pt}
\]
In fact, the steps of this circuit are\vspace{1pt}
\begin{align*}
\ket{x}\ket{0} \xrightarrow[\text{}]{I\otimes H} \ket{x}\ket{+} \xrightarrow[\text{}]{\fbox{$U'_f$}-\bullet} \frac{\ket{x}\ket{0}+(-1)^{f(x)}\ket{x}\ket{1}}{\sqrt 2} \xrightarrow[\text{}]{I\otimes H} \ket{x}\ket{f(x)},
\end{align*}
where $I$ is the $2^n$-dimensional identity operator.
To calculate the last step, there are only two cases when we fix $x$: Either $f(x)=0$ or $f(x)=1$. Firstly, we suppose that $f(x)=0$, then $I\otimes H$ is applied to $\ket{x}\ket{+}$ resulting in $\ket{x}\ket{f(x)}$, and secondly, we suppose that $f(x)=1$, then $I\otimes H$ is applied to $\ket{x}\ket{-}$ resulting in $\ket{x}\ket{f(x)}$. After performing a measurement of the last qubit in the computational basis, we obtain $f(x)$ with probability 1.

We have just shown the following circuits are equivalent:
\[
\Qcircuit @C=1.0em @R=0.6em {
\lstick{}& {/}^{{n}}\qw&\multigate{2}{U_f}&\qw&\qw    &&&& & {/}^{{n}}\qw    &\gate{U'_f}&\qw     &\rstick{} \qw \\
                &&&&&&\equiv&&& \\
\lstick{}&\qw&\ghost{U_f}       &\qw &\qw   &&&&& \gate{H}&\ctrl{-2}  &\gate{H}&\rstick{,} \qw
}\vspace{5pt}
\]
where the notation ${/}^{{n}}$ on a wire represents a $n$-qubit register.

\chapter{Simon's Problem}\label{chap:Simon}

Simon's problem was presented at a conference in 1994~\cite{Sim94} together with Shor's algorithms~\cite{Sho94}, and the full paper was published in 1997~\cite{Sim97}. Simon's algorithm is exponentially faster than the best deterministic or randomized equivalent classical algorithms. This is a remarkable but underestimated scientific contribution to quantum computing. Simon's algorithm exploits not only quantum parallelism but also maximal entanglement. This algorithm and its generalizations are described in books~\cite{KLM07, Mer07, NO08, RP11, YM08} and papers~\cite{CQ18,MS03,YHLW21}.

\section{Problem formulation}

Let $f:\{0,1\}^n \longrightarrow \{0,1\}^n$ be a function from $n$-bit strings to $n$-bit strings with the following property: there exists a nonzero hidden bit string $s\in\{0,1\}^n$ such that
\[
f(x)=f(y) \iff x\oplus y\in\{0,s\}
\] for all $x,y\in\{0,1\}^n$. This means that each point in the image is the image of exactly two points in the domain, that is, the point $f(x)$ in the image is associated with $x$ and $x\oplus s$ in the domain because $f(x)=f(x\oplus s)$ for all $x\in\{0,1\}^n$. Therefore, $f$ is a two-to-one function. Consider the following computational problem: determine $s$ by querying $f$ as few times as possible.

Our goal is to find $s$ by evaluating $f$ without knowing its implementation details. In quantum computing, the operator implementing this function is called an \textit{oracle}, because an oracle reveals $f(x)$ without showing $s$ explicitly, and $f(x)$ can be used to determine $s$, although evaluating $f$ only once is not enough. When we build the circuit of the algorithm, the portion associated with $f$ is implemented by another person, unless there is a formula for $f$ that does not reveal $s$ beforehand.

In the classical version of this problem, we have to consult a \textit{classical oracle}, and to determine $s$ with high probability, the number of queries to the function $f$ grows exponentially with the number of bits $n$ when using a classical computer. Indeed, a naive deterministic algorithm would be to calculate $f(x')$ for a fixed $x'$, then systematically search for $x$ in the domain such that $f(x)=f(x')$. It requires $2^n$ evaluations in the worst case. A randomized classical strategy based on searching for collisions requires $\Omega(2^{n/2})$ evaluations to achieve constant success probability, by the same reasoning as in the birthday paradox. The proof is the same as in the two-to-one collision problem's proof~\cite{BHT98b} or in the \textit{birthday paradox}~\cite{Das05}.

In the quantum case, $f$ is implemented using the unitary operator $U_f$ of $2n$ qubits, defined as
\[
U_f\ket{x}\ket{y}=\ket{x}\ket{y\oplus f(x)},
\]
where $x$, $y$, and $f(x)$ are $n$-bit strings, and $\oplus$ is the bitwise XOR operation or bitwise sum modulo 2. This operator uses two registers, each containing $n$ qubits. $U_f$ is a $2^{2n}$-dimensional permutation matrix. The proof that $U_f$ is unitary is an extension of the proof presented in Proposition~\ref{prop:U_f_is_unitary} on Page~\pageref{prop:U_f_is_unitary}.

Simon's algorithm has two parts. The quantum part returns an $n$-bit string $x$ obeying $x\cdot s=0$, where
\[
x\cdot s=x_{0}s_{0}+ \cdots + x_{{n-1}}s_{{n-1}}\mod 2.
\]
Knowing such $x$ is not enough to determine $s$. It is necessary to run the quantum part many times, collect multiple bit strings $x$ obeying $x\cdot s=0$, and then run the classical part of the algorithm, which reveals $s$ with probability greater than 1/2.

\section{The algorithm}

\begin{algorithm} [!ht]
\caption{Simon's algorithm} \label{algo_SimonClassical}
\KwIn{Function $f:\{0,1\}^n\longrightarrow \{0,1\}^n$ with the promise that $f(x)=f(y) \iff x\oplus y\in\{0,s\}$.}
\KwOut{$s$ with probability greater than 1/2.} \BlankLine
Run the quantum part $n-1$ times (Algorithm~\ref{algo_SimonQuantum}, assume outputs $x^{(1)},...,x^{(n-1)}$\big)\;
Solve the system of linear equations $\{x^{(1)}\cdot s\equiv 0,...,x^{(n-1)}\cdot s\equiv 0\}\mod 2$ (assume solution for $s_1,...,s_{n-1}$)\;
Take $s_{0}=0$\;
If $f(s)=f(0)$ then return $0s_1...s_{n-1}$; otherwise, return $1s_1...s_{n-1}$.
\end{algorithm}

Simon's algorithm is described in Algorithm~\ref{algo_SimonClassical} and the quantum part is described in Algorithm~\ref{algo_SimonQuantum}. The circuit of the quantum part is \vspace{3pt}
\[
\Qcircuit @C=1.7em @R=0.9em {
\lstick{\ket{0}}            & \gate{H} & \multigate{5}{\,\,\,U_f\,\,\,} & \qw&\gate{H} &\qw& \meter  & \rstick{x_0} \cw \\
\lstick{\vdots \,\,\,}      & {\vdots} &                                &&{\vdots}&&{\vdots} & \rstick{\vdots}  \\
\lstick{}                   &          &                                &&&         && \rstick{}  \\
\lstick{\ket{0}}            & \gate{H} & \ghost{\,\,\,U_f\,\,\,}        &\qw& \gate{H}&\qw& \meter  & \rstick{x_{n-1}}  \cw\\
 & & & & {\hspace{1.4cm}{}^\ket{\psi_4}} \\
\lstick{\ket{0}^{\otimes n}}& {/^n}\qw & \ghost{\,\,\,U_f\,\,\,}        & \meter  &  \rstick{z_0...z_{n-1}.}\cw  \\
\lstick{}                   &          &                                &&         & \rstick{}  \\
\ustick{\hspace{0.7cm}\ket{\psi_0}} & \ustick{\hspace{1.3cm}\ket{\psi_1}}& \ustick{\hspace{1.9cm}\ket{\psi_2}} & \ustick{\hspace{1.4cm}\ket{\psi_3}}
}\vspace{0.0cm}
\]
The states at the bottom of the circuit are used in the analysis of the algorithm. They describe the state of the qubits after each step. State $\ket{\psi_4}$ refers to the first register only. The notation ``$/^n$'' over a wire denotes that it is an $n$-qubit register.

\begin{algorithm} [!ht]
\caption{Quantum part of Simon's algorithm} \label{algo_SimonQuantum}
\KwIn{A black box $U_f$ implementing function $f:\{0,1\}^n\longrightarrow \{0,1\}^n$ with the promise that $f(x)=f(y) \iff x\oplus y\in\{0,s\}$.}
\KwOut{Point $x\in\{0,1\}^n$ such that $x\cdot s=0$.} \BlankLine
Prepare the initial state $\ket{0}^{\otimes n}\ket{0}^{\otimes n}$\;
Apply $H^{\otimes n}$ to the first register\;
Apply $U_f$\;
Measure the second register in the computational basis (assume output $z_0...z_{n-1}$)\;
Apply $H^{\otimes n}$ to the first register\;
Measure the first register in the computational basis.
\end{algorithm}

\section{Analysis of the quantum part}

After the first step, the state of the qubits is
\[
\ket{\psi_0}=\ket{0}^{\otimes n}\ket{0}^{\otimes n}.
\]
After the second step, the state of the qubits is
\begin{eqnarray*}
\ket{\psi_1} &=& \big(H\ket{0}\big)^{\otimes n}\otimes\ket{0}^{\otimes n}\\
&=& \frac{1}{\sqrt{2^n}}\sum_{x=0}^{2^n-1}\ket{x}\otimes \ket{0}^{\otimes n},
\end{eqnarray*}
where $x$ is written in the decimal notation. After the third step, the state of the qubits is
\begin{eqnarray*}
\ket{\psi_2} &=& U_f\ket{\psi_1}\\
&=& \frac{1}{\sqrt{2^n}}\sum_{x=0}^{2^n-1}U_f(\ket{x}\otimes\ket{0\cdots 0}).
\end{eqnarray*}
To simplify $\ket{\psi_2}$, we use the definition of $U_f$ to obtain
\begin{eqnarray*}
\ket{\psi_2} &=& \frac{1}{\sqrt{2^n}}\sum_{x=0}^{2^n-1}\ket{x}\otimes\ket{f(x)},
\end{eqnarray*}
because $(0...0)\oplus f(x)$ (bitwise XOR) is $f(x)$.
The fourth step is a measurement of each qubit of the second register, which we assume has returned $z_0...z_{n-1}$. State $\ket{\psi_2}$ collapses to a superposition of only two terms because there are only two points in the domain such that $f(x)=z_0...z_{n-1}$. Let $x'$ be one of those points. Then, $f(x')=f(x'\oplus s)=z_0...z_{n-1}$ and state $\ket{\psi_3}$ is
\begin{eqnarray*}
\ket{\psi_3} &=&  \left(\frac{\ket{x'}+\ket{x'\oplus s}}{\sqrt{2}}\right)\otimes \ket{z_0...z_{n-1}}.
\end{eqnarray*}
Note that we have renormalized state $\ket{\psi_3}$ in accordance with the measurement postulate. Point $x'$ is unknown. It is selected uniformly at random among the points in the domain. The information we wish to acquire, $s$, remains hidden, as $x'\oplus s$ is a random point in the domain.

In the fifth step, we only consider the first register. After applying $H^{\otimes n}$ to the state of the first register, we obtain
\begin{eqnarray*}
\ket{\psi_4} &=&  \frac{1}{\sqrt{2}}\left(H^{\otimes n}\ket{x'}+H^{\otimes n}\ket{x'\oplus s}\right).
\end{eqnarray*}
To simplify $\ket{\psi_4}$, we use Proposition~\ref{prop:Hnx} on Page~\pageref{prop:Hnx}, which states that for any $x'=(x'_0...x'_{n-1})_2$
\[
H^{\otimes n}\ket{x'} \,=\, \frac{1}{\sqrt{2^n}}\sum_{x=0}^{2^n-1}(-1)^{x'\cdot x}\ket{x},
\]
where $x'\cdot x=x'_{0}x_{0}+ \cdots + x'_{{n-1}}x_{{n-1}}\mod 2$. Then
\begin{eqnarray*}
\ket{\psi_4} &=&  \frac{1}{\sqrt{2^{n+1}}}\sum_{x=0}^{2^n-1}\left((-1)^{x'\cdot x}+(-1)^{(x'\oplus s)\cdot x}\right)\ket{x}.
\end{eqnarray*}
Now we use the fact that
\begin{align*}
(x'\oplus s)\cdot x &=(x'_{0}+ s_0)x_{0}+ \cdots + (x'_{{n-1}}+ s_{n-1})x_{{n-1}}\mod 2\\
 &= (x'_{0}x_{0}+ \cdots + x'_{{n-1}}x_{{n-1}})+(s_0x_{0}+ \cdots + s_{n-1}x_{{n-1}})\mod 2\\
 &= (x'\cdot x) + (s\cdot x) \mod 2.
\end{align*}
We have used that $x'_0\oplus s_0=x'_0+s_0\mod 2$. State $\ket{\psi_4}$ simplifies to
\begin{eqnarray*}
\ket{\psi_4} &=&  \frac{1}{\sqrt{2^{n+1}}}\sum_{x=0}^{2^n-1}(-1)^{x'\cdot x}\Big(1+(-1)^{s\cdot x}\Big)\ket{x}.
\end{eqnarray*}
Now we use
\[
1+(-1)^{s\cdot x} = \begin{cases}
			2, & \text{if $s\cdot x=0$,}\\
            0, & \text{otherwise,}
		 \end{cases}
\]
to obtain
\begin{eqnarray*}
\ket{\psi_4} &=&  \frac{1}{\sqrt{2^{n-1}}}\sum_{\substack{x=0\\s\cdot x=0}}^{2^n-1}(-1)^{x'\cdot x}\ket{x}.
\end{eqnarray*}
The sum is over $x$ such that $x\cdot s=0$. Then, the measurement's output of the first register is $x$ such that $x\cdot s=0$ with probability $|(-1)^{x'\cdot x}|^2=1$. When we run the quantum part of Simon's algorithm, we obtain partial information about $s$. This means that we need to run the quantum part multiple times to collect enough information to determine $s$.

The probability of obtaining a specific $x$ such that $x\cdot s= 0$ is $1/2^{n-1}$, which is the square of the absolute value of the amplitude of state $\ket{x}$ in $\ket{\psi_4}$. The distribution is uniformly spread across all of $n$-bit strings $x$ that satisfy $x\cdot s= 0$. On the other hand, the probability of obtaining an arbitrary $x$ such that $x\cdot s= 0$ is 1, or equivalently the probability of obtaining $x$ such that $x\cdot s\neq 0$ is 0. That is, we are certain that the output $x$ satisfies $x\cdot s= 0$. We obtain partial information about $s$ unless $x=(0...0)_2$.

Point $x'$, which obeys $f(x')=z$, plays no role in the final result nor in the final calculation of the success probability. This is advantageous because $x'$ was hiding $s$ at an earlier stage. Following the fifth step, $x'$ becomes harmless.

\section{Analysis of the classical part}

Each time we run the quantum part of Simon's algorithm, the output is an $n$-bit string $x$ such that $x\cdot s=0$. Suppose we have run it twice and the outputs are $x$ and $x'$. This means that we have obtained a homogeneous system of linear equations
\begin{align*}
x_{0}s_{0}+ \cdots + x_{{n-1}}s_{{n-1}}&\equiv 0 \mod 2,\\
x'_{0}s_{0}+ \cdots + x'_{{n-1}}s_{{n-1}}&\equiv 0 \mod 2,
\end{align*}
where $s_0,...,s_{n-1}$ are the variables (unknowns) and $x$, $x'$ are known binary coefficients. Let us calculate the probability $p(2)$ that the system is independent. The probability that the first equation is nontrivial is
\[
p_1 \,=\, 1-\frac{1}{2^n}
\]
because there are $2^n$ strings $x$ and only one is 0. To calculate the probability that the second equation is independent, we think that $x$ is an $n$-dimensional vector in a binary vector space with $2^n$ vectors. The subspace spanned by $x$ has two vectors, $x$ itself and the null vector. There are $2^n-2$ vectors that are linearly independent of $x$. Then, the probability $p_2$ that the second equation is independent is
\[
p_2 \,=\, 1-\frac{2}{2^n}.
\]
Then, the probability $p(2)$ that the two equations are independent is
\[
p(2) \,=\, p_1p_2 \,=\, \left(1-\frac{1}{2^n}\right)\left(1-\frac{2}{2^n}\right).
\]
The probability that the next equation added to the system is independent is calculated as follows. Vectors $x$ and $x'$ span a subspace with four vectors: $x$, $x'$, $x\oplus x'$, and the null vector. There are $2^n-4$ vectors that are linearly independent of $x$ and $x'$. The probability $p_3$ that the third equation is independent is
\[
p_3 \,=\, 1-\frac{2^2}{2^n}.
\]
Then, the probability $p(3)$ that the three equations are independent is
\[
p(3) \,=\, p_1p_2p_3 \,=\, \left(1-\frac{1}{2^n}\right)\left(1-\frac{2}{2^n}\right)\left(1-\frac{2^2}{2^n}\right).
\]
We proceed in this fashion until we have $n-1$ independent equations with probability
\[
p(n-1) \,=\, \prod_{i=1}^{n-1}p_i \,=\, \prod_{i=0}^{n-2}\left(1-\frac{2^i}{2^n}\right).
\]
Calculating this product is challenging. Our aim now is to find a nontrivial lower bound. If we expand the product we obtain
\begin{align*}
\left(1-\frac{1}{2^n}\right)\cdots\left(1-\frac{2^{n-2}}{2^n}\right)&= 1 - \sum_{i=0}^{n-2} \frac{2^i}{2^{n}} + \cdots.
\end{align*}
The sum is calculated using the geometric series yielding $(1/2-1/2^n)$. The remaining terms include only higher-order terms ($1/(2^n)^2,1/(2^n)^3, ...$), and the following Proposition shows that they make a positive contribution. Then,
\begin{align*}
p(n-1)\ge  \frac{1}{2}+\frac{1}{2^n}.
\end{align*}

\begin{proposition}
Let $n\ge 2$ be an integer. Then
\[
\prod_{i=0}^{n-2}\left(1-\frac{2^i}{2^n}\right)\ge  \frac{1}{2}+\frac{1}{2^n}.
\]
\end{proposition}
\begin{proof}
By induction on $n$. The base case follows after replacing $n$ with 2. Let's prove the induction step. The left-hand expression for $n+1$ is
\[
\prod_{i=0}^{n-1}\left(1-\frac{2^i}{2^{n+1}}\right)=\left(1-\frac{2^0}{2^{n+1}}\right)\prod_{i=1}^{n-1}\left(1-\frac{2^i}{2^{n+1}}\right).
\]
By manipulating the dummy index $i$, we obtain
\[
\prod_{i=0}^{n-1}\left(1-\frac{2^i}{2^{n+1}}\right)=\left(1-\frac{1}{2^{n+1}}\right)\prod_{i=0}^{n-2}\left(1-\frac{2^i}{2^{n}}\right).
\]
Let us assume that the inequality is true for $n$. Then
\[
\prod_{i=0}^{n-1}\left(1-\frac{2^i}{2^{n+1}}\right)\ge \left(1-\frac{1}{2^{n+1}}\right)\left( \frac{1}{2}+\frac{1}{2^n}\right).
\]
Expanding the right-hand side, we obtain
\[
\prod_{i=0}^{n-1}\left(1-\frac{2^i}{2^{n+1}}\right) \ge
{\frac{1}{2}}+ \frac{1}{2^{n+1}}+{\frac {1}{{2}^{n+1}}
 \left( {\frac{1}{2}}- \frac{1}{2^n} \right) }
\ge \frac{1}{2}+\frac{1}{2^{n+1}}.
\]
This completes the proof.
\end{proof}

We are not done yet because with $n-1$ independent equations we can determine $n-1$ bits of $s$. The missing bit $s_i$ can be determined by guessing, for instance, by initially assuming that $s_i=0$ and then using the classical oracle to ask whether $f(s)=f(0)$. If true, we have successfully found $s$; otherwise, we set $s_i=1$. The computational cost of running the classical part is basically the cost of solving a system of $n-1$ linear equations with $n$ variables with polynomial cost, for instance $O(n^3)$ using Gaussian elimination.

The total cost of the algorithm is $n-1$ calls of $U_f$ and a single call of $f$ plus $O(n^2)$ steps to solve the system of linear equations. The success probability is greater than 1/2.


\section{Analysis of the entanglement}

From the circuit of the algorithm, we realize that the only operator that creates or destroys entanglement is $U_f$. Then, it is enough to analyze $\ket{\psi_2}$ or $\ket{\psi_3}$.  It is simpler to analyze $\ket{\psi_3}$. There is entanglement in Simon's algorithm if and only if $\ket{\psi_3}$ is either totally or partially entangled. The state of the first register of $\ket{\psi_3}$ is
\[
\ket{\psi}\,=\,\frac{\ket{x}+\ket{x\oplus s}}{\sqrt{2}},
\]
where $x$ is a random $n$-bit string and $s$ is a fixed nonzero $n$-bit string. If $s=1...1$ and $x=1...1$, $\ket{\psi}$ is the well-known Greenberger–Horne–Zeilinger state of $n$ qubits, defined as
\[
\ket{\text{GHZ}}\,=\,\frac{\ket{0\cdots 0}+\ket{1\cdots 1}}{\sqrt{2}}.
\]
It is known that the GHZ state is genuinely multipartite entangled.\footnote{\url{https://en.wikipedia.org/wiki/Greenberger-Horne-Zeilinger_state}}
If $s=1...1$, state $\ket{\psi}$ is non-biseparable\footnote{A pure state $\ket{\psi}$ of $n$ qubits is called biseparable, if one can find a partition of the qubits in two registers $A$ and $B$ such that $\ket{\psi}=\ket{\psi_A}\otimes\ket{\psi_B}$.} for any $x$ because
\[
\ket{\psi}\,=\, X^{x_0}\otimes\cdots\otimes X^{x_{n-1}}\ket{\text{GHZ}},
\]
and $X^{x_0}\otimes\cdots\otimes X^{x_{n-1}}$ does not create or destroy entanglement.

On the other hand, we can factor state $\ket{\psi}$ for each bit 0 of $s$. Suppose that $s=01...1$, then
\[
\ket{\psi}\,=\,\ket{x_0}\otimes\frac{\ket{x_1...x_{n-1}}+\ket{\bar x_1...\bar x_{n-1}}}{\sqrt{2}},
\]
where $\bar x_i=x_i\oplus 1$. This state is not maximally entangled but it is still partially entangled if $n>2$. If $s=0...01$, then
\[
\ket{\psi}\,=\,\ket{x_0}\otimes\cdots\ket{x_{n-2}}\otimes\frac{\ket{x_{n-1}}+\ket{\bar x_{n-1}}}{\sqrt{2}},
\]
which has no entanglement at all.

In summary, if the Hamming weight of $s$ is greater than 1, there is entanglement in Simon's algorithm. The degree of entanglement increases with the Hamming weight of $s$, and the state of the first register before the measurement becomes maximally entangled when the Hamming weight of $s$ equals $n$.

\section{Circuit of the oracle}

In oracle-based algorithms, the oracle is indeed implemented by someone else, but understanding its implementation can help us grasp the whole process better. As an example, let's consider the case where $n=3$ and $s=110$, which is enough to understand the general case. Let us take $f$ as the following two-to-one function
\[
\begin{tabular}{|c|c|}
\hline
$x_0\,x_1\,x_2$ & $f(x)$ \\
\hline
$\begin{array}{c}
0\,\,\,\,0\,\,\,\,0\\
1\,\,\,\,1\,\,\,\,0
\end{array}$  &  000\\
\hline
$\begin{array}{c}
0\,\,\,\,0\,\,\,\,1\\
1\,\,\,\,1\,\,\,\,1
\end{array}$& 001 \\
\hline
$\begin{array}{c}
0\,\,\,\,1\,\,\,\,0\\
1\,\,\,\,0\,\,\,\,0
\end{array}$ & 010 \\
\hline
$\begin{array}{c}
0\,\,\,\,1\,\,\,\,1\\
1\,\,\,\,0\,\,\,\,1
\end{array}$ & 100 \\
\hline
\end{tabular}
\]
To build the circuit we need to write down the explicit 3-output truth table, which is
\[
\begin{tabular}{|c|c|c|c|}
\hline
$x_0\,x_1\,x_2$ & $f_0(x)$&$f_1(x)$&$f_2(x)$ \\
\hline
0\,\,\,\,0\,\,\,\,0 & 0 & 0 & 0 \\
\hline
0\,\,\,\,0\,\,\,\,1 & 0 & 0 & 1 \\
\hline
0\,\,\,\,1\,\,\,\,0 & 0 & 1 & 0 \\
\hline
0\,\,\,\,1\,\,\,\,1 & 1 & 0 & 0 \\
\hline
1\,\,\,\,0\,\,\,\,0 & 0 & 1 & 0 \\
\hline
1\,\,\,\,0\,\,\,\,1 & 1 & 0 & 0 \\
\hline
1\,\,\,\,1\,\,\,\,0 & 0 & 0 & 0 \\
\hline
1\,\,\,\,1\,\,\,\,1 & 0 & 0 & 1 \\
\hline
\end{tabular}
\]
Note that $f(x)=f_0(x)f_1(x)f_2(x)$, where $f_0$ to $f_2$ are Boolean functions. The truth table of $f_0$ is obtained by considering only the first column of the output, and the truth tables of $f_1$ and $f_2$ by considering the second and third columns, respectively.
Now we focus on all bits 1 in the first column of the output denoted by $f_0(x)$. There are two of them, corresponding to inputs 011 and 101. We add to the circuit two multi-controlled NOT gates, the first with controls activated by 011 and the second activated by 101, with the target on the 4th qubit, as shown in Fig.~\ref{fig:SimonOracle}.
Then, we focus on all bits 1 in the second column of the output denoted by $f_1(x)$. There are two of them, corresponding to the inputs 010 and 100. The multi-controlled NOT gates are activated by 010 and 100, respectively, with the target on the 5th qubit, as shown in Fig.~\ref{fig:SimonOracle}. The last column, denoted by $f_2(x)$, requires multi-controlled NOT gates activated by 001 and 111 with the target on the 6th qubit, as shown in Fig.~\ref{fig:SimonOracle}.

\begin{figure}[!ht]
\[
\Qcircuit @C=2.3em @R=0.8em {
\lstick{\ket{x_0}}& \ctrlo{1} & \ctrl{1}& \ctrlo{1} & \ctrl{1}& \ctrlo{1} & \ctrl{1} &  \rstick{\ket{x_0}} \qw \\
\lstick{\ket{x_1}}& \ctrl{1} & \ctrlo{1}& \ctrl{1} & \ctrlo{1}& \ctrlo{1} & \ctrl{1} &  \rstick{\ket{x_1}} \qw \\
\lstick{\ket{x_2}}& \ctrl{1} & \ctrl{1}& \ctrlo{2} & \ctrlo{2}& \ctrl{3} & \ctrl{3} & \rstick{\ket{x_2}} \qw \\
\lstick{\ket{0}}  & \targ  & \targ  & \qw& \qw& \qw& \qw& \rstick{\ket{f_0(x_0x_1x_2)}}  \qw \\
\lstick{\ket{0}}  &\qw&\qw& \targ  & \targ& \qw& \qw  &  \rstick{\ket{f_1(x_0x_1x_2)}}  \qw \\
\lstick{\ket{0}}  &\qw&\qw&\qw&\qw& \targ  & \targ  &  \rstick{\ket{f_2(x_0x_1x_2)}}  \qw \\
}\vspace*{-0.2cm}
\]
\caption{Oracle of Simon's algorithm.}\label{fig:SimonOracle}
\end{figure}

The only trivial simplification that can be immediately seen is that the first two multi-controlled NOT gates can be simplified into only one Toffoli gate with empty control on qubit 2, full control on qubit 3 and target on qubit 4.

\section{Final remarks}

The formulation of Simon's problem in the original paper~\cite{Sim97} is slightly different from the one presented here. Simon posed the problem of determining whether $f$ is one-to-one (injective) or a special kind of two-to-one characterized by a $n$-bit string $s$ such that $f(x')=f(x)$ if and only if $x'=x\oplus s$. In the latter case, we have to find $s$. This formulation goes along the line of the Deutsch-Jozsa algorithm, in which we have the promise that the oracle is either balanced or constant. We have to determine which is the case. Note that if we run the quantum part of Simon's algorithm with a one-to-one function, the output is a random $n$-bit string. Simon used this fact to prove that there exists an algorithm for a quantum Turing machine that solves Simon's problem with zero error probability in expected time $O(nT_f(n)+G(n))$, where $T_f(n)$ is the time required to compute $f(x)$, and $G(n)$ is the time required to solve an $n\times n$ linear system of equations over $\mathbb{Z}_2$.

\chapter{Shor’s Integer Factoring Algorithm}\label{chap:Shor}

Shor's algorithms were presented at a conference in 1994~\cite{Sho94}. The full paper was published in 1997~\cite{Sho97}, and reviewed in 1999~\cite{Sho99}. It describes two quantum algorithms for integer factoring and discrete logarithms that run in polynomial time. The best-known classical algorithms run in sub-exponential time. Shor's algorithms exploit not only quantum parallelism but also entanglement, being a remarkable and celebrated scientific contribution to quantum computing. The algorithm for factoring integers is the focus of this Chapter and is described in many books~\cite{Hid19, KLM07, Mer07, NC00, NO08, RP11, Sch19, SS08, YM08}. Some knowledge of group theory is helpful.

\section{Problem formulation}

Let $N$ be a composite natural number. The computational challenge is to find a nontrivial factor of $N$. Since $N$ is composite, there exist natural numbers $n_1$ and $n_2$ such that $N = n_1n_2$, with $1 < n_1, n_2 < N$. The objective is to first identify $n_1$, and then determine $n_2$ by computing $N/n_1$.

If an algorithm exists that can efficiently find a nontrivial factor of $N$, then it can also efficiently find all prime factors of $N$, because the total number of prime factors of $N$, counted with multiplicity, is at most $\log_2 N$.

\section{Preliminaries on number theory}

Although factoring integers is of primary interest due to its significant impact on breaking cryptographic methods such as RSA, we can shift our focus to the problem of finding the multiplicative order of a natural number $a$ modulo $N$. This is because efficiently solving the latter problem also leads to an efficient solution for the factoring problem.

In number theory, the problem of finding the multiplicative order of a natural number $a$ modulo $N$ aims to determine the smallest positive integer $r$ such that
\[
a^r\equiv 1 \mod N.
\]
For example, let $N=21$, which is the largest composite number factored on a quantum computer so far using Shor's algorithm~\cite{ST21}. Now pick at random a number $a$ such that $1<a<N$. Let us say $a=2$. Then, we obtain the following sequence if we keep multiplying each line by $a$ and simplifying using modular arithmetic:
\begin{align*}\hspace{2.5cm}
\begin{split}
a &\equiv 2 \,\,\,\mod 21\\
a^2 &\equiv 4 \,\,\,\mod 21\\
a^3 &\equiv 8 \,\,\,\mod 21\\
a^4 &\equiv 16 \mod 21\\
a^5 &\equiv 11 \mod 21\\
a^6 &\equiv 1 \,\,\,\mod 21
\end{split}
\begin{split}
a^7 &\equiv 2 \,\,\,\mod 21\\
a^8 &\equiv 4 \,\,\,\mod 21\\
a^9 &\equiv 8 \,\,\,\mod 21\\
a^{10} &\equiv 16 \mod 21\\
a^{11} &\equiv 11 \mod 21\\
a^{12} &\equiv 1 \,\,\,\mod 21
\end{split}
\begin{split}
\hspace{1.2cm}&\cdots \hspace{2.5cm}
\end{split}
\end{align*}
The multiplicative order of 2 modulo 21 is 6 because 6 is the smallest positive integer such that $2^6\equiv 1\mod 21$. This can be observed at the end of the first column. If we continue this sequence, the results 2, 4, and so on will repeat again and again, as $a^{r+1}\equiv a^1\equiv 2$, $a^{r+2}\equiv a^2\equiv 4$, and so on. Consequently, we have an $r$-periodic sequence.

If $r$ is even, then
\begin{equation*}
\begin{split}
\big(a^\frac{r}{2}+1\big)\big(a^\frac{r}{2}-1\big) &\equiv\, a^r-1\\
&\equiv\, 0
\end{split}
\begin{split}
&\mod N\\
&\mod N.
\end{split}
\end{equation*}
We find two numbers, $a^{r/2}+1$ and $a^{r/2}-1$, whose product is a multiple of $N$. If these numbers are neither zero nor multiples of $N$, then $a^{r/2}+1$ and $a^{r/2}-1$ must have factors whose product is $N$. We conclude that in this case, $\text{gcd}\big(a^{r/2}+1, N\big) > 1$ and $\text{gcd}\big(a^{r/2}-1, N\big) > 1$, where gcd stands for the \textit{greatest common divisor}.

For example, when $a=2$ and $r=6$, we have $a^{{r/2}}+1=9$ and $a^{{r/2}}-1=7$, and in both cases the gcd returns a nontrivial factor of 21. The method fails when $a=5$ because the multiplicative order of 5 modulo 21 is $r=6$ and $a^{{r/2}}+1=126$ and $a^{{r/2}}-1=124$. In the first case, 126 is a multiple of 21, and in the second case, gcd$(124,21)=1$.

In Shor's algorithm, we start with the number $N$, then we pick uniformly at random a number $a$ such that $1<a<N$. Before calculating the order of $a$, we check whether gcd$(a,N)>1$ because (1)~if gcd$(a,N)>1$ then there is no $r$ such that $a^r\equiv 1 \mod N$ and (2)~if gcd$(a,N)>1$ then gcd$(a,N)$ is a nontrivial factor of $N$, and then we are done because the calculation of gcd is efficient using the Euclidean algorithm.\footnote{\url{https://en.wikipedia.org/wiki/Euclidean_algorithm}} To better understand what is going on here, we split the set of numbers $\{1,...,N-1\}$ into two subsets:
\begin{align*}
S_1&=\{a: 1<a<N \text{ and gcd}(a,N)>1\},\\
S_2&=\{a: 1\le a < N \text{ and gcd}(a,N)=1\}.
\end{align*}
In our example with $N=21$, we have
\begin{align*}
S_1&=\{3, 6, 7, 9, 12, 14, 15, 18\},\\
S_2&=\{1, 2, 4, 5, 8, 10, 11, 13, 16, 17, 19, 20\}.
\end{align*}
When we pick a number $a$ at random, if $a\in S_1$, we quickly find a factor of $N$ by calculating gcd$(a,N)$ using the Euclidean algorithm. The question now is what is the cardinality of $S_1$? Is it larger than the cardinality of $S_2$? To answer this question we use the following two facts about $S_2$ for an arbitrary $N$, which is denoted by $\mathbb{Z}_N^\times$ in number theory~\cite{HW75,NZM91}:

\

\noindent
\textbf{Fact 1} $\mathbb{Z}_N^\times$ is a finite multiplicative group modulo $N$.

\

\noindent
Fact~1 is the theoretical basis that guarantees the existence of $r$ such that $a^r\equiv 1\mod N$ for any $a\in \mathbb{Z}_N^\times$. It also guarantees that the function $f(\ell)=a^\ell \mod N$ is $r$-periodic. The quantum part of Shor's algorithm is fundamentally a period-finding algorithm whose goal is to determine the period $r$. Therefore, Shor's factoring algorithm can be viewed as a special case of a quantum algorithm for finding the period of a periodic oracle. In this sense, Shor's algorithm belongs to the same class of oracle-based quantum algorithms discussed in previous chapters, where the quantum Fourier transform is used to extract the period.

\

\noindent
\textbf{Fact 2} The cardinality of $\mathbb{Z}_N^\times$ is Euler's totient function $\varphi(N)$.

\

\noindent
Fact~2 concerns the definition of Euler's totient function $\varphi(N)$, which has been widely studied in number theory. It has been established that $\varphi(N)$ is always nearly $N$ (Hardy and Wright~\cite{HW75}). Given that the cardinality of $S_1$ is $N-\varphi(N)-1$, the probability of choosing $a\in S_1$ is much smaller than the probability of choosing $a\in \mathbb{Z}_N^\times$ when $N$ is large.

Not all $a$'s in $\mathbb{Z}_N^\times$ are suitable because the order of $a$ may be odd or $a^{r/2}+1\equiv 0\mod N$. If we select an unsuitable $a$, we have to discard it and randomly pick another one. The question now is: how many $a$'s in $\mathbb{Z}_N^\times$ have even order and, at the same time, satisfy $a^{{r/2}}+1\not\equiv 0\mod N$? To answer this question we use the following fact:

\

\noindent
\textbf{Fact 3} (Theorem A4.13 of~\cite{NC00}) Suppose $N=p_1^{\alpha_1}\cdots p_m^{\alpha_m}$ is the prime factorization of an odd composite positive integer. Let $a$ be chosen uniformly at random from $\mathbb{Z}_N^\times$, and let $r$ be the order of $a$ modulo $N$. Then the probability that $r$ is even and $a^{r/2}+1\not\equiv 0\mod N$ is at least $1- 1/2^m\ge 3/4$.

\

\noindent
Fact~3 states that the probability of selecting a good $a$ is at least 3/4. Note that the case $a^{r/2}-1\equiv 0 \mod N$ never happens because by definition $r$ is the smallest integer such that $a^{r}-1\equiv 0 \mod N$. Then, if $r$ is even and $a^{r/2}+1\not\equiv 0\mod N$, we have \text{gcd}$\big(a^{{r/2}}+1,N\big)>1$ and at the same time gcd$\big(a^{{r/2}}-1,N\big)>1$; the results of these gcd's are nontrivial factors of $N$.

In summary, a  nontrivial divisor of $N$ is readily obtained as soon as we find a nontrivial integer solution to the equation
\[
x^2\equiv 1 \mod N,
\]
such that $x \not \equiv \pm 1 \mod N$ ($x$ plays the role of $a^r$ here). This is because $x^2 - 1 \equiv 0$ implies $(x+1)(x-1) \equiv 0$, which in turn implies that there exists an integer $k$ such that $kN = (x+1)(x-1)$. If $p_1$ is a prime and divides $N$, then $p_1$ divides $x+1$ or $x-1$. If $p_1$ divides only $x-1$, then there exists another prime $p_2$ that divides $N$ such that $p_2$ also divides $x+1$. In any case, $\gcd(x+1,N) > 1$. The goal of Shor's algorithm is to find an integer $x$ that satisfies the equation above by guessing an integer $a$ with an even order $r$, under the restriction that $a^{r/2} \not \equiv \pm 1 \mod N$. The order $r$ is determined by running a quantum subroutine after $a$ is randomly chosen. Once suitable $a$ and $r$ have been found, we set $x = a^{r/2}$. By calculating the greatest common divisor between $x+1$ and $N$ using the Euclidean algorithm, Shor's algorithm then outputs a non-trivial factor of $N$.

\begin{exercise}
Assume that $N=p_1p_2$, where $p_1$ and $p_2$ are prime numbers of approximately the same size. This is the most interesting case in cryptography because it is the hardest case for classical factoring algorithms. In practical applications, those primes usually have 1024 bits, which is close to 308 decimal figures.
\begin{itemize}
\item[(a)] Show that the set $S_1$ is given by $S_1=\{p_1,2p_1,\ldots,(p_2-1)p_1,\,p_2,2p_2,\ldots,(p_1-1)p_2\}$.

\item[(b)] Deduce that
$|S_1|=p_1+p_2-2$,
and show that this quantity is approximately $2\sqrt{N}$.

\item[(c)] Using this result, explain why, for large $N$, the probability of choosing at random an integer $a$ such that $\gcd(a,N)>1$ is small, and therefore why a random choice of $a$ is likely to belong to $\mathbb{Z}_N^\times$.

\end{itemize}
\end{exercise}

\section{Quantum operator for modular exponentiation}

The unitary operator $U^{(a)}_N$, which calculates the exponentiation of integer $a$ modulo $N$, is defined as
\[
U^{(a)}_N\ket{\ell}\ket{y}=\ket{\ell}\ket{y\oplus \big(a^\ell\mod N\big)},\,
\]
for $0\le \ell<q$, $0\le y< 2^n$, where $n=\lceil \log_2 N\rceil$, $\oplus$ is the bitwise XOR operation or bitwise sum modulo 2, and $q$ is the smallest power of 2 such that $q>N^2$. $U^{(a)}_N$ acts on two registers with sizes $\log_2 q$ and $n$, and it is a permutation matrix of dimension $2^nq$. The inputs of the algorithm, $a$ and $N$, come through $U^{(a)}_N$. It replaces the oracle $U_f$ in Simon's algorithm, but $U^{(a)}_N$ is no oracle because it is our task to implement it. In Shor's algorithm, $U^{(a)}_N$ is used many times with a different $a$ each time. The proof that $U^{(a)}_N$ is unitary  is an extension of the proof presented in  Proposition~\ref{prop:U_f_is_unitary} on Page~\pageref{prop:U_f_is_unitary}.

To implement $U^{(a)}_N$ efficiently, it is necessary to use the \textit{repeated squaring method}, which is an efficient algorithm to calculate modular exponentiation. For instance, if we want to calculate $3^{16} \mod 7$, we would naively multiply $3\times 3\times ...\times 3$ sixteen times, obtain a large number, and then calculate the remainder after dividing by 7. In the exponentiation by squaring, we calculate the square of 3 modulo 7, then we calculate the square of the result modulo 7, and so on four times, which requires a logarithmic number of multiplications and each result is never too large. When we calculate $a^\ell\mod N$, $\ell$ is not a power of 2 in general, but the repeated squaring method still can be used. Using this method, it is possible to find an efficient circuit of $U^{(a)}_N$ by converting classical irreversible circuits into reversible ones~\cite{MS12, NO08, PG14}.

\section{Fourier transform and its inverse}\label{seq:shor_fourier_transform}

The $q$-dimensional Fourier Transform $F_q$ is a linear operator whose action on the computational basis is
\[
F_{q}\ket{k} \,=\, \frac{1}{\sqrt{q}}\sum_{\ell=0}^{q-1}\omega^{k\ell} \ket{\ell},
\]
where $0\le k<q$ and $\omega=\e^{2\pi\ii/q}$. Note that the $(k,\ell)$-entry of $F_q$ is
\[
\big(F_q\big)_{k\ell}=\frac{\omega^{k\ell}}{\sqrt{q}}.
\]
Then, $F_q$ is a symmetric matrix. To find $F_{q}^\dagger$, we simply take the complex conjugate of each entry, which is ${\omega^{-k\ell}}/{\sqrt{q}}$ because the complex conjugate of $\omega$ is $\omega^{-1}$. Then
\[
F_{q}^\dagger\ket{k} \,=\, \frac{1}{\sqrt{q}}\sum_{\ell=0}^{q-1}\omega^{-k\ell} \ket{\ell},
\]
where $0\le k<q$.

Let us show that $F_q$ is unitary. Using the definitions of $F_q$ and $F_{q}^\dagger$, we have
\begin{align*}
\bra{k'}F_{q}^\dagger F_{q}\ket{k} &= \left(\frac{1}{\sqrt{q}}\sum_{\ell'=0}^{q-1}\omega^{-k'\ell'} \bra{\ell'}\right)\left(\frac{1}{\sqrt{q}}\sum_{\ell=0}^{q-1}\omega^{k\ell} \ket{\ell}\right)\\
  &= \frac{1}{q}\sum_{\ell=0}^{q-1}\omega^{(k-k')\ell}.
\end{align*}
Now we use the closed-form formula for the geometric series with $q$ terms, which is
\[
\sum_{k=0}^{q-1} {s}^k \,=\, \frac{1-{s}^{q}}{1-{s}},
\]
if ${s}\neq 1$. When ${s}=1$, the left-hand sum is equal to $q$. In our case ${s}=\omega^{(k-k')}$. From the definition of $\omega$, we have $\omega^q=1$, and then $\omega^{(k-k')q}=1$.  Combining those results, we obtain
\[
\frac{1}{q}\sum_{\ell=0}^{q-1}\omega^{(k-k')\ell} \,=\,
\begin{cases}
			1, & \text{if $k=k'$,}\\
            0, & \text{otherwise.}
		 \end{cases}
\]
We have just shown that
\begin{align*}
\bra{k'}F_{q}^\dagger F_{q}\ket{k} &= \delta_{kk'},
\end{align*}
that is, $F_{q}^\dagger F_{q}=I$.

In Shor's algorithm, given $N$, $q$ is the smallest power of 2 such that $q> N^2$. $F_q$ is applied only to the first register. The number of qubits of the first register is $\log_2 q$, which is at most $2n$ and at least $2n-1$, while the number of qubits of the second register is exactly $n=\lceil \log_2 N\rceil$. $F_q$ plays a role similar to the Hadamard gates at the end of Simon's algorithm.
The circuit of $F_q$ in terms of CNOT and one-qubit gates is described in Section~\ref{sec:Shor_Fourier_Circ}.

\section{The algorithm}\label{sec_subsection_shor_algorithm}

\begin{algorithm} [!ht]
\caption{Shor's algorithm} \label{algo_ShorClassical}
\KwIn{Composite integer $N$.}
\KwOut{A nontrivial factor of $N$.} \BlankLine
If $N$ is even, return $2$; otherwise, continue\;
If $N$ is a power of some prime number $p$, return $p$; otherwise, continue\;
Pick uniformly at random an integer $a$ such that $1<a<N$\label{step-3}\;
If gcd$(a,N)>1$, return gcd$(a,N)$; otherwise, continue\;
Run the quantum part with inputs $a$ and $N$ \big(Algorithm~\ref{algo_ShorQuantum}, assume output $\ell_0,...,\ell_{m-1}$\big)\label{step-5}\;
If $\ell=0$, go to Step \ref{step-5}\;
Calculate $b = \ell/q$ (the same $q$ used in the quantum part)\;
Find the convergent of the continued fraction expansion of $b$ with the largest denominator $r'$ such that $r'<N$\;
If $r'$ is even, calculate $n_1=\text{gcd}(a^{r'/2} +1,N)$; otherwise, go to Step \ref{step-3}\; 
If $1<n_1<N$, return $n_1$; otherwise, go to Step \ref{step-3}.
\end{algorithm}

Shor's algorithm is described in Algorithm~\ref{algo_ShorClassical} and the quantum part is described in Algorithm~\ref{algo_ShorQuantum}. The circuit of the quantum part is \vspace{3pt}
\[
\Qcircuit @C=1.7em @R=0.9em {
\lstick{\ket{0}}            & \gate{H} & \multigate{5}{\,\,\,U^{(a)}_N\,\,\,} & \qw&\multigate{3}{F_{q}^\dagger} &\qw& \meter  & \rstick{\ell_0} \cw \\
\lstick{\vdots \,\,\,}      & {\vdots} &                                &&{}&&{\vdots} & \rstick{\vdots}  \\
\lstick{}                   &          &                                &&&         && \rstick{}  \\
\lstick{\ket{0}}            & \gate{H} & \ghost{\,\,\,U^{(a)}_N\,\,\,}        &\qw& \ghost{F_{q}^\dagger}&\qw& \meter  & \rstick{\ell_{{m}-1}}  \cw\\
 & & & & {\hspace{1.7cm}{}^\ket{\psi_4}} \\
\lstick{\ket{0}^{\otimes {n}}}& {/^{n}}\qw & \ghost{\,\,\,U^{(a)}_N\,\,\,}        & \meter  &  \rstick{z_0...z_{{n}-1},}\cw  \\
\lstick{}                   &          &                                &&         & \rstick{}  \\
\ustick{\hspace{0.7cm}\ket{\psi_0}} & \ustick{\hspace{1.3cm}\ket{\psi_1}}& \ustick{\hspace{1.9cm}\ket{\psi_2}} & \ustick{\hspace{1.4cm}\ket{\psi_3}}
}\vspace{0.0cm}
\]
where $q$ is the least power of 2 such that $q>N^2$, $m=\log_2 q$, $n=\lceil \log_2 N\rceil$, and $a$ is an integer such that gcd$(a,N)=1$. The first register has at least $2n-1$ qubits (at most $2n$) and the second has exactly $n$ qubits. The states at the bottom of the circuit are used in the analysis of the algorithm. They describe the states of the qubits after each step. State $\ket{\psi_4}$ refers to the first register only. The notation ``$/^{n}$'' over a wire denotes that it is a $n$-qubit register.

\begin{algorithm} [!ht]
\caption{Quantum part of Shor's algorithm} \label{algo_ShorQuantum}
\KwIn{A composite integer $N$ and integer $1<a<N$ such that gcd$(a,N)=1$.}
\KwOut{$m$-bit string $\ell$ that is the nearest integer to a multiple of $q/r$ with probability greater than $3/\pi^2$, where $q=2^m$ is the smallest power of 2 such that $q>N^2$.} \BlankLine
Prepare the initial state $\ket{0}^{\otimes m}\ket{0}^{\otimes n}$, where $m=\lceil 2\log_2 N\rceil$ and $n=\lceil\log_2 N\rceil$\;
Apply $H^{\otimes \log_2 q}$ to the first register\;
Apply $U^{(a)}_N$ to both registers\;
Measure the second register in the computational basis (assume output $z_0...z_{n-1}$)\;
Apply $F^\dagger_{q}$ to the first register\;
Measure the first register in the computational basis and return the result.
\end{algorithm}

We have presented Shor's algorithm as a Las Vegas algorithm, which means that the output is always correct and the expected runtime is finite. With a small modification, it can be presented as a Monte Carlo algorithm, which means that the output may be incorrect. Then, we focus on the quantum part and we have to show that the probability of returning a nearest integer $\ell$ to a multiple of $q/r$ is lower bounded by $3/\pi^2$ because this $\ell$ allows us to obtain $r$ using a continued fraction expansion of $\ell/q$.

\section{Analysis of the quantum part}\label{sec:shor_analysis}

\subsection*{Calculation of $\ket{\psi_0}$}
After the first step, the state of the qubits is
\[
\ket{\psi_0}=\ket{0}^{\otimes m}\ket{0}^{\otimes n},
\]
where $m=\log_2 q=\lceil 2\log_2 N\rceil$.
\subsection*{Calculation of $\ket{\psi_1}$}
After the second step, the state of the qubits is
\begin{eqnarray*}
\ket{\psi_1} &=& \big(H\ket{0}\big)^{\otimes m}\otimes\ket{0}^{\otimes n}\\
&=& \frac{1}{\sqrt{q}}\sum_{\ell=0}^{q-1}\ket{\ell}\otimes \ket{0}^{\otimes n},
\end{eqnarray*}
where $\ell$ is written in the decimal notation and $q=2^m$.
\subsection*{Calculation of $\ket{\psi_2}$}
After the third step, the state of the qubits is
\begin{eqnarray*}
\ket{\psi_2} &=& U^{(a)}_N\ket{\psi_1}\\
&=& \frac{1}{\sqrt{q}}\sum_{\ell=0}^{q-1}U^{(a)}_N(\ket{\ell}\otimes\ket{0\cdots 0}).
\end{eqnarray*}
To simplify $\ket{\psi_2}$, we use the definition of $U^{(a)}_N$ to obtain
\begin{eqnarray*}
\ket{\psi_2} &=& \frac{1}{\sqrt{q}}\sum_{\ell=0}^{q-1}\ket{\ell}\otimes\ket{a^\ell\mod N},
\end{eqnarray*}
because $(0...0)\oplus a^\ell$ (bitwise XOR) is $a^\ell$. From now on, we drop the notation modulo~$N$ inside the second ket because we have no XOR operation and there is no danger in failing to recognize the correct arithmetic.

It is really important to understand the structure of $\ket{\psi_2}$ before proceeding. Expanding the sum, we obtain
\begin{equation*}
\begin{split}
\sqrt{q}\ket{\psi_2} \,\,\,\,\,\,=\,\,\,\,\, \ket{0}&\ket{1}\,\,\,+\,\,\,\\
               \ket{r}&\ket{1}\,\,\,+\,\,\,\\
               \ket{2r}&\ket{1}\,\,\,+\,\,\,\\
                  &\vdots\\
               \ket{(c-1)r}&\ket{1}\,\,\,+\,\,\,\\
\end{split}
\begin{split}
\ket{1}&\ket{a} \,\,\,+\,\,\,\\
\ket{r+1}&\ket{a} \,\,\,+\,\,\,\\
\ket{2r+1}&\ket{a} \,\,\,+\,\,\,\\
&\vdots\\
\ket{(c-1)r+1}&\ket{a} \,\,\,+\,\,\,\\
\end{split}
\begin{split}
\ket{2}&\ket{a^2} \,\,\,+\,\,\,\\
\ket{r+2}&\ket{a^2} \,\,\,+\,\,\,\\
\ket{2r+2}&\ket{a^2} \,\,\,+\,\,\,\\
&\vdots\\
\ket{(c-1)r+2}&\ket{a^2} \,\,\,+\,\,\,\\
\end{split}
\begin{split}
\cdots& \,\,\,+\,\,\,\\
\cdots& \,\,\,+\,\,\,\\
\cdots& \,\,\,+\,\,\,\\
&\color{white}{\vdots}\\
.\,.\,.\,,& \,\,\,\,\,\,\\
\end{split}
\begin{split}
\ket{r-1}&\ket{a^{r-1}} \,\,\,+\,\,\,\\
\ket{2r-1}&\ket{a^{r-1}} \,\,\,+\,\,\,\\
\ket{3r-1}&\ket{a^{r-1}} \,\,\,+\,\,\,\\
&\vdots\\
\\
\end{split}
\end{equation*}
where $c=\lceil q/r \rceil$. Indeed, we have $q=(c-1)r+r_0$, where $r_0$ is the remainder of $q$ divided by $r$. The last line has $r_0$ terms and is incomplete unless $r_0=0$.\footnote{The last line is complete if $r$ divides $q$. In this case, $r$ is a power of 2.} The possible values inside the second ket are $1$, $a$, $a^2$, ..., $a^{r-1}$. We have split $\ket{\psi_2}$ into columns that have the same second ket. The first columns have $c$ terms and the last columns have $c-1$ terms.

\subsection*{Calculation of $\ket{\psi_3}$}
The fourth step is a measurement of each qubit of the second register, which we assume has returned $z_0,...,z_{n-1}$. Then there exists $r_1$ such that $a^{r_1}=z$, where $0\le r_1<r$. State $\ket{\psi_2}$ collapses to a superposition of the first register with all terms $\ket{\ell}$ such that $a^\ell\equiv a^{r_1} \mod N$, that is, $\ell=kr+r_1$, for $0\le k<c$, yielding
\begin{eqnarray*}
\ket{\psi_3} &=&  \left(\frac{1}{\sqrt{c}}\sum_{k=0}^{c-1}\ket{k r+r_1}\right)\ket{a^{r_1}},
\end{eqnarray*}
where $c=\lceil q/r \rceil$ if $a^{r_1}$ belongs to one of the first columns $(r_1< r_0)$, and $c=\lfloor q/r \rfloor$ if $a^{r_1}$ belongs to one of the last columns $(r_1\ge r_0)$. The analysis of the algorithm uses parameter $c$ extensively. It is important to memorize its definition and be aware that it is an integer.

We have renormalized the state $\ket{\psi_3}$ as required by the measurement postulate. We do not know $r_1$ because it is randomly selected from 0 to $r-1$. The information we want to acquire, $r$, is concealed because $kr+r_1$ is a random value from 0 to $q-1$. Performing a measurement of the first register at this point would be futile. However, the probability distribution of the first register is an $r$-periodic function. This motivates the application of the inverse Fourier transform because the result is another (almost) periodic function, with a period close to $q/r$.

\subsection*{Calculation of $\ket{\psi_4}$: last state}
In the fifth step, we only consider the first register. After applying $F_q^\dagger$ to the state of the first register,\footnote{The algorithm uses $F^\dagger_q$ and not $F_q$ at this point because the goal is to convert a superposition state into a state of the computational basis, not the other way around.} we obtain
\begin{eqnarray*}
\ket{\psi_4} &=&  \frac{1}{\sqrt{c}}\sum_{k=0}^{c-1}F_q^\dagger\ket{k r+r_1}\\
             &=& \frac{1}{\sqrt{c\,q}}\sum_{k=0}^{c-1}\sum_{\ell=0}^{q-1}\omega^{-\ell(kr+r_1)}\ket{\ell}.
\end{eqnarray*}
Rearranging the order of the sums, we obtain
\begin{eqnarray*}
\ket{\psi_4} &=&  \frac{1}{\sqrt{q}}\sum_{\ell=0}^{q-1}\omega^{-\ell r_1}\left(\frac{1}{\sqrt{c}}\sum_{k=0}^{c-1}\omega^{-\ell k r}\right)\ket{\ell}.
\end{eqnarray*}
To calculate the sum inside the parentheses, we use again the closed-form formula for the geometric series, which is
\[
\sum_{k=0}^{c-1} {s}^k \,=\, \frac{1-{s}^{c}}{1-{s}},
\]
if ${s}\neq 1$. When ${s}=1$, the left-hand sum is equal to $c$.
For the sum inside the parentheses,  ${s}=\omega^{-\ell r}=\e^{-2\pi\ii \ell r/q}$.  Then
\[
\sum_{k=0}^{c-1}\omega^{-\ell k r} \,=\,
\begin{cases}
			{c}, & \text{if $q\mid (\ell r)$,}\\
            \frac{1-\omega^{-\ell c r}}{1-\omega^{-\ell r}}, & \text{otherwise,}
		 \end{cases}
\]
where the notation $q\mid (\ell r)$ means that $q$ divides $\ell r$.
State $\ket{\psi_4}$ simplifies to
\begin{eqnarray}\label{eq:shor_psi_4_final}
\ket{\psi_4} &=&   \frac{\sqrt{c}}{\sqrt{q}}\sum_{\substack{\ell=0\\ q \mid (\ell r)}}^{q-1}\, \omega^{-\ell r_1}\ket{\ell}
+
\frac{1}{\sqrt{q c}}\sum_{\substack{\ell=0\\ q \nmid (\ell r)}}^{q-1}\omega^{-\ell r_1} \frac{1-\omega^{-\ell c r}}{1-\omega^{-\ell r}} \ket{\ell},
\end{eqnarray}
where the notation $q \nmid (\ell r)$ means that $q$ does not divide $\ell r$. The first sum is over $\ell$ such that $q$ is a divisor of $(\ell r)$. The second sum is over the remaining $\ell$'s.

\subsection*{Calculation of the output}

The probability of obtaining $0\le\ell< q$ is
\begin{eqnarray*}
p(\ell) &=&
\begin{cases}
			\frac{c}{q}, & \text{if $q\mid (\ell r)$,}\\
            \frac{1}{{q\,c}}\left|\frac{1-\omega^{-\ell c r}}{1-\omega^{-\ell r}}\right|^2, & \text{otherwise,}
		 \end{cases}
\end{eqnarray*}
where $c$ was defined when we calculated $\ket{\psi_3}$.
Using the definition of $\omega$ and $\left|1-\e^{2\pi \ii \theta}\right|^2=4\sin^2(\pi\theta)$, $p(\ell)$ simplifies to
\begin{eqnarray}\label{shor_prob(l)}
p(\ell) &=&
\begin{cases}
			\frac{c}{q}, & \text{if $q\mid (\ell r)$,}\\
            \frac{\sin^2 \frac{\pi \ell r c }{q}}{{q\,c\,}\sin^2 \frac{\pi \ell r}{q}}, & \text{otherwise.}
		 \end{cases}
\end{eqnarray}
Fig.~\ref{fig:Shorprob} depicts an example of the probability distribution $p(\ell)$.

\begin{figure}[!ht]
\centering
\includegraphics[scale=0.3]{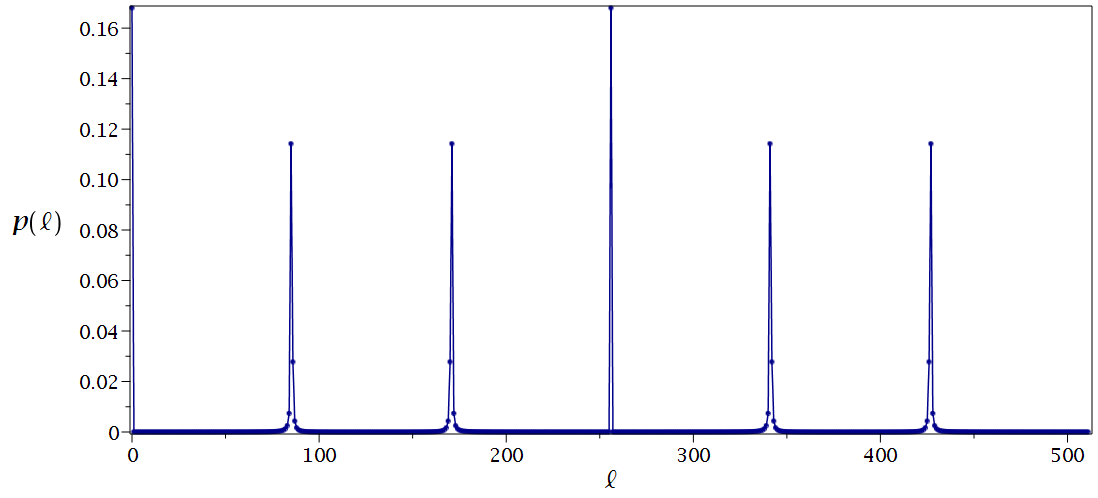}
\caption{Probability distribution $p(\ell)$ as a function of $\ell$ given by Eq.~(\ref{shor_prob(l)}) when $N=21$ ($q=2^9$, $a=2$, $r=6$, $c=85$). The dots at the tops of the peaks correspond to $\ell=$ 0, 85, 171, 256, 341, 427.}\label{fig:Shorprob}
\end{figure}

When $q \nmid (\ell r)$, it is likely that a measurement outputs $\ell$ satisfying
\[
\sin^2 \frac{\pi \ell r}{q} \approx 0
\]
because $p(\ell)$ is large when the denominator is close to zero. This implies that $\pi \ell r/q$ must be close to a multiple of $\pi$ (say $k\pi$), and then the $\ell$'s with the highest chances are
\[
\ell \approx \frac{kq}{r},
\]
where $k$ runs from 0 to $r-1$. If $\ell$ is zero or an exact multiple of $q/r$ then $p(\ell)=c/q\approx 1/r$ --- see the definition of $p(\ell)$ in Eq.~(\ref{shor_prob(l)}). This analysis explains why the peaks of Fig.~\ref{fig:Shorprob} correspond to $\ell$'s that are close to ${kq}/{r}$.

In summary, the output of the quantum part is a $m$-bit string $\ell$ such that $\ell$ is close to a multiple of $q/r$, where  $m=\log_2 q$ (number of qubits of the first register).

\subsection*{Details about the probability distribution}

In the previous analysis, it is missing to show that the numerator of $p(\ell)$ when $q \nmid (\ell r)$, the term
\[
\sin^2 \frac{\pi \ell r c }{q},
\]
is not too small when $\ell\approx kq/r$ for $0\le k< r-1$. Since this analysis is too long when $\ell$ is a discrete variable, we take an alternative route.

Let us look at $p(\ell)$ as a continuous function in terms of $\ell$ in the domain $[0,q]$. By using trigonometric identities and the fact that $c$ is an integer, it is straightforward to show that
\[
p\left(\ell + \frac{q}{r}\right) = p(\ell).
\]
When we look at $p(\ell)$ as a continuous function, we are able to show that it is a truly periodic function, while the function depicted in Fig.~\ref{fig:Shorprob} is not.

Now let us obtain the shape of $p(\ell)$. Since $p(\ell)$ is $(q/r)$-periodic, let us consider the interval $\ell\in [0,q/r]$, and besides let us restrict to $\ell$'s such that $q \nmid (\ell r)$. We use the expression
\begin{equation}\label{eq:p_ell_part}
p(\ell) \,=\, \frac{\sin^2 \frac{\pi \ell r c }{q}}{{q\,c\,}\sin^2 \frac{\pi \ell r}{q}}.
\end{equation}
The numerator of $p(\ell)$, $\sin^2(\pi \ell r c /q)$, is the square of a sinusoidal function, which is $(q/rc)$-periodic. Note that $q/rc$ is exactly 1 if $r\mid q$ and is close to 1 if $r \nmid q$ because $c$ is either $\lceil q/r\rceil$ or $\lfloor q/r\rfloor$. An example of the numerator of $p(\ell)$ is depicted in Fig.~\ref{fig:Shorp123}(a) for $\ell\in [0,q/r]$.

\begin{figure}[!ht]
\centering
\includegraphics[scale=0.125]{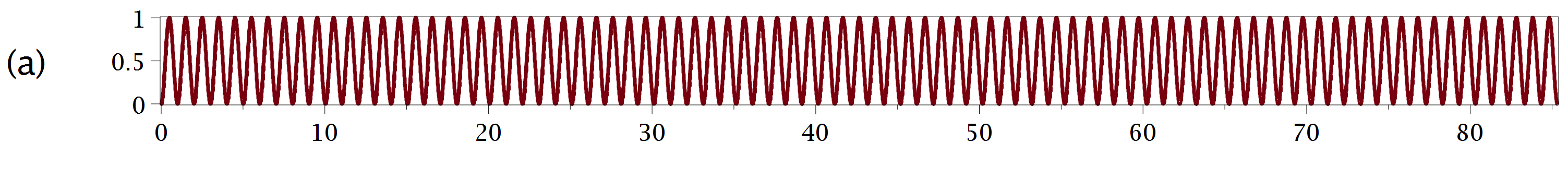}
\includegraphics[scale=0.125]{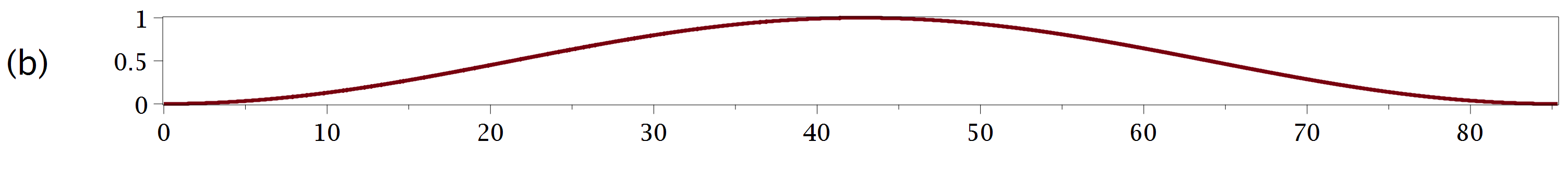}
\includegraphics[scale=0.125]{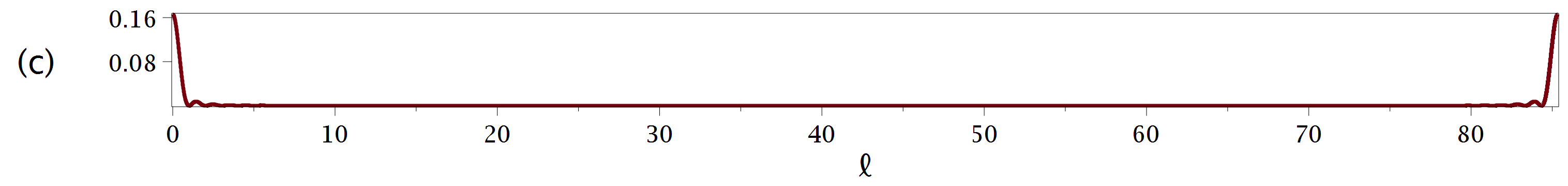}
\caption{(a) Numerator of $p(\ell)$ as a continuous function of $\ell$ for  $\ell\in [0,q/r]$ when $N=21$ ($q=2^9$, $r=6$, $c=85$). (b)~Denominator of $p(\ell)$ as a continuous function of $\ell$ without $qc$. (c)~$p(\ell)$ as a continuous function of $\ell$ for $\ell\in [0,q/r]$. Plot (c) is obtained by dividing (a) by (b) by $qc$.}\label{fig:Shorp123}
\end{figure}

The denominator of $p(\ell)$ (without $qc$), $\sin^2(\pi \ell r /q)$, is also the square of a sinusoidal function, which is zero only at $\ell=0$ and $\ell=q/r$. An example of the denominator of $p(\ell)$ (without $qc$) for the same values of $N$, $q$, and $r$ is depicted in Fig.~\ref{fig:Shorp123}(b) for $\ell\in [0,q/r]$.

If we divide the plot shown in Fig.~\ref{fig:Shorp123}(a) by the plot shown in Fig.~\ref{fig:Shorp123}(b) and divide the result by $qc$, we obtain the plot shown in Fig.~\ref{fig:Shorp123}(c), which is the continuous version of the plot shown in Fig.~\ref{fig:Shorprob} for $0\le \ell< q/r$. It is straightforward to check two facts: (1)~the largest values of $p(\ell)$ are close to $\ell=0$ and $\ell=q/r$ because the denominator of $p(\ell)$ is close to zero; and (2)~the smallest values of $p(\ell)$ are in the middle ($\ell$ close to $q/2r$) because the denominator of $p(\ell)$ is close to $qc$, which is large, far from zero. The numerator of $p(\ell)$ oscillates quickly between 0 and 1. Many more facts about $p(\ell)$ can be obtained from this analysis; for instance, it is easy to find the number of zeroes of $p(\ell)$. Using our knowledge of calculus, we check that
\[
\lim_{\ell=0^+} \frac{\sin^2 \frac{\pi \ell r c }{q}}{\sin^2 \frac{\pi \ell r}{q}}\,=\,
\lim_{\ell=\frac{q}{r}^-} \frac{\sin^2 \frac{\pi \ell r c }{q}}{\sin^2 \frac{\pi \ell r}{q}}\,=\, c^2.
\]
Then, $p(0)=p(q/r)=c/q\approx 1/r$. With this analysis, we have obtained the shape of $p(\ell)$ in the whole domain because $p(\ell)$ is $(q/r)$-periodic, that is, if we put 6 plots of Fig.~\ref{fig:Shorp123}(c) side by side, we obtain the continuous version of the plot of Fig.~\ref{fig:Shorprob}.

\begin{figure}[!ht]
\centering
\includegraphics[scale=0.3]{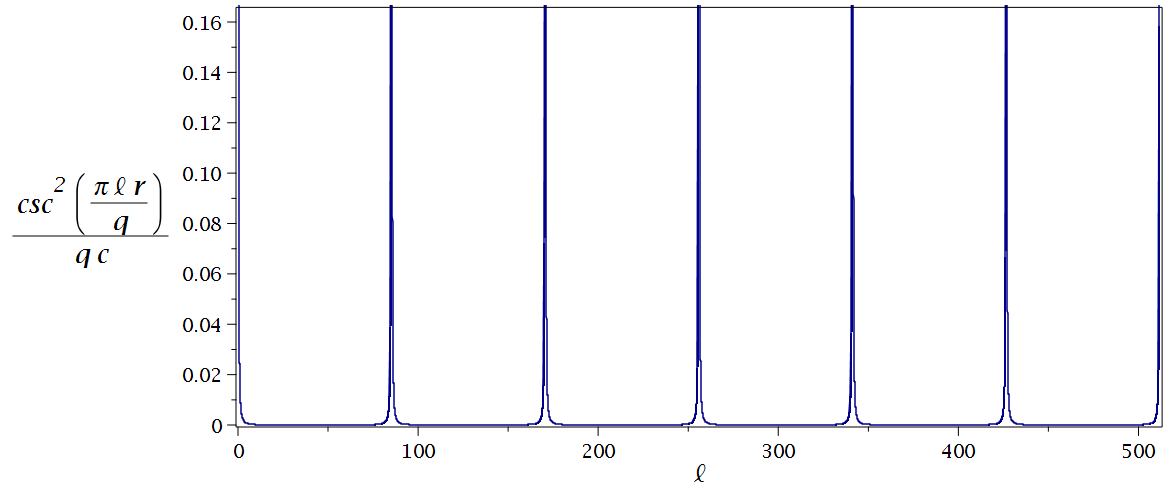}
\caption{The upper part of the envelope of $p(\ell)$ as a function of $\ell$ in the domain $[0,q]$  when $N=21$ ($q=2^9$, $r=6$, $c=85$), which is obtained using the square of the cosecant function ($\csc x=1/\sin x$).}\label{fig:Shorprob_envelope}
\end{figure}

There is an alternative route to obtain the shape of $p(\ell)$ in the domain $[0,q]$. The numerator of $p(\ell)$, $\sin^2 \frac{\pi \ell r c }{q}$, is the square of a sinusoidal function quickly oscillating inside an envelope, the lower part of which is the $\ell$-axis and the upper part is the function
\begin{equation*}
\frac{\csc^2\left(\frac{ \pi \ell r}{q}\right)}{qc},
\end{equation*}
which is obtained from Eq.~(\ref{eq:p_ell_part}) by replacing the numerator of $p(\ell)$ with 1. Fig.~\ref{fig:Shorprob_envelope} depicts the upper part of the envelope in the domain $[0,q]$, which must be compared with Fig.~\ref{fig:Shorprob}. Note that the envelope has vertical asymptotes for each $\ell$ multiple of $q/r$.

\begin{figure}[!ht]
\centering
\includegraphics[scale=0.2]{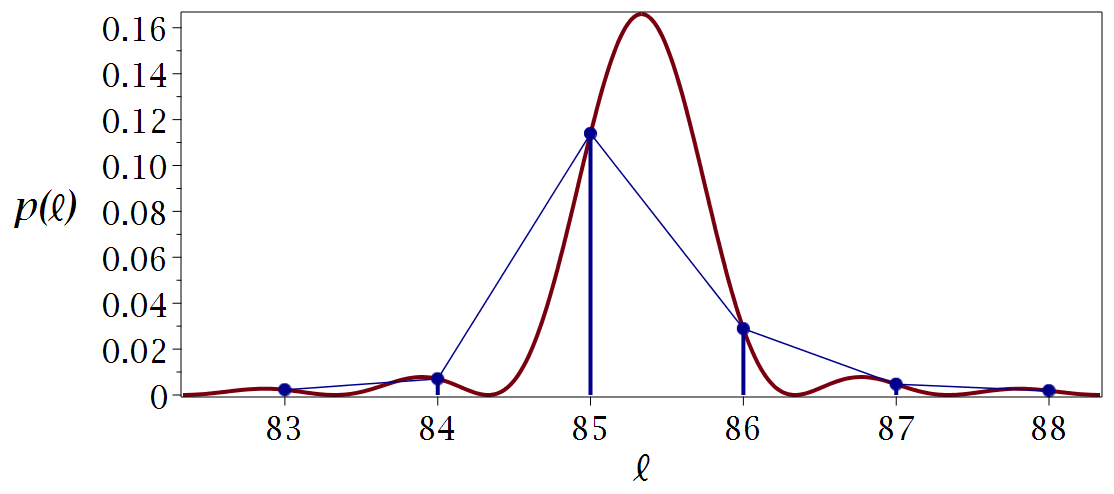}
\caption{First peak of the probability distribution $p(\ell)$ as a function of $\ell$ when $N=21$ ($q=2^9$, $r=6$, $c=85$). The red curve is the continuous version of $p(\ell)$ and the height of the peak is $1/r$. The values of $p(\ell)$ for integer $\ell$'s are highlighted in blue. The blue plot is a stretched version of the first peak of Fig.~\ref{fig:Shorprob}.}\label{fig:Shorp4}
\end{figure}

Let us set aside the envelope for now and focus on the shape of the first peak of $p(\ell)$. Fig.~\ref{fig:Shorp4} depicts the first peak in detail when $N=21$, $q=2^9$, and $r=6$, showing not only the continuous version in red but also the discrete values in blue when $\ell$ is an integer. The points in blue are the same ones as depicted in the first peak of Fig.~\ref{fig:Shorprob}, and at least four of them are clearly recognizable. Note that the first peak of Fig.~\ref{fig:Shorprob} does not reach $c/q$ because the peak is cut off before reaching the underlying summit. The underlying summit is reached only for $\ell$'s such that $q \mid (\ell r)$.

The particular case $r\mid q$ ($r$ is a power of 2) is noteworthy. In this case, $c$ is exactly equal to $q/r$, which is an integer, and $p(\ell)$ reaches the summit, whose height is exactly $1/r$, for all $\ell$'s that are multiples of $q/r$ and $p(\ell)$ goes to the bottom of the valley ($p(\ell)=0$) for all other integer values of $\ell$. In Eq.~(\ref{eq:shor_psi_4_final}), $\ket{\psi_4}$ contains only the first sum because all terms of the second sum vanish, since $\omega^{-\ell c r} = \omega^{-\ell q} = 1$. Here, the discrete version of the probability distribution is truly periodic with a period of $q/r$. If the result of the measurement is $\ell$, we obtain a candidate for $r$ by taking the denominator of the reduced fraction of $\ell/q$. This particular case arises because we are using qubits, and thus the dimension of the Hilbert space is a power of 2. Looking forward, in the general case, the concept of continued fraction expansion is necessary to accurately determine $r$.

\subsection*{Success probability of the quantum part}

The main focus of the analysis is the calculation of the success probability after one round of the quantum part of the algorithm. We have already gathered enough information about the probability distribution; we are now ready for the calculation. The trigonometric inequalities
\[
{4\alpha ^2} \le \sin^2 (\pi \alpha ) \le \pi^2\alpha ^2\,\text{ for }\, |\alpha |\le \frac{1}{2},
\]
can be proved using basic calculus. Using both sides of that inequality, we prove the following proposition:
\begin{proposition} \label{prop:shor_sin_2}
If $|\alpha |\le {1}/{2c}$ and $c\ge 1$, then
\[
\frac{\sin^2 {\pi \alpha  c }}{\sin^2 {\pi \alpha }} \ge \frac{4c^2}{\pi^2}.
\]
\end{proposition}\vspace{5pt}

Using this proposition, we are able to find an interesting lower bound on the success probability. Let us start by calculating a lower bound on the probability $p(\ell)$ when $\ell$ is the nearest integer to a multiple of $q/r$. Suppose that $q \nmid (\ell r)$ and take $\ell=\lfloor kq/r\rceil$ for $1\le k<r$, where $\lfloor\,\rceil$ is the notation for the nearest integer. Then,
\[
p\left(\left\lfloor \frac{kq}{r}\right\rceil \right) \,=\, \frac{\sin^2 \left(\pi \left\lfloor \frac{kq}{r}\right\rceil\frac{ r c }{q}\right)}{{q\,c\,}\sin^2 \left(\pi \left\lfloor \frac{kq}{r}\right\rceil\frac{ r}{q}\right)}.
\]
Now we use the trigonometric identity $\sin^2 \alpha=\sin^2(\pi k'-\alpha)$ valid for any integer $k'$ to obtain
\[
p\left(\left\lfloor \frac{kq}{r}\right\rceil \right) \,=\, \frac{\sin^2 \pi\alpha c}{{q\,c\,}\sin^2 \pi\alpha },
\]
where
\[
\alpha \,=\, k  - \left\lfloor \frac{kq}{r}\right\rceil\frac{ r }{q}.
\]
Using that
\begin{equation}\label{eq:shor_bound_cont_frac}
\left|\, \frac{kq}{r}- \left\lfloor \frac{kq}{r}\right\rceil \,   \right| \le \frac{1}{2},
\end{equation}
we obtain
\[
|\alpha| \le \frac{r}{2q}  \le \frac{1}{2\left\lfloor\frac{q}{r}\right\rfloor}=\frac{1}{2c},
\]
where $c=\lfloor q/r \rfloor$ (we will not address the case $c=\lceil q/r \rceil$). Using Proposition~\ref{prop:shor_sin_2}, we obtain
\[
p\left(\left\lfloor \frac{kq}{r}\right\rceil \right) \ge \frac{4c}{\pi^2q}.
\]
Using $c\ge q/r-1$ and $q>N^2>rN$, we obtain
\begin{equation}\label{eq_shor_lb}
p\left(\left\lfloor \frac{kq}{r}\right\rceil \right) > \frac{4}{\pi^2 r}\left(1-\frac{1}{N}\right).
\end{equation}
We can visualize this result by noticing that there are two blue bars inside the peak of $p(\ell)$ in Fig.~\ref{fig:Shorp4}. Assuming that $N>4$, we have just shown that the height of the tallest bar is larger than ${3}/{(\pi^2 r)}$, that is, it is never smaller than 30\% of the height of the peak of the continuous version.

In the quantum part of the algorithm, there are $r$ possible $k$'s. Then, a lower bound on the success probability is
\[
p_\text{succ} >  \frac{3}{\pi^2},
\]
which means that the quantum part has at least a 30\% chance of returning an $m$-bit string $\ell$ that is the nearest integer to a multiple of $q/r$. Numerical calculations with $N$ up to four digits show that the lower bound can be improved to at least 70\%. Note that this success probability includes the trivial result $\ell=0$, which is a multiple of $q/r$. If $r$ is a power of 2, the success probability is 1 in the sense that the output is an exact multiple of $q/r$.

\section{Analysis of the classical part}

Before running the quantum part of the algorithm, we need to do a classical checklist. It is necessary to check whether $N$ is a composite number, which can be completed efficiently using primality-testing algorithms\footnote{\url{https://en.wikipedia.org/wiki/Primality_test}}. We quickly check whether $N$ is even.  Furthermore, there are efficient classical algorithms to check whether $N$ is a power of some prime number $p$~\cite{Ber98}. After confirming that $N$ is a composite odd number and not a power of a prime number, we randomly choose an integer $a$ such that $1<a<N$ and check whether gcd$(a,N)=1$, which can be done efficiently using the Euclidean algorithm.

After running the quantum part of the algorithm, we assume that the output is a string $\ell$ such that $\ell$ is the nearest integer to a multiple of $q/r$, that is
\[
\ell = \left\lfloor \frac{kq}{r}\right\rceil
\]
for some $k$ such that $0\le k <r$. If $\ell=0$, which happens with a negligible probability when $N$ is not too small, we have to rerun the quantum part of the algorithm. Define
\[
b=\frac{\ell}{q},
\]
which obeys $0< b<1$. Now we use the method of continued fraction approximation to obtain the desired information $r$. A continued fraction expansion of a positive rational number $b<1$ is
\[
b=\frac{1}{b_1+\frac{1}{b_2+\frac{1}{\ddots\,\,+\frac{1}{b_z}}}},
\]
where $b_1$ to $b_z$ are positive integers, and $z$ is a positive integer. The notation for the continued fraction is $[b_1,b_2,...,b_z]$, and the successive convergents are $[b_1]$, $[b_1,b_2]$, $[b_1,b_2,b_3]$, and so on, each one closer to $b$, until the last one, which is equal to $b$. Each convergent is converted into a rational number by truncating the continued fraction expansion. In the algorithm, we have to find the convergent $[b_1,b_2,...,b_j]$ that has the largest $j$ such that the denominator of the equivalent rational number is less than $N$. The denominator of this convergent is the candidate for $r$.

For example, the successive convergents of $\ell/q$ for $\ell=85$ and $q=2^9$ are $[6]=1/6$, $[6,42]=42/253$, and $[6,42,2]=85/512=\ell/q$. When $N=21$ and $a=2$, we have to select the convergent $[6]$ and then the candidate for $r$ is the denominator of $1/6$. Luckily, $a^6\equiv 1\mod N$. Our luck would fail if we picked $\ell=171$ (see Fig.~\ref{fig:Shorprob}) because we would conclude that $r=3$.

Now we can see why we have to demand that $q>N^2$. It is a consequence of a theorem proved at the end of the chapter on continued fraction expansion in Hardy and Wright's book~\cite{HW75}.
\begin{theorem}\label{theo:hardy_wright} (Hardy and Wright)
If $k$ and $r$ are positive integers and $b$ is a positive real number and
\[
\left|\, \frac{k}{r}- b \,   \right| < \frac{1}{2r^2},
\]
then $k/r$ is a convergent of the continued fraction expansion of $b$.
\end{theorem}

\noindent
In Eq.~(\ref{eq:shor_bound_cont_frac}), we have shown that the output $\ell$ of the quantum part obeys
\begin{equation}\label{eq:shor_r_q_2q}
\left|\, \frac{k}{r}- \frac{\ell}{q} \,   \right| \le \frac{1}{2q},
\end{equation}
for some integer $k$ such that $0\le k < r$, which is unknown to us as well as $r$. In order to use Theorem~\ref{theo:hardy_wright}, we have to demand that $q>N^2$ because $N$ is an upper bound for $r$, that is, $1/q<1/N^2<1/r^2$. Then, we can use the theorem and be sure that $k/r$ is a convergent of $\ell/q$.

To understand why we have to pick the convergent $[b_1,b_2,...,b_j]$ that has the largest $j$ such that the denominator of the equivalent rational number is less than $N$, we need to consider the following facts. If we pick a convergent $[b_1,b_2,...,b_{j'}]=k'/r'$ that obeys Eq.~(\ref{eq:shor_r_q_2q}) and $r'<N$ and it does not have the largest $j'$, then the next convergent, let us say $[b_1,b_2,...,b_{j'+1}]=k''/r''$, also obeys Eq.~(\ref{eq:shor_r_q_2q}) and $r''<N$. We obtain a contradiction because on the one hand using Eq.~(\ref{eq:shor_r_q_2q}) twice we have
\[
\left| \frac{k'}{r'}-\frac{k''}{r''} \right| = \left| \left(\frac{k'}{r'}-b\right)-\left(\frac{k''}{r''}-b\right) \right| \le \frac{1}{q} < \frac{1}{N^2},
\]
and on the other hand we have
\[
\left| \frac{k'}{r'}-\frac{k''}{r''} \right| = \frac{ \left| k'r''-k''r' \right| }{r'r''}> \frac{1}{N^2}.
\]
The last inequality follows from the inequalities $\left|k'r''-k''r'\right|\ge 1$ and $r'<N$ and $r''<N$. In conclusion, there is only one convergent $[b_1,b_2,...,b_j]$ that obeys Eq.~(\ref{eq:shor_r_q_2q}) such that the denominator of the equivalent fraction is less than $N$; it is the one with the largest $j$.

The analysis is not complete yet because it may happen that gcd$(k,r)>1$. In this case, when we look at the denominator of $k/r$, we obtain a factor of $r$, not $r$ itself. In the example above with $\ell=171$, we have $k=2$ and $r=6$ and the denominator of $k/r$ yields a wrong result. We have to discard those cases and we ask how many $k$'s relatively prime with $r$ are there. The proportion of good $k$'s is $\varphi(r)/r$, where $\varphi(r)$ is Euler's totient function. There is a lower bound for $\varphi(r)$ given by~\cite{HW75,RS62}
\[
\varphi(r) > \frac{r}{4\ln(\ln (r))},\, \text{ for }\, r\ge 7.
\]
Then, a lower bound for the probability that the output $\ell$ is the nearest integer to a multiple of $q/r$ and gcd$(k,r)=1$ is
\[
p\left(\ell=\left\lfloor \frac{kq}{r}\right\rceil \text{ and gcd}(k,r)=1\right)\,>\,\frac{3}{4\pi^2\ln(\ln (r))},
\]
for $r\ge 7$.

If we consider the factoring algorithm as a Monte Carlo algorithm, which is obtained by removing the ``go to'' statements of Algorithm~\ref{algo_ShorClassical}, the algorithm will run only one time and a lower bound on the overall success probability is
\[
\frac{9}{16\pi^2\ln(\ln (r))},
\]
which is obtained from the analysis of this Section and Fact~3 \big($r$ must be even and gcd$(a^{r/2}+1,N)>1$\big). If we think of the factoring algorithm as a Las Vegas algorithm, the way it is described in Algorithm~\ref{algo_ShorClassical}, the success probability is 1 but we have to calculate an upper bound for the average number of times the quantum part of the algorithm will run until finding a factor of $N$.  Suppose that an algorithm has probability $0<p<1$ of outputting the correct result in one run of the algorithm. The probability of outputting the correct result after exactly $n$ runs is $(1-p)^{n-1}p$ because it must fail $n-1$ times before succeeding. The average number of times the algorithm will run is\footnote{The sum is calculated using that
\[
\sum_{n=1}^\infty nq^n\,=\,q\,\frac{\text{d}}{\text{d}q}\left(\sum_{n=1}^\infty q^n\right),
\]
where $q=1-p$. The sum inside the derivative is the geometric series, whose value is $q/(1-q)$ when $|q|<1$.}
\[
\sum_{n=1}^\infty n\,(1-p)^{n-1}p \,=\, \frac{1}{p}.
\]
Then, an upper bound on the average number of times that the quantum part will run in Algorithm~\ref{algo_ShorClassical} is $16\pi^2\ln(\ln (r))/9$ for $r\ge 7$. Now we are done.

\section{Circuit of the modular exponentiation}

It is easy to build the circuit of $U^{(a)}_N$ if the order of $a$ is 2 modulo $N$. In this case
\[
U^{(a)}_N\ket{\ell}\ket{0} = \begin{cases}
			\ket{\ell}\otimes\ket{1}, & \text{if $\ell$ even,}\\
            \ket{\ell}\otimes\ket{a}, & \text{if $\ell$ odd,}
		 \end{cases}
\]
because $a^\ell\equiv 1 \mod N$ if $\ell$ is even and  $a^\ell\equiv a \mod N$ if $\ell$ is odd. Besides, $\ell$ is even if $\ell_{m-1}=0$ and odd if $\ell_{m-1}=1$. Then, if $\ell$ is even, we have
\[
U^{(a)}_N\ket{\ell}\ket{0} = \ket{\ell}\otimes (\ket{0}\cdots\ket{0}X^{\ell_{m-1}\oplus 1}\ket{0}),
\]
which is implemented using a CNOT with empty control on the $m$-th qubit and target on the last qubit. If $\ell$ is odd, we have
\[
U^{(a)}_N\ket{\ell}\ket{0} = \ket{\ell}\otimes (X^{\ell_{m-1}\cdot a_0}\ket{0}\cdots X^{\ell_{m-1}\cdot a_{n-1}}\ket{0}),
\]
where $a=(a_0,...,a_{n-1})_2$. This result is obtained by applying CNOTs with all controls on the $m$-th qubit and one target for each bit 1 of $a$. The number of CNOTs is the Hamming weight of $a$.

For example, if $N=21$ and $a=13$, the circuit of $U^{(a)}_N$ is\vspace{-5pt}
\[
\Qcircuit @C=1.7em @R=1.1em {
\lstick{\text{qubit}_1\,:\,\ket{\ell_0}}         & \qw       & \qw       & \qw       & \qw        & \rstick{}\qw \\
\lstick{\vdots \,\,\,}&  &                 &&  \rstick{\vdots}\\
\lstick{\text{qubit}_{m}\,:\,\ket{\ell_{m-1}}}       & \ctrlo{5} & \ctrl{2}& \ctrl{3}& \ctrl{5} & \rstick{}\qw\\
\lstick{\ket{0}}        & \qw          & \qw        & \qw       & \qw       & \qw  &  \rstick{}  \\
\lstick{\ket{0}}        & \qw          & \targ     & \qw       & \qw       & \qw  &  \rstick{}  \\
\lstick{\ket{0}}        & \qw          & \qw       & \targ     & \qw       & \qw  &  \rstick{}  \\
\lstick{\ket{0}}        & \qw          & \qw       & \qw       & \qw       & \qw  &  \rstick{}  \\
\lstick{\ket{0}}        & \targ        & \qw       & \qw       & \targ    & \qw  &  \lstick{}  \\
}\vspace{0.0cm}
\]
because $13=(01101)_2$. This technique can be applied for any
$a\in\{4, 11, 14\}$  and $N = 15$;
$a\in\{8, 13, 20\}$  and $N = 21$;
$a\in\{10, 23, 32\}$  and $N = 33$;
$a\in\{6, 29, 34\}$  and $N = 35$;
$a\in\{14, 25, 38\}$  and $N = 39$; and so on.

The circuit of $U^{(a)}_N$ when $a^2\not\equiv 1\mod N$ can be implemented with the methods described in~\cite{MS12, NO08, PG14}.

\section{Circuit of the Fourier transform}\label{sec:Shor_Fourier_Circ}

The first decomposition of the Fourier transform in terms of basic gates can be traced back to an IBM report from 1994, which became widely available in 2002~\cite{Cop02}. A description of this decomposition based on the classical FFT is available in~\cite{MPS10}.

The basic block of the circuit of the Fourier transform $F_{2^n}$, where $n$ is the number of qubits, is the controlled gate $C(R_k)$ for $k\ge 0$, where
\[
R_k\,=\,
\left[\begin{array}{cc}
1 & 0 \\
0 & \exp \left( \frac{2\pi\ii}{2^k} \right)
\end{array} \right].
\]
The matrix representation of $C(R_k)$ is
\[
{C(R_k)} =
\begin{bmatrix}
    I_2 &  \\
    & R_k \\
\end{bmatrix}
=
\begin{bmatrix}
                 1 & 0 & 0 & 0 \\
                 0 & 1 & 0 & 0 \\
                 0 & 0 & 1 & 0 \\
                 0 & 0 & 0 & \exp \left( \frac{2\pi\ii}{2^k} \right)
               \end{bmatrix}.
\]
The set of gates $R_k$ has many special subcases, which are
\begin{align*}
R_0 &=I_2,\\
R_1 &=Z,\\
R_2 &=S,\\
R_3 &=T,
\end{align*}
where $Z$, $S$, and $T$ are the Pauli $Z$ gate, phase gate, and $\pi/8$ or $T$ gate, respectively. Note that, $R_{k+1}=\sqrt{R_k}$. Then, the sequence above means the next gate is the square root of the previous one. The same idea applies to $C(R_k)$, that is, $C(R_{k+1})=\sqrt{C(R_k)}$, but in the latter case we are calculating the square root of $(4\times 4)$-matrices.

\begin{figure}[!ht]
\centering
\includegraphics[scale=0.18]{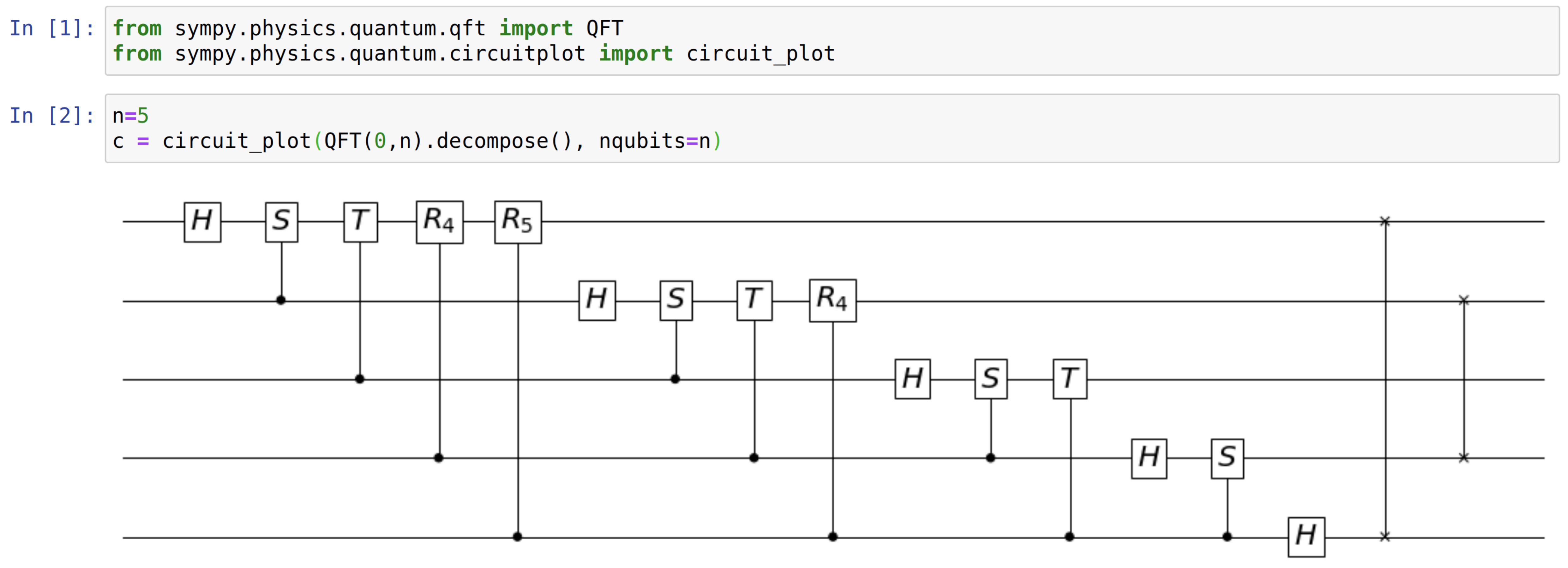}
\caption{Snapshot of a Jupyter notebook showing the decomposition of the Fourier transform $F_{2^5}$ using Python commands \emph{QFT} and \emph{circuit\_plot} of the Sympy library.}\label{fig:ShorFourierDecomp}
\end{figure}

Fig.~\ref{fig:ShorFourierDecomp} depicts the circuit of $F_{2^n}$ in terms of the Hadamard, $C(R_k)$, and swap gates when $n=5$. The structure of the circuit can be easily grasped from this example. The circuit has $n+1$ blocks. From left to right, the first block has $n$ gates, starting with $H$ and then $R_2$, $R_3$,..., $R_n$ acting on qubit 1 and controlled by qubit 2, 3,..., $n$, respectively. The next block starts again with $H$ and then $R_2$, $R_3$,..., $R_{n-1}$ acting on qubit 2 and controlled by qubit 3, 4,..., $n$, respectively. This goes on until we reach the last qubit, on which a single $H$ ($n$-th block) is applied. The last block is made of $\lfloor n/2\rfloor$ swap gates and has a simple symmetric structure. If $n$ is odd, the central qubit does not swap. The number of gates is
\[
n+(n-1)+\cdots+1+\left\lfloor \frac{n}{2}\right\rfloor\,\,=\,\,\frac{n(n+1)}{2}+\left\lfloor \frac{n}{2}\right\rfloor.
\]

Let us show that a circuit with the structure depicted in Fig.~\ref{fig:ShorFourierDecomp} implements the Fourier transform when we have $n$ qubits. Suppose that the input is $\ket{\ell}=\ket{\ell_1}\otimes\cdots\otimes\ket{\ell_n}$. When the input is a state of the computational basis, the output is an unentangled state $\ket{\psi_1}\otimes\cdots\otimes\ket{\psi_n}$. Let us start by calculating the output $\ket{\psi_1}$ of the first qubit. Since there is a swap gate inverting the states of the first and last qubit, we have
\[
\ket{\psi_1} = H\ket{\ell_n}= \frac{\ket{0}+\e^{2\pi\ii\frac{\ell_n}{2}}\ket{1}}{\sqrt 2}= \frac{\ket{0}+\e^{2\pi\ii\frac{\ell}{2}}\ket{1}}{\sqrt 2}.
\]
The second equation follows from the fact that if $\ell_n$ is 0, the output is $\ket{+}$, and if $\ell_n=1$, the output is $\ket{-}$ (because $\e^{\pi\ii}$ is $-1$). The last equation follows from the decomposition $\ell=2^{n-1}\ell_1+ 2^{n-2}\ell_2+\cdots+ 2\ell_{n-1} + \ell_n$ and from $\e^{2\pi\ii k}=1$ if $k$ is an integer.

The output of the second qubit is
\[
\ket{\psi_2} = R_2^{\ell_n}H\ket{\ell_{n-1}} =
R_2^{\ell_n}\left( \frac{\ket{0}+\e^{2\pi\ii\frac{\ell_{n-1}}{2}}\ket{1}}{\sqrt 2}\right)=
\frac{\ket{0}+\e^{2\pi\ii\left(\frac{\ell_{n-1}}{2}+\frac{\ell_{n}}{2^2}\right)}\ket{1}}{\sqrt 2}=
\frac{\ket{0}+\e^{2\pi\ii\frac{\ell}{2^2}}\ket{1}}{\sqrt 2}.
\]
The first equation is obtained from Fig.~\ref{fig:ShorFourierDecomp} using that $C(R_2)\ket{\ell_n}\ket{\ell_{n-1}}=\ket{\ell_n}\big(R_2^{\ell_n}\ket{\ell_{n-1}}\big)$. The second equation uses the same calculation described before for the first qubit. The third equation follows from
\begin{align*}
R_2^{\ell_n}\ket{0} &= \ket{0},\\
R_2^{\ell_n}\ket{1} &= \e^{2\pi \ii\frac{ \ell_n}{2^2}}\ket{1}.
\end{align*}
The last equation uses the same decomposition of $\ell$ described before.

The output of the last qubit is
\[
\ket{\psi_n} = R_n^{\ell_n}\cdots R_2^{\ell_2} H\ket{\ell_{1}} =
\frac{\ket{0}+\e^{2\pi\ii\left(\frac{\ell_{1}}{2}+\frac{\ell_{2}}{2^2}+\cdots +\frac{\ell_{n}}{2^n}\right)}\ket{1}}{\sqrt 2}=
\frac{\ket{0}+\e^{2\pi\ii\frac{\ell}{2^n}}\ket{1}}{\sqrt 2}.
\]
The first equation is obtained from Fig.~\ref{fig:ShorFourierDecomp} using that the output of the last qubit is obtained from the action of $H$, $R_2^{\ell_2}$, ..., $R_n^{\ell_n}$ on the first qubit. The second and third equations are obtained with the same sort of calculations described before for the first and second qubits.

Then, the output $\ket{\psi_1}\otimes\cdots\otimes\ket{\psi_n}$ of the circuit of Fig.~\ref{fig:ShorFourierDecomp} with $n$ qubits is
\[
\frac{\ket{0}+\e^{2\pi\ii\frac{\ell}{2}}\ket{1}}{\sqrt 2}\otimes
\frac{\ket{0}+\e^{2\pi\ii\frac{\ell}{2^2}}\ket{1}}{\sqrt 2}\otimes \cdots \otimes
\frac{\ket{0}+\e^{2\pi\ii\frac{\ell}{2^n}}\ket{1}}{\sqrt 2}.
\]
Now we convert each term into a sum
\[
\frac{1}{\sqrt 2}\sum_{k_1=0}^1\e^{2\pi\ii k_1\frac{\ell}{2}}\ket{k_1}\otimes
\frac{1}{\sqrt 2}\sum_{k_2=0}^1\e^{2\pi\ii k_2\frac{\ell}{2^2}}\ket{k_2}\otimes \cdots \otimes
\frac{1}{\sqrt 2}\sum_{k_n=0}^1\e^{2\pi\ii k_n\frac{\ell}{2^n}}\ket{k_n}.
\]
Pushing all sums to the right-hand side and combining the exponentials, we obtain
\[
\frac{1}{\sqrt{2^n}}\sum_{k_1,...,k_n=0}^1 \e^{2\pi\ii\ell\left(\frac{k_1}{2} +\cdots + \frac{k_n}{2^n} \right)} \ket{k_1,...,k_n},
\]
which is equivalent to
\[
\frac{1}{\sqrt{2^n}}\sum_{k=0}^{2^n-1} \e^{\frac{2\pi\ii\ell k}{2^n} } \ket{k}.
\]
Using the definition of the Fourier transform given in Section~\ref{seq:shor_fourier_transform}, we recognize that the last expression is $F_{2^n}\ket{\ell}$.

\begin{exercise}
Draw the circuit for the two-qubit Fourier transform based on Fig.~\ref{fig:ShorFourierDecomp} and show that the unitary matrix corresponding to this circuit is
\[
F_4=\frac{1}{2} \begin{bmatrix}
1 & \,\,\,\,1 & \,\,\,\,1 & \,\,\,\,1 \\
1 & \,\,\,\,\ii  & -1 & -\ii  \\
1 & -1 & \,\,\,\,1 & -1 \\
1 & -\ii  & -1 & \,\,\,\,\ii
\end{bmatrix}.
\]
\end{exercise}

\subsection*{Decomposition of $C(R_k)$}

The circuit that decomposes $C(R_k)$ into CNOT and single-qubit gates $R_{k+1}$ is
\[
\Qcircuit @C=1.5em @R=0.9em {
                 & \ctrl{2}  &          \qw   & & &\qw& \ctrl{2}&\qw &\ctrl{2}&\gate{R_{k+1}} & \qw   \\
                 &           & \rstick{\,\,\equiv} & &  & & & &\\
\lstick{}        & \gate{R_k}  &  \qw       & & &\gate{R_{k+1}}& \targ &\gate{R^\dagger_{k+1}} &\targ&\qw  & \rstick{.}\qw
}\vspace*{0.1cm}
\]
This is not a full decomposition in terms of universal gates because $R_{k+1}$ still needs to be expressed using a finite set of single-qubit gates. However, this step is generally handled by quantum computer compilers, so it is usually unnecessary to address it manually. The gate $R_k$ can be implemented using $R_z$, a rotation around the $z$-axis of the Bloch sphere. When $k$ is large, errors may prevent these gates from functioning correctly unless error-correcting codes are applied.

Now we show that the decomposition of $C(R_k)$ in terms of CNOT and single-qubit gates $R_{k+1}$ is correct. If the input to $C(R_k)$ is $\ket{j}\ket{\ell}$ then the output is
\begin{align*}
\ket{j}\ket{\ell} \xrightarrow[\text{}]{\bullet-\fbox{$R_k$}} \ket{j}R_k^j\ket{\ell}.
\end{align*}
If $\ket{j\ell}$ is either $\ket{00}$, $\ket{01}$, or $\ket{10}$, the output remains $\ket{j\ell}$ because either the control qubit is inactive or $R_k\ket{0}=\ket{0}$. The only nontrivial output occurs when the input is $\ket{11}$. In this case, the output becomes $\exp(2\pi\ii /2^k)\ket{11}$.

Next, let us analyze the decomposition (right-hand circuit above) to verify that it produces the same result. For inputs $\ket{00}$, $\ket{01}$, or $\ket{10}$, either the CNOT gates are inactive or they cancel out. Additionally, since $R_k\ket{0}=\ket{0}$ and $R_{k+1}^\dagger R_{k+1}=I$, the output is the same as the input in these cases. The only remaining case is the input $\ket{11}$. Let us go through the steps, using the facts that $R_{k+1}\ket{0}=\ket{0}$ and $R_{k+1}\ket{1}=\exp(2\pi\ii /2^{k+1})\ket{1}$:
\begin{align*}
&\ket{1}\ket{1} \xrightarrow[\text{}]{I\otimes R_{k+1}} \e^{{2\pi\ii}/{2^{k+1}}}\ket{1}\ket{1} \xrightarrow[\text{}]{\bullet-\oplus}  \e^{{2\pi\ii}/{2^{k+1}}}\ket{1}\ket{0} \xrightarrow[\text{}]{I\otimes R_{k+1}^\dagger} \e^{{2\pi\ii}/{2^{k+1}}}\ket{1}\ket{0} \rightarrow\xrightarrow[\text{}]{\bullet-\oplus}\\
\mbox{}\\
\vspace{-10pt}
 &\e^{{2\pi\ii}/{2^{k+1}}}\ket{1}\ket{1}  \xrightarrow[\text{}]{ R_{k+1} \otimes I} \e^{{2\pi\ii}/{2^{k+1}}} \e^{{2\pi\ii}/{2^{k+1}}}\ket{1}\ket{1}=\e^{{2\pi\ii}/{2^{k}}}\ket{1}\ket{1}.
\end{align*}
Thus, the output is $\exp(2\pi\ii /2^k)\ket{11}$, which matches the output of the left-hand circuit.

\section{Final remarks}

In his original paper~\cite{Sho99}, Shor demonstrated that the expected number of repetitions required to find a solution is $O(\log\log r)$, and the asymptotic lower bound on the success probability is $4/\pi^2$, as derived from Eq.~(\ref{eq_shor_lb}) when $N$ is large. The number of repetitions of the quantum component can be reduced by optimizing the classical post-processing step~\cite{Eke22}.

\chapter{Shor's Algorithm for the Discrete Logarithm Problem}\label{chap:DL}



The full paper on Shor's algorithms was published in 1997~\cite{Sho97} and describes not only an algorithm for integer factoring but also an exponentially faster algorithm for discrete logarithms. While fast factoring threatens RSA schemes, efficient algorithms for discrete logarithms threaten the security of cryptographic protocols such as the Diffie--Hellman key exchange. Shor's work is a remarkable and celebrated contribution to quantum computing, showing that quantum computers can solve important number-theoretic problems much more efficiently than the best known classical algorithms. Although the factoring algorithm is widely discussed in the literature, the discrete logarithm algorithm has received less attention in books~\cite{KLM07}. Some papers explore cryptographic applications of this algorithm~\cite{HSZDC20,PZ03}. In this Chapter, we focus on the discrete logarithm algorithm and explain how its efficiency arises from the periodic structure of an appropriate function. Knowledge of group theory is important.

\section{Preliminaries on number theory}

Over the integers, the \textit{logarithm} of $b$ to base $a$ is the number of times that $a$ is multiplied by itself to obtain $b$; that is, it is the integer $s$ such that $a^s = b$. Here we assume that $a$, $b$, and $s$ are positive integers. For instance, $\log_2 8 = 3$ because $2 \times 2 \times 2 = 8$. Since $s = \log_a b$, we write
\[
a^{\log_a b} = b.
\]
There are choices of $a$ and $b$ for which the logarithm $\log_a b$ does not exist; for instance, $\log_2 7$ does not exist because there is no integer $s$ such that $2^s = 7$. Two important properties of the logarithm to base $a$, when all logarithms involved exist, are
\begin{itemize}
\item $\log_a (b c)=\log_a b+\log_a c$;
\item $\log_a (b^c) = c \log_a b$.
\end{itemize}
There are two drawbacks, in the context of quantum algorithms, to defining the logarithm over the set of positive integers $\mathbb{Z}^+$: (1) depending on the choice of $a$ and $b$, $\log_a b$ may not exist, and (2) if $\log_a b$ does exist, it can be computed efficiently using the standard real Napierian logarithm or natural logarithm $\ln$, since
$$
s = \log_a b = \frac{\ln b}{\ln a},
$$
and this quantity is an integer whenever the logarithm exists.

Let us now define the \textit{discrete logarithm}. Instead of using $\mathbb{Z}^+$, we work with a subset of $\mathbb{Z}_N=\{0,1,\dots,N-1\}$, equipped with modular addition and multiplication. We are primarily interested in multiplication modulo $N$, where $N>1$. An element $a\in \mathbb{Z}_N$ has a multiplicative inverse $a^{-1}$ (that is, $aa^{-1}\equiv 1 \mod N$) if and only if $\gcd(a,N)=1$. Let $r$ be the order of $a$, that is, the smallest positive integer such that $a^r \equiv 1 \mod N$. The set of all such elements forms the multiplicative group $\mathbb{Z}_N^\times$. If $a\in \mathbb{Z}_N^\times$, the set
\begin{equation}\label{eq:G_a}
G_a=\{a^s \bmod N : 0\le s < r\}
\end{equation}
is a cyclic subgroup of $\mathbb{Z}_N^\times$, generated by $a$. For any $b \in G_a$, the discrete logarithm $\log_a b$ is defined as the integer $s \in \{0,1,\dots,r-1\}$ such that $a^s \equiv b \mod N$. In this context, no efficient classical algorithm is known for computing $\log_a b$ (in polynomial time) for arbitrary $b \in G_a$.

For example, for $N=34$ and $a=15$ we have in increasing order of $s$ in Eq.~\eqref{eq:G_a}
$$G_{15}=\{1,15, 21, 9, 33, 19, 13, 25\}.$$
Then, $\log_{15} 1=0$, $\log_{15} 15=1$, $\log_{15} 21=2$, $\log_{15} 9=3$, $\log_{15} 33=4$, and so on.

\section{Problem formulation and the main strategy}

The discrete logarithm problem can be formulated as follows. Let $G$ be a finite cyclic group generated by an element $a$, and let $b$ be an element in the subgroup generated by $a$. The goal is to determine an integer $s$ such that
\[
a^s=b.
\]
If the order of $a$ is $r$, then $s$ is defined modulo $r$, and we usually seek the unique solution in the range $0\le s<r$. In the multiplicative group $\mathbb{Z}_N^\times$, this means that, given $a$ and $b$, we want to find $s$ such that
\[
a^s \equiv b \mod N.
\]
While this problem is believed to be hard for classical computers when the group is large, Shor's quantum algorithm solves it efficiently by exploiting the periodic structure of a suitable function associated with $a$ and $b$.

Shor's order-finding algorithm can efficiently compute $\log_a b$ in some cases because the value $\mathrm{ord}(a)/\mathrm{ord}(b)$ may coincide with $\log_a b$ (see Exercise~\ref{exe:log-special-case}). Before attempting to compute $\log_a b$ using a general discrete-logarithm algorithm, we may compute $s=\mathrm{ord}(a)/\mathrm{ord}(b)$ efficiently using Shor's order-finding algorithm and then check whether $b=a^s$. If this is not the case, we use a quantum algorithm that efficiently computes $\log_a b$ for an arbitrary $b \in G_a$, as described below.

The main strategy in the part of Shor's factoring algorithm that computes the order of $a \in \mathbb{Z}_N$ (chosen at random) is to define the function
\begin{align*}
  f \colon \mathbb{Z}_{2^m} &\to \mathbb{Z}_N\\
  x &\mapsto a^x \bmod N,
\end{align*}
where $n=\left\lceil\log_2 N\right\rceil$, $m$ is an integer close to $2n$, and then exploit the fact that $f$ is $r$-periodic. We know that this function can be implemented on a quantum computer with about $3n$ qubits using the unitary operator $U_f$, defined by $U_f\ket{x}\ket{y}=\ket{x}\ket{y\oplus (a^x \bmod N)}$. It is possible to determine $r$ efficiently by applying the inverse discrete Fourier transform $F_{2^m}^\dagger$ to the first register after $U_f$ has been applied to a superposition of all computational basis states of the first register (the second register is initially set to $\ket{0}$). The method works because $f$ is periodic with period $r$.

A natural question is whether the same strategy can be used to compute $s$ from the equation $a^s \equiv b \mod N$. That is, can the discrete logarithm problem be reduced to finding the period of a suitable periodic function? The difficulty is that there is no $s$-periodic function $f(x)$ with domain $\mathbb{Z}_{2^m}$ and codomain $\mathbb{Z}_N$ that uses parameters $a$ and $b$ only. The way out is to use a function $f(x,y)$ with two variables defined as
\begin{align}\label{eq:f-x-y-log}
  f \colon \mathbb{Z}_{2^m}\times\mathbb{Z}_{2^m} &\to \mathbb{Z}_N \nonumber \\
  (x,y) &\mapsto a^x b^y \mod N.
\end{align}
The function $f$ is periodic, but the periodicity of $f$ is two-dimensional rather than one-dimensional, as we can see as follows. Using that $b=a^s$, we have
\begin{equation}\label{eq:f-x-y-log2}
f(x,y)=a^{x+sy} \mod N.
\end{equation}
This shows that the periodicity occurs in the exponent, and therefore
\[
f(x,y)=f(x',y') \Longleftrightarrow x+ sy \equiv x'+ sy' \mod r.
\]
Note that we use modulo $r$ because $a^k \equiv a^{k'}$ modulo $N$ whenever
$k \equiv k'$ modulo $r$, where $r=\textrm{ord}(a)$.

The condition $(x+ sy) \equiv (x'+ sy')$ modulo $r$ is equivalent to
\[
(x+ sy) - (x'+ sy') = kr
\]
for some integer $k$. To describe all integer solutions $(x',y')$, we introduce a free integer parameter $\ell$ and set $y' = y + \ell$. Substituting into the equation, we obtain
\begin{align}\label{eq:x'y'}
x' &= x + kr - \ell s , \nonumber \\
y' &= y + \ell,
\end{align}
for integers $k,\ell$. Conversely, every pair $(x',y')$ of the form given in Eq.~\eqref{eq:x'y'} satisfies $f(x',y') = f(x,y)$. To confirm the periodicity,
\[
f(x,y)=f(x+kr-\ell s,y+\ell),
\]
where $k,\ell\in \mathbb{Z}$, we use Eqs.~\eqref{eq:x'y'} and $a^r\equiv 1$ to compute
\[
f(x',y')\equiv  a^{x'+sy'}
\equiv a^{x+kr-\ell s +s(y+\ell)}
\equiv a^{x+sy}\equiv f(x,y) \mod N.
\]

Eqs.~\eqref{eq:x'y'} can be written in vector notation as
\[
(x',y')=(x + kr - \ell s , y + \ell)= (x,y)+k\vec{\textbf{r}}+\ell\vec{\textbf{s}},
\]
where \vspace*{-10pt}
\begin{align*}
\vec{\textbf{r}}&=(r,0),\\
\vec{\textbf{s}}&=(-s,1).
\end{align*}
Therefore, all pairs $(x',y')$ that satisfy
$f(x',y') = f(x,y)$ are obtained by adding integer combinations of the vectors $\vec{\textbf{r}}$ and $\vec{\textbf{s}}$. This shows that the periodic points of $f$ form a lattice (see next Section) in the plane generated by the two-dimensional vectors $\vec{\textbf{r}}$ and $\vec{\textbf{s}}$, where the linear combinations are built with scalars $k$ and $\ell$ in $\mathbb{Z}$. Finding the discrete logarithm $s$ is therefore equivalent to finding the lattice generated by the vectors $\vec{\textbf{r}}$ and $\vec{\textbf{s}}$.

\begin{exercise}\label{exe:log-special-case}
Let $G_a=\{a^k \bmod N : k\in \mathbb{Z}\}$, where $r$ is the order of $a$ modulo $N$. Let $b \in G_a$ with $b=a^s$ for some $0 \le s < r$.

\begin{itemize}
\item[(a)] Show that
\[
\mathrm{ord}(b)=\frac{r}{\gcd(r,s)}.
\]

\item[(b)] Prove that
\[
\log_a b = \frac{\mathrm{ord}(a)}{\mathrm{ord}(b)}
\]
if and only if $s \mid r$.

\item[(c)] Give an explicit example in $\mathbb{Z}_N^\times$ where $\mathrm{ord}(b)$ divides $\mathrm{ord}(a)$ but
\[
\log_a b \ne \frac{\mathrm{ord}(a)}{\mathrm{ord}(b)}.
\]

\item[(d)] Conclude that knowing only $\mathrm{ord}(a)$ and $\mathrm{ord}(b)$ is not sufficient, in general, to determine $\log_a b$.
\end{itemize}
\end{exercise}

\begin{exercise}\label{exe:L_r}
Let $G_a=\{a^k \bmod N : k\in \mathbb{Z}\}$, where $r$ is the order of $a$ modulo $N$.
Let $b \in G_a$ with $b=a^s$ for some $0 \le s < r$. Consider the function~\eqref{eq:f-x-y-log}.

\begin{itemize}

\item[(a)] Prove that the set
\[
L_r=\{(x,y)\in \mathbb{Z}_r^2 : x+sy \equiv 0 \mod r\}
\]
is a subgroup of $\mathbb{Z}_r^2$ with componentwise addition modulo $r$ as group operation.

\item[(b)] Show that $f(x,y)=f(x',y')$ if and only if $(x-x',y-y') \in L_r$.

\item[(c)] Show that $L_r=\{k\vec{\textbf{s}}\mod r: k\in \mathbb{Z}_r\}$ and conclude that $L_r$ is a cyclic group with exactly $r$ elements.

\item[(d)] Explain how, given $r$, a generator $(u,v)$ of $L_r$ can be used to determine $s$ modulo $r$, if $\gcd(v,r)=1$.
\end{itemize}
\end{exercise}

\section{Lattice and group structure of the periodicity}

Consider the infinite additive group $\mathbb{Z}^2=\mathbb{Z} \times \mathbb{Z}$, whose elements are pairs $(x,y)$ with addition defined componentwise:
\[
(x,y)+(x',y')=(x+x',\,y+y').
\]
We identify vectors and points in the plane, so that a vector $(x,y)$ can also be viewed as a point with coordinates $(x,y)$. Thus, the group $\mathbb{Z}^2$ can be represented geometrically by the set of all lattice points in the Euclidean plane, that is, the points whose coordinates are both integers. In this representation, the group operation corresponds to vector addition. Here, the word vector is used only in a geometric sense, since $\mathbb{Z}^2$ is an additive group rather than a vector space.

A two-dimensional lattice $L$ in $\mathbb{Z}^2$, by definition, is a set
\[
L=\{k\vec{\textbf{r}}+\ell\vec{\textbf{s}}: k,\ell\in \mathbb{Z}\},
\]
where $\vec{\textbf{r}}$ and $\vec{\textbf{s}}$ are linearly independent vectors in $\mathbb{Z}^2$. The set $\{\vec{\textbf{r}},\vec{\textbf{s}}\}$ is called a basis of the lattice. Different bases may generate the same lattice. The group $\mathbb{Z}^2$ itself is a two-dimensional lattice generated by the canonical basis $\vec{\textbf{e}}_0=(1,0)$ and $\vec{\textbf{e}}_1=(0,1)$. Since $L$ is generated by $\vec{\mathbf{r}}$ and $\vec{\mathbf{s}}$ under componentwise addition, it is a subgroup of $\mathbb{Z}^2$.
Fig.~\ref{fig:logdisfig0} shows an example of a two-dimensional lattice generated by the vectors $\vec{\textbf{r}}=(16,0)$ and $\vec{\textbf{s}}=(5,1)$ in $\mathbb{Z}^2$.

\begin{figure}[!ht]
\centering
\includegraphics[scale=0.8]{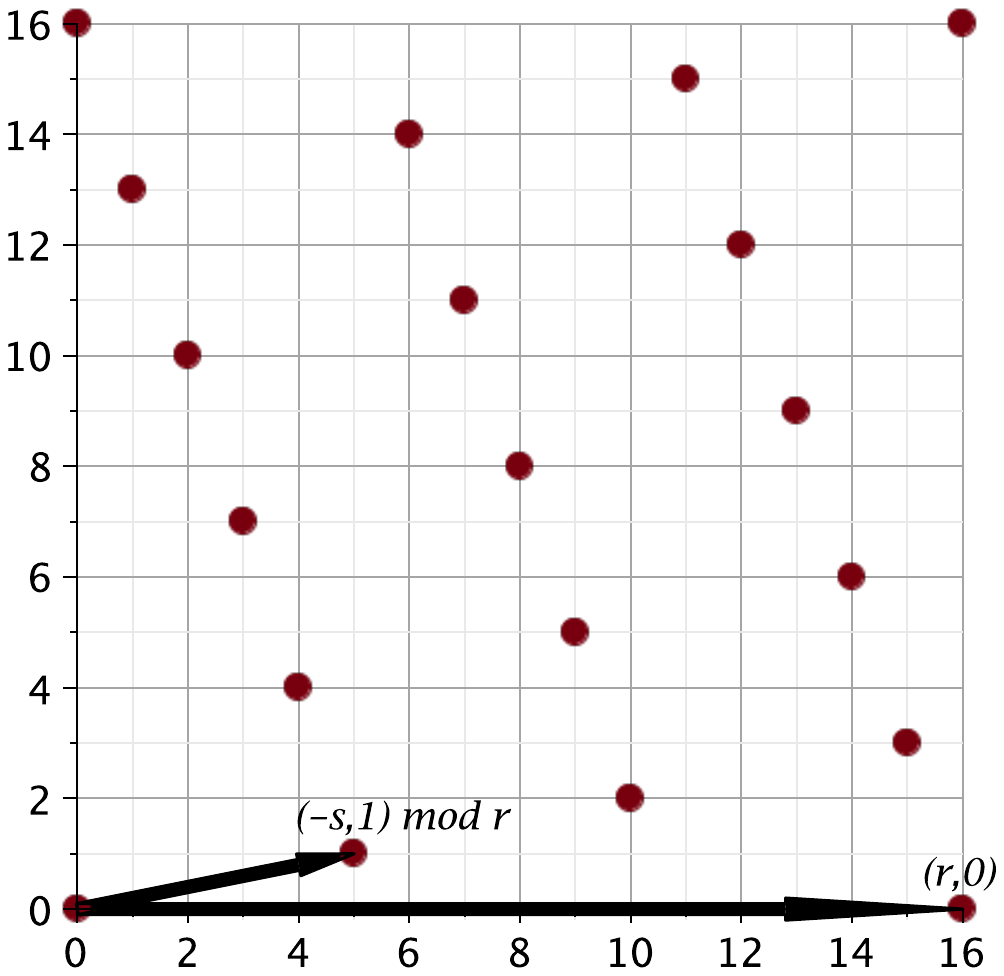}
\caption{A two-dimensional lattice generated by $(16,0)$ and $(5,1)$. Only the points with coordinates between $0$ and $16$ are shown; the lattice extends infinitely in all directions. For every point $(x,y)$ in the lattice, we have $f(x,y)=1$. For any point outside the lattice, $f(x,y)\neq 1$.}
\label{fig:logdisfig0}
\end{figure}

As an example, consider again $N=34$, but now take $a=27$. Then the cyclic group in Eq.~\eqref{eq:G_a} is
\[
G_{27}=\langle a \rangle=\{1, 27, 15, 31, 21, 23, 9, 5, 33, 7, 19, 3, 13, 11, 25, 29\}.
\]
Suppose we want to compute the discrete logarithm of $b=3$ to the base $a$, that is, $\log_a b=\log_{27} 3$. The answer is $s=11$, since $3$ appears in the 11th position of $G_{27}$ when the counting starts at $s=0$. Since the order of $a$ is $r=16$ (see the order of $G_{27}$), the two-dimensional lattice $L$ shown in Fig.~\ref{fig:logdisfig0}, with basis vectors $\vec{\textbf{r}}=(r,0)=(16,0)$ and $\vec{\textbf{s}}=(-s,1)=(-11,1)$, represents all points $(x,y)\in\mathbb{Z}^2$ such that $f(x,y)=f(0,0)=1$, where $f(x,y)=27^x 3^y \bmod 34$, according to the definition in Eq.~\eqref{eq:f-x-y-log}.


The two-dimensional lattice $L$ provides an interesting way to visualize the periodicity of the function $f(x,y)$. To understand how Shor's algorithm determines the discrete logarithm, however, it is more convenient to work with the finite group associated with this lattice. Consider the additive group
\begin{equation}\label{eq:L_r}
L_r=\{(x,y)\in \mathbb{Z}_r^2 : x+sy \equiv 0 \mod r\},
\end{equation}
which is a subgroup of $\mathbb{Z}_r^2$ (see Exercise~\eqref{exe:L_r}). The condition $x+sy \equiv 0 \mod r$ means that the exponent in Eq.~\eqref{eq:f-x-y-log2} is congruent to $0$ modulo $r$. Since $a$ has order $r$, we obtain
\[
a^{x+sy}\equiv 1 \mod N,
\]
and therefore $f(x,y)=1$ for every $(x,y)\in L_r$. Thus, $L_r$ consists exactly of the points of $L$ whose coordinates satisfy $0\le x<r$ and $0\le y<r$. The group $L_r$ is cyclic and is generated by $(-s,1)$. In contrast, the vector $(r,0)$, which is essential in describing the infinite lattice $L$, plays no role as a generator of $L_r$, because in $\mathbb{Z}_r^2$ it coincides with the identity element.

Shor's algorithm uses three registers, $\ket{x}\ket{y}\ket{f(x,y)}$, to store the point $(x,y)$ and the corresponding value $f(x,y)$. By quantum parallelism, the algorithm creates a superposition of all possible values of $x$ and $y$. To simplify the resulting linear combination, we group together the terms that have the same value of $f(x,y)$. Thus, we must partition the set of points $(x,y)$ according to their image under $f$.

This partition arises naturally from the subgroup $L_r$ of the additive group $\mathbb{Z}_r^2$. Indeed, $\mathbb{Z}_r^2$ can be decomposed into equivalence classes, namely the cosets of $L_r$ in $\mathbb{Z}_r^2$. A convenient choice of representatives is
\[
(0,0),\ (1,0),\ \ldots,\ (r-1,0).
\]
Among these representatives, only $(0,0)$ belongs to $L_r$, since $(x,0)\in L_r$ implies $x\equiv 0 \mod r$. Moreover, the values of $f$ at these representatives are all distinct: if $f(x,0)=f(x',0)$, then $a^x\equiv a^{x'} \mod N$, hence $a^{x-x'}\equiv 1 \mod N$, and since $a$ has order $r$, we must have $x\equiv x' \mod r$. Because $0\le x,x'<r$, it follows that $x=x'$. Therefore, the value of $f$ at any point $(x,y)\in\mathbb{Z}_r^2$ is equal to the value of $f$ at exactly one of these representatives. More precisely, for each fixed $x$ with $0\le x<r$, the points $(x-\ell s,\ell)$, with $0\le \ell<r$, form the coset represented by $(x,0)$, and all satisfy
\[
f(x-\ell s,\ell)=f(x,0).
\]
Geometrically, this corresponds to translating the lattice leftward by $(x,0)$ modulo $r$, that is, with cyclic boundary conditions. Therefore, when $x$ is fixed and $\ell$ runs from $0$ to $r-1$, we obtain exactly the coset represented by $(x,0)$.

\begin{exercise}
Consider the discrete logarithm problem with $N=34$, $a=27$, and $b=3$. Recall that $a$ has order $r=16$ and that the discrete logarithm of $b$ to the base $a$ is $s=11$.
\begin{itemize}
\item[(a)] List all elements of $L_r$.
\item[(b)] Find the coset representatives of $L_r$ in $\mathbb{Z}_{16}^2$ and write the corresponding cosets explicitly.
\item[(c)] Verify that $f(x,y)=27^x 3^y \bmod 34$ is constant on each coset.
\end{itemize}
\end{exercise}

\begin{exercise}
Let $G_a$ and $L_r$ be the groups defined in Eqs.~\eqref{eq:G_a} and~\eqref{eq:L_r}, respectively.
\begin{itemize}
\item[(a)] For a fixed $x\in\mathbb{Z}_r$, show that
\[
(x,0)+L_r=\{(x-\ell s,\ell): \ell\in\mathbb{Z}_r\}.
\]
Describe geometrically the elements of this coset.

\item[(b)] Show that the cosets
\[
(x,0)+L_r,\qquad x\in\mathbb{Z}_r,
\]
form a partition of $\mathbb{Z}_r^2$.

\item[(c)] Define the map
\[
f:\mathbb{Z}_r^2 \to G_a,
\qquad
f(x,y)=a^x b^y.
\]
Show that $f$ is a group homomorphism and that its kernel is $L_r$.

\item[(d)] Use the First Isomorphism Theorem to prove that
\[
\mathbb{Z}_r^2/L_r \cong G_a.
\]
Interpret this result in terms of the cosets of $L_r$ and the values of $f$.
\end{itemize}
\end{exercise}

\section{Special case: the order of \textit{a} is a power of 2}\label{sec:dislog_special_case}

Let $N$, $a$, and $b$ be known positive integers, and let $s$ be a positive integer such that $a^s \equiv b \mod N$ and gcd$(a,N) = 1$. Our goal is to find $s$, given $N$, $a$, and $b$ as input. Let $r$ be the order of $a$ modulo $N$, which can be efficiently determined using Shor's order-finding algorithm. In this Section, we address the case where $r = 2^m$ for some integer $m$, meaning that $r$ is a power of 2. In this case, there is an efficient classical algorithm called the \textit{Pohlig-Hellman algorithm} that can compute $s$ in polynomial time. We describe the quantum algorithm for this case because the Fourier transform $F_r$ can be implemented in a straightforward way in a qubit-based quantum computer, and the analysis of the algorithm is easier than in the general case.

Let $f$ be a two-variable function with domain $\mathbb{Z}_{r}\times \mathbb{Z}_{r}$ and codomain $\mathbb{Z}_N$ defined as
\begin{align*}
  f(x,y) = a^xb^y\mod N.
\end{align*}
We have shown that $f$ is periodic in the following way:
\[
f(x,y)=f(x+kr-\ell s,y+\ell).
\]
Using $f$, we define a 3-register unitary operator
\begin{equation}\label{eq:U_f_log_dis}
U_f\ket{x}\ket{y}\ket{z}\,=\,\ket{x}\ket{y}\ket{z\oplus f(x,y)},
\end{equation}
where the first and second registers have $m$ qubits each and the third register has $n=\lceil \log_2 N\rceil$ qubits. The arithmetic with the variables of the first and second registers is performed modulo $r$. The arithmetic to calculate the image $f(x,y)$ is performed modulo $N$. The symbol $\oplus$ represents the bitwise xor operation.
The algorithm that calculates the discrete logarithm when $r$ is a power of 2 is described in Algorithm~\ref{algo_discrete_log} and the circuit is depicted in Fig.~\ref{circuit:logdis}.

\begin{algorithm}[!ht]
\caption{Discrete logarithm algorithm when $r$ is a power of 2}\label{algo_discrete_log}
\KwIn{$N$, $a$, $b$, and $r$ (order of $a$).}
\KwOut{$s=\log_a b$ with probability $1/2$, where $a^s\equiv b \mod N$} \BlankLine
Prepare the initial state $\ket{0}^{\otimes m}\ket{0}^{\otimes m}\ket{0}^{\otimes n}$, where $m=\log_2 r$ and $n=\lceil\log_2 N\rceil$\vspace{3pt}\;
Apply $H^{\otimes m}\otimes H^{\otimes m}$ to the first and second registers\vspace{3pt}\;
Apply $U_f$\vspace{2pt} as defined in Eq.~(\ref{eq:U_f_log_dis})\;
Measure the third register in the computational basis\vspace{1pt}\;
Apply $F_{r}^\dagger\otimes F_{r}^\dagger$ to the first and second registers\vspace{3pt}\;
Measure the first and second registers in the computational basis, where $r_1$ and $r_2$ are the results\vspace{2pt}\;
If $\gcd(r_1,r)=1$, return $s\equiv r_2/r_1 \mod r$; otherwise, Fail.
\end{algorithm}

\begin{figure}
\begin{equation*}
\Qcircuit @C=2.1em @R=1.1em {
\lstick{\ket{0}^{\otimes m}} &\qw /^m&\gate{H^{\otimes m}}&\multigate{3}{U_f}&\qw&\gate{F_{2^m}^\dagger}&\meter &\rstick{r_1}\cw&   \\
\lstick{\ket{0}^{\otimes m}} &\qw /^m&\gate{H^{\otimes m}}&\ghost{U_f}       &\qw&\gate{F_{2^m}^\dagger}& \meter &\rstick{r_2}\cw&   \\
 & & & & & {\hspace{1.7cm}{}^\ket{\psi_4}} \\
\lstick{\ket{0}^{\otimes n}} &\qw /^n&\qw                 &\ghost{U_f}       & \meter &\cw \\
& \\
\ustick{\hspace{1.4cm}\ket{\psi_0}} & & \ustick{\hspace{1.8cm}\ket{\psi_1}}& \ustick{\hspace{1.6cm}\ket{\psi_2}} & \ustick{\hspace{1.5cm}\ket{\psi_3}}
}\vspace{0.1cm}
\end{equation*}
\caption{Circuit of the discrete logarithm algorithm when $r$ is a power of 2, where $m=\log_2 r$, $n=\lceil\log_2  N\rceil$, and $U_f$ is defined by Eq.~(\ref{eq:U_f_log_dis}). If $\gcd(r_1,r)=1$, return $s\equiv r_2/r_1 \mod r$; otherwise, Fail. The probability of returning the correct result is 1/2.}\label{circuit:logdis}
\end{figure}

\subsection*{Analysis of the algorithm}
In {step 2}, we apply $H^{\otimes m}\otimes H^{\otimes m}\otimes I$ to $\ket{0}\ket{0}\ket{0}$ obtaining
\[
\ket{\psi_1}=\frac{1}{r}\sum_{x,y=0}^{r-1}\ket{x}\ket{y}\ket{0}.
\]
In {step 3}, we apply $U_f$ to $\ket{\psi_1}$ obtaining
\[
\ket{\psi_2}=\frac{1}{r}\sum_{x,y=0}^{r-1}\ket{x}\ket{y}\ket{f(x,y)}.
\]
Before proceeding, we simplify $\ket{\psi_2}$ as much as possible using the periodicity of the function $f$. Since the image of $f$ on all points $(x-\ell s,\ell)$ for $0\le \ell < r$ is equal to the image of $(x,0)$, we can collect the third register as
\[
\ket{\psi_2}=\frac{1}{r}\sum_{x=0}^{r-1}\left(\sum_{\ell=0}^{r-1}\ket{x-\ell s}\ket{\ell}\right)\ket{f(x,0)}.
\]
In {step 4}, we measure the third register in the computational basis, obtaining the image of one of the representatives of the cosets of $L_r$ in $\mathbb{Z}_{r}^2$, as follows:
\[
\ket{\psi_3}=\frac{1}{\sqrt r}\left(\sum_{\ell=0}^{r-1}\ket{x_0-\ell s}\ket{\ell}\right)\ket{f(x_0,0)}.
\]
In {step 5}, we apply $F_{r}^\dagger\otimes F_{r}^\dagger$ to the first and second registers (we disregard the third register) obtaining
\[
\ket{\psi_4}=\frac{1}{\sqrt r}\sum_{\ell=0}^{r-1}\left(F_{r}^\dagger\ket{x_0-\ell s}\right)\left(F_{r}^\dagger\ket{\ell}\right).
\]
We now simplify $\ket{\psi_4}$ as much as possible. Using the definition of the Fourier transform $F_{r}^\dagger$, we obtain
\[
\ket{\psi_4}=\frac{1}{\sqrt r}\sum_{\ell=0}^{r-1}\left(\frac{1}{\sqrt {r}}\sum_{x=0}^{{r}-1}\omega_{r}^{x({x_0-\ell s})}\ket{x}\right)\left(\frac{1}{\sqrt {r}}\sum_{y=0}^{{r}-1}\omega_{r}^{y\ell}\ket{y}\right),
\]
where $\omega_r=\exp(2\pi\ii/r)$. Changing the order of the sums by pushing $\sum_{x,y}$ to the left, and by pushing $\sum_\ell$ and the term $\omega_{r}^{-x\ell s}$ to the second register, we obtain
\[
\ket{\psi_4}=\frac{1}{ r\sqrt{r}}\sum_{x,y=0}^{{r}-1}\omega_{r}^{xx_0}\ket{x}\,\sum_{\ell=0}^{r-1}\omega_{r}^{\ell(-xs+y)}\ket{y}.
\]
Using that
\[
\frac{1}{{r}}\sum_{\ell=0}^{r-1}\left(\omega_{r}^{-xs+y}\right)^\ell
 =
\begin{cases}
1, & \mbox{if $y=xs$,} \\
0, & \mbox{otherwise,}
\end{cases}
\]
we obtain
\[
\ket{\psi_4}=\frac{1}{ \sqrt{r}}\sum_{x,y=0}^{{r}-1}\omega_{r}^{xx_0}\ket{x}\,\delta_{y,xs}\ket{y},
\]
and by simplifying the sum over $y$, we obtain
\[
\ket{\psi_4}=\frac{1}{\sqrt r}\sum_{x=0}^{{r}-1}\omega_{r}^{xx_0}\ket{x}\,\ket{xs}.
\]

In {step 6}, we measure the first and second registers in the computational basis and obtain two results: $r_1$ and $r_2\equiv r_1s \mod r$, where $r_1$ is chosen in the interval $[0,r-1]$ uniformly at random.

In {step 7}, if $\gcd(r_1,r)=1$, we calculate $s\equiv r_2/r_1 \mod r$. Since $s<r<N$, $s$ satisfies $a^s\equiv b\mod N$. For any odd $r_1$, $\gcd(r_1,r)=1$. Since half of the values of $r_1$ are odd, the success probability is 1/2.

\section{Final remarks}

In this Chapter, we assumed for simplicity that the order $r$ is a power of $2$. In the general case, one applies the quantum Fourier transform over a larger register of size $q=2^m$, chosen so that $q$ is sufficiently large compared with $r$. Then the measurement outcomes are no longer exact multiples of $q/r$, but are concentrated near them. The value of $r$ can then be recovered by classical post-processing, for instance by using continued fractions.

Another simplifying assumption made in the exposition is that the measurement outcomes determine the desired quantity in a single run. In practice, the algorithm is probabilistic, and one usually has to repeat it a few times in order to obtain enough information to reconstruct the solution with high probability.

Finally, the algebraic structures introduced in this Chapter, such as the subgroup $L_r$ and its cosets, show that the discrete logarithm problem can be interpreted in terms of periodicity in a finite abelian group. This viewpoint is fundamental in the analysis of Shor's algorithm and connects the discrete logarithm problem with the broader framework of quantum algorithms based on the quantum Fourier transform.

\chapter{Grover's Algorithm}\label{chap:Grover}

Grover's algorithm~\cite{Gro96,Gro97} is a search algorithm initially developed for unstructured data. It can also be described in terms of an oracle, which is a function with some promise or property that can be evaluated as many times as we want, and our goal is to determine the property that the function has. This chapter follows the latter description with a focus on the circuit model. The analysis of the algorithm is based on a geometric interpretation, and as an application example, we solve an instance of a SAT problem. Grover's algorithm is optimal, that is, it cannot be improved~\cite{BBBV97,Zal99}, has already been used to create new quantum algorithms~\cite{LMOP22,MSS07}, and is described in many books~\cite{Bar09,BCS07,BH13,HIKKO14,Hir10,KLM07,KSV02,LR14,Mer07,NC00,Por18book,RP11,SS08,Wil08}.

\section{Problem formulation in terms of an oracle}

Let $N$ be a power of 2 for some integer $n$, that is, $N=2^n$. Suppose that $f:\big\{0,\ldots,N-1\big\}\rightarrow \{0,1\}$ is a Boolean function such that $f(x)=1$ if and only if $x=x_0$ for some fixed value $x_0$, that is,
\begin{equation*}\label{ag_f_x}
    f(x) = \left\{
  \begin{array}{l@{\quad}l}
    1, & \hbox{if \ensuremath{x=x_0},} \\
    0, & \hbox{otherwise.}
  \end{array}
\right.
\end{equation*}
Suppose that $x_0$ is unknown to us. How can we find $x_0$ by evaluating $f$? From a computational point of view, we want to evaluate $f$ as few as possible.

Classically, the most efficient algorithm queries this function $N$ times in the worst case, that is, the complexity of the classical algorithm in terms of the number of queries is $O(N)$. How is it done in practice? We ask someone else to generate a $n$-bit random number $x_0$. This person hides $x_0$ from us and makes a compiled subroutine of $f$. We can use the subroutine as many times as we want, but we cannot hack the code in search for $x_0$. The classical algorithm that solves this problem is an iteration that queries $f(x)$ for $x$ from 0 to $2^n-1$. As soon as $f(x)$ is 1, the program returns $x_0$.

\textit{Quantumly}, it is possible to improve the query complexity to $O(\sqrt N)$. How is it done in practice? We have to ask someone again to generate a $n$-bit random number $x_0$ and define $f$. This person, the \textit{oracle}, implements $f$ through a unitary matrix $U_{f}$ in a quantum computer. We can use $U_{f}$, but we cannot see the details of the implementation of $U_{f}$. Each time we use $U_{f}$, we add a unit to the count. We can use additional gates that obviously don't depend on $x_0$.

We have the same problem that can be solved by algorithms executed on two different machines. In the first case, a classical computer with $O(n)$ bits is used and the solution is found after $O(N)$ evaluations of $f$. In the second case, a quantum computer with $O(n)$ qubits is used and the solution is found after $O(\sqrt{N})$ evaluations of $f$. This improvement in complexity motivates the investment in quantum hardware and the development of quantum algorithms, which necessarily make use of state superposition. In the case of Grover's algorithm, applying $U_f$ to a superposition, together with additional quantum operations, allows the desired element to be identified with fewer evaluations of $f$ than are required classically.

\section{How to implement the oracle on a quantum computer} \label{subsec:implementar_f}

The first step in developing a quantum algorithm that solves Grover's problem is the implementation of the function $f$. Since $f$ is a Boolean function whose truth table has a single row with output 1, $f$ can be implemented with a multi-controlled NOT gate activated by $x_0$, as described in Section~\ref{sec:circ_Boolean_fcn} on Page~\pageref{sec:circ_Boolean_fcn}. This gate has an associated unitary matrix $U_{f}$, which is defined by its action on the computational basis as
\begin{equation*}
    U_{f} \ket{x}\ket{i}=\ket{x}\ket{i\oplus f(x)},
\end{equation*}
where $x$ is a $n$-bit string and $i$ is a bit. The first register has $n$ qubits and the second register has one qubit. If we take $i=0$, the above equation reduces to
\begin{equation*}
    U_{f}\ket{x}\ket{0}=\left\{
                             \begin{array}{l@{\quad}l}
                               \ket{x_0}\ket{1}, & \hbox{if \ensuremath{x=x_0},} \\
                               \ket{x}\ket{0}, & \hbox{otherwise,}
                             \end{array}
                           \right.
\end{equation*}
which describes the output of a multi-controlled NOT gate activated by $x_0$ (and only by $x_0$) when the input is $\ket{x}\ket{0}$. The result of the calculation of  $f(x)$ is stored in the second register while the state of the first register remains unchanged.

For example, the circuit that implements $U_{f}$ in the case $N=8$ and $x_0=6$, that is, $f(110)=1$ and $f(j)=0$ if $j\neq 110$, is
\[
\Qcircuit @C=2.3em @R=1.3em {
\lstick{\ket{1}}    & \ctrl{1} &  \rstick{\ket{1}} \qw \\
\lstick{\ket{1}}    & \ctrl{1} &  \rstick{\ket{1}} \qw \\
\lstick{\ket{0}}    & \ctrlo{1} &  \rstick{\ket{0}} \qw \\
\lstick{\ket{0}}  & \targ    &  \rstick{\ket{1}.}  \qw
}\vspace*{0.2cm}
\]
The first register has three qubits. Note that the state of the second register changes from $\ket{0}$ to $\ket{1}$ only if the input to the first register is $\ket{110}$ because 110 activates the three controls and all other 3-bit strings do not.

In Grover's algorithm, the state of the second register is always
\[
\ket{-}=\frac{\ket{0}-\ket{1}}{\sqrt 2}.
\]
Using linearity, the action of $U_{f}$ is given by
\begin{equation*}
    U_{f}\ket{x}\ket{-}=\left\{
                             \begin{array}{l@{\quad}l}
                               -\ket{x_0}\ket{-}, & \hbox{if \ensuremath{x=x_0},} \\
                               \,\,\,\,\ket{x}\ket{-}, & \hbox{otherwise.}
                             \end{array}
                           \right.
\end{equation*}
The same result is obtained from  Proposition~\ref{prop:U_fxminus} on Page~\pageref{prop:U_fxminus}, which states that
\[
U_{f}\ket{x}\ket{-}=(-1)^{f(x)}\ket{x}\ket{-}.
\]

\section{The algorithm}

Grover's algorithm uses an additional operator defined as
\begin{equation*}
    G=\big(2\, \ket{\textrm{d}}\bra{\textrm{d}} - I_N\big)\otimes I_2,
\end{equation*}
where
\[
\ket{\textrm{d}}=\frac{1}{\sqrt{N}}\sum_{j=0}^{N-1}\ket{j}.
\]
The notation ``$\ket{\textrm{d}}\bra{\textrm{d}}$'' denotes the \textit{outer product} between the vector $\ket{\textrm{d}}$ (an $N\times 1$ matrix) and the dual vector $\bra{\textrm{d}}$ (a $1\times N$ matrix). This outer product coincides with the usual matrix product. Multiplying an $N\times 1$ matrix by a $1\times N$ matrix yields an $N\times N$ matrix. Therefore, $\ket{\textrm{d}}\bra{\textrm{d}}$ is an $N\times N$ matrix, given by
\begin{equation*}
\ket{\textrm{d}}\bra{\textrm{d}} = \frac{1}{N}\begin{bmatrix}
                   1 & 1 & \cdots & 1 \\
                   1 & 1 & \cdots & 1\\
                   \vdots & \vdots & \ddots & \vdots \\
                   1 & 1 &\cdots & 1 \\
                 \end{bmatrix},
\end{equation*}
and
\begin{equation*}
\big(2\, \ket{\textrm{d}}\bra{\textrm{d}} - I_N\big) \,=\, \frac{1}{N}\begin{bmatrix}
                   (2-N) & 2 & \cdots & 2 \\
                   2 & (2-N) & \cdots & 2\\
                   \vdots & \vdots & \ddots & \vdots \\
                   2 & 2 &\cdots & (2-N) \\
                 \end{bmatrix}
\end{equation*}
which is called Grover matrix (or Grover operator).

Grover's algorithm is described in Algorithm~\ref{ag_XYZ}.

\begin{algorithm}[!ht]
\caption {Grover's algorithm} \label{ag_XYZ}
\KwIn{An integer $N$ and a function $f:\{0,...,N-1\}\rightarrow \{0,1\}$ such that $f(x)=1$ only for one point $x=x_0$ in the domain.}
\KwOut{$x_0$ with probability greater than or equal to $1-\frac{1}{N}$.} \BlankLine
Prepare the initial state $\ket{\textrm{d}}\ket{-}$  using $n+1$ qubits\;
Apply ${\left(G \, U_{f}\right)}^{t}$, where $t=\Big\lfloor\frac{\pi}{4}\sqrt N\Big\rfloor$ \;
Measure the first register in the computational basis.
\end{algorithm}

\section{Non-economical circuit of Grover's algorithm}

The goal of this Section is to find the circuit that implements the Grover operator using our understanding of implementing Boolean functions. The circuit initially uses more qubits than needed, but we will later demonstrate how to obtain a more resource-efficient version of it.

To obtain the circuit, we have to do an algebraic manipulation with the expression of the Grover matrix
$\big(2\, \ket{\textrm{d}}\bra{\textrm{d}} - I_N\big)$. Note that
\[
\ket{\textrm{d}}=H^{\otimes n}\ket{0}
\]
where $\ket{0}$ is in the decimal notation and $H^{\otimes n}=H\otimes\dots \otimes H$. Transposing the above equation, we obtain
\[
\bra{\textrm{d}}=\bra{0}H^{\otimes n}.
\]
Using $(H^{\otimes n})\cdot(H^{\otimes n})=(H\cdot H)^{\otimes n}=(I_2)^{\otimes n}=I_N$, we obtain
\[
\big(2\, \ket{\textrm{d}}\bra{\textrm{d}} - I_N\big)=
H^{\otimes n}\big(2\, \ket{0}\bra{0} - I_N\big)H^{\otimes n},
\]
where
\begin{equation*}
\big(2\, \ket{0}\bra{0} - I_N\big) = \begin{bmatrix}
                   1 & 0 & \cdots & 0 \\
                   \,\,\,\,0 & -1 & \cdots & 0\\
                   \,\,\,\,\vdots & \vdots & \ddots & \vdots \\
                   \,\,\,\,0 & 0 &\cdots & -1 \\
                 \end{bmatrix}.
\end{equation*}
Matrix $\big(2\, \ket{0}\bra{0} - I_N\big)$ acts only on the first register. However, it is simpler to implement this matrix using both registers. Let us show that it is implemented by a multi-controlled NOT gate activated by 0. Indeed, the action of $\big(2\, \ket{0}\bra{0} - I_N\big)$ on $\ket{x}$, where $x$ is a $n$-bit string, is
\[
\big(2\, \ket{0}\bra{0} - I_N\big)\ket{x}=\left\{
                             \begin{array}{l@{\quad}l}
                               \,\,\,\,\ket{0}, & \hbox{if \ensuremath{x=0},} \\
                               -\ket{x}, & \hbox{otherwise.}
                             \end{array}
                           \right.
\]
Therefore, the action of $\big(2\, \ket{0}\bra{0} - I_N\big)$ on the first register is the same as the action of $(-U_{f'})$ on both registers when $x_0=0$ (the state of the second register must be $\ket{-}$), where
\begin{equation*}
    f'(x) = \left\{
  \begin{array}{l@{\quad}l}
    1, & \hbox{if \ensuremath{x=0},} \\
    0, & \hbox{otherwise.}
  \end{array}
\right.
\end{equation*}
The minus sign in $(-U_{f'})$ changes neither the result of the algorithm nor the final probability. That is, using $G$ or $-G$ in Grover's algorithm does not change the final result.

Using these algebraic results, we conclude that a circuit that implements Grover's algorithm is
\[
\Qcircuit @C=1.9em @R=1.0em {
& & \hspace{90pt} \mbox{repeat $\big\lfloor\frac{\pi}{4}\sqrt{N}\big\rfloor$ times} & & & \\
\lstick{\ket{0}_1}        &\gate{H}                                & \ghost{U_{f}}   & \gate{H}  & \ctrlo{1}  & \gate{H} &   \meter    & \cw & \lstick{i_1}\\
& &  &   &  \controlo{1}  &  & &  \\
\lstick{\vdots\hspace{0.3cm}} & \vdots &  & \vdots  & & \vdots & \vdots & \rstick{\,\,\vdots}\\
& &    &   &  \controlo{1}  &  &   &  \\
\lstick{\ket{0}_n}        &\gate{H}                               & \ghost{U_{f}}     & \gate{H}   & \ctrlo{-1} \ctrlo{1}  & \gate{H} &    \meter    & \cw & \lstick{i_n} \\
\lstick{\ket{-}}        &\qw                                     & \multigate{-5}{U_{f}} &\qw &\targ{-1}  &   \qw & \qw & \lstick{\ket{-},}
\gategroup{2}{3}{7}{6}{.7em}{--}
}\vspace{0.3cm}
\]
where bits $i_1$, ..., $i_n$ are the outputs of the measurements. Those bits are the bits of $x_0$, that is $x_0=(i_1\dots i_n)_2$, with high probability.

Note that the depth of the circuit implementing Grover's algorithm is $\Omega(\sqrt{2^n})$. Since we still have to decompose the multi-controlled NOT gate in each iteration, this introduces an overhead of $O(\log n)$, resulting in a final depth of $O(\log n \sqrt{2^n})$ when using the decomposition of $C^n(X)$ provided in Refs.~\cite{CZFDPP24,KG25}.

\section{Economical circuit of Grover's algorithm} \label{sec:Grover_eco}

The second register of Grover's algorithm can be discarded since it is possible to make a more economical implementation of the oracle~\cite{MOJ18}. Let us start by showing how to implement the operator $\big(2\, \ket{0}\bra{0} - I_N\big)$ (modulo a global phase), which enables us to implement operator $G$ using only the first register. Let us show the equivalence of the following circuits:
$$
\Qcircuit @C=0.99em @R=1.0em {
{\ket{k_1}}    & & \ctrlo{1}  & \qw & {\ket{k_1}} &&&&&&&& &{\ket{k_1}}     & &\qw &\ctrlo{1}  & \qw & \qw& {\ket{k_1}}  \\
{\ket{k_2}}    & & \ctrlo{1}  & \qw & {\ket{k_2}} &&&&&&&& &{\ket{k_2}}     & &\qw &\ctrlo{1}  & \qw & \qw &{\ket{k_2}} \\
{\vdots}       & &            & & \vdots          &&&&&&&& &{\vdots} &       &     &   &    & &\vdots \\
\lstick{}      & &            &                   &&&&&&\equiv&& \lstick{}      & & &   &        & \\
{\ket{k_{n-1}}}\hspace{0.3cm}& & \ctrlo{-1} & \qw & \hspace{0.3cm}{\ket{k_{n-1}}} &&&&&&&& &{\ket{k_{n-1}}} &  & \qw & \ctrlo{-1}\ctrlo{1} & \qw & \qw & {\ket{k_{n-1}}}\\
{\ket{k_n}}    & & \ctrlo{-1}\ctrlo{1}  & \qw & {\ket{k_{n}}} &&&&&&&& {\ket{k_n}}    & &\gate{X} &\gate{H} & \targ{-1}  & \gate{H}& \gate{X} & \qw & {\hspace{1.7cm}(-1)^{\bar{k}_1\cdots\bar{k}_n}\ket{k_{n}}}.  \\
{\ket{-}}      & & \targ{-1}  &  \qw  & \hspace{1.5cm}{(-1)^{\bar{k}_1\cdots\bar{k}_n}\ket{-}}  &&&&&&& {}
}\vspace{0.3cm}
$$

\noindent
The output of the left-hand circuit is the $(n+1)$-qubit state
\[
\ket{k_1}\otimes\cdots\otimes\ket{k_{n}}\otimes\left((-1)^{\bar{k}_1\cdots\bar{k}_n}\ket{-}\right),
\]
which is obtained from the definition of the multi-controlled NOT gate active only when qubits $k_1, \dots, k_{n}$ are set to 0. The Kronecker product has the property
\[
\ket{v_1}\otimes(a\ket{v_2})=(a\ket{v_1})\otimes\ket{v_2}=a(\ket{v_1}\otimes\ket{v_2}),
\]
for any vectors $\ket{v_1}$, $\ket{v_2}$ and any scalar $a$. Then, we move the scalar term $(-1)^{\bar{k}_1\cdots\bar{k}_n}$ to the first register whose output is
\[
\ket{k_1}\otimes\cdots\otimes\ket{k_{n-1}}\otimes\left((-1)^{\bar{k}_1\cdots\bar{k}_n}\ket{k_n}\right)\,=\,(-1)^{\bar{k}_1\cdots\bar{k}_n}\ket{k_1}\otimes\cdots\otimes\ket{k_n}.
\]
The state of the second register is $\ket{-}$ and this register will be discarded at the end of the process.

Now let us show that the output of the right-hand circuit is the same. We use the fact that $XHXHX=-Z$. Therefore, the right-hand circuit is equivalent to
$$
\Qcircuit @C=0.99em @R=1.0em {
\lstick{\ket{k_1}}   &  &\ctrlo{1}  & \qw & \rstick{\ket{k_1}}  \\
\lstick{\ket{k_2}}    &  &\ctrlo{1}  & \qw& \rstick{\ket{k_2}} \\
{\vdots} &     &     &    & &\vdots \\
\lstick{}      & & &         & \\
\lstick{\ket{k_{n-1}}} &   & \ctrlo{-1}\ctrlo{1} & \qw & \rstick{\ket{k_{n-1}}}\\
\lstick{\ket{k_n}}    &  & \gate{-Z}  & \qw & \rstick{(-1)^{\bar{k}_1\cdots\bar{k}_n}\ket{k_{n}}.}  \\
}\vspace{0.3cm}
$$
The output is obtained using that $(-Z)\ket{k_n}=(-1)^{\bar k_n}\ket{k_n}$, and since $-Z$ is active only when qubits $k_1, \dots, k_{n-1}$ are set to 0, we obtain the overall output $(-1)^{\bar k_1\cdots \bar k_{n-1}\bar k_n}\ket{ k_1 \cdots k_{n-1}k_n}$. In conclusion, the result of the first circuit (after the elimination of the second register) is the same as the result of the second circuit. Since we are going to measure only the first register, this completes the proof that we can replace the first circuit with the second in Grover's algorithm without any loss. This replacement is not valid in all algorithms. So far we have shown that $G$ (modulo a global phase) can be implemented using only the first register.

Let us consider the oracle. If the oracle chooses $x_0=0$, the circuit of $U_f$ is a multi-controlled NOT gate active only when all qubits of the first register are set to 0. In this case, we have already shown how to implement $U_f$ using only the first register. The oracle would use the right-hand circuit depicted at the beginning of this Section. If the oracle chooses $x_0=1$, the circuit of $U_f$ is a multi-controlled NOT gate that is active only when all qubits of the first register are set to 0 except the $n$-th qubit, which is set to 1. In this case, we have the following circuit equivalence:\vspace{5pt}\\
$\mbox{\hspace{1.5cm}}\Qcircuit @C=1.17em @R=1.0em {
{\ket{k_1}}    & & \ctrlo{1}  & \qw & {\ket{k_1}} &&&&&&&& &{\ket{k_1}}     & &\qw &\ctrlo{1}  & \qw & \qw& {\ket{k_1}}  \\
{\ket{k_2}}    & & \ctrlo{1}  & \qw & {\ket{k_2}} &&&&&&&& &{\ket{k_2}}     & &\qw &\ctrlo{1}  & \qw & \qw &{\ket{k_2}} \\
{\vdots}       & &            & & \vdots          &&&&&&&& &{\vdots}        &     &   &    & &&\vdots \\
\lstick{}      & &            &                   &&&&&&\equiv&& \lstick{}      & & &   &        & \\
{\ket{k_{n-1}}}\hspace{0.3cm}& & \ctrlo{-1} & \qw & \hspace{0.3cm}{\ket{k_{n-1}}} &&&&&&&& &{\ket{k_{n-1}}}\hspace{0.2cm} &  & \qw & \ctrlo{-1}\ctrlo{1} & \qw & \qw & \hspace{0.2cm}{\ket{k_{n-1}}}\\
{\ket{k_n}}    & & \ctrl{-1}\ctrlo{1}  & \qw & {\ket{k_{n}}} &&&&&&&& &{\ket{k_n}}     & &\gate{H} & \targ{-1}  & \gate{H}&  \qw & &{\hspace{1.7cm}(-1)^{\bar{k}_1\cdots\bar{k}_{n-1}k_n}\ket{k_{n}}.}  \\
{\ket{-}}      & & \targ{-1}  &  \qw  & \hspace{2.2cm}{(-1)^{\bar{k}_1\cdots\bar{k}_{n-1}k_n}\ket{-}}  &&&&&&& {}
}\vspace{10pt}
$

\noindent
The equivalence check is similar to the previous case but now it is obtained using that $HXH=Z$ and $Z\ket{k_n}=(-1)^{k_n}\ket{k_n}$, and since $Z$ is active only when qubits $k_1, k_2,\dots, k_{n-1}$ are set to 0, we obtain the overall output $(-1)^{\bar k_1\bar k_2\cdots\bar k_{n-1}k_n}\ket{ k_1 k_2\cdots k_n}$.  This concludes the proof that the oracle would use a $n$-qubit circuit if $x_0=1$.

The remaining cases, $x_0\ge 2$, are obtained from the previous results. If the rightmost (last) bit of $x_0$ is 0, we use the circuit equivalence described at the beginning of this Section in the following way: if any control (except the $n$-th) on the left-hand circuit changes from an empty circle to a full circle, the same must happen to the corresponding controls on the right-hand circuit. If the rightmost bit of $x_0$ is 1, we use the second circuit equivalence of this Section, and in the same fashion, if any control (except the $n$-th) on the left-hand circuit changes from an empty circle to a full circle, the same must happen to the corresponding controls on the right-hand circuit. The expression for the output changes accordingly.

Thus, not only $G$ but also $U_{f}$ can be implemented with $n$ qubits by eliminating the second register and introducing two Hadamard gates plus two Pauli $X$ gates if the $n$-th qubit is activated by 0, and only two Hadamard gates if the $n$-th qubit is activated by 1.

\subsection*{Economical circuit when $N=4$}
The circuit of Grover's algorithm in the economical form when $N=4$ and $x_0=11$ is
\[\mbox{\hspace{20pt}}
\Qcircuit @C=1.5em @R=1.4em {
\lstick{\ket{0}} & \gate{H} & \qw     &\ctrl{1}&  \qw     & \gate{H} &   \qw& \qw& \ctrlo{1} &\qw &\qw & \gate{H} &   \meter    & \cw & \lstick{1\,\,}\\
\lstick{\ket{0}} & \gate{H} & \gate{H}&\targ   & \gate{H} & \gate{H} & \gate{X} & \gate{H}   & \targ  &  \gate{H} & \gate{X} & \gate{H}  &    \meter    & \cw & \lstick{1.}
\gategroup{1}{3}{2}{5}{.7em}{--}\gategroup{1}{7}{2}{11}{.7em}{--}
}
\]
The gates inside the first dashed box implement the oracle and the gates inside the second dashed box implement $\big(2\, \ket{0}\bra{0} - I_N\big)$ (modulo a global phase). This circuit can be simplified by substituting $HXH$ in the second qubit with $Z$ in two places.

The goal of Grover's algorithm is to determine $x_0$ by querying the oracle, that is, using the first dashed box, without looking at its implementation details. We have to pretend that the first dashed box is a black box. When $N=4$, there are four possible black boxes, case $x_0=11$ is one of them. The gates that implement the oracle when $x_0=00$, $x_0=01$, and $x_0=10$ are $(X\otimes XH)\text{CNOT}(X\otimes HX)$, $(X\otimes H)\text{CNOT}(X\otimes H)$, and $(I\otimes XH)\text{CNOT}(I\otimes HX)$, respectively. When $N=4$, the output is the correct one with probability 1.

\subsection*{Economical circuit for arbitrary $N$}
For an arbitrary $N$ ($N=2^n$ and $n\ge 2$), the circuit of Grover's algorithm with only $n$ qubits is
\[
\Qcircuit @C=1.8em @R=1.0em {
								&         & \hspace{90pt} \mbox{repeat $\big\lfloor\frac{\pi}{4}\sqrt{N}\big\rfloor$ times} & & & \\
\lstick{\ket{0}_1\,\,\,\,\,\,}  &\gate{H} &\ghost{u_{f}}& \gate{H}  & \ctrlo{1}  & \gate{H} &   \meter    & \cw & {i_1}\\
								&         &                    &   &  \controlo{1}  &  & &  \\
\lstick{\vdots\hspace{0.3cm}} 	& \vdots  &                    & \vdots  & & \vdots & \vdots & & {\,\,\vdots}\\
\lstick{\hspace{0.3cm}} 		&         &                    &   & &  &  & \rstick{\,\,}\\
\lstick{\ket{0}_{n-1}}        	&\gate{H} &\ghost{u_{f}}& \gate{H}  &  \ctrlo{1}\ctrlo{-1}  & \gate{H} &      \meter    & \cw & {i_{n-1}}\\
\lstick{\ket{0}_n\,\,\,\,\,}  	&\gate{H} &\multigate{-5}{u_{f}}& \gate{Z}   & \targ{-1} \ctrlo{-1}  & \gate{Z} &    \meter    & \cw & {i_n,}
\gategroup{2}{3}{7}{6}{.7em}{--}}\vspace{0.3cm}
\]
where the matrix $u_{f}$ is the economical version of $U_{f}$ and the circuit for $u_{f}$ for an arbitrary $x_0=(i_1\dots i_n)_2$ is
\[
\Qcircuit @C=1.1em @R=1.0em {
&\ghost{u_f}&\qw&&& &\gate{X^{i_1}}&\gate{X}& \ctrl{1}           &\gate{X}&\gate{X^{i_1}}&\qw\\
&& &&&                    &&&   \\
\vdots &&\vdots  &  &&& \vdots  & \vdots & & \vdots  & {\,\,\vdots}\\
&& &&\equiv&                    &&&   \\
&\ghost{u_f}&\qw&&& &\gate{X^{i_{n-1}}}&\gate{X}& \ctrl{1}\ctrlo{-1} &\gate{X}&\gate{X^{i_{n-1}}}&\qw\\
&\multigate{-5}{u_f}&\qw&&&&\gate{X^{1-i_n}}&\gate{H}& \targ              &\gate{H}&\gate{X^{1-i_n}}& \qw & .
}\vspace{0.3cm}
\]

\section{Analysis of the algorithm}\label{sec:ana-Grover}

Why does Grover's algorithm work correctly? We answer this question using the economical form of the algorithm. The operators in this case are $u_{f}$, which is defined as
\[
u_f \,=\, \sum_x (-1)^{f(x)}\ket{x}\bra{x},
\]
and
\[
g \,=\, 2\, \ket{\textrm{d}}\bra{\textrm{d}} - I_N.
\]
Operator $u_f$ is the economical version of $U_f$, that is, $u_f$ is a $N$-dimensional operator whose action on the computational basis is
\begin{equation*}
    u_{f}\ket{x}=\left\{
                             \begin{array}{l@{\quad}l}
                               -\ket{x_0}, & \hbox{if \ensuremath{x=x_0},} \\
                               \,\,\,\,\ket{x}, & \hbox{otherwise.}
                             \end{array}
                           \right.
\end{equation*}
In its turn, $g$ is the economical version of operator $G$. The economical version of Grover's algorithm is described in Algorithm~\ref{ag_economical}.

\begin{algorithm}[!ht]
\caption {Grover's algorithm (economical version~)} \label{ag_economical}
\KwIn{An integer $N$ (power of 2) and a function $f:\{0,...,N-1\}\rightarrow \{0,1\}$ such that $f(x)=1$ only for one point $x=x_0$ in the domain.}
\KwOut{$x_0$ with probability greater than or equal to $1-\frac{1}{N}$.} \BlankLine
Prepare the initial state $\ket{\textrm{d}}$ using $n$ qubits \;
Apply ${\left(g \, u_{f}\right)}^{t}$, where $t=\Big\lfloor\frac{\pi}{4}\sqrt N\Big\rfloor$ \;
Measure all qubits in the computational basis.
\end{algorithm}

The goal of the algorithm is to find $x_0$, which is a $n$-bit string. It will be accomplished if the state of the qubits just before the measurement is $\ket{x_0}$ because the measurement in this case returns $x_0$. The analysis of the algorithm that we now start to describe is based on a geometric interpretation of vector reflections~\cite{Aha98}. At the beginning of the algorithm, the state of the qubits is $\ket{\textrm{d}}$. For large $N$, $\ket{\textrm{d}}$ is almost orthogonal to $\ket{x_0}$. Fig.~\ref{fig:Grover_analysis} shows vectors $\ket{\textrm{d}}$ and $\ket{x_0}$, where $\theta/2$ is the angle between $\ket{\textrm{d}}$ and the horizontal axis. Any other representation of those vectors can be used in the analysis of the algorithm provided that $\ket{\textrm{d}}$ is almost orthogonal to $\ket{x_0}$.

\begin{figure}[ht!]
\centering
\begingroup%
  \makeatletter%
  \providecommand\color[2][]{%
    \errmessage{(Inkscape) Color is used for the text in Inkscape, but the package 'color.sty' is not loaded}%
    \renewcommand\color[2][]{}%
  }%
  \providecommand\transparent[1]{%
    \errmessage{(Inkscape) Transparency is used (non-zero) for the text in Inkscape, but the package 'transparent.sty' is not loaded}%
    \renewcommand\transparent[1]{}%
  }%
  \providecommand\rotatebox[2]{#2}%
  \ifx\svgwidth\undefined%
    \setlength{\unitlength}{198.83047861bp}%
    \ifx\svgscale\undefined%
      \relax%
    \else%
      \setlength{\unitlength}{\unitlength * \real{\svgscale}}%
    \fi%
  \else%
    \setlength{\unitlength}{\svgwidth}%
  \fi%
  \global\let\svgwidth\undefined%
  \global\let\svgscale\undefined%
  \makeatother%
  \begin{picture}(1,0.61197067)%
    \put(0,0){\includegraphics[width=\unitlength,page=1]{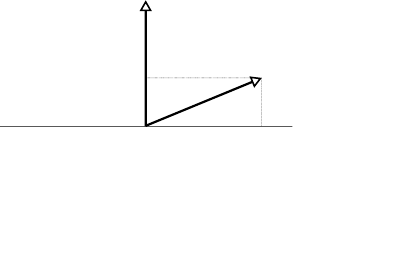}}%
    \put(0.64597669,0.42806614){\color[rgb]{0,0,0}\makebox(0,0)[lb]{\smash{$\ket{\text{d}}$}}}%
    \put(0.37016459,0.53009315){\color[rgb]{0,0,0}\makebox(0,0)[lb]{\smash{$\ket{x_0}$}}}%
    \put(0,0){\includegraphics[width=\unitlength,page=2]{Grover_analysis.pdf}}%
    \put(0.47603825,0.31705346){\color[rgb]{0,0,0}\makebox(0,0)[lb]{\smash{${\theta}/{2}$}}}%
  \end{picture}%
\endgroup%
\caption{Depiction of vectors $\ket{x_0}$ and $\ket{\textrm{d}}$.}
\label{fig:Grover_analysis}
\end{figure}

The angle $\theta$ is very small for large $N$, and in this case, $\theta/2$ is a good approximation of $\sin(\theta/2)$.
In addition, the sine of an angle is equal to the cosine of the complement, that is,
\[
\frac{\theta}{2}\approx \sin\frac{\theta}{2}=\cos\left(\frac{\pi}{2}-\frac{\theta}{2}\right).
\]
As $(\pi-\theta)/2$ is the angle between $\ket{x_0}$ and $\ket{\textrm{d}}$, by definition of the inner product, $\cos \,(\pi-\theta)/2$ is the inner product of vectors $\ket{x_0}$ and $\ket{\textrm{d}}$, the result of which is
\[
\frac{\theta}{2}\approx \sin\frac{\theta}{2}=\cos\left(\frac{\pi}{2}-\frac{\theta}{2}\right)=\braket{x_0}{\textrm{d}}=\frac{1}{\sqrt{N}}.
\]
Therefore,
\[
\theta\approx \frac{2}{\sqrt{N}}.
\]

The first step of Algorithm~\ref{ag_economical} is the preparation of the initial state $\ket{\textrm{d}}$. The next step is to apply $u_{f}$ to $\ket{\textrm{d}}$. The action of $u_{f}$ on $\ket{\textrm{d}}$ (written in the computational basis) inverts the sign of the amplitude of $\ket{x_0}$ and does not change the other amplitudes. The amplitude of $\ket{x_0}$ is the orthogonal projection of $\ket{\text{d}}$ on the vertical axis---see Fig.~\ref{fig:Grover_analysis}, which is inverted by the action of $u_{f}$. Geometrically, the action of $u_{f}$ is represented by a reflection of $\ket{\textrm{d}}$ about the horizontal axis. The angle between the vectors $\ket{\textrm{d}}$ and $(u_{f}\ket{\textrm{d}})$ is $\theta$, as shown in Fig.~\ref{fig:Grover_analysis_2}.

\begin{figure}[!ht]
\centering
\begingroup%
  \makeatletter%
  \providecommand\color[2][]{%
    \errmessage{(Inkscape) Color is used for the text in Inkscape, but the package 'color.sty' is not loaded}%
    \renewcommand\color[2][]{}%
  }%
  \providecommand\transparent[1]{%
    \errmessage{(Inkscape) Transparency is used (non-zero) for the text in Inkscape, but the package 'transparent.sty' is not loaded}%
    \renewcommand\transparent[1]{}%
  }%
  \providecommand\rotatebox[2]{#2}%
  \newcommand*\fsize{\dimexpr\f@size pt\relax}%
  \newcommand*\lineheight[1]{\fontsize{\fsize}{#1\fsize}\selectfont}%
  \ifx\svgwidth\undefined%
    \setlength{\unitlength}{288.97106456bp}%
    \ifx\svgscale\undefined%
      \relax%
    \else%
      \setlength{\unitlength}{\unitlength * \real{\svgscale}}%
    \fi%
  \else%
    \setlength{\unitlength}{\svgwidth}%
  \fi%
  \global\let\svgwidth\undefined%
  \global\let\svgscale\undefined%
  \makeatother%
  \begin{picture}(1,0.44099056)%
    \lineheight{1}%
    \setlength\tabcolsep{0pt}%
    \put(0,0){\includegraphics[width=\unitlength,page=1]{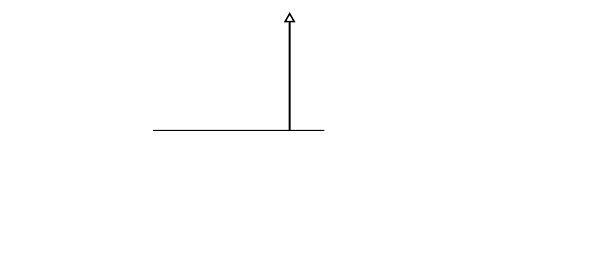}}%
    \put(0.67092912,0.30262045){\color[rgb]{0,0,0}\makebox(0,0)[lt]{\lineheight{0}\smash{\begin{tabular}[t]{l}$\ket{\text{d}}$\end{tabular}}}}%
    \put(0.49301386,0.36843396){\color[rgb]{0,0,0}\makebox(0,0)[lt]{\lineheight{0}\smash{\begin{tabular}[t]{l}$\ket{x_0}$\end{tabular}}}}%
    \put(0,0){\includegraphics[width=\unitlength,page=2]{Grover_analysis_2.pdf}}%
    \put(0.56130869,0.21024733){\color[rgb]{0,0,0}\makebox(0,0)[lt]{\lineheight{0}\smash{\begin{tabular}[t]{l}${\theta}$\end{tabular}}}}%
    \put(0,0){\includegraphics[width=\unitlength,page=3]{Grover_analysis_2.pdf}}%
    \put(0.67056725,0.13027895){\color[rgb]{0,0,0}\makebox(0,0)[lt]{\lineheight{0}\smash{\begin{tabular}[t]{l}$u_{f}\ket{\text{d}}$\end{tabular}}}}%
    \put(0.1021699,0.21087724){\color[rgb]{0,0,0}\makebox(0,0)[lt]{\lineheight{0}\smash{\begin{tabular}[t]{l}$\mbox{}$\end{tabular}}}}%
    \put(0,0){\includegraphics[width=\unitlength,page=4]{Grover_analysis_2.pdf}}%
  \end{picture}%
\endgroup%
\caption{Vector $u_{f}\ket{\textrm{d}}$ is a reflection of $\ket{\textrm{d}}$ about the horizontal axis.}
\label{fig:Grover_analysis_2}
\end{figure}

The next step is to apply $g=\big(2\, \ket{\textrm{d}}\bra{\textrm{d}} - I_N\big)$. Let us show that the action of $g$ is a reflection about the axis defined by $\ket{\textrm{d}}$. This proof is done in two parts. First, we show that $\ket{\textrm{d}}$ is invariant under the action of $g$. Second, we show that the action of $g$ on $\ket{\textrm{d}^\perp}$ inverts the sign of $\ket{\textrm{d}^\perp}$, where $\ket{\textrm{d}^\perp}$ is any vector orthogonal to $\ket{\textrm{d}}$. The first step follows from
\[
g\ket{\textrm{d}}=\big(2\, \ket{\textrm{d}}\bra{\textrm{d}} - I_N\big)\ket{\textrm{d}}=2\, \ket{\textrm{d}}\braket{\textrm{d}}{\textrm{d}} - \ket{\textrm{d}}= \ket{\textrm{d}},
\]
because $\braket{\textrm{d}}{\textrm{d}}=1$. The second step follows from
\[
g\ket{\textrm{d}^\perp}=\big(2\, \ket{\textrm{d}}\bra{\textrm{d}} - I_N\big)\ket{\textrm{d}^\perp}=2\, \ket{\textrm{d}}\braket{\textrm{d}}{\textrm{d}^\perp} - \ket{\textrm{d}^\perp}=- \ket{\textrm{d}^\perp},
\]
because $\braket{\textrm{d}}{\textrm{d}^\perp}=0$.

\begin{figure}[ht!]
\centering
\begingroup%
  \makeatletter%
  \providecommand\color[2][]{%
    \errmessage{(Inkscape) Color is used for the text in Inkscape, but the package 'color.sty' is not loaded}%
    \renewcommand\color[2][]{}%
  }%
  \providecommand\transparent[1]{%
    \errmessage{(Inkscape) Transparency is used (non-zero) for the text in Inkscape, but the package 'transparent.sty' is not loaded}%
    \renewcommand\transparent[1]{}%
  }%
  \providecommand\rotatebox[2]{#2}%
  \newcommand*\fsize{\dimexpr\f@size pt\relax}%
  \newcommand*\lineheight[1]{\fontsize{\fsize}{#1\fsize}\selectfont}%
  \ifx\svgwidth\undefined%
    \setlength{\unitlength}{290.47107028bp}%
    \ifx\svgscale\undefined%
      \relax%
    \else%
      \setlength{\unitlength}{\unitlength * \real{\svgscale}}%
    \fi%
  \else%
    \setlength{\unitlength}{\svgwidth}%
  \fi%
  \global\let\svgwidth\undefined%
  \global\let\svgscale\undefined%
  \makeatother%
  \begin{picture}(1,0.43871326)%
    \lineheight{1}%
    \setlength\tabcolsep{0pt}%
    \put(0,0){\includegraphics[width=\unitlength,page=1]{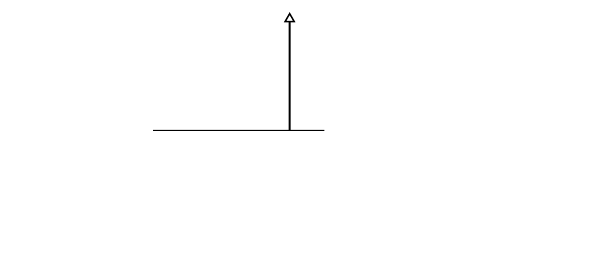}}%
    \put(0.64240179,0.3510683){\color[rgb]{0,0,0}\makebox(0,0)[lt]{\smash{\begin{tabular}[t]{l}$g\, u_{f}\ket{\text{d}}$\end{tabular}}}}%
    \put(0.49046792,0.36653135){\color[rgb]{0,0,0}\makebox(0,0)[lt]{\smash{\begin{tabular}[t]{l}$\ket{x_0}$\end{tabular}}}}%
    \put(0,0){\includegraphics[width=\unitlength,page=2]{Grover_analysis_3.pdf}}%
    \put(0.55757465,0.21223707){\color[rgb]{0,0,0}\makebox(0,0)[lt]{\smash{\begin{tabular}[t]{l}${\theta}$\end{tabular}}}}%
    \put(0.67226845,0.15542629){\color[rgb]{0,0,0}\makebox(0,0)[lt]{\smash{\begin{tabular}[t]{l}$u_{f}\ket{\text{d}}$\end{tabular}}}}%
    \put(0,0){\includegraphics[width=\unitlength,page=3]{Grover_analysis_3.pdf}}%
    \put(0.68217813,0.26588947){\color[rgb]{0,0,0}\makebox(0,0)[lt]{\smash{\begin{tabular}[t]{l}$\ket{\text{d}}$\end{tabular}}}}%
    \put(0.54982856,0.25096721){\color[rgb]{0,0,0}\makebox(0,0)[lt]{\smash{\begin{tabular}[t]{l}${\theta}$\end{tabular}}}}%
    \put(0,0){\includegraphics[width=\unitlength,page=4]{Grover_analysis_3.pdf}}%
  \end{picture}%
\endgroup%
\caption{Vector $g\,u_{f}\ket{\textrm{d}}$ is a reflection of $u_{f}\ket{\textrm{d}}$ about $\ket{\textrm{d}}$. }
\label{fig:Grover_analysis_3}
\end{figure}

Fig.~\ref{fig:Grover_analysis_3} depicts $g\, u_{f} \ket{\textrm{d}}$ and shows that the action of $g u_{f}$ rotates the initial state by $\theta$ degrees towards $\ket{x_0}$. Since $\theta$ is a small angle, this improvement is modest but promising. It is easy to see that the second application of $g u_{f}$ repeats the process of rotating by $\theta$ degrees towards $\ket{x_0}$. We want to know how many iterations $r$ are needed so that $r \theta=\pi/2$. The number of iterations is
\[
r=\left\lfloor\frac{\pi}{2\theta}\right\rfloor=\left\lfloor\frac{\pi}{4}\sqrt{N}\right\rfloor.
\]

\begin{figure}[ht!]
\centering
\begingroup%
  \makeatletter%
  \providecommand\color[2][]{%
    \errmessage{(Inkscape) Color is used for the text in Inkscape, but the package 'color.sty' is not loaded}%
    \renewcommand\color[2][]{}%
  }%
  \providecommand\transparent[1]{%
    \errmessage{(Inkscape) Transparency is used (non-zero) for the text in Inkscape, but the package 'transparent.sty' is not loaded}%
    \renewcommand\transparent[1]{}%
  }%
  \providecommand\rotatebox[2]{#2}%
  \newcommand*\fsize{\dimexpr\f@size pt\relax}%
  \newcommand*\lineheight[1]{\fontsize{\fsize}{#1\fsize}\selectfont}%
  \ifx\svgwidth\undefined%
    \setlength{\unitlength}{284.58403681bp}%
    \ifx\svgscale\undefined%
      \relax%
    \else%
      \setlength{\unitlength}{\unitlength * \real{\svgscale}}%
    \fi%
  \else%
    \setlength{\unitlength}{\svgwidth}%
  \fi%
  \global\let\svgwidth\undefined%
  \global\let\svgscale\undefined%
  \makeatother%
  \begin{picture}(1,0.44778868)%
    \lineheight{1}%
    \setlength\tabcolsep{0pt}%
    \put(0,0){\includegraphics[width=\unitlength,page=1]{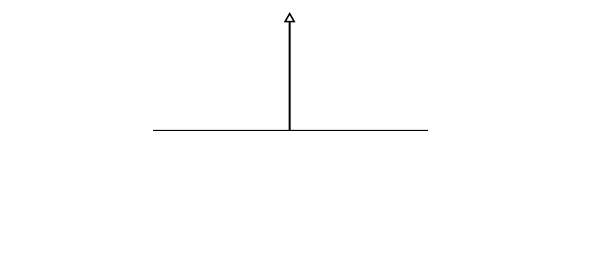}}%
    \put(0.50061396,0.37411358){\color[rgb]{0,0,0}\makebox(0,0)[lt]{\lineheight{0}\smash{\begin{tabular}[t]{l}$\ket{x_0}$\end{tabular}}}}%
    \put(0,0){\includegraphics[width=\unitlength,page=2]{Grover_analysis_4.pdf}}%
    \put(0.42717343,0.32197165){\color[rgb]{0,0,0}\makebox(0,0)[lt]{\lineheight{0}\smash{\begin{tabular}[t]{l}$\,\,\,\frac{{\theta}}{2}$\end{tabular}}}}%
    \put(0,0){\includegraphics[width=\unitlength,page=3]{Grover_analysis_4.pdf}}%
    \put(0.31654627,0.40661015){\color[rgb]{0,0,0}\makebox(0,0)[lt]{\lineheight{0}\smash{\begin{tabular}[t]{l}$\ket{\psi}$\end{tabular}}}}%
    \put(0,0){\includegraphics[width=\unitlength,page=4]{Grover_analysis_4.pdf}}%
    \put(0.40541062,0.26764944){\color[rgb]{0,0,0}\makebox(0,0)[lt]{\lineheight{0}\smash{\begin{tabular}[t]{l}$a$\end{tabular}}}}%
    \put(0,0){\includegraphics[width=\unitlength,page=5]{Grover_analysis_4.pdf}}%
    \put(0.55367373,0.24027354){\color[rgb]{0,0,0}\makebox(0,0)[lt]{\lineheight{0}\smash{\begin{tabular}[t]{l}${\theta}/2$\end{tabular}}}}%
    \put(0.65346851,0.34775932){\color[rgb]{0,0,0}\makebox(0,0)[lt]{\lineheight{0}\smash{\begin{tabular}[t]{l}$\ket{\text{d}}$\end{tabular}}}}%
  \end{picture}%
\endgroup%
\caption{Vector $\ket{\psi}$ is the final state before measurement and $a$ is the norm of the projection of $\ket{x_0}$ on $\ket{\psi}$. The angle between $\ket{\psi}$ and $\ket{x_0}$ is less than or equal to $\theta/2$. }
\label{fig:Grover_analysis_4}
\end{figure}

It remains to calculate the success probability. After $r$ iterations, the state of the qubits is
\[
\ket{\psi}={\left(g \, u_{f}\right)}^{\big\lfloor\frac{\pi}{4}\sqrt N\big\rfloor}\ket{\textrm{d}}.
\]
Vector $\ket{\psi}$ is almost orthogonal to $\ket{\textrm{d}}$ at this point, as depicted in Fig.~\ref{fig:Grover_analysis_4}. The angle between $\ket{\psi}$ and $\ket{x_0}$ is less than or equal to $\theta/2$.
The success probability is greater than or equal to the absolute square of the amplitude of $\ket{x_0}$ in the decomposition of $\ket{\psi}$ in the computational basis. This amplitude is $a$ as shown in Fig.~\ref{fig:Grover_analysis_4}. The orthogonal projection of $\ket{x_0}$ on the final state is at least $\cos(\theta/2)$. Therefore, the success probability $p=|a|^2$ satisfies
\begin{eqnarray*}
p \,\,\ge\,\,  \cos^2\frac{\theta}{2}
\,\,\ge\,\,  {1-\sin^2\frac{\theta}{2}}
\,\,\ge\,\,  {1-\frac{1}{N}}.
\end{eqnarray*}
The case $N=4$ is special because $\theta=60^\circ$, since $\sin(\theta/2)=1/\sqrt{N}$. With one application of $gu_{f}$, the vector $\ket{\textrm{d}}$ rotates by $60^\circ$ and coincides with $\ket{x_0}$. In this case, the success probability is exactly $p=1$.

\begin{exercise}
Show that in the two-dimensional invariant subspace spanned by $\ket{x_0}$ and $\ket{\textrm{d}}$, the operator $g u_f$ acts as a rotation matrix
\[
R_y(4\theta)=
\begin{bmatrix}
\cos(2\theta) & -\sin(2\theta) \\
\sin(2\theta) & \cos(2\theta)
\end{bmatrix}.
\]
\end{exercise}

\begin{exercise}
Grover~\cite{Gro97} called the operator $g=2\ket{\mathrm d}\bra{\mathrm d}-I$ an ``inversion about average''. This means that, if
\[
\bar a=\frac{1}{N}\sum_{j=0}^{N-1} a_j
\]
is the average of the amplitudes of
\[
\ket{\psi}=\sum_{j=0}^{N-1} a_j\ket{j},
\]
then each amplitude \(a_j\) is transformed into \(2\bar a-a_j\).

\begin{enumerate}
    \item[(a)] Show that
    \[
    \bar a=\frac{\braket{\mathrm d}{\psi}}{\sqrt{N}}
    \]
    and
    \[
    g\ket{\psi}
    =
    \sum_{j=0}^{N-1}(2\bar a-a_j)\ket{j}.
    \]
    Hence justify Grover's expression ``inversion about average''.

    \item[(b)] In Grover's algorithm, after the oracle is applied, the amplitude of the marked state has its sign reversed, while the other amplitudes remain unchanged. Using the interpretation of \(g\) as an inversion about average, explain qualitatively why this step tends to increase the amplitude of the marked state.

    \item[(c)] Compare this interpretation with the geometric description of Grover's algorithm presented in this chapter.
\end{enumerate}
\end{exercise}

\begin{exercise} \label{exer-grover-m} (Generalization to multiple marked elements)
In the analysis of Grover’s algorithm, we assumed that there is a single marked element $x_0$. Suppose now that there are $m$ marked elements, that is, the Boolean function $f$ satisfies $f(x)=1$ for exactly $m$ distinct values of $x$.
\begin{itemize}
\item[(a)] Show that the state space relevant to the evolution of the algorithm is the two-dimensional subspace spanned by
\[
\ket{w}=\frac{1}{\sqrt{m}}\sum_{x:\,f(x)=1}\ket{x}
\quad \text{and} \quad
\ket{r}=\frac{1}{\sqrt{N-m}}\sum_{x:\,f(x)=0}\ket{x}.
\]

\item[(b)] Show that if the initial state is $\ket{\textrm{d}}$, then
\[
\sin^2\theta=\frac{m}{N},
\]
where $\theta$ is defined as in the single-marked case.

\item[(c)] Prove that each Grover iteration rotates the state vector by an angle $2\theta$ in the subspace spanned by $\ket{w}$ and $\ket{r}$.

\item[(d)] Determine the number of iterations $t$ that maximizes the success probability and show that
\[
t=\left\lfloor \frac{\pi}{4}\sqrt{\frac{N}{m}}\right\rfloor.
\]

\item[(e)] Compute the corresponding success probability.
\end{itemize}
\end{exercise}

\section{Solving SAT with Grover's algorithm}\label{sec:SAT}

The Boolean satisfiability problem (SAT) is the problem of determining if there exists an assignment of values that satisfies a given Boolean formula. For example, consider the Boolean formula
\[
f(a,b,c)\,=\, a\wedge (c\vee (\bar b\wedge c)).
\]
This formula is satisfiable because the assignment $abc=101$ evaluates to True, that is, $f(1,0,1)=1$. If no such assignment exists, the formula is unsatisfiable. In general, it is hard to decide whether the formula is satisfiable or not because SAT is an NP-complete problem.

Since we can implement a circuit that evaluates the formula for any bit string, as shown in Section~\ref{sec:circ_Boolean_fcn}, we use this circuit as the oracle in Grover's algorithm. If we know beforehand the number $m$ of satisfying assignments of the formula, the number of iterations in the multi-marked version of Grover's algorithm~\cite{BBHT98,Por18book} is
\[
t=\left\lfloor\frac{\pi}{4}\sqrt{\frac{N}{m}}\right\rfloor
\]
(see Exercise~\ref{exer-grover-m}).

To implement the circuit that evaluates the formula $f(a,b,c)$, we use the techniques of Section~\ref{sec:circ_Boolean_fcn}. Since $f(a,b,c)$ has exactly two assignments that evaluate to True, $abc=101$ and $abc=111$, the number of iterations is $t=1$. Then, the circuit
\[
\Qcircuit @C=1.3em @R=0.7em {
\lstick{\ket{0}}&\gate{H}&\qw      & \qw     &\ctrl{4}& \qw     &\qw         &\gate{H}&\ctrlo{1}&\gate{H}&\meter&\rstick{i_1} \cw \\
\lstick{\ket{0}}&\gate{H}&\ctrlo{1}& \qw     &\qw     & \qw     &\ctrlo{1}   &\gate{H}&\ctrlo{1}&\gate{H}&\meter&\rstick{i_2} \cw \\
\lstick{\ket{0}}&\gate{H}&\ctrl{1} &\ctrlo{1}&\qw     &\ctrlo{1}&\ctrl{1}    &\gate{Z}&\targ    &\gate{Z}&\meter&\rstick{i_3} \cw \\
\lstick{\ket{0}}&\qw     & \targ   &\ctrlo{1}&\qw     &\ctrlo{1}& \targ      &\qw&\rstick{\ket{0}} \qw \\
\lstick{\ket{1}}&\qw     & \qw     & \targ   &\ctrl{1}& \targ   & \qw        &\qw&\rstick{\ket{1}} \qw \\
\lstick{\ket{-}}&\qw     & \qw     & \qw     &\targ   & \qw     & \qw        &\qw&\rstick{\ket{-}.} \qw
}\vspace*{0.cm}
\]
returns a satisfying assignment $i_1i_2i_3$ with probability 1, that is, $f(i_1,i_2,i_3)=1$. In general, the output has a success probability strictly smaller than 1, but greater than or equal to $1-1/N$. The only exception is the case $t=1$, which happens when $\theta=\pi/3$, or equivalently when
\[
\sin\frac{\theta}{2}=\sqrt{\frac{m}{N}}=\frac{1}{2}.
\]
The angle between vectors $\ket{\text{d}}$ and $\ket{x_0^\perp}$ is $\pi/6$ and after one rotation, the state of the algorithm is $\pi/6+\pi/3=\pi/2$, exactly equal to the marked state $\ket{x_0}$, as can be seen from the analysis of Section~\ref{sec:ana-Grover}.

If we do not know beforehand the number $m$ of satisfying assignments, we can rerun Grover's algorithm with a varying number of iterations and check after each run whether the measured output is a satisfying assignment. A strategy with provable performance for this case is described in~\cite{BBHT98}.

\begin{exercise}
Consider the Boolean formula
\[
f(a,b,c,d)=(a\vee b)\wedge(\bar a\vee c)\wedge(\bar b\vee d).
\]
\begin{itemize}
\item[(a)] Determine the satisfying assignments.
\item[(b)] Determine the number $m$ of satisfying assignments.
\item[(c)] Compute the required number of Grover iterations.
\item[(d)] Draw a quantum circuit that uses Grover's algorithm to solve this SAT instance, taking as oracle a circuit that computes $f(a,b,c,d)$.
\end{itemize}
\end{exercise}


\section{Final remarks}

The same technique for implementing the oracle using only $n$ qubits, analyzed in this Chapter, can be applied to the implementation of the Deutsch-Jozsa algorithm with only $n$ qubits. Note that the state of the second register in the Deutsch-Jozsa circuit before applying $U_f$ is $\ket{-}$. What we have to do is discard the second register and replace $U_f$ in the Deutsch-Jozsa circuit with $u_f$, described at the end of Section~\ref{sec:Grover_eco}.

\chapter{Phase Estimation and Applications}\label{chap:QPE}

Kitaev published the quantum phase estimation algorithm as a preprint in 1995~\cite{Kit95}, and later as a section in a book in Russian, which was translated into English~\cite{KSV02}. Kitaev's method is based on a procedure for measuring an eigenvalue of a unitary
operator; that is, given a unitary operator $U$ and one of its eigenvectors $\ket{\psi}$, the algorithm finds the eigenvalue $\exp(2\pi\ii \phi)$, so that $U\ket{\psi}=\exp(2\pi\ii \phi)\ket{\psi}$, where $\phi$ is the phase of the eigenvalue. This algorithm provides an alternative way of factoring integers and calculating discrete logarithms. Not only that, it is used in many applications such as quantum counting.
This algorithm has been described in many books~\cite{BCS07,BH13,Hid19,KLM07,NC00, SS08}.


\section{Quantum phase estimation algorithm}

Suppose we have a ${{n}}$-qubit unitary operator $U$ and we know one of its eigenvectors $\ket{\psi}$. We do not know the eigenvalue associated with $\ket{\psi}$, but we know that its analytical expression is  $\e^{2\pi \ii\phi}$, where $0\le \phi<1$ ($\phi$ is unknown), because $U$ is unitary. We assume for now that $\phi=0.\phi_1\cdots\phi_{m}$ for some integer ${m}$, where $\phi_1$, ..., $\phi_{m}$ are bits, that is, the phase of the eigenvalue $\e^{2\pi \ii \phi}$ is a rational multiple of $2\pi$. The goal of the phase estimation algorithm is to determine $\phi$ using $U$ as an oracle and $\ket{\psi}$ as an input.

\subsection*{Basic block}

The basic block of the circuit of the quantum phase estimation algorithm depends on an integer $0\le j< m$ and is given by\vspace*{3pt}
\[
\Qcircuit @C=2.0em @R=1.6em {
\lstick{\ket{0}}        & \gate{H} & \ctrl{1}  &  \rstick{\frac{\ket{0}+\e^{2\pi \ii \phi 2^j}\ket{1}}{\sqrt 2}} \qw \\
\lstick{\ket{\psi}}     & {/}^{{n}}\qw      & \gate{U^{2^j}}     &  \rstick{\ket{\psi}.} \qw
}\vspace{0.2cm}
\]
To verify the correctness of the output of the basic block, we have to use
\[
U\ket{\psi}\,=\,\e^{2\pi \ii \phi}\ket{\psi},
\]
and also
\[
U^{2^j}\ket{\psi}\,=\,\e^{2\pi \ii \phi\,2^j}\ket{\psi}.
\]
The output of the basic block is obtained by applying the controlled $U^{2^j}$ operator on $(H\ket{0})\otimes \ket{\psi}$, that is,
\[
C\big(U^{2^j}\big)\left(\frac{\ket{0}\ket{\psi}+\ket{1}\ket{\psi}}{\sqrt 2}\right)\,=\,
\frac{\ket{0}\ket{\psi}+\ket{1}U^{2^j}\ket{\psi}}{\sqrt 2}\,=\,
\frac{\ket{0}+\e^{2\pi \ii \phi 2^j}\ket{1}}{\sqrt 2}\otimes\ket{\psi}.
\]
Here we see an example of the \textit{phase kickback process} \index{phase kickback} because the phase was produced by the action of $U$ on the second register but it appears as a relative phase of the first qubit after $\ket{\psi}$ has been collected.

There is an alternative way of writing the eigenvalue of $U^{2^j}$ associated with $\ket{\psi}$. Using that $\phi=0.\phi_1\cdots\phi_{m}$ in binary, then
\[
\phi\,=\,\frac{\phi_1}{2}+\frac{\phi_2}{2^2}+\cdots+\frac{\phi_{m}}{2^{m}}.
\]
Multiplying by $2^j$, we obtain
\[
\phi\,2^j\,=\,2^{j-1}\phi_1+\cdots+2\phi_{j-1}+\phi_j+\frac{\phi_{j+1}}{2}+\cdots+\frac{\phi_{m}}{2^{{m}-j}}.
\]
It is straightforward to check that
\[
\exp\big(2\pi\ii \phi\,2^j\big)\,=\,\exp\left(2\pi\ii \left(\frac{\phi_{j+1}}{2}+\cdots+\frac{\phi_{m}}{2^{{m}-j}}\right)\right)
\]
because $\exp\big(2\pi\ii \,2^{j-1}\phi_1\big)=\dots=\exp\big(2\pi\ii \phi_j\big)=1$. Then,
\[
\exp\big(2\pi\ii \phi\,2^j\big)\,=\,\exp\left(2\pi\ii \,0.\phi_{j+1}\cdots\phi_{m}\right).
\]
Note that the first digits of $\phi$ were eliminated.

The implementation of $U^{2^j}$ is not necessarily performed by $2^j$ applications of $U$. This method is inefficient if $m$ is large. The implementation depends on specific applications of the phase estimation algorithm. For instance, if $U$ performs modular arithmetic, the \textit{repeated squaring method} is employed.


The circuit of the quantum phase estimation algorithm has two blocks. The first is made of $m$ basic blocks and the second is the inverse Fourier transform. Let us start by describing the first block.

\subsection*{First block}

The circuit of the first block comprises $m$ basic blocks with a common second register, as depicted in Fig.~\ref{fig:interblockpe}. The first register has ${m}$ qubits with input $\ket{0}^{\otimes {m}}$ and the second register has ${{n}}$ qubits with input $\ket{\psi}$. The output is a direct consequence of each basic block, which uses $U^{2^j}$, where $j$ runs from 0 to ${m}-1$. The order of the controlled operations is irrelevant, but $j$ must be 0 for the ${m}$-th qubit, $j$ must be 1 for the $({m}-1)$-th qubit, and so on.

\begin{figure}[!ht]
\hspace{1cm}
\Qcircuit @C=1.2em @R=1.0em {
\lstick{\ket{0}}    &\qw & \qw            &\qw& \qw            & \qw & \cdots &&\qw     & \qw            &\gate{H}& \ctrl{6}           & \rstick{\frac{\ket{0}+\e^{2\pi \ii \phi 2^{{m}-1}}\ket{1}}{\sqrt 2}} \qw\\
\lstick{\ket{0}}    &\qw      & \qw            &\qw& \qw            & \qw & \cdots &&\gate{H}& \ctrl{5}       &\qw& \qw                & \rstick{\frac{\ket{0}+\e^{2\pi \ii \phi 2^{{m}-2}}\ket{1}}{\sqrt 2}} \qw\\
\lstick{\vdots\ \ } & \vdots  &                &&                &     &        &&        &                &&                    &&& \rstick{\vdots}\\
\lstick{ }          &         &                &&                &     &        &&        &                &&                    &       \\
\lstick{\ket{0}}    & \qw     &\qw&\gate{H}        & \ctrl{2}       & \qw & \cdots &&\qw     & \qw            &\qw& \qw                & \rstick{\frac{\ket{0}+\e^{2\pi \ii \phi 2^1}\ket{1}}{\sqrt 2}} \qw\\
\lstick{\ket{0}}    &\gate{H} & \ctrl{1}       &\qw& \qw            & \qw & \cdots &&\qw     & \qw            &\qw& \qw                & \rstick{\frac{\ket{0}+\e^{2\pi \ii \phi 2^0}\ket{1}}{\sqrt 2}} \qw\\
\lstick{\ket{\psi}} &/^{{n}}\qw & \gate{U^{2^0}} &\qw& \gate{U^{2^1}} & \qw & \cdots &&\qw& \gate{U^{2^{{m}-2}}}&\qw& \gate{U^{2^{{m}-1}}} & \rstick{\ket{\psi}}\qw
\gategroup{6}{2}{7}{3}{0.7em}{--}\gategroup{5}{4}{7}{5}{0.7em}{--}\gategroup{2}{9}{7}{10}{0.7em}{--}\gategroup{1}{11}{7}{12}{0.7em}{--}}
\caption{The first block of the phase estimation circuit is made of $m$ basic blocks sharing the second register.}
\label{fig:interblockpe}
\end{figure}

The output of the first register of the first block is
\[
\frac{\ket{0}+\e^{2\pi \ii \phi 2^{{m}-1}}\ket{1}}{\sqrt 2}\otimes
\frac{\ket{0}+\e^{2\pi \ii \phi 2^{{m}-2}}\ket{1}}{\sqrt 2}\otimes \cdots \otimes
\frac{\ket{0}+\e^{2\pi \ii \phi 2^{0}}\ket{1}}{\sqrt 2}.
\]
This output can be simplified into a very neat expression. In order to do so, let us replace each term with an equivalent term using a binary sum and collecting the denominators
\[
\frac{1}{\sqrt{2^{m}}} \sum_{\ell_1=0}^1\e^{2\pi \ii \phi 2^{{m}-1}\ell_1}\ket{\ell_1}\otimes
\sum_{\ell_2=0}^1\e^{2\pi \ii \phi 2^{{m}-2}\ell_2}\ket{\ell_2}\otimes \cdots \otimes
\sum_{\ell_m=0}^{1}\e^{2\pi \ii \phi 2^{0}\ell_{m}}\ket{\ell_{m}}.
\]
Pushing all sums to the beginning of the expression and combining all exponentials, we obtain
\[
\frac{1}{\sqrt{2^{m}}} \sum_{\ell_1,...,\ell_{m}=0}^1\e^{2\pi \ii \phi (2^{{m}-1}\ell_1+\cdots+2^{0}\ell_{m})}\ket{\ell_1}\otimes...\otimes\ket{\ell_{m}}.
\]
Converting binary numbers into the decimal notation, we obtain
\[
\frac{1}{\sqrt{2^{m}}} \sum_{\ell=0}^{2^{m}-1}\e^{2\pi \ii \,\phi\,\ell}\ket{\ell}.
\]
This is the neat expression we were looking for. Let us summarize the first block:
\[
\ket{0}^{\otimes {m}}\otimes\ket{\psi}\xrightarrow[\text{ block }]{\text{first}}
\left(\frac{1}{\sqrt{2^{m}}} \sum_{\ell=0}^{2^{m}-1}\e^{2\pi \ii \,\phi\,\ell}\ket{\ell}\right)\otimes\ket{\psi},
\]
where $m$ is the number of qubits of the first register, $\ket{\psi}$ is an eigenvector of $U$ with eigenvalue $\exp(2\pi \ii \,\phi)$, and $\phi=0.\phi_1...\phi_m$. In the next subsection, we show that the output of the first register is
\[
F_{2^{m}}\ket{\phi_1,\dots,\phi_{m}},
\]
where $F_{2^{m}}$ is the Fourier transform, defined in Section~\ref{seq:shor_fourier_transform}. Then, we write
\[
\ket{0}^{\otimes {m}}\otimes\ket{\psi}\xrightarrow[\text{ block }]{\text{first}}
\left(F_{2^{m}}\ket{\phi_1,\dots,\phi_{m}}\right)\otimes\ket{\psi},
\]
where $0.\phi_1\cdots\phi_m$ is the phase of the eigenvalue associated with eigenvector $\ket{\psi}$ of $U$.

\subsection*{Full circuit of the quantum phase estimation (QPE)}

Now we show that the second block of the QPE algorithm is the inverse Fourier transform. In the last Subsection, we have shown that the output of the first register of the first block is
\[
\frac{1}{\sqrt{2^{m}}} \sum_{\ell=0}^{2^{m}-1}\e^{2\pi \ii \,\phi\,\ell}\ket{\ell}.
\]
On the other hand, the action of the Fourier transform $F_{2^{m}}$ on a generic state $\ket{j}$ of the computational basis is
\[
F_{2^{m}}\ket{j}\,=\,
\frac{1}{\sqrt{2^{m}}} \sum_{\ell=0}^{2^{m}-1}\e^{\frac{2\pi \ii j\ell}{2^{m}}}\ket{\ell}.
\]
Inverting the equation, taking $j=\left(\phi_1...\phi_{m}\right)_2=(2^{m}\phi)_{10}$ and $\ket{j}=\ket{\phi_1}\otimes...\otimes\ket{\phi_{m}}$, we obtain
\[
F_{2^{m}}^\dagger\left(\frac{1}{\sqrt{2^{m}}} \sum_{\ell=0}^{2^{m}-1}\e^{2\pi \ii \,\phi\,\ell}\ket{\ell}\right)\,=\,\ket{2^m\phi}\,=\,\ket{\phi_1}\otimes\cdots\otimes\ket{\phi_{m}}.
\]
If we apply the inverse Fourier transform to the output of the first block, the result is a state of the computational basis equal to $\ket{\phi_1}\otimes...\otimes\ket{\phi_{m}}$. This means that a measurement in the computational basis reveals with certainty each fractional bit of $\phi$ because we are assuming that $\phi$ is represented with $m$ bits. In the general case, the result of the algorithm is a good $m$-bit estimate of $\phi$, which is denoted by $\tilde\phi$, that is, $\tilde\phi\approx \phi2^m$.

The full circuit of the phase estimation algorithm is depicted in Fig.~\ref{fig:phaseestimation}. The algorithm is described in Algorithm~\ref{algo_phase_estimation}.

\begin{figure}[!ht]
\centering
\begin{align*}
\Qcircuit @C=1.0em @R=1.0em {
\lstick{\ket{0}}    & \gate{H}  & \qw            & \qw            & \qw & \cdots && \qw               & \ctrl{5}           & \multigate{4}{F_{2^{m}}^\dagger} & \qw    & \meter & \cw &\tilde{\phi}_1\\
\lstick{\ket{0}}    & \gate{H}  & \qw            & \qw            & \qw & \cdots && \ctrl{4}          & \qw                & \ghost{F_{2^{m}}^\dagger}	& \qw  & \meter & \cw &\tilde{\phi}_2\\
\lstick{\vdots\ \ } & \vdots    &                &                &     & \iddots&&                   &                    &\pureghost{F_{2^{m}}^\dagger}    &        & \vdots &     \\
\lstick{\ket{0}}    & \gate{H}  & \qw            & \ctrl{2}       & \qw & \cdots && \qw               & \qw                & \ghost{F_{2^{m}}^\dagger}	& \qw  & \meter & \cw &\,\,\,\,\,\,\tilde{\phi}_{{m}-1}\\
\lstick{\ket{0}}    & \gate{H}  & \ctrl{1}       & \qw            & \qw & \cdots && \qw               & \qw                & \ghost{F_{2^{m}}^\dagger}	& \qw  & \meter & \cw &\tilde{\phi}_{m}\\
\lstick{\ket{\psi}} & /^{{n}} \qw   & \gate{U^{2^0}} & \gate{U^{2^1}} & \qw & \cdots && \gate{U^{2^{{m}-2}}}& \gate{U^{2^{{m}-1}}} & \rstick{\ket{\psi}}\qw
}
\end{align*}
\caption{Full circuit of the QPE algorithm.}
\label{fig:phaseestimation}
\end{figure}

\begin{algorithm}[!ht]
\caption{Quantum phase estimation algorithm}\label{algo_phase_estimation}
\KwIn{Eigenvector $\ket{\psi}$ of $U$.}
\KwOut{Number $\tilde{\phi}$, where $\exp(2\pi\ii \phi)$ is the eigenvalue of $\ket{\psi}$ and $\tilde\phi\approx \phi2^m$.} \BlankLine
Prepare the initial state $\ket{0}^{\otimes {m}}\otimes\ket{\psi}$\vspace{3pt}\;
Apply $H^{\otimes {m}}$ to the first register\;
For $\ell$ in $[0,{m}-1]$ apply the controlled operation $C^{{m}-\ell}\left(U^{2^\ell}\right)$, where the control qubit is ${m}-\ell$ and the target is the 2nd register\vspace{2pt}\;
Apply $F_{2^{m}}^\dagger$ to the first register\vspace{3pt}\;
Measure the first register in the computational basis.
\end{algorithm}


\section{Application to order-finding}

Let $N$ and $a$ be positive integers so that $1<a<N$ and $\gcd(a,N)=1$. The multiplicative order of $a$ modulo $N$ is the smallest positive integer $r$ that obeys
\[
a^r\equiv 1\mod N.
\]
Given $a$ and $N$, order-finding is the problem of calculating $r$. In this Section, we show how to solve the order-finding problem efficiently using the phase estimation algorithm, thereby providing an alternative to Shor's factoring algorithm.

The strategy is to replace $\ket{\psi}$ in Algorithm~\ref{algo_phase_estimation} by $\ket{1}$ (the second vector of the computational basis of the second register) and to choose $U$ as the unitary operator that multiplies the input by $a$ modulo $N$, that is,
\begin{equation}\label{eq:U_order_finding_QPE}
U\ket{y}\,=\,\ket{ay\mod N},
\end{equation}
where $0\le y <N$ and $U\ket{y}=\ket{y}$ otherwise. The input is a vector $\ket{y}$ of the computational basis of the second register. We may think that $y$ is represented in the decimal system. The output is also a vector $\ket{y'}$ of the computational basis of the second register, which is obtained by calculating $ ay\equiv y'$ modulo $N$. $U$ is a unitary operator because $\gcd(a,N)=1$. $U^\dagger$ is defined accordingly using $a^{-1}$ modulo $N$, that is,
\[
U^\dagger\ket{y}\,=\,\ket{a^{-1}y\mod N}.
\]
The motivation for using $U$ here is that repeated application of $U$ produces successive powers of $a$; in fact, $U^j\ket{y}=\ket{a^jy}$. The number of qubits ${{n}}$ in the second register must be large enough to represent the states on which $U$ acts, so we take ${{n}}=\lceil \log_2 N\rceil$.

In order to understand order-finding as a phase estimation algorithm, let us find the eigenvectors of $U$. It is straightforward to obtain a 1-eigenvector because the set $\{a^0,a^1,...a^{r-1}\}$, where $r$ is the order of $a$ modulo $N$, is invariant under multiplication by $a$. Then, the normalized vector
\[
\ket{\psi_0}=\frac{1}{\sqrt r}\sum_{\ell=0}^{r-1}\ket{a^\ell}
\]
is a 1-eigenvector of $U$. The remaining ones are constructed using the entries of the Fourier transformation $F^\dagger_r$. Define
\begin{equation}\label{eq:fourierphaseest}
\ket{\psi_k}\,=\,\frac{1}{\sqrt r}\sum_{\ell=0}^{r-1}\e^{-\frac{2\pi\ii k\ell}{r}}\ket{a^\ell}.
\end{equation}
Now let us check that each $\ket{\psi_k}$ is an eigenvector of $U$. In fact,
\begin{eqnarray*}
U\ket{\psi_k} &=& \frac{1}{\sqrt r}\sum_{\ell=0}^{r-1}\e^{-\frac{2\pi\ii k\ell}{r}}\ket{a^{\ell+1}}\\
&=& \frac{1}{\sqrt r}\sum_{\ell=0}^{r-1}\e^{-\frac{2\pi\ii k(\ell-1)}{r}}\ket{a^{\ell}}\\
&=& \e^{\frac{2\pi\ii k}{r}}\ket{\psi_k}.
\end{eqnarray*}
We conclude that $\ket{\psi_k}$ is an eigenvector of $U$ with eigenvalue $\exp({2\pi\ii k}/{r})$ for $0\le k<r$, whose phase is $k/r$. If we are able to prepare the input to the second register of the phase estimation algorithm as $\ket{\psi_k}$ for some $k$, we will obtain an approximation of $k/r$ as the output, and then find a candidate for the order of $a$ using the continued fraction expansion. If the input to the second register is $\ket{\psi_k}$, the output of the first block is
\begin{equation}\label{eq:outputintblock}
\ket{0}\otimes\ket{\psi_k}\xrightarrow[\text{ block }]{\text{first}}\frac{1}{\sqrt{2^{m}}} \sum_{\ell=0}^{2^{m}-1}\e^\frac{2\pi \ii k\ell}{r}\ket{\ell}\ket{\psi_k}.
\end{equation}

We cannot prepare $\ket{\psi_k}$ as an input to the phase estimation algorithm but we can find a known vector that is spanned by the set of vectors $\{\ket{\psi_0},...\ket{\psi_{r-1}}\}$. Using Eq.~(\ref{eq:fourierphaseest}) and the geometric series, we show that
\[
\ket{a^\ell}\,=\,\frac{1}{\sqrt r}\sum_{k=0}^{r-1}\e^{\frac{2\pi\ii k\ell}{r}}\ket{\psi_k}.
\]
The simplest choice is $\ket{a^0}=\ket{1}$, which is given by
\begin{equation}\label{eq:ket-1}
\ket{1}\,=\,\frac{1}{\sqrt r}\sum_{k=0}^{r-1}\ket{\psi_k}.
\end{equation}
Using transformation~(\ref{eq:outputintblock}) for each $k$, the output of the first block is
\[
\ket{0}\otimes\ket{1}\xrightarrow[\text{ block }]{\text{first}}\frac{1}{\sqrt{r2^{m}}}\sum_{k=0}^{r-1} \sum_{\ell=0}^{2^{m}-1}\e^\frac{2\pi \ii k\ell}{r}\ket{\ell}\ket{\psi_k}.
\]
To simplify the output, we use Eq.~(\ref{eq:fourierphaseest}). Then,
\begin{eqnarray*}
\text{output} &=&\frac{1}{\sqrt{r2^{m}}}\sum_{k=0}^{r-1} \sum_{\ell=0}^{2^{m}-1}\e^\frac{2\pi \ii k\ell}{r}\ket{\ell}\left(\frac{1}{\sqrt r}\sum_{\ell'=0}^{r-1}\e^{-\frac{2\pi\ii k\ell'}{r}}\ket{a^{\ell'}}\right).
\end{eqnarray*}
Inverting the order of the sums and combining the exponents, we obtain
\begin{eqnarray*}
\text{output}
&=&\frac{1}{\sqrt{2^{m}}} \sum_{\ell=0}^{2^{m}-1}\sum_{\ell'=0}^{r-1}\left(\frac{1}{r}\sum_{k=0}^{r-1}\e^\frac{2\pi \ii k(\ell-\ell')}{r}\right)\ket{\ell}\ket{a^{\ell'}}.
\end{eqnarray*}
Using the geometric series, the expression inside the parenthesis is 1 if $\ell=\ell'$ and 0 otherwise. This means that the output of the first block is
\begin{eqnarray*}
\ket{0}^{\otimes {m}}\otimes\ket{0}^{\otimes {(n-1)}}\ket{1}\xrightarrow[\text{ block }]{\text{first}}\frac{1}{\sqrt{2^{m}}} \sum_{\ell=0}^{2^{m}-1}\ket{\ell}\ket{a^{\ell}}.
\end{eqnarray*}
This is the same state as in the standard Shor's factoring algorithm (state $\ket{\psi_2}$ of Section~\ref{sec:shor_analysis}) just before applying the inverse Fourier transform $F_{2^{m}}$, that is, we consider the part of the algorithm where the input is $\ket{0}^{\otimes {m}}\ket{0}$ and then the Hadamard gate is applied on each qubit of the first register and then $U_f$, where $f(x)=a^x\mod N$:
\begin{eqnarray*}
\ket{0}^{\otimes {m}}\otimes\ket{0}^{\otimes n}\xrightarrow[\hspace{1cm}]{\,\,U_f\cdot (H^{\otimes m}\otimes I)\,\,}\frac{1}{\sqrt{2^{m}}} \sum_{\ell=0}^{2^{m}-1}\ket{\ell}\ket{a^{\ell}}.
\end{eqnarray*}
This means that the phase estimation version yields the same result and the analysis of the success probability is exactly the same as in Shor's factoring algorithm if we choose $m$ so that $m=\lceil 2\log_2 N\rceil$. The number of qubits of the first register must be close to twice the number of qubits of the second register. Fig.~\ref{fig:Shor_QPE} depicts the quantum  part of Shor's algorithm using the quantum phase estimation, where $U$ is given by Eq.~(\ref{eq:U_order_finding_QPE}) and $a$ is picked uniformly at random in $\mathbb{Z}_N^\times$. Note that the input $\ket{1}$ to the second register is the second vector of the computational basis. This circuit replaces Algorithm~\ref{algo_ShorQuantum} of Section~\ref{sec_subsection_shor_algorithm}.
\begin{figure}[!ht]
\centering
\[
\Qcircuit @C=1.0em @R=1.0em {
\lstick{\ket{0}}    & \gate{H}  & \qw            & \qw            & \qw & \cdots &  & \ctrl{4}           & \multigate{3}{F_{2^{m}}^\dagger} & \qw    & \meter & \cw &\ell_0\\
\lstick{\vdots\ \ } & \vdots    &                &                &     &   \iddots &    &                    &\pureghost{F_{2^{m}}^\dagger}    &        & \vdots &     \\
\lstick{\ket{0}}    & \gate{H}  & \qw            & \ctrl{2}       & \qw & \cdots & & \qw                & \ghost{F_{2^{m}}^\dagger}	& \qw  & \meter & \cw &\,\,\,\,\,\,\ell_{{m}-2}\\
\lstick{\ket{0}}    & \gate{H}  & \ctrl{1}       & \qw            & \qw & \cdots & & \qw                & \ghost{F_{2^{m}}^\dagger}	& \qw  & \meter & \cw & \,\,\,\,\,\,\ell_{m-1}\\
\lstick{\ket{0...01}} & /^{{n}} \qw   & \gate{U^{2^0}} & \gate{U^{2^1}} & \qw & \cdots && \gate{U^{2^{{m}-1}}} & \rstick{}\qw
}
\]
\caption{Quantum part of Shor's algorithm based on QPE. $U$ is given by Eq.~(\ref{eq:U_order_finding_QPE}). The output is the same as the one described in Algorithm~\ref{algo_ShorQuantum} of Section~\ref{sec_subsection_shor_algorithm}.}
\label{fig:Shor_QPE}
\end{figure}

There is an interesting special case. If we somehow know that the order $r$ is a power of 2, we can take $m=n=\lceil \log_2 N\rceil$. If $r$ is a power of 2, the phase of the eigenvalue $\exp({2\pi\ii k}/{r})$ is a rational multiple of $2\pi$, and the phase estimation algorithm returns an exact value $k2^m/r$. In this case, we do not need to calculate the continued fraction expansion of the result. Instead, we simply divide the output by $2^m$ and select the denominator as the candidate for the order of $a$.

It is simpler to check that Shor's factoring algorithm works correctly when we use Kitaev's version. Let's suppose that the order $r$ is a power of 2. After compacting the circuit of Fig.~\ref{fig:Shor_QPE}, the output of the QPE algorithm when the input is $\ket{0}\ket{1}$ (in the decimal notation) is shown below
\[
\Qcircuit @C=1.7em @R=0.9em {
\lstick{\ket{0}}   &{/}^m \qw &\multigate{2}{\,\,\text{QPE}_U\,\,} &  \meter& \rstick{\tilde{\phi}_k=\frac{k2^m}{r}}\cw \\
\lstick{}          &          &                                    &     & \lstick{}  \\
\lstick{\ket{1}}&{/}^n \qw &\ghost{\,\,\text{QPE}_U\,\,}        & \qw& \lstick{\ket{\psi_k},} \\
& & & & {\hspace{1.7cm}}
}\vspace{0pt}
\]
where $0\le k<r$ is selected uniformly at random. The result of the last circuit is obtained by using Eq.~(\ref{eq:ket-1}) and the fact that the phase of the eigenvalue associated with $\ket{\psi_k}$ is $\phi_k=k/r$. After applying the QPE algorithm, we obtain
$$\ket{0}\ket{1} \xrightarrow[\text{ }]{\,\,\,\,\text{QPE}_U\,\,\,\,} \frac{1}{\sqrt r} \sum_{k=0}^{r-1}  |{\tilde{\phi}_k}\rangle\ket{\psi_k},$$
where $\tilde{\phi}_k=\frac{k2^m}{r}$. Then, after a measurement of the first register, the output is necessarily $\tilde{\phi}_k$ for some $0\le k<r$ picked uniformly at random, and the outcome of the second register is necessarily the eigenvector $\ket{\psi_k}$. If the order $r$ is not a power of 2, the output $\tilde{\phi}_k$ is a good approximation for $\frac{k2^m}{r}$, and it is likely that ${\tilde{\phi}_k}$ be a nearest integer to a multiple of $k/r$.  A candidate $r'$ for the multiplicative order of $a$ modulo $N$ is obtained by selecting the convergent of the continued fraction expansion of $\tilde{\phi}_k/2^m$ that has the largest denominator $r'$ such that $r'<N$. The lower bound for the success probability of the quantum part determined in Section~\ref{sec:shor_analysis} is valid here.

How do we implement $U^{2^j}$ efficiently for $0\le j<m$? Ref.~\cite{MS12} addresses this question. Note that
\[
U^{2^j}\ket{y}\,=\,\ket{a^{2^j}y\mod N}.
\]
Since $a^{2^j}$ can be calculated efficiently in $O(n^2)$ steps using the repeated squaring method, instead of applying $U$ repeatedly $2^j$ times, for each $j$ we implement an operator $U_j\ket{y}=\ket{zy}$ after calculating $z=a^{2^j}$ using the repeated squaring method. In this case, the first block can be computed in $O(n^3)$ steps.

\section{Application to discrete logarithm}

Let $N$, $a$, and $b$ be known positive integers and let $s$ be a positive integer such that $a^s\equiv b \mod N$ and gcd$(a,N)=1$. Our goal is to find $s$ given $N$, $a$, and $b$ as input. This is the same problem addressed in Section~\ref{sec:dislog_special_case} on Page~\pageref{sec:dislog_special_case}. Now we show how to solve the discrete logarithm problem using the phase estimation algorithm, thereby providing an alternative version of Shor's algorithm for the discrete logarithm problem.


The strategy we use here is the same one used in the order-finding algorithm based on phase estimation. Recall that, when we described the order-finding algorithm, the output of the first block is the same as in the original Shor's factoring algorithm right before the action of the inverse Fourier transform. Now, the state of the qubits right before the action of $F_r^\dagger\otimes F_r^\dagger$ in Shor's discrete logarithm algorithm described in Section~\ref{sec:dislog_special_case} is
\[
\frac{1}{r}\sum_{x,y=0}^{r-1}\ket{x}\ket{y}\ket{a^xb^y \mod N}.
\]
If we wish to produce this state using the phase estimation algorithm, we need to use three registers and two unitary operators:
\begin{eqnarray*}
U_a\ket{x}&=&\ket{ax \mod N},\\
U_b\ket{y}&=&\ket{by \mod N}.
\end{eqnarray*}
$U_b$ is a unitary operator because $\gcd(b,N)=1$. Indeed, the inverse of $b$ is $a^{r-s}$, where $r$ is the order of $a$ modulo $N$. In the new algorithm, the action of $U_a$ is controlled by the first register, and the action of $U_b$ is controlled by the second register, as described in Fig.~\ref{fig:logdis_phes}. Note that $U_a$ and $U_b$ act on the same register.

\begin{figure}[!ht]
\centering
\begin{align*}
\Qcircuit @C=1.0em @R=0.5em {
\lstick{\ket{0}}    & \gate{H}  &\qw      &  \qw  &\ctrl{6}&\qw     &\qw    & \qw    &\qw& \multigate{2}{F_{2^{m}}^\dagger} & \qw    & \meter & \cw &\tilde \phi_1\\
\lstick{\vdots\ \ } & \vdots    &         &\iddots&        &        &       &        &   & \pureghost{F_{2^{m}}^\dagger}    &        & \vdots &     \\
\lstick{\ket{0}}    & \gate{H}  & \ctrl{4}&  \qw  &\qw     &\qw     &\qw    & \qw    &\qw& \ghost{F_{2^{m}}^\dagger}	& \qw  & \meter & \cw &\tilde \phi_m\\
\lstick{\ket{0}}    & \gate{H}  & \qw     &  \qw  &\qw     &\qw     &\qw    &\ctrl{3}&\qw& \multigate{2}{F_{2^{m}}^\dagger}& \qw &\meter &\cw&\tilde \phi'_1\\
\lstick{\vdots\ \ } & \vdots    &         &       &        &        &\iddots&        &   & \pureghost{F_{2^{m}}^\dagger}    &        & \vdots &     \\
\lstick{\ket{0}}    & \gate{H}  & \qw     &  \qw  &\qw     &\ctrl{1}&\qw    & \qw    &\qw& \ghost{F_{2^{m}}^\dagger}	& \qw  & \meter & \cw &\tilde \phi'_m\\
\lstick{\ket{\psi_k}} & /^{{n}} \qw &\gate{U_a^{2^{0}}}&\qw& \gate{U_a^{2^{{m}-1}}}  & \gate{U_b^{2^0}}& \qw  & \gate{U_b^{2^{m-1}}} & \rstick{\ket{\psi_k}}\qw
}
\end{align*}
\caption{Circuit of the discrete logarithm algorithm based on QPE, where $\ket{\psi_k}$ is a common eigenvector of $U_a$ and $U_b$. Since we are not usually able to prepare $\ket{\psi_k}$ as the input to the third register, we use $\ket{1}$ in its place, which can be represented as a linear combination of $\ket{\psi_k}$, for $0\le k<r$.}
\label{fig:logdis_phes}
\end{figure}

Let us check that the vectors
\begin{equation*}
\ket{\psi_k}\,=\,\frac{1}{\sqrt r}\sum_{\ell=0}^{r-1}\e^{-\frac{2\pi\ii k\ell}{r}}\ket{a^\ell},
\end{equation*}
which are eigenvectors of $U_a$, are also eigenvectors of $U_b$. In fact, using that $b=a^s$, we have
\begin{eqnarray*}
U_b\ket{\psi_k} &=& \frac{1}{\sqrt r}\sum_{\ell=0}^{r-1}\e^{-\frac{2\pi\ii k\ell}{r}}\ket{a^{\ell+s}}\\
&=& \frac{1}{\sqrt r}\sum_{\ell=0}^{r-1}\e^{-\frac{2\pi\ii k(\ell-s)}{r}}\ket{a^{\ell}}\\
&=& \e^{\frac{2\pi\ii ks}{r}}\ket{\psi_k}.
\end{eqnarray*}
We conclude that $\ket{\psi_k}$ is an eigenvector of $U_b$ with eigenvalue $\exp({2\pi\ii ks}/{r})$ for $0\le k<r$. Better yet, $\ket{\psi_k}$ is an eigenvector of $U_a$ and $U_b$ simultaneously.  If we are able to prepare the input to the third register of the circuit depicted in Fig.~\ref{fig:logdis_phes} as $\ket{\psi_k}$ for some $k$, we will obtain an estimate of $\tilde \phi\approx k/r$ as the output of the first register and an estimate of $\tilde \phi'\approx ks/r$ as the output of the second register modulo $N$. If the input to the third register is $\ket{\psi_k}$, the output of the first block is
\begin{equation*}
\ket{0}^{\otimes {m}}\ket{0}^{\otimes {m}}\ket{\psi_k}\xrightarrow[\text{ block }]{\text{first}}
\left(\frac{1}{\sqrt{2^{m}}} \sum_{\ell=0}^{2^{m}-1}\e^\frac{2\pi \ii k\ell}{r}\ket{\ell}\right)
\left(\frac{1}{\sqrt{2^{m}}} \sum_{\ell'=0}^{2^{m}-1}\e^\frac{2\pi \ii k s\ell' }{r}\ket{\ell'}\right)
\ket{\psi_k},
\end{equation*}
where $m$ is the number of qubits of the first and second registers, and ${{n}}=\lceil \log_2 N\rceil$ of the third register.  The first term inside the parentheses is obtained by replacing $\phi$ with $k/r$ and the second term by replacing $\phi$ with $ks/r$ in the output of the first block of the phase estimation algorithm. Simplifying the output, we have
\begin{equation*}
\ket{0}^{\otimes {m}}\ket{0}^{\otimes {m}}\ket{\psi_k}\xrightarrow[\text{ block }]{\text{first}}
\left(\frac{1}{2^{m}} \sum_{\ell,\ell'=0}^{2^{m}-1}\e^\frac{2\pi \ii k(\ell+s\ell')}{r}\ket{\ell}\ket{\ell'}\right)
\ket{\psi_k}.
\end{equation*}

Usually, we are not able to prepare $\ket{\psi_k}$ as an input to the phase estimation algorithm but we can use
\[
\ket{1}\,=\,\frac{1}{\sqrt r}\sum_{k=0}^{r-1}\ket{\psi_k}
\]
again instead of $\ket{\psi_k}$.
In this case, the input to and output of the first block are
\begin{equation*}
\ket{0}^{\otimes {m}}\ket{0}^{\otimes {m}}\ket{1}\xrightarrow[\text{ block }]{\text{first}}\frac{1}{\sqrt r}\sum_{k=0}^{r-1}
\left(\frac{1}{2^{m}} \sum_{\ell,\ell'=0}^{2^{m}-1}\e^\frac{2\pi \ii k(\ell+s\ell')}{r}\ket{\ell}\ket{\ell'}\right)
\ket{\psi_k},
\end{equation*}
where the input $\ket{1}$ is written in the decimal notation. We use the definition of $\ket{\psi_k}$ given by Eq.~(\ref{eq:fourierphaseest}) to simplify the output. We start writing
\begin{eqnarray*}
\frac{1}{\sqrt r}\sum_{k=0}^{r-1}
\left(\frac{1}{2^{m}} \sum_{\ell,\ell'=0}^{2^{m}-1}\e^\frac{2\pi \ii k(\ell+s\ell')}{r}\ket{\ell}\ket{\ell'}\right)
\left(\frac{1}{\sqrt r}\sum_{{k'}=0}^{r-1}\e^{-\frac{2\pi\ii k{k'}}{r}}\ket{a^{{k'}}}\right),
\end{eqnarray*}
and by pushing the sums over $\ell$, $\ell'$, ${k'}$ to the left and combining the exponents, we obtain
\begin{eqnarray*}
\frac{1}{2^{m}} \sum_{\ell,\ell'=0}^{2^{m}-1}\sum_{{k'}=0}^{r-1}\left(\frac{1}{r}\sum_{k=0}^{r-1}\e^\frac{2\pi \ii k(\ell+ s\ell' -{k'})}{r}\right)\ket{\ell}\ket{\ell'}\ket{a^{{k'}}}.
\end{eqnarray*}
The expression inside the parenthesis is 1 if $\ell+ s\ell' ={k'}$ and 0 otherwise. This means that the action of the first block is
\begin{eqnarray*}
\ket{0}^{\otimes {m}}\ket{0}^{\otimes {m}}\ket{1}\xrightarrow[\text{ block }]{\text{first}}\frac{1}{2^{m}} \sum_{\ell,\ell'=0}^{2^{m}-1}\ket{\ell}\ket{\ell'}\ket{a^{\ell+ s\ell' }}.
\end{eqnarray*}
Using that $a^s=b\mod N$, $\ell\rightarrow x$, and $\ell'\rightarrow y$, we have
\begin{eqnarray*}
\ket{0}^{\otimes {m}}\ket{0}^{\otimes {m}}\ket{1}\xrightarrow[\text{ block }]{\text{first}}\frac{1}{2^{m}} \sum_{x,y=0}^{2^{m}-1}\ket{x}\ket{y}\ket{a^{x}b^{y}}.
\end{eqnarray*}
The output of the first block is the same state as in the standard Shor's discrete-logarithm algorithm just before applying the inverse Fourier transforms $F_{2^{m}}\otimes F_{2^{m}}$ (state $\ket{\psi_2}$ of Algorithm~\ref{algo_discrete_log} on Page~\pageref{algo_discrete_log}). This means that the phase estimation version yields the same result and the analysis of the success probability is exactly the same as in Shor's algorithm if we choose $m$ appropriately.

\section{Application to quantum counting}

In the context of Grover's algorithm, we have an oracle $f:\{0,...,N-1\}\rightarrow \{0,1\}$ that is a Boolean function defined as
\begin{equation*}\label{eq:grover_f_x}
    f(x) = \left\{
  \begin{array}{l@{\quad}l}
    1, & \hbox{if \ensuremath{x\in M},} \\
    0, & \hbox{otherwise,}
  \end{array}
\right.
\end{equation*}
where $M$ is a subset of the domain and $N=2^n$. We say that $x$ is marked if $x\in M$. The optimal number of steps of Grover's algorithm depends on $|M|$; indeed, it is given by  $\frac{\pi}{4}\sqrt{N/|M|}$. If the cardinality of $M$ is unknown, it is possible to find a marked element by repeatedly guessing the runtime of Grover's algorithm~\cite{BBHT98}. An alternative method is by solving the quantum counting problem~\cite{BHMT02}.

The quantum counting problem asks what is the cardinality of $M$ given function $f$ as an oracle. A classical solution cannot perform better than $\Omega(N)$ queries to the oracle because all domain elements must be checked. The quantum algorithm can find the solution in $O\big(\sqrt{|M|\,N}\big)$ queries to the oracle.

Before addressing the quantum counting problem, let us review some key points of Grover's algorithm with many marked elements, which is an extension of the algorithm presented in Chapter~\ref{chap:Grover}. There is an economical version of the algorithm, which uses only one $n$-qubit register. The initial state is the uniform superposition of the computational basis given by
\[
\ket{\text{d}}\,=\,\frac{1}{\sqrt N}\sum_{j=0}^{N-1}\ket{j},
\]
and the algorithm consists of $\frac{\pi}{4}\sqrt{N/|M|}$ applications of the evolution operator
\[
U\,=\,G\,U_f,
\]
where
\[
G\,=\,2\ket{\text{d}}\bra{\text{d}}-I
\]
and
\[
U_f\,=\,\sum_{x=0}^{N-1}(-1)^{f(x)}\ket{x}\bra{x}.
\]
The analysis of the algorithm is performed by using the  $\e^{\pm{\ii \theta}}$-eigenvectors of $U$, which are given in terms of the superposition of marked states\footnote{Some references call $\ket{M}$ as ``good state'' and $\ket{M^\perp}$ as ``bad state''.} $\ket{M}$ and the superposition of unmarked states $\ket{M^\perp}$:
\begin{equation}\label{eq:eigenvecsphase}
\ket{\psi^\pm}\,=\,\frac{\ket{M}\pm \,\ii \,\ket{M^\perp}}{\sqrt 2},
\end{equation}
where
\[
\sin\frac{\theta}{2}\,=\,\sqrt{\frac{|M|}{N}},
\]
and
\begin{eqnarray*}
\ket{M}&=&\frac{1}{\sqrt{|M|}}\sum_{x\in M}\ket{x},\\
\ket{M^\perp}&=&\frac{1}{\sqrt{N-|M|}}\sum_{x\not\in M}\ket{x}.
\end{eqnarray*}
It is straightforward to check that $\braket{M^\perp}{M}=0$, and
\[
\ket{\text{d}}\,=\,\sqrt{\frac{|M|}{N}}\,\ket{M}+\sqrt{1-\frac{|M|}{N}}\,\,\ket{M^\perp}.
\]
Using the equation above, the definition of $\sin(\theta/2)$, and Eq.~(\ref{eq:eigenvecsphase}), we obtain
\[
\ket{\text{d}}\,=\,\frac{\e^{\frac{\ii \theta}{2}}\ket{\psi^+}-\e^{-\frac{\ii \theta}{2}}\ket{\psi^-}}{\ii \sqrt 2}.
\]

Now we come back to the quantum counting problem using the phase estimation algorithm. Since the eigenvalue of $\ket{\psi^+}$ is $\exp(\ii \theta)$, where $\sin(\theta/2)=\sqrt{|M|/N}$, we would obtain an approximation for $|M|$ if we use $\ket{\psi^+}$ as the input to the second register of the phase estimation algorithm with $U=GU_f$. If we do not know how to prepare $\ket{\psi^+}$, then the strategy is to replace $\ket{\psi}$ in Algorithm~\ref{algo_phase_estimation} by a known vector that belongs to the subspace spanned by $\ket{\psi^+}$ and $\ket{\psi^-}$. The best candidate is $\ket{\text{d}}$, which can be easily prepared by applying $H^{\otimes n}$ to $\ket{0}^{\otimes n}$. In this case, the number of qubits of the second register must be ${{n}}=\log_2 N$ and the output of the first block is
\[
\ket{0}^{\otimes {m}}\ket{\text{d}}\xrightarrow[\text{ block }]{\text{first}}\frac{\e^{\frac{\ii \theta}{2}}}{\ii \sqrt{2^{m+1}}}{\sum_{\ell=0}^{2^{m}-1}\e^{{2\pi\ii \phi^+\ell}}\ket{\ell}\ket{\psi^+}-\frac{\e^{-\frac{\ii \theta}{2}}}{\ii \sqrt{2^{m+1}}}\sum_{\ell=0}^{2^{m}-1}\e^{{2\pi\ii \phi^-\ell}}\ket{\ell}\ket{\psi^-}},
\]
where $\phi^+=\theta/2\pi$ for the first term and  $\phi^-=(2\pi-\theta)/2\pi$ for the second term. After applying the inverse Fourier transform $F_{2^m}^\dagger$, we obtain the following output of the full circuit
\[
\frac{\e^{\frac{\ii \theta}{2}}}{\ii \sqrt 2}\,\ket{2^m\tilde\phi^+}-
\frac{\e^{-\frac{\ii \theta}{2}}}{\ii \sqrt 2}\,\ket{2^m\tilde\phi^-},
\]
where $\tilde\phi$ is an $m$-bit estimate of $\phi$. After a measurement in the computational basis, we learn an estimate of $\phi^+$ or $\phi^-$ with equal probability. Let $\tilde\phi$ be the measurement result. Using that $\sin (\theta/2)=\sqrt{|M|/N}$, $\phi^+=\theta/2\pi$, and $\phi^-=(2\pi-\theta)/2\pi$, the estimate of $|M|$ is $N\sin^2(\pi\tilde\phi)$ because in the first case we obtain an estimate of $|M|$ using $|\widetilde M|=N\sin^2(\pi\tilde\phi^+)$, and in the second case $|\widetilde M|=N\sin^2(\pi-\pi\tilde\phi^-)=N\sin^2(\pi\tilde\phi^-)$.

How many qubits does the first register have? This is the tricky part. Note that $m$ cannot be equal to $n$ because the number of applications of $U$ would be $2^0+\cdots +2^{n-1}=2^n-1$. Then, the number of queries to $f$ would be $O(N)$. If we choose $m=n/2$, the number of queries to $f$ would be $\sqrt{N}$, but in this case we obtain an estimate $|\widetilde M|$ such that
\begin{equation}\label{chap10_eq_tildeM_M}
\big|\,|\widetilde M|-|M|\,\big|\,=\,O\big(\sqrt{|M|}\big).
\end{equation}
This estimate is not good. For instance, suppose that $|M|$ is around $N/2$. If we wish to know the number of marked elements, and we obtain $|\widetilde M|$ with an error as big as $O(\sqrt{N})$, we don't have a good result. To understand what is the problem here, which does not arise in the factoring and discrete logarithm algorithms, we have to analyze carefully the range of values of angle $\theta$ we are trying to estimate.

For this analysis, let us assume that $0<|M|\ll N$, or more formally, $|M|=o(\sqrt{N})$. The expression $\sin (\theta/2)=\sqrt{|M|/N}$ can be written asymptotically as
\[
\theta = \frac{2\sqrt{|M|}}{\sqrt{N}}+O\left(\frac{|M|}{N}\right).
\]
This means that $\theta/2\pi$ represented in terms of binary digits is of the form $0.0\cdots 01\cdots$, where the number of 0's before the first 1 is around $n/2-\log_2(|M|)/2$. If we choose the size of the first register so that $m$ is less than $n/2-\log_2(|M|)/2$, it is likely that we obtain a 0 as the output of the counting algorithm, which is wrong. If we choose $m=n/2$, we will obtain around $\log_2(|M|)/2$ correct significant bits of $\theta/2\pi$. This is an imprecise estimate of $|M|$ compatible with Eq.~(\ref{chap10_eq_tildeM_M}). In fact, $|M|$ has $\log_2(|M|)$ bits, and we need to know most of the significant bits in order to have a good estimate of $|M|$.

\chapter{HHL Algorithm}\label{chap:HHL}

The HHL algorithm finds an approximate solution to systems of linear equations by leveraging quantum phase estimation~(QPE). In specific cases it offers an exponential speedup over classical methods, although in many situations classical algorithms remain competitive. Proposed by Harrow, Hassidim, and Lloyd in 2009~\cite{HHL09}, the algorithm encodes the solution vector in a quantum state, making it particularly useful for applications in machine learning and optimization~\cite{SP21}. When the system is expressed in matrix form, the algorithm assumes that the matrix is well-conditioned, sparse, and admits efficient Hamiltonian simulation, and that the input vector can be efficiently prepared. The output of the algorithm is not the explicit solution in the classical sense; rather, it provides access to a quantum-encoded representation of the solution. As a result, the HHL algorithm is most useful in scenarios where one wishes to extract global properties of the solution, such as expectation values or inner products. It can also be used as a subroutine in larger quantum algorithms.

Applications of the HHL algorithm include quantum versions of support vector machines and other machine learning methods~\cite{WBL12,Eza22,PK24}. Following the publication of the algorithm, several experimental implementations were reported using different quantum computing platforms~\cite{Cai13,Bar14,Pan14,ZSCX17,LJL19} and some extensions to the algorithm~\cite{Amb10,WZP18}. A key ingredient in the implementation of HHL is Hamiltonian simulation, which has been extensively studied and significantly improved over the years~\cite{Lloy96,AT03,BACS07,BCCKS15,Low19,ALL23}. Knowledge of how Hamiltonians appear in quantum mechanics is important for this topic. More recently, alternative quantum algorithms for solving systems of linear equations have been proposed that are better suited for noisy intermediate-scale quantum (NISQ) devices~\cite{BLCSCS23,Cer21}. Detailed descriptions and analyses of the HHL algorithm can be found in several research papers and textbooks~\cite{Lin22,LR14,Her25}.

\section{Problem formulation}

\subsection*{The classical problem}

A system of $N$ linear equations with $N$ variables $x_1$, ..., $x_N$ (the unknowns) can be written in the standard form
$$A\vec{x}\,=\,\vec{b},$$
where $A$ is an $N \times N$ matrix with entries in $\mathbb{C}$, representing the coefficients of the system, and the entries of the column vector $\vec{b}$ are known constants in $\mathbb{C}$. The problem is to find $\vec{x}$ given $A$ and $\vec{b}$, assuming that $A$ is non-singular (i.e., it has no zero eigenvalues). In this case, the solution is
$$\vec{x}=A^{-1}\,\vec{b}.$$

An important parameter in the analysis of algorithms for solving linear systems is the \emph{condition number} of the matrix $A$, denoted by $\kappa$, which measures how sensitive the solution is to perturbations in the input data. It is defined as
\[
\kappa = \frac{|\lambda_{\max}|}{|\lambda_{\min}|},
\]
where $\lambda_{\max}$ and $\lambda_{\min}$ are the eigenvalues of $A$ with the largest and smallest magnitudes, respectively. A well-conditioned matrix has $\kappa \approx 1$, meaning that the eigenvalues have similar magnitudes and the system is numerically stable. If $\kappa \gg 1$, the system is said to be ill-conditioned, meaning that the solution may be highly sensitive to small perturbations in the data.

On a classical computer, a standard method for solving such systems is Gauss--Jordan elimination, which requires $O(N^3)$ steps. More efficient algorithms have been developed that reduce the asymptotic complexity of matrix inversion. In particular, using fast matrix multiplication techniques, the solution can be computed in $O(N^{2.373})$ steps.\footnote{\url{https://en.wikipedia.org/wiki/Computational_complexity_of_mathematical_operations}} However, these algorithms are mainly of theoretical interest and are generally not practical for most applications.

In many situations the matrix $A$ has additional structure that can be exploited algorithmically. An important example occurs when the matrix is \emph{sparse}. A matrix is said to be $s$-sparse if each row (and column) contains at most $s$ non-zero elements, where typically $s \ll N$. In such cases, iterative methods can solve the system more efficiently than direct matrix inversion. Examples include the Conjugate Gradient method and other Krylov subspace algorithms~\cite{GV13}. For these methods, the computational complexity typically is $O(N s \kappa \log(1/\epsilon))$, where $s$ is the sparsity parameter of the matrix, $\kappa$ is the condition number of $A$, and $\epsilon$ is the desired precision of the solution. These algorithms avoid explicitly computing $A^{-1}$ and instead iteratively approximate the solution vector.

\subsection*{The quantum problem}\label{sec:qprob}

The quantum version of the problem is slightly different. We assume that $N=2^n$ and that we have $n$ error-free qubits that can encode the vector $\vec b$ into a quantum state $\ket{b}$ (after normalization if necessary). Suppose $A$ is an $N$-dimensional Hermitian matrix, which allows us to simulate the unitary evolution $\e^{-\ii A t}$ for suitable values of $t$. Assume the unknown vector $\ket{x}$ satisfies
$$
A\ket{x}=\ket{b}.
$$
The goal is to prepare a quantum state proportional to $\ket{x}$, where
$$
\ket{x}=A^{-1}\ket{b}.
$$

There is an important issue to address. Even if $\ket{b}$ has norm $1$, the vector $\ket{x}$ will generally not have norm $1$ because $A$ is not unitary unless $A^2 = I$. Therefore $\ket{x}$ cannot be the direct output of a quantum circuit. The best we can achieve is a normalized version of $\ket{x}$, namely
\[
\frac{\ket{x}}{\big\|\,\ket{x}\,\big\|}.
\]
The goal of the algorithm is not to output the vector $\vec x$ explicitly, but rather to prepare a quantum state proportional to the solution of the linear system. To remain close to the classical formulation, if we are given a matrix $A$ and a vector $\vec{b}$, the first step is to check whether $\vec{b}$ has norm $1$. If not, we encode
\[
\ket{b} = \frac{\vec{b}}{\|\vec{b}\|}.
\]

The second step is to check whether $A$ is Hermitian. The algorithm assumes that $A$ is Hermitian so that it can be used as a Hamiltonian in the phase estimation procedure. If $A$ is non-Hermitian, we can instead use the extended Hermitian matrix
\[
\left[\begin{array}{cc}
 0 & A \\
 A^\dagger & 0
\end{array}\right],
\]
which requires only one additional qubit. The inverse of this matrix can be written as
\[
\ket{0}\bra{1}\otimes (A^\dagger)^{-1}+\ket{1}\bra{0}\otimes A^{-1}.
\]
Using this embedding, the linear system $A \vec x = \vec b$ can be rewritten as an equivalent Hermitian system. In this case we take $\ket{0}\ket{b}$ as input, and the goal is to produce a normalized state proportional to $\ket{1}A^{-1}\ket{b}$.

It is important to note that the entries of $\ket{x}$ are not obtained directly. Instead, the output of the circuit is, at best, a normalized state proportional to $\ket{x}$. Furthermore, since the HHL algorithm relies on the QPE algorithm, certain approximations are introduced during the process, resulting in a state that approximates the normalized $\ket{x}$. The authors of the original paper~\cite{HHL09} propose using this approach to compute the expectation value of an observable $M$, so that the quantum computer outputs $\bra{x} M \ket{x}$. At the beginning of this Chapter, we provided references to several other applications of the HHL algorithm.

It is not practical to use the HHL algorithm to determine each entry of $\ket{x}$ in order to solve the linear system, since this would require an expensive quantum computer and at least as many steps as classical algorithms. The entries of a quantum state cannot be obtained from a single copy. Recovering all entries of $\ket{x}$ classically would require quantum state tomography, which generally needs many copies of the state and removes the potential speedup.

In the original HHL algorithm, the runtime is $O\!\left({\kappa^{2}s^{2}\,\text{polylog}(N)/\epsilon}\right)$, where $s$ is the sparsity, $\kappa$ is the condition number, and $\epsilon$ is the desired precision. Consequently, if $\kappa$ is large, the potential quantum speedup may be significantly reduced because higher precision is required in the phase estimation step.

\section{Review of QPE algorithm}\label{sec:reviewQPE}

QPE is the quantum phase estimation algorithm, whose details are described in Chapter~\ref{chap:QPE}. Suppose we have a ${{n}}$-qubit unitary operator $U$ and we know one of its eigenvectors $\ket{\psi}$. We do not know the eigenvalue associated with $\ket{\psi}$, but we know that its analytical expression is  $\e^{2\pi \ii\phi}$, where $0\le \phi<1$ ($\phi$ is unknown), because $U$ is unitary. We assume for now that $\phi=0.\phi_1\cdots\phi_{m}$ for some integer ${m}$, where $\phi_1$, ..., $\phi_{m}$ are bits, that is, the phase $\phi$ of the eigenvalue $\e^{2\pi \ii \phi}$ is a rational number. The goal of the phase estimation algorithm is to determine $\phi$ using $U$ as an oracle and $\ket{\psi}$ as an input.

In practical applications of QPE, we are not able to input $\ket{\psi}$ directly. Instead, we input a known state that is easy to prepare and analyze the algorithm using the eigenvectors of $U$. We assume that $\{\ket{u_j}: j=0,\ldots,2^n-1\}$ is an orthonormal eigenbasis of $U$, so that
\[
U\ket{u_j}=\e^{2\pi \ii\phi_j}\ket{u_j},
\]
where $0\le \phi_j<1$ is the phase of the eigenvalue $\e^{2\pi \ii\phi_j}$ associated with the eigenvector $\ket{u_j}$. To simplify the presentation, we assume that the phases of all eigenvalues are rational numbers that can be represented with at most $m$ bits. This is not a severe restriction, since the rational numbers form a dense subset of $\mathbb{R}$, and any real number can be obtained as the limit of a sequence of rational numbers~\cite{Rud76}. Fig.~\ref{fig:interblockpe2} summarizes the algorithm by showing the output when the input to the second register is an eigenvector of $U$. The output of the first register is $\ket{2^m \phi_j}$, denoted by $\ket{\tilde{\phi}_j}$, where $\tilde{\phi}_j$ is an integer. This occurs because multiplying a binary fraction with $m$ bits by $2^m$ shifts the binary point $m$ places to the right. Hence, the output belongs to the $m$-qubit computational basis.


\begin{figure}[h!]
\[
\Qcircuit @C=1.7em @R=0.9em {
\lstick{\ket{0}}   &{/}^m \qw &\multigate{2}{\,\,\text{QPE}_U\,\,} & \qw & \rstick{\hspace{-21pt}|{\tilde\phi_j}\rangle} \\
\lstick{}          &          &                                    &     & \lstick{}  \\
\lstick{\ket{u_j}}&{/}^n \qw &\ghost{\,\,\text{QPE}_U\,\,}        & \qw& \rstick{\hspace{-20pt}\ket{u_j}} \\
& & & & {\hspace{1.7cm}}
}\vspace{-10pt}
\]
\caption{Schematic representation of the quantum phase estimation (QPE) algorithm, where $U\ket{u_j} = \exp(2\pi \ii \phi_j)\ket{u_j}$ and $\tilde{\phi}_j = 2^m \phi_j$. When $\phi_j$ is not a rational number, $\tilde{\phi}_j$ (with $m$ bits) serves as a good approximation of the integer part of $2^m \phi_j$.}
\label{fig:interblockpe2}
\end{figure}

Suppose that the input to the second register is the state $\ket{b}$, as described in Section~\ref{sec:qprob}. Let $b_j$ denote the coefficients of $\ket{b}$ in the eigenbasis of $U$. Then,
\[
\ket{b}\,=\,\sum_{j=0}^{2^n-1} b_j \ket{u_j}.
\]
Using linearity, the output of QPE$_U$ in this case is
\[
\ket{0}\ket{b} \xrightarrow[\text{ }]{\,\,\,\,\text{QPE}_U\,\,\,\,} \sum_{j=0}^{2^n-1} b_j\, |{\tilde{\phi}_j}\rangle\ket{u_j}.
\]
A measurement of the first register at this point would return $\tilde{\phi}_j$ with probability $|b_j|^2$. If the phase of the $j$-th eigenvalue is not a rational number, QPE returns a $m$-bit estimate of this phase. As a side effect of this process, the state of the second register becomes $\ket{u_j}$, which means that we are able to prepare the $j$-th eigenvector state of $U$ with probability $|b_j|^2$ in the second register.

In the HHL algorithm, there is no measurement at this stage, and QPE is used in the following way. Since $A$ is Hermitian, we take
\begin{equation}\label{eq:iAt-U}
U=\e^{-\ii A t}.
\end{equation}
Since $A$ and $U$ commute and are normal operators\footnote{An operator $M$ is called normal if it commutes with its adjoint, that is, $MM^\dagger = M^\dagger M$. Hermitian and unitary operators are examples of normal operators.}, their eigenvectors can be chosen to form a common orthonormal eigenbasis.\footnote{A standard result in linear algebra states that commuting normal operators can be simultaneously diagonalized, that is, they admit a common orthonormal eigenbasis.} Let $\{\ket{u_j}:j=0,\ldots,2^n-1\}$ be a common eigenbasis of $U$ and $A$, so that
\[
A\ket{u_j}=\lambda_j\ket{u_j}
\qquad\text{and}\qquad
U\ket{u_j}=\e^{-\ii\lambda_j t}\ket{u_j}
=\e^{2\pi \ii\phi_j}\ket{u_j}.
\]
Therefore,
\[
2\pi\phi_j \equiv -\lambda_j t \mod{2\pi}.
\]
Hence, QPE determines the quantity $-\lambda_j t/(2\pi)$ only modulo $1$. In this sense, $t$ acts as a scale factor that converts eigenvalues into phases, determining how the spectrum of $A$ is mapped into the interval $[0,1)$ used by QPE. In order to recover $\lambda_j$ from the estimated phase, $t$ must be chosen so that distinct eigenvalues in the relevant spectral range produce distinct phases. Otherwise, if two different eigenvalues $\lambda_j$ and $\lambda_k$ satisfy
\[
\lambda_j t \equiv \lambda_k t \mod{2\pi},
\]
then they produce the same phase and QPE cannot distinguish them.

On the other hand, choosing $t$ very small is not a good general strategy to avoid this ambiguity. Indeed, if $t$ is too small, then the phases associated with the eigenvalues also become small, and the separation between nearby phases is reduced by the same factor. In particular, the eigenvalues with the smallest magnitude, which are the most important for the matrix inversion step, produce phases of order $|\lambda_{\min}|t$. As a consequence, phase estimation requires higher precision, and therefore more bits, in order to resolve these phases accurately and recover the corresponding eigenvalues. Thus, $t$ must be chosen as a compromise: large enough that the relevant phases can be resolved efficiently, but not so large that distinct eigenvalues become indistinguishable because of the modulo $2\pi$ ambiguity.

The choice of $t$ is related to a rescaling of $A$ because the operator used in QPE has the form $\e^{-\ii At}$. At this stage, rescaling $t$ or rescaling $A$ is equivalent, since both simply replace $A$ by the effective matrix $tA$. Later, however, the re-scaling of $A$ must also be taken into account in the matrix inversion step.

\begin{exercise}
Using Eq.~\eqref{eq:iAt-U}:
\begin{enumerate}
\item[(a)] Prove that $A$ and $U$ commute. \textit{Hint:} Use the power-series expansion
\[
\e^{-\ii At}=\sum_{k=0}^{\infty}\frac{(-\ii tA)^k}{k!}.
\]

\item[(b)] Prove that, for every $j$,
\[
U(A\ket{u_j})=\e^{2\pi \ii\phi_j}\,A\ket{u_j}.
\]

\item[(c)] Assume that the eigenvalues $\e^{2\pi \ii\phi_j}$ of $U$ are nondegenerate. Prove that, for every $j$, there exists $\lambda_j\in\mathbb{R}$ such that
\[
A\ket{u_j}=\lambda_j\ket{u_j}.
\]

\item[(d)] Under the assumptions of item (c), prove that
\[
\e^{-\ii \lambda_j t}=\e^{2\pi \ii\phi_j}.
\]
\end{enumerate}
\end{exercise}

\section{Implementing the circuit of $U$}\label{sec:impl-circ-U}

The circuit that implements QPE$_U$ is described in Chapter~\ref{chap:QPE}, but that procedure assumes that a circuit for $U$ is already available. We presented a method for decomposing any unitary operator $U$ into universal gates in Section~\ref{sec:decomp-Univ} on Page~\pageref{sec:decomp-Univ}. However, as discussed in that Section, the method is inefficient unless there is an efficient way to express $U$ as a product of two-level matrices. Since the input to the HHL algorithm is $A$, not $U$, we still face the additional problem of computing $U$.

Fortunately, there are alternative ways to implement $\e^{-\ii A t}$, which are collectively referred to as Hamiltonian simulation. If $A$ is the Hamiltonian of a physical system, this means that the energy levels of the system correspond to the eigenvalues of $A$. If we implement a circuit for $U=\e^{-\ii A t}$ on a quantum computer, we say that the quantum computer simulates the physical system whose energy levels are described by $A$, that is, if the initial state is $\ket{\psi_0}$, the state at time $t$ is $\e^{-\ii A t}\ket{\psi_0}$.

The algorithm assumes that the unitary $\e^{-\ii A t}$ can be implemented efficiently. This is possible when the matrix $A$ is $k$-local or sparse and efficiently row-computable, using Hamiltonian simulation techniques, which are detailed in Section~\ref{sec:HS}.

\section{Operator for Inverting Numbers}\label{sec:op-inv-num}

The HHL algorithm employs an additional operator defined as
\begin{equation}\label{eq:A-inv}
A_\text{inv} \ket{0}\ket{\lambda}
=
\left(\sqrt{1-\frac{C^2}{\lambda^2}}\,\ket{0}
+\frac{C}{\lambda}\,\ket{1}\right)\ket{\lambda},
\end{equation}
where the first register consists of a single qubit and the second register consists of $m$ qubits. The state $\ket{\lambda}$ stores, in binary form, the estimate of an eigenvalue $\lambda$ obtained from the quantum phase estimation procedure. At this point, $\lambda$ may refer either to an eigenvalue of $A$ itself or to an eigenvalue of a rescaled matrix used in the phase estimation step. This distinction does not affect the preparation of the final normalized output state, but it must be taken into account when discussing the inversion step.

The parameter $C$ is a real constant chosen at the beginning of the algorithm. On the one hand, a larger value of $C$ increases the success probability of the algorithm. On the other hand, $C$ cannot exceed the smallest eigenvalue in magnitude among those encoded in the second register, since the term $C/\lambda$ must remain within the interval $[-1,1]$. For this reason, it is customary to rescale the matrix used in the phase estimation step so that its eigenvalues lie in a convenient range, typically
\begin{equation}\label{eq:lambda-range}
\frac{1}{\kappa}\le |\lambda|\le 1,
\end{equation}
where $\kappa$ is the condition number. Under this convention, we must choose
\begin{equation}\label{eq:C-range}
0<C\le \frac{1}{\kappa}.
\end{equation}
The role of $\kappa$ in the performance of the HHL algorithm will become clearer in the analysis of the success probability. After this rescaling, one typically chooses $t\in(0,\pi)$ so that the eigenvalues are mapped to phases without introducing ambiguity modulo $2\pi$.

\begin{exercise}
In Eq.~\eqref{eq:A-inv}, the operator $A_{\mathrm{inv}}$ is defined by its action on $\ket{0}\ket{\lambda}$.
\begin{itemize}
\item[(a)] Explain why the expression above does not fully define a linear operator on the joint Hilbert space of the ancilla qubit and the eigenvalue register.

\item[(b)] Define $A_{\mathrm{inv}}$ as a $\lambda$-controlled rotation acting on the ancilla qubit:
\[
A_{\mathrm{inv}}
=
\sum_{\lambda}
R_y\!\left(2\pi \phi\right)\otimes
\ket{\lambda}\!\bra{\lambda},
\]
where $\phi = \frac{1}{\pi}\arcsin\!\left(\frac{C}{\lambda}\right)$.
Show that $A_{\mathrm{inv}}$ is unitary.

\item[(c)] Verify that the action of this operator on $\ket{0}\ket{\lambda}$ reproduces Eq.~\eqref{eq:A-inv}.
\end{itemize}
\end{exercise}

\section{The algorithm}

Having described the problem of solving linear systems and reviewed the quantum phase estimation (QPE) technique, we now present the HHL algorithm. This algorithm leverages QPE$_U$ to encode the eigenvalues of the Hermitian matrix $A$, followed by controlled operations to approximate matrix inversion. The procedure consists of state preparation, eigenvalue estimation, conditional rotations, measurement, and uncomputation steps, as outlined in Algorithm~\ref{algo_HHL}. Fig.~\ref{circuit:HHL} illustrates the quantum circuit implementing this process.

\begin{algorithm}[!ht]
\caption {HHL algorithm} \label{algo_HHL}
\KwIn{A non-singular Hermitian matrix $A$ of dimension $2^n$, the operator $U=\e^{-\ii A t}$, and an $n$-qubit input state $\ket{b}=\sum_{\ell=0}^{2^n-1}c_\ell\e^{\ii \alpha_\ell}\ket{\ell}$.}
\KwOut{The quantum state $\frac{\ket{x}}{\|\ket{x}\|}$, where $\ket{x} = A^{-1}\ket{b}$.}
\BlankLine
Prepare the initial state $\ket{0}\otimes\ket{0}^{\otimes m}\otimes \ket{b}$ in a 3-register quantum circuit with $1+m+n$ qubits (Section~\ref{sec:state-preparation})\;
Apply QPE$_U$ to the second and third registers (Sections~\ref{sec:reviewQPE},~\ref{sec:impl-circ-U},~\ref{sec:HS})\;
Apply $A_\text{inv}$ to the first and second registers (Sections~\ref{sec:op-inv-num} and~\ref{sec:impl-op-inv})\;
Measure the first register in the computational basis. If the result is 0, go to Line 1, otherwise continue\;
Apply QPE$_U^\dagger$ to the second and third registers\;
\end{algorithm}

\begin{figure}[h!]
\begin{equation*}
\Qcircuit @C=1.9em @R=1.9em {
\lstick{\ket{0}}&\qw	/^1								&  \qw                                             &\multigate{1}{A_\text{inv}}   &\meter &\rstick{\begin{cases}
0, & \text{restart, }\\
1, & \text{continue.}
\end{cases}}\cw&   \\
\lstick{\ket{0}^{\otimes m}} &\qw /^m&\multigate{1}{\text{QPE}_U}&\ghost{A_\text{inv}}                 &\multigate{1}{\text{QPE}^\dagger_U} & \rstick{\ket{0}^{\otimes m}}\qw&   \\
\lstick{\ket{b}} &\qw /^n							& \ghost{\text{QPE}_U}            & \qw                               & \ghost{\text{QPE}^\dagger_U}             &\rstick{\frac{1}{\left\|\ket{x}\right\|}\,\ket{x}}\qw \\
\ustick{\hspace{1.4cm}\ket{\psi_0}} & & \ustick{\hspace{2.3cm}\ket{\psi_1}}& \ustick{\hspace{1.8cm}\ket{\psi_2}} & \ustick{\hspace{2.3cm}\ket{\psi_3}}
}\vspace{0.1cm}
\end{equation*}
\caption{Circuit of the HHL algorithm. }\label{circuit:HHL}
\end{figure}

\section{Analysis}

If each eigenvalue can be represented with at most $m$ bits, the output of QPE$_U$ is exact. However, the HHL algorithm still introduces approximations due to an additional operator used to implement $A_\text{inv}$, which will be described later. We now calculate each quantum state in the circuit shown in Fig.~\ref{circuit:HHL}.

\subsection*{Calculation of $\ket{\psi_0}$}

The initial state of the algorithm is
\begin{align*}
\ket{\psi_0} &= \ket{0} \otimes \ket{0}^{\otimes m} \otimes \ket{b}.
\end{align*}
Using the expansion $\ket{b} = \sum_j b_j \ket{u_j}$ in the eigenbasis of $U$ (and $A$), we obtain
\begin{align*}
\ket{\psi_0} &= \ket{0} \otimes \sum_{j=0}^{2^n-1} b_j \ket{0}^{\otimes m} \ket{u_j}.
\end{align*}

\subsection*{Calculation of $\ket{\psi_1}$}

Applying QPE$_U$ with $U = \e^{-\ii A t}$, we obtain
\begin{equation*}
\ket{\psi_1} \,=\, \ket{0} \otimes \sum_{j=0}^{2^n-1} b_j \,\ket{\tilde{\phi}_j} \ket{u_j},
\end{equation*}
where $\tilde{\phi}_j \approx \phi_j 2^m$, in the notation of Section~\ref{sec:reviewQPE}.
The quantity $\tilde{\phi}_j$ is an $m$-bit integer representing an approximation of the phase of the eigenvalue $\e^{2\pi \ii \phi_j}$. From $\phi_j$ we can recover $\lambda_j$ (the eigenvalue of $A$) through the relation $2\pi\phi_j=-\lambda_j t \mod 2\pi$. For simplicity of notation, in what follows we denote the state $\ket{\tilde{\phi}_j}$ by $\ket{\lambda_j}$. Then
\begin{equation*}
\ket{\psi_1} \,=\, \ket{0} \otimes \sum_{j=0}^{2^n-1} b_j \,\ket{\lambda_j} \ket{u_j}.
\end{equation*}
The conversion of $\phi_j$ into $\lambda_j$ is addressed later in the implementation of $A_\text{inv}$.

\subsection*{Calculation of $\ket{\psi_2}$}

Applying $A_\text{inv}$ to the first and second registers of $\ket{\psi_1}$, we obtain
\begin{equation*}
\ket{\psi_2} \,=\, \sum_{j=0}^{2^n-1} b_j \left(\sqrt{1 - \frac{C^2}{{\lambda}_j^2}}\,\ket{0} + \frac{C}{{\lambda}_j}\,\ket{1} \right) \ket{{\lambda}_j} \ket{u_j}.
\end{equation*}

\subsection*{Measurement}

The next step is to perform a measurement of the first qubit in the computational basis. The probability of obtaining the result $1$ is
\begin{align*}
p_1 &= \bra{\psi_2}\big(\ket{1}\bra{1} \otimes I \otimes I\big)\ket{\psi_2}\\
   &= C^2\sum_{j=0}^{2^n-1}\frac{|b_j|^2}{{\lambda}_j^2}.
\end{align*}
Let us establish a connection between $p_1$ and the norm of $\ket{x}$. Using the fact that $\ket{x} = A^{-1}\ket{b}$ and expanding both $A$ and $\ket{b}$ in the eigenbasis $\{\ket{u_j}\}$ of $A$, we obtain
\begin{equation*}
\ket{x} \,=\, \sum_{j=0}^{2^n-1}\frac{b_j}{{\lambda}_j}\ket{u_j}.
\end{equation*}
The norm of $\ket{x}$ is
\begin{equation}\label{eq:C-p1}
\left\|\ket{x}\right\|
\,=\,
\left\|A^{-1}\ket{b}\right\|
=
\sqrt{\sum_{j=0}^{2^n-1}\frac{|b_j|^2}{{\lambda}_j^2}}
= \frac{\sqrt{p_1}}{C},
\end{equation}
Recall that QPE$_U$ returns an approximation of $\lambda_j$, therefore this expression is an approximation. If we repeat the algorithm up to this point multiple times, we can estimate $p_1$. Since $C$ can be chosen within a certain range, this also allows us to estimate $\|\ket{x}\|$.

The algorithm requires that after performing the measurement, we check the outcome. If it is $0$, we rerun the algorithm. If it is $1$, the unnormalized state of the quantum computer will be
\begin{equation*}
C\,\sum_{j=0}^{2^n-1}\frac{b_j}{{\lambda}_j}\,\ket{1}\ket{{\lambda}_j}\ket{u_j},
\end{equation*}
whose norm is
\begin{equation*}
C\,\sqrt{\sum_{j=0}^{2^n-1}\left(\frac{|b_j|}{{\lambda}_j}\right)^2}
= \sqrt{p_1}.
\end{equation*}
Using Eq.~\eqref{eq:C-p1}, the state of the qubits after a successful measurement (outcome $1$) is
\begin{equation*}
\ket{\psi'_2} = \frac{1}{\left\|\ket{x}\right\|}
\sum_{j=0}^{2^n-1}\frac{b_j}{{\lambda}_j}\,\ket{1}\ket{{\lambda}_j}\ket{u_j}.
\end{equation*}
Note that $\ket{\psi'_2}$ does not depend on $C$, whereas the success probability $p_1$ does.

The role of the condition number becomes apparent at this point. After rescaling $A$, we assume that the restrictions \eqref{eq:lambda-range} and \eqref{eq:C-range} are satisfied. If we take the largest possible value $C=1/\kappa$, the success probability is
\[
p_1
=
\frac{1}{\kappa^2}
\sum_{j=0}^{2^n-1}\frac{|b_j|^2}{{\lambda}_j^2}.
\]
Since $1/\kappa \le |\lambda_j| \le 1$ and $\sum_j |b_j|^2=1$, it follows that
\[
\frac{1}{\kappa^2} \le p_1 \le 1.
\]
Thus the success probability can be as small as $1/\kappa^2$. If $\kappa$ is large, the algorithm may need many repetitions before the outcome $1$ is obtained. In particular, since the success probability scales as $1/\kappa^2$, the expected number of repetitions required to obtain the desired outcome is $O(\kappa^2)$. Although the post-measurement state does not depend on $C$, the number of repetitions required to obtain it does, which is one of the ways in which $\kappa$ affects the performance of the HHL algorithm.

\subsection*{Calculation of $\ket{\psi_3}$}

The current state of the circuit is $\ket{\psi'_2}$. Using the fact that QPE$_U\ket{0}^{\otimes m} \ket{u_j} = \ket{{\lambda}_j} \ket{u_j}$, when we apply QPE$^\dagger_U$ to the second and third registers, we obtain
\begin{align*}
\ket{\psi_3} &\,=\,\frac{1}{\left\|\ket{x}\right\|}\sum_{j=0}^{2^n-1}\frac{b_j}{{\lambda}_j}\,\ket{1}\ket{0}^{\otimes m}\ket{u_j}\\
             &\,=\,\ket{1}\otimes\ket{0}^{\otimes m}\otimes \frac{\ket{x}}{\left\|\ket{x}\right\|}\,.
\end{align*}
The result in the third register is not exactly $\ket{x}/\left\|\ket{x}\right\|$ because QPE$_U$ does not return an exact value for each $\lambda_j$ in the general case.


\section{Implementing the operator for inverting numbers}\label{sec:impl-op-inv}

The implementation of $A_\text{inv}$ is based on a unitary operator $U_y$, which acts on $m+1$ qubits as follows
\begin{equation}\label{eq:Uphi}
U_y\ket{j}\ket{\phi_1\cdots\phi_m}=R_y(2\pi\phi)\ket{j}\ket{\phi_1\cdots\phi_m},
\end{equation}
where $\phi = 0.\phi_1...\phi_m$, meaning that $\phi_1, \dots, \phi_m$ are the binary digits of a value in the range $0 \leq \phi < 1$. Alternatively, we can express $\phi$ as
\[
\phi\,=\,\frac{\phi_1}{2^1}+\cdots+\frac{\phi_m}{2^{m}}.
\]
The operator $U_y$ can be implemented using the circuit shown in Fig.~\ref{fig:Uphi}.

\begin{figure}[!h]
\[
\Qcircuit @C=1.9em @R=1.5em {
\lstick{\ket{j}}     & \gate{R_y\left(\frac{\pi}{2^0}\right)} &\gate{R_y\left(\frac{\pi}{2^1}\right)}&\qw&\qw&\gate{R_y\left(\frac{\pi}{2^{m-1}}\right)}    &\rstick{R_y(2\pi\phi)\ket{j}} \qw \\
\lstick{\ket{\phi_1}}&\ctrl{-1}        &\qw                                 &\qw&\qw\,\,\,&\qw                                 &\rstick{\ket{\phi_1}} \qw \\
\lstick{\ket{\phi_2}}&\qw              &\ctrl{-2}                           &\qw&\qw&\qw                                 &\rstick{\ket{\phi_2}} \qw \\
\lstick{\vdots\,\,\,}&                 &                                    &   &\ddots& &\rstick{\,\,\,\vdots}\\
\lstick{\ket{\phi_m}}&\qw              &\qw                                 &\qw&\qw&\ctrl{-4}                                 &\rstick{\ket{\phi_m}} \qw
}
\]
\caption{Circuit that implements $U_y$ given by Eq.~(\ref{eq:Uphi}).}\label{fig:Uphi}
\end{figure}

To show that the circuit functions as intended, observe that the action of $R_y(\frac{\pi}{2^0})$ on $\ket{j}$, controlled by $\ket{\phi_1}$, is equivalent to applying $R_y(\frac{\pi\phi_1}{2^0})$ on $\ket{j}$. This follows because if $\phi_1 = 0$, then $R_y(0)$ is the identity operator, leaving the output as $\ket{j}$, whereas if $\phi_1 = 1$, the output is $R_y(\frac{\pi}{2^0})\ket{j}$. The same reasoning applies to the other controlled gates in the circuit. Consequently, the final output of the first qubit is $R_y(2\pi\phi)\ket{j}$ because
\begin{align*}
R_y(2\pi\phi)\ket{j}&=R_y\left(\frac{2\pi\phi_1}{2^1}+\cdots+\frac{2\pi\phi_m}{2^{m}}\right)\ket{j}\\
&=R_y\left(\frac{\pi\phi_m}{2^{m-1}}\right)\cdots R_y\left(\frac{\pi\phi_2}{2^{1}}\right)R_y\left(\frac{\pi\phi_1}{2^{0}}\right)\ket{j}.
\end{align*}
This confirms that the circuit in Fig.~\ref{fig:Uphi} correctly implements the operator $U_y$ as defined in Eq.~\eqref{eq:Uphi}.

Using that
\[
R_y(2\pi\phi) = \left[\begin{array}{cc}
 \cos(\pi\phi) &  -\sin(\pi\phi) \vspace{2pt}\\
 \sin(\pi\phi)   &  \,\,\,\,\cos(\pi\phi)
\end{array}\right],
\]
we obtain
\[
U_y\ket{0}\ket{\phi} \,=\, \cos(\pi\phi)\ket{0}\ket{\phi}+\sin(\pi\phi)\ket{1}\ket{\phi}.
\]
If we choose
\begin{equation}\label{eq:phi-lambda}
\phi =\frac{1}{\pi} \arcsin\frac{C}{\lambda},
\end{equation}
the resulting state matches the effect of applying $A_\text{inv}$ to $\ket{0}\ket{\lambda}$ (see Eq.~\eqref{eq:A-inv}), when considering only the first register. This condition implies that $C$ must be smaller than or equal to the smallest eigenvalue of $A$, since otherwise $\arcsin(C/\lambda)$ would be undefined. On the other hand, $C$ should not be too small, as this would result in a very small probability $p_1$ (see Eq.~\eqref{eq:C-p1}), assuming that $\left\|\ket{x}\right\|$ is not too large. At this point, we must choose a specific value for $C$, which will be used to define the next operator, $U_\text{inv}$.

To obtain $A_\text{inv}$ from $U_y$, we introduce another unitary operator, $U_\text{inv}$, which acts as
\begin{equation*}
U_\text{inv}\ket{\lambda}=\ket{\phi},
\end{equation*}
where $\phi$ is given by Eq.~\eqref{eq:phi-lambda}. We assume that $\lambda$ and $\phi$ are represented using the same number of bits. If a higher precision is needed for $\phi$, we can introduce an additional register, allowing $U_\text{inv}$ to act as $U_\text{inv}\ket{0}^a\ket{\lambda}=\ket{\phi}$, where $a$ represents the number of extra bits used to store $\phi$. Under the assumption that $\lambda$ and $\phi$ have the same number of bits, we can express $A_\text{inv}$ as
\[
A_\text{inv}=(I\otimes U_\text{inv}^{-1})\,U_y\,(I\otimes U_\text{inv}).
\]

The circuit for $U_\text{inv}$ is designed using classical computation. To implement this operator, we compute $\phi = \frac{1}{\pi} \arcsin (C / \lambda)$ classically, ensuring that the process is reversible. This computation can be performed efficiently using polynomial approximation (such as a Chebyshev series) or Newton's method. To make it quantum-compatible, we construct a reversible circuit for $U_\text{inv}$ that takes $\lambda$ as input and produces $\phi$ using reversible arithmetic operations such as controlled additions and multiplications. Although the design of $U_\text{inv}$ is based on classical computational techniques, its reversible implementation on a quantum computer allows it to operate correctly on quantum states, including those in superposition or entangled with other registers.

\section{Hamiltonian Simulation}\label{sec:HS}

Hamiltonian simulation is the task of implementing the unitary time-evolution operator generated by a Hamiltonian $H$, namely
\[
U(t)=\e^{-\ii Ht}.
\]
This problem plays a central role in quantum computing because the dynamics of quantum systems are governed by Hamiltonians, and many quantum algorithms rely on the ability to simulate such dynamics efficiently. Notable examples include algorithms for quantum chemistry, condensed matter physics, and the HHL algorithm for solving systems of linear equations.

In this Section, we discuss several approaches to Hamiltonian simulation. We begin with the special case of diagonal Hamiltonians, where the implementation is particularly simple. We then show how to simulate exponentials of arbitrary Pauli strings, which allows us to treat more general Hamiltonians expressed in the Pauli basis. Finally, we introduce product-formula approximations, such as the Trotter--Suzuki formulas, which make it possible to simulate the time evolution generated by sums of non-commuting terms.

\subsection*{Diagonal Hamiltonians}

Suppose that the Hamiltonian to be simulated is diagonal and has the form
\[
H=\operatorname{diag}(h_x)_{x\in\{0,1\}^n},
\]
where $n$ is the number of qubits. For instance, for two qubits we have
\[
H=\operatorname{diag}(h_{00},\,h_{01},\,h_{10},\,h_{11})=
\begin{bmatrix}
h_{00} & 0 & 0 & 0 \\
0 & h_{01} & 0 & 0 \\
0 & 0 & h_{10} & 0 \\
0 & 0 & 0 & h_{11}
\end{bmatrix}.
\]

Any diagonal Hamiltonian can be written as
\begin{equation}\label{eq:hamil-H}
H=\sum_{s\in\{0,1\}^n}\alpha_s\, Z^s,
\end{equation}
where $Z^s$ (that is, $Z$ raised to the bit string $s=s_1\ldots s_n$) means
\[
Z^s = Z^{s_1}\otimes Z^{s_2}\otimes \cdots \otimes Z^{s_n},
\]
and, naturally, $Z^0=I$.

\begin{exercise}
\begin{enumerate}
\item[(a)] Let $\ket{x}$ be a state of the computational basis, where $x=x_1\ldots x_n$ and $x_j\in\{0,1\}$. Show that
\[
Z^s\ket{x}=(-1)^{s\cdot x}\ket{x},
\]
where
\[
s\cdot x=s_1x_1+\cdots+s_nx_n \mod 2.
\]

\item[(b)] Using the expansion \eqref{eq:hamil-H}, show that, for every $x\in\{0,1\}^n$,
\[
h_x=\sum_{s\in\{0,1\}^n}\alpha_s(-1)^{s\cdot x}.
\]

\item[(c)] Prove that
\[
\sum_{x\in\{0,1\}^n} (-1)^{(s+r)\cdot x}
=
\begin{cases}
2^n, & s=r,\\
0, & s\neq r.
\end{cases}
\]

\item[(d)] Deduce that
\[
\alpha_s=\frac{1}{2^n}\sum_{x\in\{0,1\}^n} h_x\,(-1)^{s\cdot x}.
\]
\end{enumerate}
\end{exercise}\vspace{10pt}

An alternative notation for Eq.~\eqref{eq:hamil-H} is
\[
H=\sum_{S\subseteq \{1,\dots,n\}} \alpha_S \prod_{j\in S} Z_j,
\]
where we sum over all subsets of $\{1,\dots,n\}$, and $Z_j$ means a Pauli $Z$ acting on qubit $j$ and identity on the others. For instance, $Z_3=I\otimes I\otimes Z$ when $n=3$. Note that $Z^{011}=Z^0\otimes Z^{1}\otimes Z^{1}=Z_2 Z_3$.

When we exponentiate $H$, we obtain the time-evolution operator
\[
U(t)=\e^{-\ii Ht}.
\]
Since $H$ is diagonal, all terms in its expansion commute with each other. Therefore,
\begin{equation}\label{eq:exp-i-H-t}
\e^{-\ii Ht}
=
\prod_{s\in\{0,1\}^n} \e^{-\ii \alpha_s t Z^s},
\end{equation}
because $\e^{A+B}=\e^A\e^B$ when $A$ and $B$ commute. Eq.~\eqref{eq:exp-i-H-t} is a product of unitary operators, and the circuit is built sequentially. Thus, the problem reduces to implementing each factor
\[
\e^{-\ii \alpha_s t Z^s}.
\]

There are three cases to consider: (1)~If $s$ has Hamming weight $0$ (that is, $s=0\ldots 0$), then $Z^s=I$ and
\[
\e^{-\ii \alpha t Z^s}=\e^{-\ii \alpha t}\, I.
\]
Usually this term is ignored because it contributes only a global phase. In the HHL algorithm, however, the Hamiltonian simulation is used inside the $\mathrm{QPE}_U$ procedure, and this phase plays a non-trivial role there; (2)~If $s$ has Hamming weight $1$ (that is, $s$ has a single $1$), then $Z^s$ is simply a Pauli operator $Z_j$ acting on qubit $j$. Using
\[
R_z(\phi)=\e^{-\ii\phi Z/2},
\]
we obtain
\[
\e^{-\ii \alpha t Z_j}=R_z(2\alpha t),
\]
which acts only on the $j$-th qubit; and (3)~If $s$ has Hamming weight $k\ge 2$, then $Z^s$ is a tensor product of $Z$ operators acting on those $k$ qubits, for example
\[
\e^{-\ii \alpha t\, Z_{j_1}Z_{j_2}\cdots Z_{j_k}}.
\]

\begin{exercise}
Let $\phi\in\mathbb{R}$ and let $U$ be a single-qubit unitary operator. Denote by $C(U)$ the controlled-$U$ operation with the first qubit as control and the second qubit as target.

\begin{enumerate}
\item[(a)] Show that, in general,
\[
C(\e^{\ii\phi}U)\neq \e^{\ii\phi}\,C(U).
\]
Compute both operators explicitly using the definition
\[
C(U)=|0\rangle\langle0|\otimes I + |1\rangle\langle1|\otimes U.
\]

\item[(b)] Explain why multiplying $U$ by the global phase $\e^{\ii\phi}$ does not affect the action of $U$ on a single qubit, but becomes observable when the operation is controlled.

\item[(c)] Show that
\[
C(\e^{\ii\phi}U)=(|0\rangle\langle0|+\e^{\ii\phi}|1\rangle\langle1|)\otimes I \; C(U),
\]
and interpret the first factor as a phase gate acting on the control qubit.

\item[(d)] In the Hamiltonian simulation discussed above, the case where $s$ has Hamming weight $0$ produces a term $\e^{-\ii\alpha t}I$, which is usually considered a global phase. Explain why this phase cannot be ignored when the unitary $\e^{-\ii Ht}$ is used inside the Quantum Phase Estimation (QPE) procedure of the HHL algorithm.
\end{enumerate}
\end{exercise}

\begin{exercise}\label{ex:R_z-phi}
Recall that
\[
R_z(\phi)=
\begin{bmatrix}
\e^{-\ii\phi/2} & 0 \\
0 & \e^{\ii\phi/2}
\end{bmatrix}.
\]
Show that for $x\in\{0,1\}$,
\[
R_z(\phi)\ket{x}=\e^{-\ii \frac{\phi}{2} (-1)^x}\ket{x}.
\]
\end{exercise}

\begin{figure}[!h]
\[
\Qcircuit @C=1.3em @R=1.4em {
\lstick{\ket{x_1}}&\qw& \ctrl{3} & \qw      & \qw       & \qw & \qw & \qw & \qw             & \qw & \ctrl{3} & \qw \\
\lstick{\ket{x_2}}&\qw& \qw      & \ctrl{2} & \qw       & \qw & \qw & \qw & \qw             & \ctrl{2} & \qw & \qw \\
\lstick{\ket{x_3}}&\qw& \qw      & \qw      & \ctrl{1}  & \qw & \qw & \qw & \ctrl{1}        & \qw      & \qw & \qw \\
\lstick{\ket{x_4}}&\qw& \targ    & \targ    & \targ     & \qw & \gate{R_z(\phi)}   & \qw & \targ    & \targ& \targ & \qw\\
& \ket{\psi_0} &&&& \ket{\psi_1} && \ket{\psi_2}&&&& \ket{\psi_3}
}
\]
\caption{Implementation of $\e^{-\ii \alpha t\, Z_1 Z_2 Z_3 Z_4}$ when $\phi=2\alpha t$.}\label{fig:exp-i-H}
\end{figure}

The circuit that implements $\exp(-\ii \alpha t\, Z_{j_1}Z_{j_2}\cdots Z_{j_k})$ has gates only on qubits $j_1$, $j_2$, \ldots, $j_k$. An example with four qubits is sufficient to illustrate the general idea. Suppose we want to implement $U(t)=\exp(-\ii \alpha t\, Z_{1}Z_{2}Z_{3}Z_{4})$ when $n=4$. Consider the circuit in Fig.~\ref{fig:exp-i-H}. The input is a computational-basis state:
\[
\ket{\psi_0}= \ket{x_1x_2x_3x_4}.
\]
After applying the sequence of $n-1$ CNOTs, the output is
\[
\ket{\psi_1}= \ket{x_1x_2x_3}\otimes\ket{x_1\oplus x_2\oplus x_3\oplus x_4}.
\]

Now we calculate the action of $R_z(\phi)$ on the $n$-th qubit:
\[
\ket{\psi_2}= \ket{x_1x_2x_3}\otimes R_z(\phi)\ket{x_1\oplus x_2\oplus x_3\oplus x_4}.
\]
Using Exercise~\ref{ex:R_z-phi}, we obtain
\begin{align*}
\ket{\psi_2} &= \ket{x_1x_2x_3}\otimes\e^{-\ii \frac{\phi}{2} (-1)^{x_1\oplus x_2\oplus x_3\oplus x_4}}\ket{x_1\oplus x_2\oplus x_3\oplus x_4}\\
             &= \e^{-\ii \frac{\phi}{2} (-1)^{x_1}(-1)^{x_2}(-1)^{x_3}(-1)^{x_4}}\ket{x_1x_2x_3}\otimes \ket{x_1\oplus x_2\oplus x_3\oplus x_4}.
\end{align*}
In the second equality, we used the identity
\[
(-1)^{x_1\oplus \cdots \oplus x_n}
=
(-1)^{x_1}\cdots(-1)^{x_n},
\qquad x_j\in\{0,1\}.
\]

In the last step, we apply the same $n-1$ CNOTs in reverse order, which produces the output
\[
\ket{\psi_3} = \e^{-\ii \frac{\phi}{2} (-1)^{x_1}(-1)^{x_2}(-1)^{x_3}(-1)^{x_4}}\ket{x_1x_2x_3x_4}.
\]
Since
\[
Z_1Z_2Z_3Z_4\ket{x_1x_2x_3x_4}
=
(-1)^{x_1}(-1)^{x_2}(-1)^{x_3}(-1)^{x_4}\ket{x_1x_2x_3x_4},
\]
we can rewrite $\ket{\psi_3}$ as
\[
\ket{\psi_3}=\e^{-\ii \frac{\phi}{2} Z_1Z_2Z_3Z_4}\ket{x_1x_2x_3x_4}.
\]
The implementation of the desired operator $\e^{-\ii \alpha t\, Z_1 Z_2 Z_3 Z_4}$ is obtained by taking $\phi=2\alpha t$.

In the worst case, the implementation of a diagonal Hamiltonian requires an exponential number of gates. Indeed, a general diagonal Hamiltonian on $n$ qubits has $2^n$ independent parameters, which appear as the coefficients $\alpha_s$ in the expansion~\eqref{eq:hamil-H}. After exponentiation, the time-evolution operator is written as the product~\eqref{eq:exp-i-H-t}, which contains up to $2^n$ unitary factors of the form $\e^{-\ii \alpha_s t Z^s}$. Each factor can be implemented using $O(n)$ elementary gates (a sequence of CNOTs together with one $R_z$ gate). Therefore, the total number of gates required to implement $\e^{-\ii Ht}$ is $O(n2^n)$ in the worst case.

In many relevant situations, however, the Hamiltonian has a special structure.
A diagonal Hamiltonian is called $k$-local if each term in its decomposition acts nontrivially on at most $k$ qubits, where $k$ is a constant independent of $n$.
In the expansion~\eqref{eq:hamil-H}, this means that the coefficients $\alpha_s$ are nonzero only for bit strings $s$ whose Hamming weight is at most $k$. Consequently, in the product representation~\eqref{eq:exp-i-H-t}, only those factors of the form $\e^{-\ii \alpha_s t Z^s}$ with Hamming weight at most $k$ appear. The number of such terms grows only polynomially with $n$, and each of them can be implemented using $O(n)$ elementary gates. Therefore, the unitary operator $\e^{-\ii Ht}$ can be implemented using a number of gates that is polynomial in $n$, making the simulation efficient. Many physically motivated Hamiltonians arising in condensed matter physics and quantum chemistry are $k$-local.

\begin{exercise} (Circuit for diagonal Hamiltonians) \label{exe:PPTrick}
Let
\[
U=\prod_{j=1}^{n-1}\mathrm{CNOT}_{j,n}.
\]

\begin{enumerate}

\item[(a)] Show that for any computational-basis state $\ket{x_1\ldots x_n}$,
\[
U\ket{x_1\ldots x_n}
=
\ket{x_1\ldots x_{n-1}}\otimes
\ket{x_1\oplus x_2\oplus\cdots\oplus x_n}.
\]

\item[(b)] Show that
\[
U^\dagger Z_n U = Z_1Z_2\cdots Z_n.
\]

\item[(c)] Using the definition
\[
R_z(\phi)=\e^{-\ii \phi Z/2},
\]
show that
\[
U^\dagger R_z(\phi)_n U
=
\e^{-\ii \frac{\phi}{2} Z_1Z_2\cdots Z_n}.
\]

\item[(d)] Conclude that the circuit
\[
U^\dagger R_z(\phi)_n U
\]
implements the unitary operator
\[
\e^{-\ii \frac{\phi}{2} Z_1Z_2\cdots Z_n}.
\]

\end{enumerate}
\end{exercise}

\subsection*{Non-diagonal Hamiltonians}

In the previous subsection we considered the simulation of diagonal Hamiltonians, whose eigenbasis coincides with the computational basis. In many relevant situations, however, the Hamiltonian contains non-diagonal terms, that is, terms involving Pauli operators $X$ or $Y$. Such Hamiltonians cannot be directly written as functions of $Z$ operators only, and therefore their implementation requires additional steps.

A general Hamiltonian acting on $n$ qubits can always be expanded as a linear combination of Pauli strings,
\begin{equation}\label{eq:H-beta-P}
H=\sum_{j} \beta_j P_j,
\end{equation}
where each $P_j$ is a tensor product of Pauli operators
\[
P_j\in\{I,X,Y,Z\}^{\otimes n}.
\]
This representation is often called the \emph{Pauli decomposition} of the Hamiltonian (see Exercise~\ref{exe:Pauli-basis}). Each term $P_j$ acts nontrivially only on the qubits where the Pauli operators differ from the identity.

To implement the time-evolution operator
\[
U(t)=\e^{-\ii Ht},
\]
we cannot in general write the exponential as a product of exponentials as in Eq.~\eqref{eq:exp-i-H-t}, because different Pauli strings may not commute. A common approach is to approximate the evolution using product formulas such as the \emph{Trotter--Suzuki decomposition}. Before discussing this approximation, we first show how to implement the unitary operator
\[
\e^{-\ii \phi P},
\]
where $P$ is an arbitrary Pauli string.

Any Pauli string
\[
P=P_1\otimes P_2\otimes\cdots\otimes P_n,
\]
with $P_j\in\{I,X,Y,Z\}$, can be transformed into a tensor product of $Z$ operators by a suitable change of basis applied independently to each qubit. Specifically, we use the identities
\begin{equation}\label{eq:HZH-SH}
H Z H = X,
\qquad
S H\, Z\, H S^\dagger = Y.
\end{equation}
Therefore, if $P_j=X$ we conjugate qubit $j$ by a Hadamard gate $H$, and if $P_j=Y$ we conjugate qubit $j$ by the unitary $HS^\dagger$. Under these transformations, the Pauli string $P$ is mapped to a product of $Z$ operators acting on the same set of qubits.

More precisely, let $V$ be the unitary operator that applies the appropriate basis change to each qubit so that
\[
V P V^\dagger = Z_{j_1}Z_{j_2}\cdots Z_{j_k},
\]
where $\{j_1,\ldots,j_k\}$ are the qubits on which $P$ acts nontrivially. The remaining qubits are associated with wires on which no gates act, since the corresponding factors of the Pauli string are identities. Then
\[
\e^{-\ii \phi P}
=
V^\dagger
\left(\e^{-\ii \phi Z_{j_1}Z_{j_2}\cdots Z_{j_k}}\right)
V.
\]
The central unitary operator in this expression is exactly of the type studied in the previous subsection, and it can therefore be implemented using a sequence of CNOT gates and a single $R_z$ rotation.

Consequently, the implementation of $\e^{-\ii \phi P}$ proceeds in three steps. First, we apply the basis-change circuit $V$ that maps the Pauli string $P$ to a product of $Z$ operators. Second, we implement the unitary $\e^{-\ii \phi Z_{j_1}\cdots Z_{j_k}}$ using the circuit described in Exercise~\ref{exe:PPTrick}. Finally, we apply $V^\dagger$ to undo the basis change. This procedure reduces the simulation of arbitrary Pauli strings to the implementation of diagonal operators.

In general, the Hamiltonian $H=\sum_j \beta_j P_j$ is a sum of Pauli strings that do not necessarily commute with each other. In this case we cannot write $\e^{-\ii Ht}$ as a product of terms of the form $\e^{-\ii \beta_j t P_j}$, because the Pauli strings $P_j$ do not necessarily commute with each other. The identity $\e^{A+B}=\e^A\e^B$ holds only when $A$ and $B$ commute. Therefore, for a generic Hamiltonian we need an approximation that expresses the exponential of the sum in terms of a product of exponentials that can be implemented individually.

\begin{exercise} (Pauli strings as an operator basis) \label{exe:Pauli-basis}
Let $\mathcal{P}_n=\{I,X,Y,Z\}^{\otimes n}$ be the set of $n$-qubit Pauli strings.

\begin{enumerate}
\item[(a)] Show that $\mathcal{P}_n$ contains $4^n$ operators.

\item[(b)] Show that the vector space of all $2^n\times2^n$ complex matrices has dimension $4^n$.

\item[(c)] For any operators $B$ and $C$ acting on $n$ qubits, show that Pauli strings are orthogonal with respect to the Hilbert--Schmidt inner product
\[
\langle B,C\rangle=\operatorname{Tr}(B^\dagger C),
\]
that is,
\[
\operatorname{Tr}(P_iP_j)=2^n\delta_{ij}.
\]

\item[(d)] Consider the vector space of all $2^n\times2^n$ complex matrices endowed with the Hilbert--Schmidt inner product. Conclude that $\mathcal{P}_n$ forms an orthogonal basis of this vector space.

\item[(e)] Using the orthogonality of the Pauli strings, deduce that any operator $C$ acting on $n$ qubits can be written as
\[
C=\sum_{P\in\mathcal{P}_n}\alpha_P P,
\qquad
\alpha_P=\frac{1}{2^n}\operatorname{Tr}(PC).
\]

\item[(f)] Conclude that any Hamiltonian $H$ admits the expansion~\eqref{eq:H-beta-P}.
Show that the fact that both $H$ and $P_j$ are Hermitian implies that the coefficients $\beta_j$ in~\eqref{eq:H-beta-P} are real.
\end{enumerate}
\end{exercise}

\begin{exercise}
Consider the unitary operator
\[
U(\phi)=\e^{-\ii \phi X_1Y_2Z_3}.
\]

\begin{enumerate}

\item[(a)] Using the identities~\eqref{eq:HZH-SH}, show that the Pauli string $X_1Y_2Z_3$ can be converted into a product of $Z$ operators by a suitable change of basis. In particular, define
\[
V=H_1H_2S^\dagger_2
\]
and prove that
\[
\e^{-\ii \phi X_1Y_2Z_3}
=
V^\dagger\, \e^{-\ii \phi Z_1Z_2Z_3}\, V.
\]

\item[(b)] The circuit below is proposed to implement the unitary operator $\e^{-\ii \phi X_1Y_2Z_3}$:
\[
\Qcircuit @C=1.2em @R=1.2em {
\lstick{\ket{x_1}} & \gate{H} & \qw              &\ctrl{2} & \qw      & \qw                 & \qw      & \ctrl{2} & \qw      & \gate{H}   & \qw \\
\lstick{\ket{x_2}} & \gate{H} & \gate{S^\dagger} & \qw     & \ctrl{1} & \qw                 & \ctrl{1} & \qw      & \gate{S} & \gate{H}   & \qw \\
\lstick{\ket{x_3}} & \qw      & \qw              & \targ   & \targ    & \gate{R_z(2\phi)} & \targ    & \targ    & \qw      & \qw        & \qw
}
\]

Show that this circuit indeed implements the operator $\e^{-\ii \phi X_1Y_2Z_3}$.

\end{enumerate}
\end{exercise}

\subsection*{Trotter-Suzuki approximation}

A widely used method is the \emph{Trotter--Suzuki approximation}. The basic idea is to split the time evolution into many small steps and approximate the exponential of the sum by a product of exponentials of the individual terms. If
\[
H=\sum_{j=1}^m H_j,
\]
then the first-order Trotter formula gives
\begin{equation}\label{eq:Trotter}
\e^{-\ii Ht}
=
\left(
\prod_{j=1}^m
\e^{-\ii H_j t/r}
\right)^r
+O\!\left(\frac{t^2}{r}\right),
\end{equation}
where $r$ is the number of Trotter steps. As $r$ increases, the approximation becomes more accurate.

In our context, each term $H_j$ is proportional to a Pauli string, that is, $H_j=\beta_j P_j$ as in Eq.~\eqref{eq:H-beta-P}. Therefore, each unitary operator $\e^{-\ii H_j t/r}$ can be implemented using the techniques described above: a change of basis that converts the Pauli string into a product of $Z$ operators, followed by the circuit described in Exercise~\ref{exe:PPTrick}, and then the inverse change of basis. Consequently, the Trotter formula allows us to approximate the time evolution generated by an arbitrary Hamiltonian using a sequence of circuits that implement exponentials of Pauli strings.

Let us explore the meaning of the Trotter formula~\eqref{eq:Trotter} by proving that
\begin{equation}\label{eq:H-expansion}
e^{-\ii H\delta} =\prod_{j=1}^m e^{-\ii H_j\delta}+O(\delta^2).
\end{equation}
For convenience, we write $\delta=t/r$. The order term $O(\delta^2)$ in this expression means that the difference between the two operators has norm of order $\delta^2$. More precisely, there exist constants $\delta_0>0$ and $c>0$ such that, for all $0<\delta<\delta_0$,
\[
\left\|
\prod_{j=1}^m e^{-\ii H_j\delta}
-
e^{-\ii H\delta}
\right\|
\le c\,\delta^2,
\]
where $c$ may depend on the operators $H_1,\ldots,H_m$ and $\|M\|$ is the norm of operator $M$.\footnote{The \emph{operator norm} (or spectral norm) of a matrix $M$ is defined as
\[
\|M\| = \max_{\|\ket{\psi}\|=1} \|M\ket{\psi}\|.
\]
Equivalently, $\|M\|$ is the largest singular value of $M$.}

One way to justify the approximation~\eqref{eq:H-expansion} is to expand the exponentials as power series. For each term we have
\[
\e^{-\ii H_j \delta}
=
I-\ii H_j\delta+O(\delta^2).
\]
Multiplying the $m$ factors on the right-hand side gives
\[
\prod_{j=1}^m
\left(I-\ii H_j\delta+O(\delta^2)\right)
=
I-\ii\delta\sum_{j=1}^m H_j+O(\delta^2).
\]
On the other hand, since $H=\sum_{j=1}^m H_j$, we have
\[
\e^{-\ii H\delta}
=
I-\ii H\delta+O(\delta^2)
=
I-\ii\delta\sum_{j=1}^m H_j+O(\delta^2).
\]
Therefore the two operators agree up to first order in $\delta$, which implies Eq.~\eqref{eq:H-expansion}.

To obtain the Trotter formula~\eqref{eq:Trotter}, we repeat the approximation in Eq.~\eqref{eq:H-expansion} for each of the $r$ time slices. Since the total evolution over time $t$ can be written as
\[
\e^{-\ii Ht}
=
\left(\e^{-\ii H\delta}\right)^r,
\]
we obtain
\[
\e^{-\ii Ht}
\approx
\left(
\prod_{j=1}^m \e^{-\ii H_j \delta}
\right)^r.
\]
The error in a single time slice is $O(\delta^2)$. After $r$ steps, these local errors combine to give a total error, which is obtained by multiplying $O(\delta^2)$ by $O(r)$. This yields the first-order Trotter formula~\eqref{eq:Trotter}.

To illustrate the procedure, consider the three-qubit Hamiltonian
\[
H = Z_1Z_2 + X_2X_3 .
\]
This Hamiltonian is the sum of two Pauli strings that do not commute because both terms act nontrivially on qubit $2$ with different Pauli operators. As a result, we cannot split $\e^{-\ii Ht}$ into only two separate exponentials without introducing a large error. Therefore, the Trotter approximation is required. Using the first-order Trotter formula with $r$ time steps, we obtain
\[
\e^{-\ii Ht}
\approx
\left(
\e^{-\ii Z_1Z_2\, t/r}
\,\e^{-\ii X_2X_3\, t/r}
\right)^r .
\]
Thus, the evolution generated by $H$ is approximated by repeatedly applying the two simpler unitaries $\e^{-\ii Z_1Z_2\, t/r}$ and $\e^{-\ii X_2X_3\, t/r}$.

Each factor in this product can be implemented using the techniques described in the previous sections. The unitary operator $\e^{-\ii Z_1Z_2\, t/r}$ is diagonal in the computational basis and can therefore be implemented using the circuit for exponentials of products of $Z$ operators. In particular, we apply a CNOT gate with qubit $1$ as control and qubit $2$ as target, then apply the single-qubit rotation $R_z(2t/r)$ on qubit $2$, and finally apply the same CNOT gate again to uncompute the parity.

The unitary operator $\e^{-\ii X_2X_3\, t/r}$ is not diagonal in the computational basis, but it can be reduced to the previous case by a change of basis. Using the identity $HXH=Z$, we first apply Hadamard gates on qubits $2$ and $3$, which transforms the operator $X_2X_3$ into $Z_2Z_3$. We then implement the unitary $\e^{-\ii Z_2Z_3\, t/r}$ using the same CNOT–$R_z$–CNOT construction described above. Finally, we apply the Hadamard gates again on qubits $2$ and $3$ to return to the original basis. The circuit that implements $\e^{-\ii (Z_1Z_2+X_2X_3)t}$ is described in Fig.~\ref{fig:non-comm-Hamil}.

\begin{figure}[!h]
\centering
\[
\Qcircuit @C=1.em @R=1.em {
\lstick{\ket{x_1}} & \qw      & \qw      & \qw        &\qw       &\qw     & \qw      & \ctrl{1} & \qw               & \ctrl{1} & \qw  &&\lstick{\quad ...}\\
\lstick{\ket{x_2}} & \gate{H} & \ctrl{1} & \qw                    & \ctrl{1} & \gate{H} & \qw& \targ             & \gate{R_z(2t/r)} & \targ  & \qw&&&&&&&\lstick{\text{repeat }r\text{ times}}\\
\lstick{\ket{x_3}} & \gate{H} & \targ    & \gate{R_z(2t/r)}       & \targ    & \gate{H} & \qw&\qw               & \qw & \qw & \qw&&\lstick{\quad ...}
\gategroup{1}{2}{3}{6}{1.em}{--}
\gategroup{1}{8}{3}{10}{1.3em}{--}
}
\]
\caption{Circuit implementing one Trotter step for the Hamiltonian $H=Z_1Z_2+X_2X_3$. The left dashed box implements $\e^{-\ii X_2X_3\,t/r}$. The right dashed box implements $\e^{-\ii Z_1Z_2\,t/r}$. Repeating this step $r$ times approximates the evolution operator $\e^{-\ii Ht}$.}\label{fig:non-comm-Hamil}
\end{figure}

\begin{exercise}
\begin{enumerate}

\item[(a)]
Let $U$ and $A$ be unitary operators acting on the target system. Show that
\[
C(A^\dagger U A) = (I\otimes A^\dagger)\, C(U)\, (I\otimes A),
\]
where $C(U)$ denotes the controlled version of $U$ (see Chapter~\ref{chap:2} for more details).

\item[(b)]
Show that the controlled version of the first-order Trotter approximation of $\e^{-\ii Ht}$ (circuit in Fig.~\ref{fig:non-comm-Hamil}) is
\[
\Qcircuit @C=1.em @R=0.6em {
\lstick{\ket{q_0}} & \qw      & \qw      & \ctrl{4}   &\qw       &\qw     & \qw      & \qw & \ctrl{3}               & \qw & \qw  &&\lstick{\quad ...}\\
& & \\
\lstick{\ket{x_1}} & \qw      & \qw      & \qw        &\qw       &\qw     & \qw      & \ctrl{1} & \qw               & \ctrl{1} & \qw  &&&&&&&\lstick{\text{repeat }r\text{ times}}\\
\lstick{\ket{x_2}} & \gate{H} & \ctrl{1} & \qw                    & \ctrl{1} & \gate{H} & \qw& \targ             & \gate{R_z(2t/r)} & \targ  & \qw&&\lstick{\quad ...}\\
\lstick{\ket{x_3}} & \gate{H} & \targ    & \gate{R_z(2t/r)}       & \targ    & \gate{H} & \qw&\qw               & \qw & \qw & \qw&&\lstick{.}
\gategroup{1}{2}{5}{6}{1.2em}{--}
\gategroup{1}{8}{5}{10}{1.2em}{--}
}
\]
\end{enumerate}
\end{exercise}

\subsection*{Second-order Trotter--Suzuki approximation}

The first-order Trotter formula provides an approximation whose error scales as $O(t^2/r)$. This error arises because the exponential of a sum of non-commuting operators cannot be written exactly as a product of exponentials of the individual terms. A better approximation can be obtained using a symmetric product formula, known as the \emph{second-order Trotter--Suzuki approximation}~\cite{SS99}.

Let again
\[
H=\sum_{j=1}^m H_j.
\]
The second-order Trotter--Suzuki formula approximates the time evolution operator as
\begin{equation}\label{eq:TS2}
\e^{-\ii Ht}
=
\left(
\prod_{j=1}^{m-1}
\e^{-\ii H_j t/(2r)}
\;
\e^{-\ii H_m t/r}
\;
\prod_{j=m-1}^{1}
\e^{-\ii H_j t/(2r)}
\right)^r
+
O\!\left(\frac{t^3}{r^2}\right).
\end{equation}
The sequence of exponentials is applied first in the forward order and then in the reverse order. This symmetric structure cancels the leading error terms that appear in the first-order approximation, resulting in a smaller overall error.

To understand the idea, consider the case of two terms, $H = A + B$. The second-order formula takes the form
\[
\e^{-\ii (A+B)t}
\approx
\left(
\e^{-\ii A t/(2r)}
\e^{-\ii B t/r}
\e^{-\ii A t/(2r)}
\right)^r.
\]
Compared with the first-order formula
\[
\e^{-\ii (A+B)t}
\approx
\left(
\e^{-\ii A t/r}
\e^{-\ii B t/r}
\right)^r,
\]
the symmetric arrangement reduces the approximation error. In particular, the error of the second-order formula scales as $O(t^3/r^2)$, which converges faster as $r$ increases.

In the context of Hamiltonian simulation, each operator $\e^{-\ii H_j t/(2r)}$ is again the exponential of a Pauli string. Therefore the same techniques described earlier can be used to implement these operators: a change of basis that converts the Pauli string into a product of $Z$ operators, followed by the circuit implementing $\e^{-\ii \phi Z_{j_1}\cdots Z_{j_k}}$, and finally the inverse change of basis. The second-order Trotter--Suzuki formula is often preferred in practice because it achieves higher accuracy without significantly increasing the complexity of each Trotter step.

In many applications the Hamiltonian has additional structure that makes the Trotter--Suzuki approximation particularly useful. Suppose that the Hamiltonian admits a Pauli decomposition
\[
H=\sum_{j=1}^m \beta_j P_j,
\]
where each $P_j$ is a Pauli string acting on $n$ qubits. In this case the operators that appear in the second-order Trotter--Suzuki formula are exponentials of Pauli strings of the form
\[
\e^{-\ii \beta_j P_j t/(2r)}.
\]
As discussed earlier, each such unitary can be implemented by applying a suitable change of basis that converts the Pauli string into a product of $Z$ operators, followed by a circuit implementing $\e^{-\ii \phi Z_{j_1}\cdots Z_{j_k}}$, and finally undoing the basis change.

If the Hamiltonian is $k$-local, each Pauli string $P_j$ acts nontrivially on at most $k$ qubits, where $k$ is independent of $n$. Consequently, the number of terms $m$ in the Pauli decomposition grows at most polynomially with $n$, and each exponential $\e^{-\ii \beta_j P_j t/(2r)}$ can be implemented using a circuit whose size is $O(\text{poly}(n))$. Therefore, each Trotter step of the second-order approximation can be implemented efficiently, and the overall simulation of the time-evolution operator $\e^{-\ii Ht}$ requires a number of elementary gates that is polynomial in the number of qubits.

\begin{exercise} (Second-order Trotter--Suzuki formula)
Assume that the first-order product formula has already been established. Let
\[
S(\delta)=\e^{A\delta/2}\e^{B\delta}\e^{A\delta/2},
\]
where $A$ and $B$ are bounded operators and $\delta$ is a small real parameter.

\begin{enumerate}

\item[(a)] Using the power-series expansion of the exponential, expand $S(\delta)$ up to terms of order $\delta^2$ and show that
\[
S(\delta)
=
I+(A+B)\delta+\frac{(A+B)^2}{2}\delta^2+O(\delta^3).
\]

\item[(b)] Deduce that
\[
\e^{A\delta/2}\e^{B\delta}\e^{A\delta/2}
=
\e^{(A+B)\delta}+O(\delta^3).
\]

\item[(c)] Let $H=A+B$ and take $\delta=t/r$. Explain why
\[
\e^{-\ii Ht}
=
\left(\e^{-\ii Ht/r}\right)^r.
\]
Then use item (b) to show that
\[
\e^{-\ii Ht}
=
\left(
\e^{-\ii A t/(2r)}
\e^{-\ii B t/r}
\e^{-\ii A t/(2r)}
\right)^r
+
O\!\left(\frac{t^3}{r^2}\right).
\]

\item[(d)] Consider the symmetric product
\[
S'(\delta)=
\prod_{j=1}^{m-1}
\e^{-\ii H_j \delta/2}
\;
\e^{-\ii H_m \delta}
\;
\prod_{j=m-1}^{1}
\e^{-\ii H_j \delta/2}.
\]

Use the result of item (b) repeatedly to argue that
\[
S'(\delta)
=
\e^{-\ii H\delta}+O(\delta^3).
\]
\item[(e)] Taking again $\delta=t/r$, conclude that Eq.~\eqref{eq:TS2} is correct.

\end{enumerate}
\end{exercise}

\subsection*{Sparse Hamiltonians}

Another important class of Hamiltonians that can be simulated efficiently are the \emph{sparse Hamiltonians}. Let $H$ be a $2^n\times 2^n$ Hermitian matrix acting on $n$ qubits. We say that $H$ is \emph{$s$-sparse} if each row and each column of $H$ contains at most $s$ nonzero entries, where $s$ grows at most polynomially with $n$.

In addition to sparsity, one usually assumes that the nonzero entries of $H$ can be efficiently located and computed. More precisely, there must exist an efficient classical procedure (or oracle) that, given a row index $x$ and an integer $k\in\{1,\ldots,s\}$, returns the column index and value of the $k$-th nonzero entry in that row. Under these assumptions, it is possible to simulate the time-evolution operator $\e^{-\ii Ht}$ efficiently.

Sparse Hamiltonians arise naturally in many physical and algorithmic settings. For example, the Hamiltonian describing a particle moving on a graph corresponds to the adjacency matrix of the graph, which is sparse whenever the degree of each vertex is bounded. Similarly, many lattice models in condensed matter physics lead to Hamiltonians in which each basis state is connected to only a small number of other states.

Efficient algorithms for simulating sparse Hamiltonians have been developed using several techniques, including higher-order product formulas, quantum walks, and more recent approaches based on block-encoding and qubitization. These methods allow the implementation of the unitary operator $\e^{-\ii Ht}$ using a number of quantum gates that scales polynomially with the number of qubits, the sparsity parameter $s$, and the evolution time $t$.

\begin{exercise} ($k$-local Hamiltonians are sparse)
Let $H$ be an $n$-qubit Hamiltonian of the form
\[
H=\sum_{\ell=1}^{L} H_\ell,
\]
where each term $H_\ell$ acts nontrivially on at most $k$ qubits, with $k$ independent of $n$. Such Hamiltonians are called $k$-local.

\begin{enumerate}

\item[(a)] Show that for each term $H_\ell$ there exists a subset
$S_\ell\subseteq\{1,\dots,n\}$ with $|S_\ell|\le k$ such that
\[
H_\ell=\widetilde H_\ell\otimes I,
\]
up to a permutation of tensor factors, where $\widetilde H_\ell$ acts only on the qubits in $S_\ell$.

\item[(b)] Let $\ket{x}$ be a computational-basis state.
Explain why
\[
\langle y|H_\ell|x\rangle\neq 0
\]
can occur only if the bit strings $x$ and $y$ differ at positions belonging to $S_\ell$.

\item[(c)] Deduce that for a fixed basis state $\ket{x}$ there are at most $2^{|S_\ell|}$ basis states $\ket{y}$ such that
\[
\langle y|H_\ell|x\rangle\neq 0.
\]

\item[(d)] Conclude that each row of the matrix representing $H_\ell$ in the computational basis contains at most $2^k$ nonzero entries.

\item[(e)] Using the decomposition $H=\sum_{\ell=1}^{L}H_\ell$, show that each row of $H$ has at most
\[
s\le L\,2^k
\]
nonzero entries.

\item[(f)] Conclude that if $k$ is constant and $L$ grows at most polynomially with $n$, then $H$ is $s$-sparse with $s=\operatorname{poly}(n)$.

\end{enumerate}

\end{exercise}

\subsection*{Simulation of 1-sparse Hamiltonians}

In this subsection we discuss the simulation of Hamiltonians $H$ that are \emph{1-sparse}. Each row and each column of $H$ contains at most one nonzero entry. In the sparse-Hamiltonian simulation model, the matrix $H$ is not provided explicitly as input, since writing down all its entries would require exponential space in $n$. Instead, the algorithm assumes access to an \emph{oracle description} of the Hamiltonian. Given a row index $x\in\{0,1\}^n$, the oracle returns the column index $y(x)$ of the unique nonzero entry in that row (if it exists), together with the value of the matrix element $H_{x,y(x)}$. If the rule or logic determining the location of the nonzero entries is unknown, then it is not possible to use this algorithm efficiently. The inputs to the simulation algorithm are therefore the quantum state $\ket{\psi}$ on which the evolution will act, the evolution time $t$, and oracle access to the functions that specify the position and value of the nonzero entries of $H$.

To formalize the oracle description, we introduce two unitary operators that encode the structure of the Hamiltonian. We assume that the circuit uses three registers. The first register stores the index $x$ of a computational-basis state and consists of $n$ qubits. The second register also consists of $n$ qubits and is used to store the column index $y(x)$ of the nonzero entry in row $x$. The third register stores the value of the corresponding matrix element and contains enough qubits to represent the number $H_{x,y}$ with the desired precision.

The first oracle returns the position of the nonzero entry in a given row. For a $1$-sparse Hamiltonian this oracle implements the mapping
\begin{equation}\label{eq:O_y}
O_y\ket{x}\ket{0}\ket{0}=\ket{x}\ket{y(x)}\ket{0},
\end{equation}
where $y(x)$ is the column index of the unique nonzero entry in row $x$ (if the row has no nonzero entry, we take $y(x)=x$ by convention). The oracle $O_y$ acts nontrivially only on the first two registers.

The second oracle returns the value of that entry and acts as
\begin{equation}\label{eq:O_H}
O_H\ket{x}\ket{y}\ket{0}
=
\ket{x}\ket{y}\ket{H_{x,y}},
\end{equation}
where $H_{x,y}$ is the corresponding matrix element of the Hamiltonian. Since $H$ is Hermitian and $1$-sparse, the mapping $x\mapsto y(x)$ has the property that $y(y(x))=x$ whenever the nonzero entry exists. These two oracles provide all the information needed by the simulation algorithm without requiring explicit access to the full $2^n\times2^n$ matrix of $H$.

Let us describe how the Hamiltonian acts on the computational basis.
Since $H$ is $1$-sparse, each row $x$ contains at most one nonzero entry, which is located in column $y(x)$.
Therefore, the action of $H$ on a computational-basis state $\ket{x}$ is particularly simple and can be written as
\[
H\ket{x}=H_{x,y(x)}\ket{y(x)},
\]
where $H_{x,y(x)}$ is the corresponding matrix element.
If the row $x$ has no nonzero entries, then $H\ket{x}=0$.
Because $H$ is Hermitian, the relation $H_{y(x),x}=H^*_{x,y(x)}$ holds, which implies that the state $\ket{y(x)}$ is coupled back to $\ket{x}$.
Thus the Hamiltonian connects computational-basis states in pairs (or leaves them uncoupled when the row is zero), a structure that will be useful for constructing the simulation algorithm.

Consider a pair of basis states $\{\ket{x},\ket{y(x)}\}$ such that
$H_{x,y(x)}\neq 0$. The Hamiltonian couples these two states but does not
connect them to any other basis state because $H$ is $1$-sparse.
Therefore the two-dimensional subspace spanned by
$\{\ket{x},\ket{y(x)}\}$ is invariant under the action of $H$.
In the ordered basis $\{\ket{x},\ket{y(x)}\}$ the Hamiltonian takes the form
\[
H_2=
\begin{bmatrix}
0 & H_{x,y(x)} \\
H^*_{x,y(x)} & 0
\end{bmatrix}.
\]
Thus the global Hamiltonian can be viewed as a direct sum of independent
$2\times2$ blocks acting on such pairs of basis states (together with
$1\times1$ zero blocks corresponding to uncoupled states).
Consequently, the simulation of $\e^{-\ii Ht}$ reduces to implementing the
time evolution generated by each of these $2\times2$ Hamiltonians.
Their action can be applied in quantum superposition, allowing the algorithm to exploit quantum parallelism. A sequential application of the $2\times2$ Hamiltonians associated with each pair of basis states would in general be inefficient.

Let us now compute the time-evolution operator generated by the
$2\times2$ Hamiltonian $H_2$ (which depends on $x$). Write
\[
H_{x,y(x)} = \omega\,\e^{\ii \phi}, \qquad \omega\ge 0,
\]
so that
\[
H_2=
\omega\begin{bmatrix}
0 & \e^{\ii \phi} \\
\e^{-\ii \phi} & 0
\end{bmatrix}.
\]
A direct calculation shows that
\[
(H_2)^2 = \omega^2 I .
\]
Using the power-series expansion of the exponential, we obtain
\begin{align}
\e^{-\ii H_2 t}
&= \cos(\omega t)\, I
-
\ii\,\frac{\sin(\omega t)}{\omega}\, H_2 \\
&=
\begin{bmatrix}
\cos(\omega t) & -\ii\,\e^{\ii \phi}\sin(\omega t) \\
-\ii\,\e^{-\ii \phi}\sin(\omega t) & \cos(\omega t)
\end{bmatrix}.
\end{align}
Therefore the evolution generated by $H$ performs a rotation in the
two-dimensional subspace spanned by $\{\ket{x},\ket{y(x)}\}$, mixing the
two basis states with an angle determined by $\omega t$.

To translate this structure into a quantum circuit, the simulation algorithm uses the oracle description of the Hamiltonian. Starting from a basis state $\ket{x}$, the first step is to compute the column index $y(x)$ using the oracle $O_y$. This produces the state
\[
\ket{x}\ket{0}\ket{0}\xrightarrow{O_y}\ket{x}\ket{y(x)}\ket{0}.
\]
Next, the value oracle $O_H$ is applied to obtain the matrix element $H_{x,y(x)}$ in the third register,
\[
\ket{x}\ket{y(x)}\ket{0}\xrightarrow{O_H}\ket{x}\ket{y(x)}\ket{H_{x,y(x)}}.
\]
These registers now contain all the information required to implement the evolution generated by the $2\times2$ block associated with the pair $\{\ket{x},\ket{y(x)}\}$.

Using the information stored in the third register, which contains the value $H_{x,y(x)}=\omega\e^{\ii\phi}$, the circuit performs a controlled transformation that mixes the amplitudes of the two basis states. Schematically, the transformation acts as
\begin{equation}\label{eq:x-y-H-trans}
\ket{x}\ket{y(x)}\ket{H_{x,y(x)}}
\longrightarrow
\cos(\omega t)\,\ket{x}\ket{y(x)}\ket{H_{x,y(x)}}
-\ii\,\e^{-\ii\phi}\sin(\omega t)\,\ket{y(x)}\ket{x}\ket{H_{x,y(x)}}.
\end{equation}
Thus the amplitudes associated with $\ket{x}$ and $\ket{y(x)}$ are mixed exactly as required by the operator $\e^{-\ii H_2 t}$.

Finally, the auxiliary registers are returned to their initial state by applying the inverse oracles to the whole superposition. We first apply $O_H^\dagger$, which gives
\[
\cos(\omega t)\,\ket{x}\ket{y(x)}\ket{0}
-\ii\,\e^{-\ii\phi}\sin(\omega t)\,\ket{y(x)}\ket{x}\ket{0}.
\]
Next, we apply $O_y^\dagger$ to the first two registers, which returns
the final state
\[
\left(\cos(\omega t)\,\ket{x}
-\ii\,\e^{-\ii\phi}\sin(\omega t)\,\ket{y(x)}\right)\ket{0}\ket{0}.
\]
Thus the second and third registers are restored to $\ket{0}$, while the first register undergoes exactly the desired evolution in the two-dimensional subspace spanned by $\{\ket{x},\ket{y(x)}\}$. In all these expressions, the parameter $\omega$ and the phase $\phi$ depend on $x$.

\begin{exercise}
Assume that the Hamiltonian is real, so that
\[
H_{x,y(x)}=\omega \qquad (\omega\in\mathbb{R})
\]
and therefore $H_2=\omega X$. Recall that $\omega$ depends on $x$. The goal of this exercise is to construct a circuit implementing transformation~\eqref{eq:x-y-H-trans}.
\begin{enumerate}

\item[(a)]
Introduce an ancilla qubit (fourth register) and show that the two states $\ket{x}\ket{y(x)}\ket{\omega}\ket{0}_a$ and $\ket{y(x)}\ket{x}\ket{\omega}\ket{1}_a$ span a two-dimensional subspace. Construct a circuit that maps
\[
\ket{x}\ket{y(x)}\ket{\omega}\ket{0}_a
\longrightarrow
\ket{x}\ket{y(x)}\ket{\omega}\ket{0}_a
+
\ket{y(x)}\ket{x}\ket{\omega}\ket{1}_a .
\]
Hint: use a Hadamard gate on the ancilla followed by a controlled
SWAP acting on the first two registers.

\item[(b)]
Show that the unitary operator
\[
R_x(2\omega t)=
\exp(-i\,\omega t\,X)
\]
acts on the ancilla qubit as
\[
\ket{0}_a
\longrightarrow
\cos(\omega t)\ket{0}_a
-
i\sin(\omega t)\ket{1}_a .
\]

\item[(c)]
Apply the rotation $R_x(2\omega t)$ to the ancilla qubit and verify
that the global state becomes
\[
\cos(\omega t)\ket{x}\ket{y(x)}\ket{\omega}\ket{0}_a
-
i\sin(\omega t)\ket{y(x)}\ket{x}\ket{\omega}\ket{1}_a .
\]

\item[(d)]
Finally, undo the circuit of part (a) to return the ancilla qubit to
$\ket{0}_a$ and show that the transformation on the first three
registers is the one described in~\eqref{eq:x-y-H-trans} with $\phi=0$.

\end{enumerate}

\end{exercise}

\subsection*{Simulation of $s$-sparse Hamiltonians}

We now extend the ideas developed for $1$-sparse Hamiltonians to the more general case of $s$-sparse Hamiltonians with $s>1$. Recall that a Hamiltonian $H$ acting on $n$ qubits is called $s$-sparse if each row and each column of its matrix representation contains at most $s$ nonzero entries. The key idea of the simulation algorithm is to reduce this problem to the case already studied. More precisely, one can decompose $H$ as a sum of Hamiltonians
\[
H=\sum_{j=1}^{m} H^{(j)},
\]
where each $H^{(j)}$ is $1$-sparse and the number of terms $m$ is at most proportional to the sparsity parameter $s$. Each Hamiltonian $H^{(j)}$ therefore couples basis states in disjoint pairs and can be simulated using the method described in the previous subsection. The remaining task is to combine these individual evolutions in order to approximate the full time-evolution operator $\e^{-\ii Ht}$.

To obtain the decomposition of $H$ into $1$-sparse Hamiltonians, it is convenient to interpret the matrix of $H$ as defining a graph structure. Consider the simple graph $G$ whose vertices are the computational-basis states $\ket{x}$, $x\in\{0,1\}^n$, where two vertices $x$ and $y$ are adjacent whenever $H_{x,y}\neq 0$. Since $H$ is Hermitian, $H_{x,y}\neq0$ implies $H_{y,x}\neq0$, and therefore the graph $G$ is undirected. Because $H$ is $s$-sparse, each vertex is connected to at most $s$ other vertices. Thus, the maximum degree $\Delta(G)$ is upper bounded by $s$. The goal is to assign a color to each edge of this graph so that no two edges of the same color share a vertex. Such a coloring partitions the set of edges into disjoint groups, where within each group every vertex is incident to at most one edge. Each color class therefore defines a Hamiltonian $H^{(j)}$ in which every row and column has at most one nonzero entry. In other words, each $H^{(j)}$ is $1$-sparse.

A coloring with this property can be constructed using at most $s+1$ colors. This is guaranteed by Vizing's theorem, which states that every simple graph $G$ with maximum degree $\Delta(G)$ belongs to one of two classes: $\chi'(G) = \Delta(G)$ (Class~1) or $\chi'(G) = \Delta(G)+1$ (Class~2), where $\chi'(G)$ is the edge-chromatic number. Let $c(x,y)\in\{1,\ldots,s+1\}$ denote the color assigned to the edge connecting $x$ and $y$. We then define the Hamiltonians
\[
H^{(j)}_{x,y} =
\begin{cases}
H_{x,y}, & \text{if } c(x,y)=j,\\[4pt]
0, & \text{otherwise}.
\end{cases}
\]
By construction, the matrix $H^{(j)}$ contains at most one nonzero entry in each row and column, and therefore it is $1$-sparse. Moreover,
\[
H=\sum_{j=1}^{s+1} H^{(j)} .
\]

Each Hamiltonian $H^{(j)}$ can therefore be simulated using the procedure developed for $1$-sparse Hamiltonians. The algorithm queries the oracle that specifies the location and value of the nonzero entries of $H$ and determines, for a given pair of basis states $x$ and $y(x,k)$ corresponding to the $k$-th nonzero entry in row $x$, which color $j=c(x,y(x,k))$ has been assigned to that edge. This information identifies the unique $1$-sparse Hamiltonian $H^{(j)}$ responsible for coupling those two basis states. Consequently, the time evolution generated by each $H^{(j)}$ can be implemented using the same circuit structure described in the previous subsection.

Finally, the simulation of the full Hamiltonian $H$ is obtained by combining the evolutions generated by the $1$-sparse Hamiltonians $H^{(j)}$. Since the Hamiltonians $H^{(j)}$ generally do not commute with one another, the operator $\e^{-\ii Ht}$ cannot in general be written exactly as a product of the operators $\e^{-\ii H^{(j)}t}$. Instead, we approximate the evolution using the Trotter--Suzuki product formulas introduced earlier. In this approach, the operator $\e^{-\ii Ht}$ is approximated by a product of operators of the form $\e^{-\ii H^{(j)}t/r}$ applied in sequence. Because each $H^{(j)}$ is $1$-sparse, each of these factors can be implemented using the circuit developed in the previous subsection. Consequently, the overall simulation of $\e^{-\ii Ht}$ can be performed with a number of oracle queries and quantum gates that scales polynomially with the sparsity parameter $s$, the evolution time $t$, and the number of qubits $n$.

In the sparse-Hamiltonian simulation model used in the original analysis of the HHL algorithm, the complexity acquires an additional factor of $s$ because the algorithm must determine which of the at most $s$ nonzero entries in a given row corresponds to the current edge being simulated. As a result, locating the appropriate matrix element may require $O(s)$ oracle queries, leading to an overall complexity that scales as $O(s^2)$.

\section{State preparation}\label{sec:state-preparation}

State preparation is a fundamental task in quantum computing, enabling the encoding of classical information into quantum states. Given an arbitrary normalized quantum state of the form
\[
\ket{\psi} = \sum_{\ell=0}^{2^n-1} c_\ell\e^{\ii \alpha_\ell}\ket{\ell},
\]
where $c_\ell$ are non-negative real numbers ($c_\ell \geq 0$ for all $\ell$) and $\alpha_\ell$ are phases, our goal is to construct a quantum circuit that efficiently prepares this state. The procedure consists of two main steps:
\begin{enumerate}
    \item {Amplitude Encoding:} Applying controlled $R_y$ rotations to encode the amplitudes $c_\ell$.
    \item {Phase Encoding:} Applying controlled $R_z$ rotations to encode the phases $\alpha_\ell$.
\end{enumerate}
The complete circuit can be schematically represented as follows:
\[
\Qcircuit @C=2.3em @R=1.3em {
\lstick{q_1:\,\ket{0}}&\multigate{2}{\,\,\,U_{\theta}\,\,\,}& \multigate{2}{\,\,\,U_{\beta}\,\,\,} &\qw  \\
\lstick{\vdots \,\,\,}&                               & & &{\ket{\psi}.}  \\
\lstick{q_n:\,\ket{0}}&\ghost{\,\,\,U_{\theta}\,\,\,}       & \ghost{\,\,\,U_{\beta}\,\,\,}      & \qw  \gategroup{1}{1}{3}{4}{.7em}{\}}
}\vspace{6pt}
\]
The algebraic formulation of this process is given by
\[
\ket{\psi}=U_{\beta}U_{\theta}\ket{0}^{\otimes n}.
\]
Below, we provide a detailed construction of the circuits for $U_{\theta}$ and $U_{\beta}$, which generate the desired state up to a global phase.

\subsection*{Amplitude encoding}

The first unitary $U_{\theta}$ consists of a sequence of controlled $R_y$ gates that prepare the intermediate state
\[
\ket{\psi_0} = U_{\theta}\ket{0}^{\otimes n}= \sum_{\ell=0}^{2^n-1} c_\ell\ket{\ell}.
\]
For $n=3$, the explicit structure of this part of the circuit is
\[
\Qcircuit @C=0.7em @R=1.0em {
\lstick{q_1:\,\ket{0}} & \gate{R_y(\theta^1_{1})}&\ctrlo{1}                          &\ctrl{1}                             &\ctrlo{1}                         &\ctrlo{1}                          &\ctrl{1}                          &\ctrl{1}                          & \rstick{} \qw \\
\lstick{q_2:\,\ket{0}} & \qw                                  &\gate{R_y(\theta^2_1)}&\gate{R_y(\theta^2_2)}&\ctrlo{1}                          &\ctrl{1}                             &\ctrlo{1}                          &\ctrl{1}                         &  \qw& \rstick{\ket{\psi_0}.} \\
\lstick{q_3:\,\ket{0}} & \qw                                  &\qw                                   &\qw                                   &\gate{R_y(\theta^3_1)}&\gate{R_y(\theta^3_2)}&\gate{R_y(\theta^3_3)}&\gate{R_y(\theta^3_4)}& \rstick{} \qw
\gategroup{1}{1}{3}{9}{.7em}{\}}}\vspace{0.2cm}
\]
The upper index $j$ in the angle $\theta^j_k$ indicates the qubit $q_j$ on which the $R_y$ gate acts, while the lower index $k$ ranges from 1 to $2^{j-1}$. From this example, we can easily generalize the circuit for larger $n$. For instance, when $n=4$, we must add, at the end of the circuit (in any order), eight controlled $C^3(R_y)$ gates with qubit $q_4$ as the target, ensuring that all possible combinations of full and empty controls are covered.

The angles $\theta^j_k$, for $1\le j\le n$ and $1\le k \le 2^{j-1}$, are computed as
\begin{equation}\label{eq:sin2-SP}
\sin^2 \left(\frac{\theta^j_k}{2}\right)=
\frac{{\sum_{\ell=0}^{2^{(n - j)}-1} c_{\ell+(2k - 1)2^{(n - j)} }^2}}
{{\sum_{\ell=0}^{2^{(n - j+1)}-1} c_{\ell+(k - 1)2^{(n - j+1)} }^2}}, \quad \theta^j_k\in [-\pi,0].
\end{equation}
Note that the right-hand side is a non-negative number between 0 and 1, and if $\theta^j_k$ is a solution to \eqref{eq:sin2-SP}, then $-\theta^j_k$ is also a solution. For the circuit described above, we must take $\theta^j_k$ in the range $[-\pi,0]$. If instead we take $\theta^j_k$ in the range $[0,\pi]$, we must reverse the circuit and the order of the multi-controlled $R_y$ gates within the blocks that have the same number of controls.

Given the state $\ket{\psi}$, we can compute $c_\ell=\big|\braket{\ell}{\psi}\big|$ for $\ell$ from 0 to $2^n-1$ and then determine all angles $\theta^j_k$ using formula~\eqref{eq:sin2-SP}. By substituting these angles into the circuit, we can prepare the quantum computer in the state $\ket{\psi_0}$.

To prepare the state $\ket{b}$ used in the HHL algorithm, only the first part of the circuit is needed, and Eq.~\eqref{eq:sin2-SP} is sufficient. However, we will complete the description of the state preparation algorithm as it has broader applications.

\begin{exercise}
Show that:
\begin{enumerate}
\item[(a)] $\theta_k^j=\frac{\pi}{2}$ for all $j$ and $k$, if
\[
\ket{\psi}=\frac{1}{\sqrt{2^n}}\sum_{\ell=0}^{2^n-1} \ket{\ell}.
\]

\item[(b)] $U_\theta=R_y\left(\frac{\pi}{2}\right)^{\otimes n}$.

\item[(c)] $U_\theta\ket{0}^{\otimes n}=H^{\otimes n}\ket{0}^{\otimes n}$.
\end{enumerate}
\end{exercise}

\subsection*{Phase encoding}

The second part of the circuit, $U_{\beta}$, consists of controlled $R_z$ gates that encode the phases $\alpha_\ell$, producing the final state up to a global phase. For $n=3$, it takes the form
\[
\Qcircuit @C=0.7em @R=1.0em {
\lstick{}& \gate{R_z(\beta^1_{1})}&\ctrlo{1}                          &\ctrl{1}                             &\ctrlo{1}                         &\ctrlo{1}                          &\ctrl{1}                          &\ctrl{1}                          & \rstick{} \qw \\
\lstick{\ket{\psi_0}\,\,\,}& \qw                                  &\gate{R_z(\beta^2_1)}&\gate{R_z(\beta^2_2)}&\ctrlo{1}                          &\ctrl{1}                             &\ctrlo{1}                          &\ctrl{1}                         &  \qw& \rstick{\e^{\ii \phi}\ket{\psi}.} \\
\lstick{}& \qw                                  &\qw                                   &\qw                                   &\gate{R_z(\beta^3_1)}&\gate{R_z(\beta^3_2)}&\gate{R_z(\beta^3_3)}&\gate{R_z(\beta^3_4)}& \rstick{} \qw
\gategroup{1}{1}{3}{1}{.7em}{\{}
\gategroup{1}{1}{3}{9}{.7em}{\}}
}\vspace{0.2cm}
\]
This part has the same structure as the previous one. As before, the upper index $j$ in the angle $\beta^j_k$ denotes the qubit $q_j$ on which the $R_z$ gate acts.

The angles $\beta^j_k$, for $1\le j\le n$ and $1\le k \le 2^{j-1}$, are given by
\begin{equation}\label{eq:beta-j-k}
\beta^j_k = \sum_{\ell  = 0}^{2^{(n - j)}-1}  \frac{\alpha_{\ell+(2k - 1)2^{(n - j)} } - \alpha_{\ell+(2k - 2)2^{(n - j )} }}{2^{n - j}}.
\end{equation}
The global phase factor introduced by this transformation is
\[
\phi = \frac{1}{2^n} \sum_{\ell=0}^{2^n-1} \alpha_\ell.
\]
In most applications, this global phase factor is irrelevant.

Given the state $\ket{\psi}$, the phases $\alpha_\ell$ are given by
\[
\alpha_\ell = \arg \braket{\ell}{\psi},
\]
if $\braket{\ell}{\psi} \neq 0$; otherwise, $\alpha_\ell = 0$. The argument function, $\arg(z)$, extracts the phase of a complex number. Specifically, for a complex number $z = |z| \e^{\ii \theta}$, $\arg(z) = \theta$.

\subsection*{Complexity}

Although this method provides an exact way to prepare arbitrary states, it requires $O(2^n)$ gates. Ref.~\cite{MVBS05} was one of the first to introduce this method. The theoretical lower bound is $O(2^n/n)$~\cite{PB11}. In practical implementations, approximate methods or variational approaches are often preferred~\cite{Cer21}.

If the state $\ket{b}$ has multiple repeated entries, the state preparation circuit is going to be shorter, and in some cases it is $O(\text{poly}(n))$. Apendix A of~\cite{PM22} has discussed many of those cases using an anzats of the state preparation circuit with $n$ multi-controled $R_y$ gates.

\subsection*{Application to quantum machine learning}

State preparation techniques are particularly useful in quantum machine learning, where classical data is encoded into quantum states. A common task is encoding $N$ real numbers $\alpha_0, \dots, \alpha_{N-1}$ into phase information, producing the state
\[
\ket{\psi_1} = \frac{1}{\sqrt{2^n}}\sum_{\ell=0}^{2^n-1} \e^{\ii \alpha_\ell} \ket{\ell}.
\]
This can be achieved by replacing $U_{\theta}$ with Hadamard transformations
\[
\ket{\psi_1}= U_{\beta}H^{\otimes n}\ket{0}^{\otimes n}.
\]
This approach is frequently used in quantum kernel methods and variational quantum circuits.

\chapter{Final Remarks}\label{chap:conc}

Most quantum algorithms analyzed in this work can be cast into the oracle-based framework. The query complexity of an algorithm based on an oracle or black box is the number of queries. It does not matter how difficult it is to implement the oracle unless we aim to solve a practical problem. In practical problems, it is our task to implement the oracle, and then the cost of each evaluation matters. Take Shor's factoring algorithm as an example. The oracle in this case is an $r$-periodic function, and our goal is to find $r$. We have seen that the function in Shor's algorithm is modular exponentiation, which can be implemented efficiently in terms of the input size using the repeated squaring method.

Any classical deterministic algorithm can be represented as an $n$-input and $m$-output function $f:{0,1}^n\longrightarrow {0,1}^m$, which can be viewed as a collection of $m$ $n$-bit Boolean functions. Therefore, any classical algorithm can be implemented on a quantum computer with two registers of sizes $n$ and $m$ using the operator
\[
U_f\ket{x}\ket{y}\,=\,\ket{x}\ket{y\oplus f(x)}.
\]
To exploit quantum parallelism, we need to apply $H^{\otimes n}$ to the first register before applying $U_f$. After applying $U_f$, we obtain a superposition state, which becomes useful only after we perform some quantum post-processing that produces the desired output. Most of the quantum algorithms we have analyzed can be cast into the following circuit:
\[
\Qcircuit @C=1.7em @R=0.9em {
\lstick{\ket{0}}            & \gate{H} & \multigate{5}{\,\,\,U_f\,\,\,} & \multigate{3}{\begin{array}{c}\text{post}\\ \text{processing}\end{array}} & \meter  & \rstick{i_0} \cw \\
\lstick{\vdots \,\,\,}      & {\vdots} &                                &{}&{\vdots} & \rstick{\vdots}  \\
\lstick{}                   &          &                                &         && \rstick{}  \\
\lstick{\ket{0}}            & \gate{H} & \ghost{\,\,\,U_f\,\,\,}        & \ghost{\begin{array}{c}\text{post}\\ \text{processing}\end{array}}& \meter  & \rstick{i_{n-1}}  \cw\\ & & & & {\hspace{1.7cm}} \\
\lstick{\ket{0}^{\otimes {m}}}& {/^{m}}\qw & \ghost{\,\,\,U_f\,\,\,}        & \meter  &  \rstick{j_0...j_{m-1}.}\cw
}\vspace{0.0cm}
\]
For the Deutsch-Jozsa, Bernstein-Vazirani, Simon, and Shor factoring algorithms, the quantum post-processing consists of either applying Hadamard gates to all qubits or using the inverse Fourier transform. They have the structure outlined above with a few adaptations. Some of these algorithms also require classical post-processing, which is not represented in the quantum circuit.

Grover's algorithm does not have the structure outlined above because the oracle and the post-processing are repeated many times before measurement. On the other hand, Grover's algorithm provides a polynomial speedup, in contrast to the exponential speedup of Simon's and Shor's algorithms. The extension of the general structure that includes Grover's algorithm is \vspace{-0.2cm}
\[
\Qcircuit @C=1.7em @R=0.9em {
& & \hspace{90pt} \mbox{repeat $k$ times} & & & \\
\lstick{\ket{0}}            & \gate{H} & \multigate{5}{\,\,\,U_f\,\,\,} & \multigate{3}{\begin{array}{c}\text{post}\\ \text{processing}\end{array}} & \meter  & \rstick{i_0} \cw \\
\lstick{\vdots \,\,\,}      & {\vdots} &                                &{}&{\vdots} & \rstick{\vdots}  \\
\lstick{}                   &          &                                &         && \rstick{}  \\
\lstick{\ket{0}}            & \gate{H} & \ghost{\,\,\,U_f\,\,\,}        & \ghost{\begin{array}{c}\text{post}\\ \text{processing}\end{array}}& \meter  & \rstick{i_{n-1}}  \cw\\ & & & & {\hspace{1.7cm}} \\
\lstick{\ket{0}^{\otimes {m}}}& {/^{m}}\qw & \ghost{\,\,\,U_f\,\,\,}    &\qw    & \meter  &  \rstick{j_0...j_{m-1}.}\cw
\gategroup{2}{3}{7}{4}{.7em}{--}
}\vspace{0.0cm}
\]
The number of repetitions $k$ is 1 for the Deutsch-Jozsa, Bernstein-Vazirani, Simon, and Shor algorithms, and $k$ is $\lfloor\pi\sqrt{2^n}/4\rfloor$ for Grover's algorithm.
The measurement of the second register is unnecessary. It is included because it helps in the analysis of the algorithm.

The second register of the Deutsch-Jozsa, Bernstein-Vazirani, and Grover algorithms has only one qubit ($m=1$), whose state during the computation is $\ket{-}$, which is obtained by applying $X$ and $H$ on the last qubit before $U_f$. The oracle for those cases obeys
\[
U_f\ket{x}\ket{-}\,=\,(-1)^{f(x)}\ket{x}\ket{-}.
\]
This means that the second register can be eliminated, yielding a simpler version of the circuit with the following form: \vspace{-3pt}
\[
\Qcircuit @C=1.7em @R=0.9em {
& & \hspace{90pt} \mbox{repeat $k$ times} & & & \\
\lstick{\ket{0}}            & \gate{H} & \multigate{3}{\,\,\,u_f\,\,\,} & \multigate{3}{\begin{array}{c}\text{post}\\ \text{processing}\end{array}} & \meter  & \rstick{i_0} \cw \\
\lstick{\vdots \,\,\,}      & {\vdots} &                                &{}&{\vdots} & \rstick{\vdots}  \\
\lstick{}                   &          &                                &         && \rstick{}  \\
\lstick{\ket{0}}            & \gate{H} & \ghost{\,\,\,u_f\,\,\,}        & \ghost{\begin{array}{c}\text{post}\\ \text{processing}\end{array}}& \meter  & \rstick{i_{n-1}.}  \cw\\ & & & & {\hspace{1.7cm}}
\gategroup{2}{3}{5}{4}{.7em}{--}
}\vspace{0.0cm}
\]
As before, $k=1$ for the Deutsch-Jozsa and Bernstein-Vazirani algorithms, $k=\lfloor\pi\sqrt{2^n}/4\rfloor$ for Grover's algorithm, and
\[
u_f\ket{x}=(-1)^{f(x)} \ket{x}.
\]

Shor's algorithm for discrete logarithms shows how to extend the structure of the circuit when the function $f$ has more than one variable. Suppose that $f$ has two variables. Then, $U_f$ is defined as
\[
U_f\ket{x_1}\ket{x_2}\ket{y}\,=\,\ket{x_1}\ket{x_2}\ket{y\oplus f(x_1,x_2)}.
\]
This means that we need a circuit with three registers, and the general structure of the algorithm remains the same, up to small changes, as follows:
\[
\Qcircuit @C=1.7em @R=1.6em {
& & \hspace{185pt} \mbox{repeat $k$ times} & & & \\
\lstick{\ket{0}^{\otimes n_1}}& {/^{n_1}}\qw&\gate{H^{\otimes n_1}} & \multigate{2}{\,\,\,U_f\,\,\,} & \multigate{1}{\begin{array}{c}\text{post}\\ \text{processing}\end{array}} & \meter  & \rstick{i_0...i_{n_1-1}} \cw \\
\lstick{\ket{0}^{\otimes n_2}}& {/^{n_2}}\qw&\gate{H^{\otimes n_2}} & \ghost{\,\,\,U_f\,\,\,}        & \ghost{\begin{array}{c}\text{post}\\ \text{processing}\end{array}}& \meter  & \rstick{i'_0...i'_{n_2-1}}  \cw\\
\lstick{\ket{0}^{\otimes {n_3}}}& {/^{n_3}}  \qw&\qw & \ghost{\,\,\,U_f\,\,\,}        & \qw    & \meter  &  \rstick{j_0...j_{n_3-1}.}\cw
\gategroup{2}{4}{4}{5}{1.4em}{--}
}\vspace{0.0cm}
\]
These circuit patterns provide a useful template for understanding more advanced quantum algorithms.

Tables~\ref{table:sizes},~\ref{table:entangled}, and~\ref{table:effi} summarize key features of these basic quantum algorithms.

\begin{table}[h!]
\centering
\begin{tabular}{|c|c|c|c|c|}
\hline
\textit{Algorithm} & \textit{1st reg.} & \textit{2nd reg.} & $k$ & \textit{Post-processing} \\
\hline
Deutsch-Jozsa & $n$ & 1 & $1$ & $H^{\otimes n}$ \\
\hline
Bernstein-Vazirani & $n$ & 1 & $1$ & $H^{\otimes n}$ \\
\hline
Simon & $n$ & $n$ & $1$ & $H^{\otimes n}$ \\
\hline
Shor (factoring)  & $2n$ or $2n-1$ & $n$ & $1$ & $F_{q}^\dagger$ \\
\hline
Grover & $n$ & 1 & $\sqrt{2^n}$ & $2\ket{\text{d}}\bra{\text{d}}-I$ \\
\hline
\end{tabular}
\caption{Overview of the basic algorithms, showing the number of qubits in the first and second registers, and the number of repetitions of the dashed box. The last column describes the post-processing, where $\ket{\text{d}}=\ket{+}^{\otimes n}$ and $q$ is either $2^{2n}$ or $2^{2n-1}$.  }\label{table:sizes}
\end{table}

\begin{table}[h!]
\centering
\begin{tabular}{|c|c|c|}
\hline
\textit{Algorithm}  &  \textit{Oracle} & \textit{Entangled} \\
\hline
Deutsch-Jozsa  & $f$ is balanced or constant  &  depends on $f$\\
\hline
Bernstein-Vazirani & $f$ is linear:  $f(x)=s\cdot x$ & no \\
\hline
Simon  & $f(x)=f(y) \iff x\oplus y\in\{0,s\}$  & depends on $|s|$ \\
\hline
Shor (factoring)  &  $f$ is periodic & yes\\
\hline
Grover   &  $f(x)=1$ iff $x=x_0$ & yes \\
\hline
\end{tabular}
\caption{Summary of the basic algorithms, including a brief description of the oracle and whether entanglement is present. }\label{table:entangled}
\end{table}

\begin{table}[h!]
\centering
\begin{tabular}{|c|c|c|}
\hline
\textit{Algorithm} & \textit{Quantum version} & \textit{Classical version} \\
\hline
Deutsch-Jozsa & $O(1)$ & $O(1)$  \\
\hline
Bernstein-Vazirani & $O(1)$ & $O(n)$  \\
\hline
Simon & $O(n^2)$ & $O(\sqrt{2^n})$  \\
\hline
Shor (factoring)  & $O(n^2\log n)$ & $\text{e}^{(1+o(1))\sqrt{n}\sqrt{\log n}}$  \\
\hline
Grover & $O(\sqrt{2^n})$ & $O({2^n})$  \\
\hline
\end{tabular}
\caption{Comparison of quantum and classical time complexities for the basic quantum algorithms, assuming randomized algorithms for the classical cases. For Shor's algorithm, we assume a fast multiplication method~\cite{HH21} for the quantum algorithm. For the other algorithms, we assume the oracle's implementation is $O(1)$. }\label{table:effi}
\end{table}


\newpage
\newpage
\setlength{\bibsep}{2.0pt}

\end{document}